\documentclass[12pt]{amsbook}
\usepackage[all,cmtip]{xy}

\usepackage{a4wide}
\usepackage{pxfonts}
\usepackage{wasysym}
\usepackage{amssymb,amsxtra,mathrsfs}
\usepackage{eucal}
\usepackage{setspace,upgreek}
\usepackage{mathtools}
\usepackage{url}
\usepackage[hidelinks]{hyperref}

\newcommand{\mf}[1]{\mathfrak{#1}}
\newcommand{\mc}[1]{\mathcal{#1}}
\newcommand{\mh}[1]{\mathscr{#1}}

\newcommand{\mr}[1]{\mathrm{#1}}

\newcommand{\ps}[1]{\pmb{\mathsf{#1}}}

\newcommand{\ov}{\overline}

\newcommand{\lr}[2]{\langle\, #1,\,#2\,\rangle}
\newcommand{\up}{\upharpoonright}
\newcommand{\ep}{\upvarepsilon} 
\newcommand{\ze}{\mathsf{0}} 
\newcommand{\un}{\mathsf{1}} 
\newcommand{\N}{\mathbb{N}}
\newcommand{\M}{\mathbb{M}}
\newcommand{\K}{\mathbb{K}}
\newcommand{\R}{\mathbb{R}}
\newcommand{\C}{\mathbb{C}}
\newcommand{\Z}{\mathbb{Z}}
\newtheorem{theorem}{Theorem}
\newtheorem{definition}[theorem]{Definition}

\newtheorem{postulate}[theorem]{Postulate}

\newtheorem{posit}[theorem]{Posit}

\newtheorem{proposition}[theorem]{Proposition}
\newtheorem{corollary}[theorem]{Corollary}
\newtheorem{lemma}[theorem]{Lemma}

\newtheorem{notation}[theorem]{Notation}
\theoremstyle{definition}
\newtheorem{convention}[theorem]{Convention}

\newtheorem{remark}[theorem]{Remark}
\numberwithin{theorem}{section}
\begin{document}
\title[Dynamical Patterns]
{The $2-$category of species of dynamical patterns}

\author{Benedetto Silvestri}

\email{bndsilvestri@proton.me}

\address{440 W Kelly Ave Apt 6 Jackson, WY 83001-8544 USA (via express couriers signature required)}

\urladdr{\url{https://arxiv.org/a/silvestri_b_1}}

\urladdr{https://orcid.org/0000-0002-8042-4932}



\maketitle
\tableofcontents
\flushbottom

\begin{abstract}
A new category $\mathfrak{dp}$, called of dynamical patterns
addressing a primitive, nongeometrical concept of dynamics, is defined 
and employed to construct a $2-$category $2-\mathfrak{dp}$,
where the irreducible plurality of species of context-depending dynamical patterns
is organized.
We propose a framework characterized by the following 
additional features.
A collection of experimental settings is associated with any species,
such that each one of them induces 
a collection of experimentally detectable trajectories.
For any connector $T$, a morphism between species, 
any experimental setting $E$ of its target species 
there exists a set such that with each of its elements $s$
remains associated 
an experimental setting $T[E,s]$ of its source species,
$T[\cdot,s]$ is called charge associated with $T$ and $s$.
The vertical composition of connectors 
is contravariantly represented in terms of charge composition.
The horizontal composition of connectors and $2-$cells of $2-\mathfrak{dp}$
is represented in terms of charge transfer.
A collection of trajectories induced by $T[E,s]$ 
corresponds to a collection of trajectories induced by $E$
(equiformity principle).
Context categories, species and connectors are organized respectively as $0,1$ and $2$ cells 
of $2-\mathfrak{dp}$ with factorizable functors via $\mathfrak{dp}$ as $1-$cells 
and as $2-$cells, arranged themself to form objects of categories,
natural transformations between $1-$cells obtained as horizontal composition
of natural transformations between the corresponding factors.
We operate a nonreductionistic interpretation 
positing that the physical reality holds the structure of 
$2-\mathfrak{dp}$, where the fibered category $\mathfrak{Cnt}$ of connectors 
is the only empirically knowable part.
In particular each connector exists as an irreducible entity of the physical reality,
and empirically detectable through the charges associated with it 
and experimentally represented my means of its equiformity principle.
The algebraic structure of $\mathfrak{Cnt}$ is experimentally 
detectable in terms of charge composition and charge transfer.
$\mathfrak{dp}$ widely generalizes the category of $C^{\ast}-$dynamical systems.
The dynamical group is replaced by a 
$\mathsf{top}-$enriched category called dynamical category,
the group action by the dynamical functor
namely a functor of $\mathsf{top}-$enriched categories 
from the dynamical category to
the category of unital topological $\ast-$algebras naturally enriched over $\mathsf{top}$,
finally an equivariant map between $C^{\ast}-$dynamical systems
is replaced by a couple formed by a functor $f$ between dynamical categories
and a natural transformation from the composition of 
the dynamical functor of the source with $f$
to the dynamical functor of the target.
As an emblematic model we show that the equivariance
under diffeomorphic actions
of the flow of complete perfect fluids on general spacetimes
is assembled into a species $\mathsf{a}$ 
on the category $\mathsf{St}$ of spacetimes and complete vector fields,
with smooth maps relating the vector fields as morphisms.
As a result the equivalence principle in general relativity
emerges as the equiformity principle of 
the identity connector of $\mathsf{a}$. 
Said a quantum gravity a suitable species $\mathsf{b}$ on $\mathsf{St}$ 
such that the underlying topological $\ast-$algebras are noncommutative, 
then the existence of a connector from $\mathsf{a}$ to $\mathsf{b}$ 
enables a quantum realization 
of the velocity of maximal integral curves of 
complete vector fields over spacetimes.
When applied to Robertson-Walker spacetimes we establish that the Hubble parameter, 
the acceleration of the scale function and new constraints for its positivity 
evaluated on a subset of the range of the galactic time of a geodesic $\alpha$, 
are expressed in terms of a quantum realization of the velocity of $\alpha$.
As a result the existence of a connector satisfying these constraints 
implies a positive acceleration and represents an alternative to the dark energy hypothesis.
\end{abstract}
\section{Introduction}
\label{intr}
In order to establish when physical theories 
may be considered equivalent in all spacetimes,
Fewster and Verch \cite{fv} define 
locally covariant theories and their embeddings
in terms of the category of functors from the category
of globally hyperbolic spacetimes $\mr{Loc}$ to the abstract 
category of physical systems $\mr{Phys}$,
for which the category $\mr{CA}^{\ast}$ 
of $C^{\ast}-$algebras and $\ast-$morphisms,
represents a model.
A similar concept in the special case of $\mr{CA}^{\ast}$ 
was previously discussed in
Brunetti, Fredenhagen and Verch \cite{bfv},
where in order to address in a general covariant setting the concept of quantum field,
they defined any locally covariant theory as a functor from essentially 
$\mr{Loc}$ to $\mr{CA}^{\ast}$.
They regarded a quantum field as a natural transformation
between functors obtained by composing 
locally covariant theories with
the forgetful functor from $\mr{CA}^{\ast}$ 
to the category of topological spaces.
In \cite{fv} and \cite{bfv}
the target categories $\mr{Phys}$ and $\mr{CA}^{\ast}$ respectively
are interpreted essentially as the collection of kinematical systems with embeddings as morphisms,
allowing the dynamical transformations to be realized in terms of morphisms.
In what follows the categories modeling $\mr{Phys}$ are named kinematical categories. 
\par
We instead retain that dynamics are actualized by objects of a category,
they in general are not byproducts of geometrical transformations, 
although they could be covariant under geometrical actions.
Thus our initial posit which will be later extended, reads as follows.
\par
\emph
{Dynamics is a primitive collection of entities organized to form the category $\mf{dp}$}.
\par 
The meaning of primitive will be later formalized.
$\mf{dp}$ is constructed in Def. \ref{09081303} and Cor. \ref{09100952} 
and called category of dynamical patterns.
It is a nontrivial generalization of the category $\mf{ds}$ of 
topological dynamical systems with equivariant maps as morphisms,
where a topological dynamical system is determined by a
morphism in the category of topological groups 
whose target is the group of continuous $\ast-$automorphisms
of a unital topological $\ast-$algebra endoved with the 
topology of simple convergence.
Let us call dynamical group the source object of a dynamical system. 
\par
In constructing $\mf{dp}$ the dynamical group is replaced by a 
\emph{dynamical category},
that is a category enriched\footnote{to be precise 
is a weaker version of the standard enrichment
we call quasi-enrichment.}
over $\mr{top}$,
the category of topological spaces, 
and the group morphism replaced by the 
\emph{dynamical functor},
namely a functor of $\mr{top}-$enriched categories 
from the dynamical category to the 
naturally $\mr{top}-$enriched category $\mr{tsa}$ of 
unital topological $\ast-$algebras and continuous $\ast-$morphisms.
Morphisms of $\mf{dp}$ are couples formed by a 
\emph
{$\mr{top}-$functor $f$ from the dynamical category of the target 
to the dynamical category of the source
and a 
natural transformation from the composition of 
the dynamical functor of the source composed with $f$,
to the dynamical functor of the target}, 
generalizing the concept of equivariant morphisms in 
$\mf{ds}$.
\par
$\mr{CA}^{\ast}$ is trivially embedded in $\mf{cds}$,
the subcategory of $\mf{ds}$ formed by $C^{\ast}-$dynamical systems,
by the map $\mc{A}\mapsto\lr{\mc{A},\{Id_{\mc{A}}\}}{Id_{\mc{A}}\mapsto Id_{\mc{A}}}$ of objects
and the map $f\mapsto(Id_{d(f)}\mapsto Id_{c(f)},f)$ of morphisms;
while $\mf{ds}$ is embedded in $\mf{dp}$.
The reasons to prefer in the definition of $\mf{dp}$
the category of topological $\ast-$algebras rather than $\mr{CA}^{\ast}$
or an abstract $\mr{top}-$enriched category reside in what follows.
In the first case it is in order to embed not only $\mf{cds}$
but also $W^{\ast}-$dynamical systems which are continuous w.r.t. the sigma-weak topology.
In the second case it is in order to model physical theories, 
which require a framework to produce experimentally detectable values.
Finally as we shall see below, the requisite to address in a compact and elegant way
the geometric equivariance of the flow generated in particular by perfect fluids on spacetimes,
forces us to move from dynamical groups to dynamical categories.
The following properties characterize $\mf{dp}$.
\begin{enumerate}
\item
A species contextualized in a category, is a functor 
from this category to $\mf{dp}$\footnote{Exactly
species are functors valued in the category $\mf{Chdv}$, 
defined in Def. \ref{12312135} and Cor. \ref{31122146},
however there exists a canonical
functor from $\mf{dp}$ to $\mf{Chdv}$ permitting to associate 
with any functor valued in $\mf{dp}$ a species.
It is worthwhile remarking that the
natural transformation between functors 
with values in $\mr{tsa}$ 
present in the definition of the morphisms of $\mf{dp}$,
is replaced in the $\mf{Chdv}$ case with a
natural transformation
between functors 
with values in the category 
$\mr{ptsa}$ whose object set is as $\mr{tsa}$ but whose morphism set is the subset 
of linear positive continuous maps $T$ between unital topological $\ast-$algebras such that 
$T(\un)\leq\un$.
By dropping the request on $T$ of being an algebra morphism
will allow later to consider connectors between classical and quantum species,
but at this stage of the discussion this point is irrelevant.},
and it encodes a collection of context-depending \emph{dynamical patterns}
equivariant under action of the morphisms of the context category.
\item
A collection of experimental settings $Exp(\mr{a})$ is associated with any species $\mr{a}$.
Roughly an experimental setting is a couple $(\mf{S},\mr{R})$,
where $\mr{R}$ maps any context $M$ into a subcategory of the 
dynamical category of the dynamical pattern $\mr{a}(M)$,
while $\mf{S}$ maps $M$ into a $\mr{R}_{M}-$fibered family of
continuous positive linear functionals over 
the topological $\ast-$algebras underlying $\mr{a}(M)$.
Each functional stands for a statistical ensemble 
whose strength\footnote{strength with the meaning used in \cite{dl}.} 
is the value the functional assumes at the identity.
$\mr{R}$ and $\mf{S}$ are equivariant with respect to 
the geometrical action of the context category 
and most importantly $\mf{S}$ is
equivariant under the action of the 
\emph{dynamical subcategories represented by $\mr{R}$ via the conjugate of the dynamical functors}
Def. \ref{01161844}.
The fact that we consider dynamical subcategories reflects 
the eventual occurrance of broken dynamical symmetries.
\item
To any species $\mr{a}$,
any context $M$
and any couple of objects $x,y$ of the dynamical category of 
$\mr{a}(M)$ 
is assigned a family of experimentally detectable trajectories
whose initial conditions are represented by couples of 
statistical ensembles and observables and such that the dynamics is realized
by morphisms from $x$ to $y$ via the \emph{dynamical functor} of $\mr{a}(M)$
Def. \ref{01162119}.
By restriction of the initial conditions a family of trajectories
can be assigned to any experimental setting of $\mr{a}$.
\item
A connector is a natural transformation 
between species contextualized over the same category
and it is decoded by the two diagrams exposed in 
Lemma \ref{01011521},
in which the category $\mf{dp}$ trasparently determines
the dynamical nature of the connector.
We have the following properties
\begin{enumerate}
\item
\textbf{Charges}.
Any connector $\mr{T}$ induces a set valued map $\Upgamma(\cdot,\mr{T})$
over the set $Exp(c(\mr{T}))$ of experimental settings of 
the target species of $\mr{T}$ and a function $\mr{T}[\cdot,\cdot]$, 
mapping any couple $(\mf{Q},s)$ where $\mf{Q}\in Exp(c(\mr{T}))$ and 
$s\in\Upgamma(\mf{Q},\mr{T})$
into an experimental setting $\mr{T}[\mf{Q},s]\in Exp(d(\mr{T}))$
of the source species of $\mr{T}$
Thm. \ref{10081910}\eqref{10081910st1}.
The map $\mr{T}[\cdot,s]$ is called charge associated with $\mr{T}$ and $s$.
If $\mr{T}$ connects species of dynamical systems 
the degeneration is removed and $\mr{T}[\mf{Q}]$ stands for $\mr{T}[\mf{Q},s]$.
\item
\textbf{The vertical composition of connectors is 
contravariantly represented as charge composition}.
\begin{enumerate}
\item
\emph{general connectors}:
under suitable hypothesis 
there exists a contravariant representation
of the vertical composition of connectors 
in terms of composition of charges Cor. \ref{10151636};
\item
\emph{connectors of species of dynamical systems}:
if $\mr{a}$ is a species of dynamical systems the result is stronger, 
fixed a context category $\mf{D}$, the assignements 
$\mr{a}\mapsto Exp(\mr{a})$ and $\mr{T}\mapsto\mr{T}[\cdot]$
determine a contravariant functor at values in $\mr{set}$
and defined on the functor
category of species contextualized on $\mf{D}$ with connectors as morphisms.
\end{enumerate}
\item
\textbf{The horizontal composition of connectors with $2-$cells is represented as charge transfer}.
For any connector $\mr{T}$ and any $2-$cell $\mr{L}$
$\ast-$composable to the right
with $\mr{T}$ we have that $\mr{T}\ast\mr{L}$ is a connector, 
such that the charge $(\mr{T}\ast\mr{L})[\cdot,r]$,
for a suitable $r$ depending by $s$,
maps the pullback through $\mr{y}$
of any experimental setting 
$\mf{Q}$ of the target species of $\mr{T}$ into an experimental setting 
which is included in the pullback through $\mr{x}$ of the experimental setting
obtained by mapping $\mf{Q}$ through the charge $\mr{T}[\cdot,s]$.
Here $\mr{x}$ and $\mr{y}$ are the source and target of $\mr{L}$ respectively,
and $s$ is a suitable element of $\Upgamma(\mf{Q},\mr{T})$
Cor. \ref{11200910}.
\item
\textbf{Equiformity principle}\footnote{The precise and general statement 
for links is established in Prp. \ref{01162038} and physically interpreted in Prp. \ref{12011304},
while in Thm. \ref{10081910} we show that any connector is a link.}.
Let $\mr{T}$ be a connector from the species $\mr{a}$ to the species $\mr{b}$,
$\mf{Q}$ be an experimental setting of $\mr{b}$
and $s\in\Upgamma(\mf{Q},\mr{T})$.
Thus for all contexts $M,N$ 
and morphisms $\phi:M\to N$  
we have that 
the map obtained by conjugate action of $\mr{T}_{1}^{m}(N)$
over any suitable trajectory of $\mr{a}$, relative to $N$
and assigned to $\mr{T}[\mf{Q},s]$ 
equals 
the map obtained by conjugate action of $\mr{b}_{1}^{m}(\phi)$
over a trajectory of $\mr{b}$, relative to $M$
and assigned to $\mf{Q}$.
Here 
$\mr{T}_{1}^{m}(N)$ is the morphism map of a functor, determined by $\mr{T}$,
from
the dynamical category of $\mr{b}(N)$
to 
the dynamical category of $\mr{a}(N)$ 
and
$\mr{b}_{1}^{m}(\phi)$ is the morphism map of a functor, determined by $\mr{b}$ and $\phi$,
from
the dynamical category of $\mr{b}(N)$
to 
the dynamical category of $\mr{b}(M)$
Thm. \ref{10081910}\eqref{10081910st5}. 
\end{enumerate}
\item
Species provide a great variety of context-depending dynamics
which cannot be modeled by employing functors valued in $\mr{CA}^{\ast}$
or which would require ad hoc constructions involving morphisms
in the context category.
\begin{enumerate}
\item
The diffeomorphic equivariance of the flow generated by perfect fluids 
on spacetimes is the emblematic model of a species in $\mf{dp}$
Cor. \ref{01191431}, Thm. \ref{09201707} and Thm. \ref{11151519}.
\item
The two symmetries of the relative Cauchy evolution
established in \cite[Prp. 3.7 and Prp. 3.8]{fv}, 
are essentially two specific manifestations 
of the equiformity principle see Prp. \ref{11171604} and the comment following it.
\label{11191107}
\end{enumerate}
\end{enumerate}
Context categories, species and connectors form 
$0,1$ and $2$ cells respectively of a $2-$category $2-\mf{dp}$,
such that the collection of all connectors 
can be organized in a fibered\footnote{fibered here means simply 
a collection of categories labelled by some set.}
category $\mf{Cnt}$  over a subset of couples of species
and be provided with a partial internal operation,
the vertical composition $\circ$, 
and a module structure over the collection of $2-$cells 
induced by the horizontal composition Prp. \ref{10052123}.
\par
We regard $2-\mf{dp}$, and in particular $\mf{Cnt}$, 
a \emph{nontrivial dynamics-oriented} generalization 
of the category of covariant sectors in algebraic quantum field theory.
We point out what follows.
\begin{enumerate}
\item
New it is the construction of $\mf{dp}$ to model a dynamical pattern
as a functor between $\mr{top}-$enriched categories. 
In our framework dynamical phenomena 
in general 
reflect structural properties of primitive entities
rather than be a byproduct of geometric transformations
induced by morphisms in the context category.
\label{12022035}
\item
New it is the structure of experimental setting.
Our definition extends that of state space in \cite{bfv},
since it includes the dynamics
and it extends 
the state space associated with a covariant sector in 
Doplicher, Haag and Roberts \cite{dhr2},
since it extends the Poincar\'e action to a dynamical category action.
More in general we introduce the concept of equivariance under action of 
dynamics of \emph{non-geometrical} origin.
\item
New it is the use of natural transformations and in particular connectors to construct
charges and as a result state spaces which are covariant under dynamical action.
Thus a connector extends the concept of covariant sector of \cite{dhr2}
to include dynamics of \emph{non-geometrical} origin.
\item
New it is the use of vertical composition of natural transformations
and in particular of connectors,
to generalize the concept of charge composition.
\item
New it is the concept of charge transfer.
\item
New it is the \textbf{equiformity principle} in particular 
for dynamics of \emph{non-geometric} origin.
As a natural transformation a connector is an \emph{embedding} of species, 
in such a role is analog to a natural transformation between theories 
\cite{fv,bfv}.
Nevertheless a connector embeds functors valued in $\mf{dp}$ rather than in $\mr{Phys}$
or in any kinematical category.
\emph
{Exactly because of the peculiar structure of $\mf{dp}$ which encodes 
directly the concept of \textbf{dynamics}, 
the connector encrypts in a natural way empirical information 
-
concerning the correspondence between 
trajectories associated with its target species
with those associated with its source species
- 
decoded in terms of its equiformity principle.
Now since $\mf{dp}$ permits to address dynamics which are
\textbf{not} geometrically determined\footnote{given a functor 
$\mr{a}$ from $\mf{D}$ to $\mf{Chdv}$ 
and an object $M$ of $\mf{D}$
we say that the dynamics of $\mr{a}(M)$ is geometrically determined
if the morphism map $\uptau_{\mr{a}(M)}$ of its dynamical functor factorizes through $\mr{a}_{3}$.}
\label{geom},
it results impossible to convey the above information 
by exploiting $\mr{Phys}$ or any kinematical category\footnote{more 
specifically it is impossible to obtain 
Lemma \ref{01011521} if we replace $\mf{Chdv}$ with $\mr{CA}^{\ast}$
or more in general with $\mr{Phys}$.}.
Moreover even if the dynamics is geometrically determined, 
the equiformity principle unifies diverse symmetries by explaining them as particular consequences.
More explicitly the first commutative diagram in Lemma \ref{01011521} integrates them as 
specific outcomes of the commutativity of subdiagrams,
see Prp. \ref{11171604} and the comment following it.}
\end{enumerate}
Our interpretation later referred as $\mr{Ep}$, is as follows Posit \ref{10031521}.
\par
\emph
{The physical reality is intrinsically dynamical and it is structured by $2-\mf{dp}$,
the only knowable part being $\mf{Cnt}$ the fibered category of connectors.
Each connector is an irreducible entity empirically detectable by means 
of the charges associated with it and by employing its equiformity principle.
The structure of $\mf{Cnt}$ as a whole is empirically detectable by the
representation of the vertical composition in terms of charge composition 
and by the representation of the horizontal composition in terms of charge transfer.}
\par
We shall return later on the equiformity principle,
here we have two remarks.
Firstly let $\un_{\mf{dp}}$ be the identity species contextualized in $\mf{dp}$, 
and $\un_{\un_{\mf{dp}}}$ be the identity natural transformation whose source and target 
species equal $\un_{\mf{dp}}$,
then $\un_{\un_{\mf{dp}}}$ is a connector and 
it is primitive meaning that  
it is an identity with respect to the vertical and horizontal composition in $\mf{Cnt}$.
Therefore if we identify $\mf{dp}$ with the connector $\un_{\un_{\mf{dp}}}$,
then according to $\mr{Ep}$, the category $\mf{dp}$ 
is an existing primitive entity, thus making precise the meaning primitive 
used above to characterize $\mf{dp}$.
Secondly the empirical representation of a species 
emerges in terms
of the collection of the experimental settings 
generated by all its sectors
and in terms of the equiformity principle of 
all its sectors,
where a sector is a connector whose source equal the target.
\par
As we announced the dynamical pattern approach appears 
to be required in order to encode 
with only one structure
the equivariance under diffeomorphic actions of 
the flow of complete perfect fluids on spacetimes.
To show this we briefly describe 
the construction of the functor $\mf{a}$ valued in $\mf{dp}$ and defined on the category $\mr{St}_{n}$.
Roughly $\mr{St}_{n}$ is the category of the couples $(\mc{M},U)$, 
where $\mc{M}$ is a $n-$dimensional spacetime and $U$ is an observer field on $\mc{M}$,
namely a complete timelike unit future-pointing smooth vector field on $\mc{M}$,
with smooth maps between spacetimes preserving the orientation
and relating the observer fields, as morphisms. 
\par
In particular with $n=4$ the observer field $U$ on $\mc{M}$ can be the component of a 
perfect fluid $(\uprho,p,U)$ on $\mc{M}$, where the integral curves of $U$
describe the trajectories of particles moving in the gravitational field
described by the metric tensor of $\mc{M}$ and subject to a density energy
and density pression $\uprho$ and $p$ respectively, \cite[Def. 12.4]{one}.
As a result we have that if $(\mc{M},U)$ and $(\mc{N},V)$ are objects of $\mr{St}_{4}$,
such that $U$ is the component of a perfect fluid on $\mc{N}$ and there exists
a diffeomorphism relating $U$ and $V$, then $V$ is the component of a perfect fluid on $\mc{N}$,
Thm. \ref{11151519}
\par
As a first step we define the collection of $\mr{vf}-$topologies 
Def. \ref{12051029}
formed by functions $\upxi$ mapping each object $(\mc{M},U)$ of $\mr{St}_{n}$ 
into a relevant topology on $\mc{A}(M)$ making it a topological $\ast-$algebra,
where $\mc{A}(M)$ is the commutative $\ast-$algebra 
of complex valued smooth maps on $M$ the manifold supporting $\mc{M}$. 
Now fixed a $\mr{vf}-$topology $\upxi$, 
the idea behind the construction of $\mf{a}$ 
it is to associate with any object $(\mc{M},U)$ of $\mr{St}_{n}$ the dynamical pattern
whose dynamical category denoted by $[M,U]$
holds as object set the collection of open subsets of $\mc{M}$,
with morphisms the real numbers that via the flow of the complete vector field $U$
map one open set into the other.
The dynamical functor $F_{[M,U]}$ is such that its object map sends any open subset $W$
of $\mc{M}$ to $\mc{A}(W)$
provided by the topology inherited by the topology $\upxi_{[M,U]}$, 
while its morphism map $F_{[M,U]}^{m}$ sends any real number $t$ into the conjugate 
on $\mc{A}(W)$ of the flow of $U$ evaluated in $t$.
Finally for any morphism $\phi$ between two objects $(\mc{M},U)$ and $(\mc{N},V)$ 
the value in $\phi$
of the morphim map of $\mf{a}$ is the couple formed by a $\mr{top}-$functor $f_{\phi}$
from $[N,V]$ to $[M,U]$
and a natural transformation $T_{\phi}$
from $F_{[M,U]}\circ f_{\phi}$ to $F_{[N,V]}$
Thm. \ref{09201707} and Cor. \ref{01191431}.
The definition of $\mr{vf}-$topology is intrinsically related to $\mr{St}_{n}$
and provides the minimal requirements in order to ensure 
the continuity of $F_{[M,U]}^{m}$ and $T_{\phi}$.
\par
In the remaining of this introduction we discuss 
interpretational features of the equiformity principle.
Let us start by remarking from the above example that
\begin{enumerate}
\item
The identity natural transformation of the functor $\mf{a}$
realizes the equivalence principle of general relativity
in particular providing 
diffeomorphic covariance of the integral curves 
of complete perfect fluids 
Cor. \ref{01201634}.
\label{11091230b}
\item
Item \eqref{11091230b} suggests to interpret 
the equiformity principle induced by any connector as
a generalized equivalence principle between the source 
and target species.
\label{11091230c}
\end{enumerate}
In extreme synthesis we can say that the equiformity principle of $\mr{T}$
roughly affirms that
\emph
{a collection of trajectories 
assigned to any experimental setting $\mf{Q}$ of the \emph{target} species of $\mr{T}$
corresponds
to a collection of trajectories 
assigned to the experimental setting $\mr{T}[\mf{Q},s]$ of the \emph{source} species of $\mr{T}$
for any $s\in\Upgamma(\mf{Q},\mr{T})$}.
\par
Here what we point out is the correspondence between target and source species.
Let us analyze some consequences of this principle 
when the source is a classical species
and the target is a quantum species.
Here by classical (quantum) species
we mean a species such that it is commutative 
(noncommutative)
the algebra associated with any context and any object of the dynamical category 
of the species.
Then the principle establishes that 
classical and quantum trajectories correspond.
\par
Notice that we are saying that 
\emph
{the dynamical evolution of classical observables
when measured against classical statistical ensembles, 
equals the dynamical evolution of suitable quantum observables
when measured against suitable quantum statistical ensembles}.
Incidentally the main outcome of the reductionistic point of view is 
to regard general relativity a coarse grain approximation of,
and then worthy to be reduced to, 
a theory where spacetime emerges from a more fundamental quantum entity,
or at least where the gravitational field is quantized.
Thus in both cases according to the reductionistic view, 
general relativity is compelled to be reduced and then replaced by a quantum theory
of gravity.
\emph
{Instead according $\mr{Ep}$ classical and quantum species \emph{coexist} 
and this coexistence is empirically detectable in terms of the equiformity principle of
the connectors between them.}
\par
More precisely assume that there exist $\mf{b}$ and $\mr{T}$,
where $\mf{b}$ is a (strict) quantum gravity, namely
a suitable functor from the category $\mr{St}_{n}$ to $\mf{dp}$ 
such that for any context $(\mc{M},U)$ of $\mr{St}_{n}$
the dynamical category of $\mf{b}(\mc{M},U)$ is $[M,U]$
and for any open set $W$ of $\mc{M}$
the topological $\ast-$algebra $\mc{A}_{\mf{b}((\mc{M},U))}(W)$ is noncommutative, 
while $\mr{T}$ is a natural transformation from 
the species of general relativity $\mf{a}$ 
to the quantum gravity species $\mf{b}$.
Thus the equiformity principle of $\mr{T}$
establishes in particular the
\emph
{quantum realization of the velocity of maximal integral curves of 
the complete vector field $U$}
Cor. \ref{11091638}.
\par
Clearly this sort of classical-quantum coexistence is automatically 
precluded by the actual reductionistic paradigm, 
however the equiformity principle produces experimentally testable equalities, 
that can be employed in order to opt for the paradigm embodied by $\mr{Ep}$ or for reductionism.
\par 
The paper is organized as follows.
We start in section \ref{11301512} by introducing 
the propensity map slightly generalizing the usual state-effect duality.
In section \ref{11301644} we introduce the category of dynamical patterns and 
the category of channels and devices. 
Define the concept of species, experimental settings of a species, links between 
species and the fundamental equiformity principle for links.
We prove that any connector is a link between  
any experimental setting of its target and 
a suitable experimental setting of its source, thus providing 
an equiformity principle.
Then we establish charge composition and charge transfer of connectors.
These represent three of the five main results of the paper.
In section \ref{11301655} we introduce the general language to address the 
$2-$category $2-\mf{dp}$ and the fibered category of connectors.
In section \ref{10222056}
we construct in the fourth main result an example of species of dynamical patterns
namely the $n-$dimensional classical gravity species $\mr{a}^{n}$,
and state that the equivalence principle of general relativity
emerges as the equiformity principle of the connector associated with $\mr{a}^{n}$. 
Finally in section \ref{10181450} we define the collection of quantum gravity species
and prove that the existence of a connector $\mr{T}$ from $\mr{a}^{n}$ 
and an $n-$dimensional strict quantum gravity species 
provides a quantum realization of the velocity of maximal integral curves of 
complete vector fields over spacetimes Cor. \ref{11091638}.
As an application to Robertson-Walker spacetimes 
we establish that the Hubble parameter, 
the acceleration of the scale function 
and the new constraints of its positivity evaluated over a subset of the range 
of the galactic time of a geodesic $\alpha$, 
are expressed in terms of a quantum realization of the velocity of $\alpha$
Thm. \ref{04151343} and Cor. \ref{04171614}.
\emph
{Therefore the positivity of the acceleration 
follows as a result of the existence of $\mr{T}$
satisfying these constraints 
rather than the existence of dark energy}.
\section{Terminology and preliminaries}
\label{not1}
In the entire paper 
we consider the Zermelo-Fraenkel theory together the axiom of universes 
stating that for any set there exists a universe 
containing it as an element \cite[p. $10$]{ks}.
See \cite[Expose I Appendice]{sga4} for the definition and properties of universes, 
see also \cite[p. $22$]{mcl} and \cite[$\S1.1$]{bor}. 
\subsubsection{Sets}
Given two sets $A,B$ let $\mc{P}(A)$ denote the power set of $A$,
$\hookrightarrow$ denotes the set embedding
and for maps $g,f$ such that $d(g)=c(f)$ and $B\subset d(f)$
we let $(g\circ f)(B)$ denote the image of $B$ through the map
$g\circ f\circ(B\hookrightarrow d(f))$.
Let $\mr{ev}_{(\cdot)}$ denote the evaluation map, i.e. if $F:A\to B$ 
is any map and $a\in A$, then $\mr{ev}_{a}(F)\coloneqq F(a)$.
$\R$ and $\C$ are the fields of real and complex numbers provided by the standard topology,
set $\N_{0}\coloneqq\N-\{0\}$ and $\R_{0}\coloneqq\R-\{0\}$, 
while $\widetilde{\R}\coloneqq\R\cup\{\infty\}$
provided by the topology of one-point compactification.
If $A$ is any set then $\un_{A}$ or simply $\un$, is the identity map on $A$.
For any semigroup $S$ let $S^{op}$ denote the opposite semigroup, 
while for any subset $A$ of $S$ let $\langle A\rangle$ denote 
the subsemigroup of $S$ generated by $A$.
If $S,X$ are topological spaces, let $\mc{C}(S,X)$ denote the set of continuous maps on $S$ and into $X$,
and $Op(S)$, $Ch(S)$ and $Comp(S)$ denote the sets of open, closed and compact subsets of
$S$, while $\mc{B}(S)$ the $\sigma-$field of Borel subsets of $S$. If $A\subset S$ then $\ov{A}$ 
or $\mr{Cl}(A)$ denotes the closure of $A$ in $S$. 
\subsubsection{Categories}
Let $A$, $B$ and $C$ be categories. 
Let $\mr{Obj}(A)$ denote the set of the objects of $A$, often we let $x\in A$ denote $x\in \mr{Obj}(A)$. 
For any $x,y\in A$ let $\mr{Mor}_{A}(x,y)$ be the set of morphisms of $A$ from $x$ to $y$, 
also denoted by $A(x,y)$, let $\un_{x}$ the unit 
morphism of $x$, while 
$\mr{Inv}_{A}(x,y)=\{f\in \mr{Mor}_{A}(x,y)\,\vert\,
(\exists g\in \mr{Mor}_{A}(y,x))(f\circ g=\un_{y},g\circ f=\un_{x}))\}$ 
denotes the possibly empty set of invertible morphisms from $x$ to $y$.
Set $\mr{Mor}_{A}=\bigcup\{\mr{Mor}_{A}(x,y)\,\vert\,x,y\in A\}$ and $\mr{Inv}_{A}=\bigcup\{\mr{Inv}_{A}(x,y)\,\vert\,x,y\in A\}$, while
$\mr{Aut}_{A}(y)=\mr{Inv}_{A}(y,y)$, for $y\in A$, and $\mr{Aut}_{A}=\bigcup\{\mr{Aut}_{A}(x)\,\vert\,x\in A\}$.
For any $T\in \mr{Mor}_{A}(x,y)$ we set $d(T)=x$ and $c(T)=y$, while the composition on $\mr{Mor}_{A}$ is always denoted by $\circ$.
Let $A^{op}$ denote the opposite category of $A$, \cite[p. $33$]{mcl}.
Let $\mr{Fct}(A,B)$ denote the category of functors from $A$ to $B$ and 
natural transformations provided by pointwise composition,
see \cite[$1.3$]{ks}, \cite[p.$40$]{mcl}, \cite[p. $10$]{bor}.
Let us identify any $F\in\mr{Fct}(A,B)$ with the couple $(F_{o},F_{m})$, 
where $F_{o}:\mr{Obj}(A)\to \mr{Obj}(B)$ said the object map of $F$, 
while $F_{m}:\mr{Mor}_{A}\to \mr{Mor}_{B}$ such that 
$F_{m}^{x,y}:\mr{Mor}_{A}(x,y)\to \mr{Mor}_{B}(F_{o}(x),F_{o}(y))$ where $F_{m}^{x,y}=F_{m}\up \mr{Mor}_{A}(x,y)$ for all $x,y\in A$.
Often and only when there is no risk of confusion we adopt 
the standard convention to denote $F_{o}$ and $F_{m}$ simply by $F$.
Let $\un_{A}\in\mr{Fct}(A,A)$ be the identity functor on $A$, defined in the obvious way.
Let $\circ:\mr{Fct}(B,C)\times\mr{Fct}(A,B)\to\mr{Fct}(A,C)$ 
denote the vertical composition of functors \cite[Def. $1.2.10$]{ks},
\cite[p. $14$]{mcl}, for any $\upsigma\in\mr{Fct}(A,B)$, the identity morphism 
$\un_{\upsigma}$ of $\upsigma$ in the category $\mr{Fct}(A,B)$ 
is such that $\un_{\upsigma}(M)=\un_{\upsigma_{o}(M)}$, for all $M\in A$.
Let $\beta\ast\alpha\in \mr{Mor}_{\mr{Fct}(A,C)}(H\circ F,K\circ G)$ be the Godement product 
or horizontal composition between the 
natural transformations $\beta$ and $\alpha$, where $H,K\in\mr{Fct}(B,C)$ and $F,G\in\mr{Fct}(A,B)$ while 
$\beta\in \mr{Mor}_{\mr{Fct}(B,C)}(H,K)$ and $\alpha\in \mr{Mor}_{\mr{Fct}(A,B)}(F,G)$, see \cite[Prp. $1.3.4$]{bor} 
(or \cite[p. $42$]{mcl} where it is used the symbol $\circ$ instead of $\ast$). 
The product $\ast$ is associative, in addition
we have the following rule for all $\gamma\in \mr{Mor}_{\mr{Fct}(A,B)}(H,L)$, $\alpha\in \mr{Mor}_{\mr{Fct}(A,B)}(F,H)$, 
and $\delta\in \mr{Mor}_{\mr{Fct}(B,C)}(K,M)$, $\beta\in \mr{Mor}_{\mr{Fct}(B,C)}(G,K)$, see \cite[Prp. $1.3.5$]{bor}
\begin{equation}
\label{bor135}
(\delta\ast\gamma)\circ(\beta\ast\alpha)=(\delta\circ\beta)\ast(\gamma\circ\alpha).
\end{equation}
We have 
\begin{equation}
\label{20061403}
\beta\ast\un_{F}=\beta\circ F_{o},
\end{equation}
moreover $\un_{G}\ast\un_{F}=\un_{G\circ F}$, and $\un_{F}\circ\un_{F}=\un_{F}$.
For all categories $\mf{D},\mf{F}$, 
all $\mr{a},\mr{b},\mr{c}\in\mr{Fct}(\mf{D},\mf{F})$, all
$\mr{T}\in\prod_{O\in\mf{D}}\mr{Mor}_{\mf{F}}(\mr{a}(O),\mr{b}(O))$ 
and $\mr{S}\in\prod_{O\in\mf{D}}\mr{Mor}_{\mf{F}}(\mr{b}(O),\mr{c}(O))$ 
we set $S\circ T\in\prod_{O\in\mf{D}}\mr{Mor}_{\mf{F}}(\mr{a}(O),\mr{c}(O))$ 
such that $(S\circ T)(M)\coloneqq S(M)\circ T(M)$ for all
$M\in\mf{D}$.
\par
Let $B$ be a category, for any $x\in B$ set $\Delta_{x}\coloneqq\bigcup\{\mr{Mor}_{B}(x,y)\,\vert\, y\in B\}$ and 
$\Gamma_{x}\coloneqq\bigcup\{\mr{Mor}_{B}(y,x)\,\vert\, y\in B\}$.
Let $(\cdot)^{\dagger}$ and $(\cdot)_{\star}$ be maps defined on $\mr{Mor}_{B}$ such that for any $T\in \mr{Mor}_{B}$ we have
$T^{\dagger}:\Delta_{c(T)}\to\Delta_{d(T)}$ $S\mapsto S\circ T$, 
and $T_{\star}:\Gamma_{d(T)}\to\Gamma_{c(T)}$ $S\mapsto T\circ S$, 
moreover let $(\cdot)^{\ast}\coloneqq(\cdot)^{\dagger}\circ(\cdot)^{-1}\up \mr{Inv}_{B}$, thus
\begin{equation}
\label{08161111}
\begin{aligned}
(T\circ J)^{\dagger}&=J^{\dagger}\circ T^{\dagger},
\\
(T\circ J)_{\star}&=T_{\star}\circ J_{\star},
\\
(U\circ W)^{\ast}&=U^{\ast}\circ W^{\ast},
\end{aligned}
\end{equation}
for any $T,J\in \mr{Mor}_{B}$, such that $(T,J)\in Dom(\circ)$, 
and $U,W\in \mr{Inv}_{B}$ such that $(U,W)\in Dom(\circ)$. 
Moreover whenever it does not make confusion let $T^{\dagger}$, $T_{\star}$ and $W^{\ast}$ 
denote also their restrictions.
\par
Let $\mc{V}$ be a universe, then following \cite{ks}, a set is $\mc{V}-$small
if it is isomorphic to a set belonging to $\mc{V}$, and a set is a $\mc{V}-$set if it belongs to $\mc{V}$.
A $\mc{V}-$category $A$ is a category such that $\mr{Mor}_{A}(M,N)$ is $\mc{V}-$small for any $M,N\in A$,
and $A$ is $\mc{V}-$small if it is a $\mc{V}-$category and $\mr{Obj}(A)$ is $\mc{V}-$small. 
Let us add the following definition:
$\mc{T}$ is a $\mc{V}-$type category if it is a $\mc{V}-$category such that $\mr{Obj}(\mc{T})\simeq B$ and $B\subseteq\mc{V}$.
Clearly any $\mc{V}-$small category is a $\mc{V}-$type category since $\mc{V}\subset\mc{P}(\mc{V})$.
Moreover if $\mc{V}_{0}$ is a universe such that $\mc{V}\in\mc{V}_{0}$, then 
$\mc{V}\subset\mc{V}_{0}$ since $\mc{V}_{0}\subset\mc{P}(\mc{V}_{0})$, hence  
any $\mc{V}-$type category is a $\mc{V}_{0}-$small category.
\begin{proposition}
\label{12311604}
Let $A,B$ be categories and $\mc{V}$ a universe, thus 
\begin{enumerate}
\item
\begin{equation*}
\begin{aligned}
\mr{Obj}(\mr{Fct}(A,B))
&\subset
\mc{P}(\mr{Obj}(A)\times \mr{Obj}(B))
\times
\mc{P}(\mr{Mor}_{A}\times \mr{Mor}_{B}),
\\
\mr{Mor}_{\mr{Fct}(A,B)}(f,g)
&\subset
\prod_{x\in \mr{Obj}(A)}\mr{Mor}_{B}(f_{o}(x),g_{o}(x)),\,
\forall f,g\in \mr{Obj}(\mr{Fct}(A,B));
\end{aligned}
\end{equation*}
\label{12311104}
\item
if $A$ is $\mc{V}-$small and $B$ is a $\mc{V}-$category,
thus $\mr{Fct}(A,B)$ is a $\mc{V}-$category;
\label{12311604st2}
\item
if $A$ and $B$ are $\mc{V}-$small categories then $\mr{Fct}(A,B)$ is a $\mc{V}-$small category.
\label{12311604st3}
\end{enumerate}
\end{proposition}
\begin{proof}
St.\eqref{12311104} follows easily by the definitions.
St.\eqref{12311604st2} follows since the second equality in st.\eqref{12311104} and \cite[Def. 1.1.1.]{ks}.
For any $\mc{V}-$small category $C$ we have $\mr{Mor}_{C}\in\mc{V}$ since 
\cite[Def. 1.1.1.(v)]{ks}, thus the st.\eqref{12311604st3} follows 
by st.\eqref{12311604st2}, the first equality in st.\eqref{12311104},
and \cite[Def. 1.1.1.(viii,iv)]{ks}. 
\end{proof}
Let $\mc{V}-\mr{set}$ be the category of $\mc{V}-$sets, functions as morphisms with map composition, 
and $\mc{V}-\mr{cat}$ be the $2-$category whose object set is the set of $\mc{V}-$small categories,
and for any $A,B$ $\mc{V}-$small categories $\mr{Mor}_{\mc{V}-\mr{cat}}(A,B)$ is the 
$\mc{V}-$small category $\mr{Fct}(A,B)$, see Prp. \ref{12311604}\eqref{12311604st3}.
\par
\textbf{In the remaining of the paper we assume fixed three universes $\mc{U},\mc{U}_{0},\mc{U}_{1}$ 
such that $\mc{U}\in\mc{U}_{0}\in\mc{U}_{1}$}, 
the existence of $\mc{U}_{0}$ and then $\mc{U}_{1}$ fixed $\mc{U}$ 
being ensured by the axiom of universes. 
\par
Let $\mr{set}$ and $\mr{cat}$ denote $\mc{U}-\mr{set}$ and $\mc{U}_{0}-\mr{cat}$ respectively,
while let $\mr{Cat}$ denote $\mc{U}_{1}-\mr{cat}$. 
\par
Thus $\mr{set}$ is a $\mc{U}-$category but it is \textbf{not} an object of $\mc{U}-\mr{cat}$.
However $\mr{set}$ is a subcategory of $\mc{U}_{0}-\mr{set}$
and any $\mc{U}-$type category is an object of $\mr{cat}$,
note that $\mr{set}$ is an $\mc{U}-$type category so an object of $\mr{cat}$.
Next $\mc{U}_{0}\subset\mc{U}_{1}$ hence $\mr{cat}$ is a $2-$subcategory of $\mr{Cat}$.
\par
\textbf{For any structure $S$ whenever we say ``the set of the $S$'s'', we always mean the subset of 
those elements of $\mc{U}$ satisfying the axioms of $S$}.
Therefore for what said, letting $\mc{S}$ be the category whose object 
set is the set of the $S$'s and $\mr{Mor}_{\mc{S}}(a,b)$
is the set of the morphisms of the structure $S$ from $a$ to $b$, with $a,b\in \mr{Obj}(\mc{S})$, then
$\mc{S}$ is an object of $\mr{cat}$, 
and via the forgetful functor is equivalent to a subcategory $\tilde{\mc{S}}$ of $\mr{set}$.
Let $\mc{S}$ be called the category of structure $S$, and let $\tilde{\mc{S}}$ denote the image of $\mc{S}$
via the forgetful functor.
\par
Let $\mr{top}$ be the category such that $\mr{Obj}(\mr{top})$ is the set of topological spaces,
$\mr{Mor}_{\mr{top}}(X,Y)=\mc{C}(X,Y)$ for all $X,Y\in\mr{top}$, and map composition as morphism composition.
Let $\mr{tg}$ be the category such that $\mr{Obj}(\mr{tg})$ is the set of topological groups, 
$\mr{Mor}_{\mr{tg}}(G,H)$ is the set of continuous group morphisms on $G$ and into $H$, for all $G,H\in\mr{tg}$,
and map composition as morphism composition.
\par
In \cite[Def. $1.3.2$]{lei} is presented a definition of enrichment over a category, 
here we need a weaker definition. 
Let $A$ be a subcategory of $\mr{set}$ and $B$ a category.
$B$ is said to be $A-$quasi enriched if
\begin{itemize}
\item
$\mr{Mor}_{B}(x,y)\in A$ for all $x,y\in B$; 
\item
for all $T\in \mr{Mor}_{B}$ and $y\in B$;
\begin{itemize}
\item
$T^{\dagger}\up \mr{Mor}_{B}(c(T),y)\in \mr{Mor}_{A}(\mr{Mor}_{B}(c(T),y),\mr{Mor}_{B}(d(T),y))$, 
\item
$T_{\star}\up \mr{Mor}_{B}(y,d(T))\in \mr{Mor}_{A}(\mr{Mor}_{B}(y,d(T)),\mr{Mor}_{B}(y,c(T)))$.
\end{itemize}
\end{itemize}
Notice that $B$ is a $\mc{U}-$category and any subcategory of $B$ is $A-$quasi enriched.
Let $B$ be a $A-$quasi enriched category, $C$ is called a $A-$quasi enriched subcategories of $B$ iff 
$C$ is a $A-$quasi enriched category, $C$ is a subcategory of $B$ and the injection map from $\mr{Mor}_{C}(x,y)$ 
to $\mr{Mor}_{B}(x,y)$ is a morphism of $A$ for all $x,y\in C$.
Clearly if $B$ is $A-$quasi enriched then $B^{op}$ is $A-$quasi enriched, 
while if any category is $A-$enriched then it is $A-$quasi enriched.
If $B,C$ are $A-$quasi enriched define $\mr{Fct}_{A}(B,C)$ the subset of the $F\in\mr{Fct}(B,C)$ such that 
for all $u,v\in B$
\begin{equation*}
F_{m}^{u,v}\in \mr{Mor}_{A}(\mr{Mor}_{B}(u,v), \mr{Mor}_{C}(F_{o}(u),F_{o}(v))).
\end{equation*} 
If $D$ is $A-$quasi enriched, 
then $G\circ F\in\mr{Fct}_{A}(B,D)$ for any $F\in\mr{Fct}_{A}(B,C)$ and $G\in\mr{Fct}_{A}(C,D)$,
where $\circ$ is the vertical composition of functors. 
For any structure $S$ we convein that $\mc{S}-$quasi enriched means $\tilde{\mc{S}}-$quasi enriched.
\subsubsection{Preordered topological linear spaces}
Let $\K\in\{\R,\C\}$, and $X,Y$ be topological linear spaces over $\K$ ($\K-$t.l.s. often simply t.l.s. if $\K=\C$). 
Let $\mf{L}(X,Y)$ denote the $\K-$linear space of continuous $\K-$linear maps 
from $X$ to $Y$, set $\mf{L}(X)\coloneqq\mf{L}(X,X)$ and $X^{\ast}\coloneqq\mf{L}(X,\K)$.
$\mf{L}_{s}(X,Y)$ is the $\K-$t.l.s. whose underlying linear space is $\mf{L}(X,Y)$ provided by the topology of simple 
convergence well-set definition indeed the topology of $\mf{L}_{s}(X,Y)$ is compatible with the structure of linear space of 
$Y^{X}$ since \cite[Prp.4($EV_{I}$) I.7; Prp.1 III.13]{EVT}. Next $\mf{L}_{w}(X,Y)$ is the locally convex linear space whose 
underlying linear space is $\mf{L}(X,Y)$ provided by the topology generated by the following set of seminorms 
$\{q_{(\upphi,x)}\mid(\upphi,x)\in Y^{\ast}\times X\}$, where $q_{(\upphi,x)}(A)\doteq|\upphi(Ax)|$. 
There exists a unique category $\K-\mr{tls}$ (or simply $\mr{tls}$ if $\K=\C$) 
whose object set is the set of all $\K-$t.l.s.'s, $\mr{Mor}_{\mr{tls}}(X,Y)=\mf{L}(X,Y)$ for all $X,Y\in\mr{tls}$,
and map composition as morphism composition.
$\mr{Aut}_{\K-\mr{tls}}(X)$ will be provided with the topology induced by the one in $\mf{L}_{s}(X)$
\begin{equation}
\label{08160957}
(\cdot)^{\ast}\in \mr{Mor}_{\mr{tg}}(\mr{Aut}_{\K-\mr{tls}}(X),\mr{Aut}_{\K-\mr{tls}}(X_{s}^{\ast})).
\end{equation}
If $A,B$ have richer structure than that of t.l.s., then $\mf{L}(A,B)$ stands for $\mf{L}(X,Y)$ where 
$X,Y$ are the underlying t.l.s. underlying $A,B$ respectively, similarly for $A^{\ast}$ and $\mf{L}(A)$.
In case $X$ is a normed space we assume $\mf{L}(X)$ to be provided by the topology generated by the usual $\sup-$norm.
If $\mr{X}$ is any structure including as a substructure the one of normed space say $\mr{X}_{0}$, 
for example the normed space underlying any normed algebra,
we let $\mf{L}(\mr{X})$ denote the normed space $\mf{L}(\mr{X}_{0})$.
If $X,Y$ are Hilbert spaces and $U\in\mf{L}(X,Y)$ is unitary then $\mr{ad}(U)\in\mf{L}(\mf{L}(X),\mf{L}(Y))$ 
denotes the isometry defined by $\mr{ad}(U)(a)\coloneqq UaU^{-1}$, for all $a\in\mf{L}(X)$.
\par
A preordered topological linear space (p.t.l.s) consists of a couple formed by an object $X$ of $\R-\mr{tls}$ and by
a preorder on it, i.e. a reflexive, transitive relation $\geq$ providing $X$ with the structure of preordered linear space,
see for the definition \cite[II.15]{BourTVS}, and such that the set $X^{+}\coloneqq\{x\in X\,\vert\, x\geq\ze\}$ of positive 
elements of $X$ is closed. $a\leq b$ stands for $b\geq a$ while $x>y$ stands for $x\geq y$ and $x\neq y$, likewise for 
$y<x$. Let $[a,b]\coloneqq\{x\in X\,\vert\, a\leq x\leq b\}$ and $]a,b[\coloneqq\{x\in X\,\vert\,a<x<b\}$, for any 
$a,b\in X$. Any function from $X$ to $Y$ mapping $X^{+}$ into $Y^{+}$ is called positive.
If $X$ is any $\R-$linear space and $C$ is a pointed convex cone in $X$, i.e. $\ze\in C$, $C+C\subseteq C$ and
$\lambda\cdot C\subseteq C$ 
for all $\lambda>0$, then the relation $x\geq y$ iff $x-y\in C$ provides $X$ by the structure of preordered linear space,
see \cite[II.12, Prp. 13]{BourTVS} or \cite[p. 20]{schm} (convex cone called wedge and preordered vector space called
ordered vector space in \cite{schm}). For any two p.t.l.s.'s $X,Y$, define 
\begin{equation}
\label{08271622}
\mr{P}(X,Y)\coloneqq\{T\in\mf{L}(X,Y)\,\vert\, T(X^{+})\subseteq Y^{+}\},
\end{equation}
set $\mr{P}(X)\coloneqq\mr{P}(X,X)$ and $\mf{P}_{X}\doteqdot\mr{P}(X,\R)$.
Note that $T\in\mr{P}(X,Y)$ is an order morphism since it is linear and $a\geq b$ iff $a-b\geq\ze$.
$\mr{P}(X,Y)$ is a pointed convex cone and it is closed in $\mf{L}_{s}(X,Y)$ since $Y^{+}$ is closed, 
hence the relation $T\geq S\Leftrightarrow T-S\in\mr{P}(X,Y)$ provides $\mf{L}_{s}(X,Y)$ with the 
structure of p.t.l.s. said canonical. Note that $T\geq S$ iff $T(x)\geq S(x)$, for all $x\in X^{+}$ and $S,T\in\mf{L}(X,Y)$.
By construction $\mf{L}_{s}(X,Y)^{+}=\mr{P}(X,Y)$, in particular
$(X_{s}^{\ast})^{+}=\mf{P}_{X}=\mr{P}(X,\R)=\{\upphi\in X^{\ast}\,\vert\,\upphi(X^{+})\subseteq\R^{+}\}$.
Let $\mr{K}(X)\coloneqq\{T\in\mf{L}(X)\,\vert\,\ze\leq T\leq \un\}$. 
There exists a unique category $\mr{ptls}$ whose object set is the set of all the p.t.l.s's,
$\mr{Mor}_{\mr{ptls}}(X,Y)=\mr{P}(X,Y)$ for all $X,Y\in\mr{ptls}$ and map composition as morphism composition.
$\mr{Aut}_{\mr{ptls}}(X)$ will be provided by the topology induced by the one in $\mf{L}_{s}(X)$.
Let $\mr{tls}_{\geq}$ denote the full subcategory of $\R-\mr{tls}$ whose object set equals $\mr{Obj}(\mr{ptls})$.
\subsubsection{Topological $\ast-$algebras and their order structures}
Here as a topological algebra (t.a.) we mean a complex associative algebra with a topology of Hausdorff providing it an 
object of $\mr{tls}$, and such that the product is \textbf{separately} continuous. For any (unital) t.a. $\mc{B}$ the 
semigroup $\mc{B}$ has to be understood as the (unital) multiplicative semigroup underlying $\mc{B}$. As a topological 
$\ast-$algebra ($\ast-$t.a.) we mean a unital involutive algebra with a topology providing it a t.a. and w.r.t. which the 
involution is a continuous map. Let $\un$ denote by abuse of language the unit of any unital algebra. 
If $X$ is in $\mr{tls}$ then $\mf{L}_{s}(X)$ is a unital t.a. provided by the composition of maps as the product.
For any two $\ast-$t.a. $\mc{A},\mc{B}$ define $Hom^{\ast}(\mc{A},\mc{B})$ 
to be the set of $\uptau\in\mf{L}(\mc{A},\mc{B})$
such that $\uptau(\un)=\un$, $\uptau(ab)=\uptau(a)\uptau(b)$ and $\uptau(a^{\ast})=\uptau(a)^{\ast}$ for all $a,b\in\mc{A}$. 
There exists a unique category $\mr{tsa}$ such that $\mr{Obj}(\mr{tsa})$ is the set of the topological $\ast-$algebras,
$\mr{Mor}_{\mr{tsa}}(\mc{A},\mc{B})=Hom^{\ast}(\mc{A},\mc{B})$ for all $\mc{A},\mc{B}\in\mr{tsa}$, and map composition as 
morphism composition.
$\mr{Aut}_{\mr{tsa}}(\mc{A})$ will be provided by the topology induced by the one in $\mf{L}_{s}(\mc{A})$. 
Notice that according the general convention used here for categories, 
for all $\mc{A}\in\mr{tsa}$ the symbol $\un_{\mc{A}}$ denotes the identity morphism 
of the object $\mc{A}$, i.e. the identity map from $\mc{A}$ to itself. 
This is the reason instead to denote simply by $\un$ the unit of $\mc{A}$.
\par
For any subset $S$ of $\mc{B}$ set $S'\coloneqq\{a\in\mc{B}\,\vert\,(\forall b\in S)(ab=ba)\}$ and $S''\coloneqq (S')'$ 
said the commutant and  bicommutant of $S$ respectively. Let $\mc{A}$ be a topological $\ast-$algebra. 
Define $\mc{A}_{ob}\coloneqq\{a\in\mc{A}\,\vert\, a=a^{\ast}\}$, $\Pr(\mc{A})\coloneqq\{p\in\mc{A}_{ob}\,\vert\, p p=p\}$,
$\mr{U}(\mc{A})\coloneqq\{u\in\mc{A}\,\vert\,u^{\ast}=u^{-1}\}$, $\mc{A}^{\star}\coloneqq\{a^{\ast}a\,\vert\, a\in\mc{A}\}$,
$\mc{A}^{\natural}\coloneqq\{\sum_{k=1}^{n}a_{k}^{\ast}a_{k}\mid a\in\mc{A}^{n},n\in\Z_{>0}\}$ and the set of positive elements 
of $A$
\begin{equation*}
\mc{A}^{+}\coloneqq\mr{Cl}(\mc{A}^{\natural}).
\end{equation*}
\textbf
{Our definition of positive elements and therefore 
of partial order in $\mc{A}$ differs from the one in \cite{schm} where instead
it is used the set $\mc{A}^{\natural}$, and differs from that in \cite{fra}. 
However the three definitions 
of partial order in $\mc{A}$ agree in case $\mc{A}$ is a locally $C^{\ast}-$algebra}, see below. 
The reason to choose such a set of positive elements,
resides in the fact that we want to have the set $\mr{P}(\mc{A},\mc{B})$ closed w.r.t. the topology of simple
convergence and the set $\mr{Ef}(\mc{A})$ (if the product is jointly continuous also $\Uptheta_{\mc{A}}$) closed in $\mc{A}$,
see below.
\par
$\ze,\un\in\mc{A}^{+}$ since $\un^{\ast}\un=\un$. 
$\mc{A}_{ob}$ is an $\R-$linear space and it is closed in $\mc{A}$ the involution being 
continuous, 
hence $\mc{A}^{+}\subset\mc{A}_{ob}$ and clearly it is the closure of $\mc{A}^{\natural}$ in the $\R-$t.l.s. $\mc{A}_{ob}$,
moreover $\mc{A}^{\natural}$ is a pointed convex cone of $\mc{A}_{ob}$ thus $\mc{A}^{+}$ is a pointed convex cone since 
\cite[II.13, Prp. 14]{BourTVS}.
Hence $\mc{A}^{+}$ 
is a pointed closed convex cone in the $\R-$t.l.s. $\mc{A}_{ob}$, therefore by defining $a\geq b$ iff $a,b\in\mc{A}_{ob}$
and $a-b\in\mc{A}^{+}$ 
we provide $\mc{A}_{ob}$ with the structure of p.t.l.s. Now let $\geq$ be called standard and denote the relation 
in $\mc{A}$ inherited by $\geq$ on $\mc{A}_{ob}$, and for any $\mc{B}\in\mr{tsa}$ by abuse of language set 
\begin{equation}
\label{08101915}
\mr{P}(\mc{A},\mc{B})\coloneqq\{T\in\mf{L}(\mc{A},\mc{B})\,\vert\, T(\mc{A}^{\star})\subseteq\mc{B}^{+},
T(\mc{A}_{ob})\subseteq\mc{B}_{ob}\}.
\end{equation} 
Set $\mr{P}(\mc{A})\coloneqq\mr{P}(\mc{A},\mc{A})$, $\mf{P}_{\mc{A}}\coloneqq\mr{P}(\mc{A},\C)$, 
$\mf{P}_{\mc{A}}^{\natural}\coloneqq\{\upphi\in\mf{P}_{\mc{A}}\,\vert\,\upphi(\un)\neq 0\}$, and
$\mf{E}_{\mc{A}}\coloneqq\{\upomega\in\mf{P}_{\mc{A}}\,\vert\,\upomega(\un)=1\}$.
$\mf{E}_{\mc{A}}$ is a closed convex set of $\mc{A}_{s}^{\ast}$, and with any 
$\upphi\in\mf{P}_{\mc{A}}^{\natural}$ one associates the 
continuous state $\upphi(\un)^{-1}\upphi$.
We have
\begin{equation}
\label{08111513}
\mr{P}(\mc{A},\mc{B})=\{T\in\mf{L}(\mc{A},\mc{B})\,\vert\, T(\mc{A}^{+})\subseteq\mc{B}^{+},T(\mc{A}_{ob})\subseteq
\mc{B}_{ob}\},
\end{equation}
since linearity and continuity, moreover 
\begin{equation}
\label{09101459}
(\forall T\in\mf{L}(\mc{A},\mc{B}))
(T\in\mr{P}(\mc{A},\mc{B})\Leftrightarrow T\up_{\mc{A}_{ob}}^{\mc{B}_{ob}}\in\mr{P}(\mc{A}_{ob},\mc{B}_{ob})).
\end{equation}
Any element in $\mr{P}(\mc{A},\mc{B})$ is an order morphism and
\begin{equation}
\label{08161036}
\mr{Mor}_{\mr{tsa}}(\mc{A},\mc{B})\subset\mr{P}(\mc{A},\mc{B}),
\end{equation}
so the maps $\mr{Obj}(\mr{tsa})\ni\mc{A}\mapsto\mc{A}_{ob}$ and 
$\mr{Mor}_{\mr{tsa}}(\mc{A},\mc{B})\ni T\mapsto T\up_{\mc{A}_{ob}}^{\mc{B}_{ob}}$
determine a functor from $\mr{tsa}$ to $\mr{ptls}$.
$\mr{P}(\mc{A},\mc{B})$ is a pointed convex cone and it is closed in $\mf{L}_{s}(\mc{A},\mc{B})$ since $\mc{B}^{+}$ 
is closed, so closed also in the $\R-$t.l.s. $\mf{L}_{s}(\mc{A},\mc{B})_{\R}$ underlying $\mf{L}_{s}(\mc{A},\mc{B})$.
Therefore the relation $T\geq S\Leftrightarrow T-S\in\mr{P}(X,Y)$ provides $\mf{L}_{s}(\mc{A},\mc{B})_{\R}$ with the 
structure of p.t.l.s. said canonical.
Let us convein to denote the p.t.l.s. $\mf{L}_{s}(\mc{A},\mc{B})_{\R}$ 
by $\mf{L}_{s}(\mc{A},\mc{B})$ or simply $\mf{L}(\mc{A},\mc{B})$, 
and to denote the p.t.l.s. $(\mc{A}_{s}^{\ast})_{\R}$ by $\mc{A}_{s}^{\ast}$ or simply $\mc{A}^{\ast}$. 
Let $\mr{K}(\mc{A})\coloneqq\{T\in\mf{L}(\mc{A})\,\vert\,\ze\leq T\leq \un\}$. 
Thus by construction $\mf{L}_{s}(\mc{A},\mc{B})^{+}=\mr{P}(\mc{A},\mc{B})$, in particular
\begin{equation}
\label{08271616}
(\mc{A}_{s}^{\ast})^{+}=\mf{P}_{\mc{A}}=\mr{P}(\mc{A},\C).
\end{equation}
Let $\mr{P}_{\mc{A}}$ be the set of the linear functionals $\upphi$ on $\mc{A}$ such that $\upphi(\mc{A}^{+})\subseteq\R^{+}$,
not to be confused with the set $\mr{P}(\mc{A})$. Let $\mr{E}_{\mc{A}}$ be called the set of states of $\mc{A}$ and defined 
as the subset of the $\uppsi\in\mr{P}_{\mc{A}}$ such that $\uppsi(\un)=1$.
A functional in $\mc{A}$ is hermitian if $\upphi=\upphi^{\ast}$, where $\upphi^{\ast}(a)\coloneqq\ov{\upphi(a^{\ast})}$ 
for all $a\in\mc{A}$. 
If $\upphi\in\mr{P}_{\mc{A}}$ 
(respectively $\upphi\in\mr{E}_{\mc{A}}$) then $\upphi$ is positive (respectively a state) w.r.t. 
the definition in \cite{schm} and \cite{fra} i.e. $\upphi(\mc{A}^{\star})\subseteq\R^{+}$, hence $\upphi$ is hermitian 
since \cite[Lemma 12.3]{fra}, in particular $\upphi(\mc{A}_{ob})\subseteq\R$, therefore by \eqref{08111513} we 
obtain $\mf{P}_{\mc{A}}=\mr{P}_{\mc{A}}\cap\mc{C}(\mc{A},\C)$ and $\mf{E}_{\mc{A}}=\mr{E}_{\mc{A}}\cap\mc{C}(\mc{A},\C)$ so,
by linearity and continuity we conclude that 
\begin{equation}
\label{08102057}
\mf{P}_{\mc{A}}=\{\upphi\in\mc{A}^{\ast}\,\vert\,(\forall a\in\mc{A})(\upphi(a^{\ast}a)\geq 0)\}.
\end{equation}
Notice that since \cite{fra} defines to be positive a linear functional 
in $\mc{A}$ if it maps $\mc{A}^{\star}$ into $\R^{+}$, 
then by \eqref{08102057} the set $\mf{P}_{\mc{A}}$ 
(respectively $\mf{E}_{\mc{A}}$) is the set of continuous positive linear functionals 
(respectively continuous states) w.r.t. the definition in \cite{fra}, hence our order structure
in $\mc{A}^{\ast}$ coincides with the one in \cite{fra}. 
If $\mc{A}$ is a locally $C^{\ast}-$algebra $\mc{A}^{\star}$is a closed cone since \cite[Cor. 10.16, Thm. 10.17]{fra}, 
then $\mc{A}^{\star}=\mc{A}^{\natural}=\mc{A}^{+}$ and 
our order structure in $\mc{A}$ equals those in \cite{schm} and \cite{fra}. 
If $\mc{A}$ is a $C^{\ast}-$algebra then $\mr{P}_{\mc{A}}=\mf{P}_{\mc{A}}$ and $\mr{E}_{\mc{A}}=\mf{E}_{\mc{A}}$.
Set
\begin{equation*}
\mr{Ef}(\mc{A})\coloneqq\{a\in\mc{A}\,\vert\,\ze\leq a\leq\un\};
\end{equation*}
the set of effects of $\mc{A}$, it is closed since $\mc{A}^{+}$ it is so, moreover $\Pr(\mc{A})\subset\mr{Ef}(\mc{A})$ 
indeed $\un-p\in\Pr(\mc{A})$ and clearly $\Pr(\mc{A})\subset\mc{A}^{+}$. $\uppsi(e)\in[0,\uppsi(\un)]$ for any 
$\uppsi\in\mr{P}_{\mc{A}}$ and $e\in\mr{Ef}(\mc{A})$ since $\uppsi$ is an order morphism. 
We can define in $\mr{Ef}(\mc{A})$ a partial sum as the restriction
of the sum in $\mc{A}$ in the set of the $(a,b)\in\mr{Ef}(\mc{A})\times\mr{Ef}(\mc{A})$ such that $a+b\in\mr{Ef}(\mc{A})$. 
The domain of the partial sum is not empty since it contains the set of the $(p,q)\in\Pr(\mc{A})\times\Pr(\mc{A})$
such that $p q=\ze$, indeed $p+q\in\Pr(\mc{A})$. 
Define 
\begin{equation}
\label{11181224}
\begin{cases}
\ep_{\mc{A}}:\mc{A}\to\mc{A}^{\mc{A}},\,a\mapsto(b\mapsto aba^{\ast}),
\\
\updelta_{\mc{A}}:\mc{A}\to\mc{A}^{\mc{A}},\,a\mapsto(b\mapsto a^{\ast}ba);
\end{cases}
\end{equation}
clearly $\updelta_{\mc{A}}=\ep_{\mc{A}}\circ(\ast)$, and $\ep_{\mc{A}}(a)$ is a $\ast-$preserving 
continuous linear map for every $a\in\mc{A}$. If the product in $\mc{A}$ is jointly continuous, then $\ep_{\mc{A}}$ or 
simply $\ep$ is continuous. For all $a\in\mc{A}$ we have $\ep(a)(\mc{A}^{+})\subseteq\mc{A}^{+}$, since 
$\ep(a)(\mc{A}^{\natural})\subseteq\mc{A}^{\natural}$ and $\ep(a)$ is continuous, moreover by letting 
$\mc{A}^{op}$ be the opposite multiplicative semigroup of $\mc{A}$ and by letting $\mr{smg}$ be the category of semigroups
and their morphisms, we have 
\begin{equation}
\label{08091352} 
\begin{cases}
\ep_{\mc{A}}\in\mr{Mor_{smg}}(\mc{A},\mr{P}(\mc{A}));
\\
\updelta_{\mc{A}}\in\mr{Mor_{smg}}(\mc{A}^{op},\mr{P}(\mc{A})).
\end{cases}
\end{equation}
In general $\ep(a)(xy)\neq\ep(a)(x)\ep(a)(y)$ unless $a^{\ast}a=\un$ for instance $a\in\mr{U}(\mc{A})$, while in general 
$\ep(a+b)(d)\neq\ep(a)(d)+\ep(b)(d)$ unless $d\in\{a\}^{\prime}\cap\{b\}^{\prime}$ and $ab^{\ast}=\ze$.
\par
If $S$ is any structure richer than the structure of topological spaces, then $\mc{S}$ is a $\mr{top}-$quasi enriched 
category by providing for all $x,y\in\mc{S}$ the set $\mr{Mor}_{\mc{S}}(x,y)$ with the product topology, therefore in case 
$y$ is a uniform space, with the topology of simple convergence. Note that any of the following categories
$\mr{tg}$, $\K-\mr{tls}$, with $\K\in\{\R,\C\}$, $\mr{ptls}$, $\mr{tls}_{\geq}$ and $\mr{tsa}$ are examples
of such a $\mc{S}$. If $\mr{B}$ is a $\mr{top}-$quasi enriched category,
then $\mr{Aut}_{\mr{B}}(x)$ is a possibly trivial group whose product is separately continuous for all $x\in\mr{B}$. 
If $\mr{C}$ is a $\mr{top}-$enriched category, then $\mr{Aut}_{\mr{C}}(x)$ is a topological group for all $x\in\mr{C}$ such 
that the inversion map in $\mr{Aut}_{\mr{C}}(x)$ is continuous.
A topological groupoid is a $\mr{top}-$enriched groupoid (i.e. a category in which all the morphisms are invertible)
with continuous inversion, so a topological group is (identifiable with) a $\mr{top}-$enriched groupoid with a unique 
object whose inversion map is continuous.
\section{Propensity map}
\label{11301512}
Our framework described in \textbf{Postulate \ref{11030928}} and Def. \ref{08081839} is established over the concept of  
propensity map \textbf{Def.\ref{08081546}} interpreted as empirical representation of the propensity. 
The term propensity is used here with a different meaning with respect to the one ascribed usually to it.
We retain propensity a primitive, measuring independent, structural property of any triplet formed by a 
channel of statistical ensembles $J$, a statistical ensemble $\upomega$ and an effect $e$, and that this property admits 
an empirical representation in terms of frequency Rmk.\ref{08101005}.
Here we maintain the standard meaning of statistical ensemble,\footnote{see for instance \cite[p.116]{dgg}
and \cite[p.246]{dl}.} namely an ensemble of identical preparations, providing non-interacting copies of a 
system, called simply preparations, realized by macroscopic apparatuses under well defined and repeatable conditions. 
The propensity map slightly generalizes the usual state-effect duality Rmk. \ref{08101104}.
As an application of our framework we prove in Schr\"{o}dinger picture the generalized Wigner formula for a sequence of 
measuring processes of semiobservables in \textbf{Thm.\ref{11191154}} and physically interpret it in Rmk.\ref{11191905}; 
while we prove it independently for a sequence of von Neumann measuring processes associated with discrete observables 
in Cor.\ref{10211852} and physically interpret it in Rmk.\ref{11181512}.
The semantics developed in Def. \ref{08081839} permits to show the Wigner formula in a more intuitive 
and less involved fashion in Schr\"{o}dinger picture compared with the Heisenberg picture as usually
done for instance in \cite{oza3}. 
\subsection{Channels, devices and operations}
\begin{lemma}
\label{08080010pre}
$\mr{tls}$ is $\imath(\mr{tls})-$quasi enriched, 
with $\imath$ here the forgetful functor from $\mr{tls}$ to $\mr{set}$,
if we provide $\mf{L}(X,Y)$ with the topology of simple convergence for every $X,Y\in\mr{tls}$,
moreover the maps $X\mapsto X_{s}^{\ast}$ and $T\mapsto T^{\dagger}$
determine an element in 
$\mr{Fct}_{\mr{tls}}(\mr{tls}^{op},\mr{tls})$\footnote{where we let 
$\mr{Fct}_{\mr{tls}}(\mr{tls}^{op},\mr{tls})$ denote $\mr{Fct}_{\imath(\mr{tls})}(\mr{tls}^{op},\mr{tls})$.}.
In other words if $X,Y\in\mr{tls}$, then
\begin{enumerate}
\item
$\dagger\in\mf{L}(\mf{L}_{s}(X,Y),\mf{L}_{s}(Y_{s}^{\ast}, X_{s}^{\ast}))$ and $T^{\dagger}(\upomega)=\upomega\circ T$
for every $T\in\mf{L}(X,Y)$ and $\upomega\in Y^{\ast}$;
\label{08080010prest1}
\item
Let $Z\in\mr{tls}$, $S\in\mf{L}(Y,Z)$ and $T\in\mf{L}(X,Y)$ thus, $(S\circ T)^{\dagger}=T^{\dagger}\circ S^{\dagger}$.
\label{08080010prest2}
\end{enumerate}
\end{lemma}
\begin{proof}
In this proof if $W\subseteq B^{A}$, then for every $a\in A$ we let $\mr{ev}_{a}^{W}:W\to B$, $g\mapsto g(a)$ and 
let $\mr{ev}_{a}$ denote $\mr{ev}_{a}^{W}$ whenever this is not cause of confusion.
In addition if $\mf{B}\in\mr{Bf}(A)$ i.e. $\mf{B}$ is a base of a filter of a set $A$, then $\mf{F}_{\mf{B}}^{A}$ or simply 
$\mf{F}_{\mf{B}}$ is the filter of $A$ generated by $\mf{B}$.
St.\ref{08080010prest2} is trivial so, let us prove st.\eqref{08080010prest1}.
Let $S\in\mf{L}(X,Y)$ clearly $S^{\dagger}(Y^{\ast})\subseteq X^{\ast}$ and $S^{\dagger}$ is linear. 
Next let $\mf{N}$ be a filter of $Y^{\ast}$ and $\upphi\in\lim\mf{N}$ with respect to $Y_{s}^{\ast}$ thus, 
\begin{equation}
\label{08211237}
(\forall y\in Y)(\upphi(y)=\lim_{\mf{N}}\mr{ev}_{y}).
\end{equation}
First we claim that $S^{\dagger}(\upphi)\in\lim_{\mf{N}}S^{\dagger}$ with respect to $X_{s}^{\ast}$.
This is equivalent to 
$S^{\dagger}(\upphi)\in\lim\mf{F}_{S^{\dagger}(\mf{N})}$ with respect to $X_{s}^{\ast}$ 
which is equivalent to state that for all $x\in X$
\begin{equation*}
\begin{aligned}
\upphi(Sx)
&=
\lim_{\mf{F}_{S^{\dagger}(\mf{N})}}\mr{ev}_{x}
\\
&=
\lim
\mf{F}_{\mr{ev}_{x}(\mf{F}_{S^{\dagger}(\mf{N})})}
\\
&=
\lim
\mf{F}_{\mr{ev}_{x}(S^{\dagger}(\mf{N}))}
\\
&=
\lim
\mf{F}_{\mr{ev}_{S(x)}(\mf{N})}
\\
&=
\lim_{\mf{N}}\mr{ev}_{S(x)};
\end{aligned}
\end{equation*}
where in the third equality we used the fact 
\begin{equation}
\label{08210826}
f(\mf{F}_{\mf{B}})\simeq_{A}f(\mf{B});
\end{equation}
with $f:Z\to A$, $\mf{B}\in\mr{Bf}(Z)$ and $\mf{C}\simeq_{A}\mf{D}$ iff by definition 
$\mf{C},\mf{D}\in\mr{Bf}(A)\wedge\mf{F}_{\mf{C}}=\mf{F}_{\mf{D}}$.
But the above equality follows by \eqref{08211237} so, our first claim follows 
and therefore $S^{\dagger}\in\mf{L}(Y_{s}^{\ast},X_{s}^{\ast})$ which proves that st.\eqref{08080010prest1} is well-set.
Next we claim to show that the map $\dagger$ is continuous with respect to the topologies of simple convergence.
In order to do that let $\mf{G}$ be a filter of $\mf{L}(X,Y)$ and let $T\in\lim\mf{G}$ with respect to $\mf{L}_{s}(X,Y)$, 
our second claim to be proven is that $T^{\dagger}\in\lim_{\mf{G}}\dagger$ with respect to $\mf{L}_{s}(Y_{s}^{\ast},X_{s}^{\ast})$.
Now 
$T\in\lim\mf{G}$ with respect to $\mf{L}_{s}(X,Y)$ 
iff 
$(\forall x\in X)(T(x)\in\lim_{\mf{G}}\mr{ev}_{x})$ with respect to $Y$
i.e.
\begin{equation}
\label{08211333}
(\forall x\in X)(T(x)\in\lim\mf{F}_{\mr{ev}_{x}(\mf{G})}\text{ with respect to $Y$}).
\end{equation}
Instead 
$T^{\dagger}\in\lim_{\mf{G}}\dagger$ with respect to $\mf{L}_{s}(Y_{s}^{\ast},X_{s}^{\ast})$ 
iff
$T^{\dagger}\in\lim\mf{F}_{\dagger(\mf{G})}$ with respect to $\mf{L}_{s}(Y_{s}^{\ast},X_{s}^{\ast})$ 
namely
$(\forall\uppsi\in Y^{\ast})(T^{\dagger}(\uppsi)\in\lim_{\mf{F}_{\dagger(\mf{G})}}\mr{ev}_{\uppsi})$ with respect to $X_{s}^{\ast}$ 
which is equivalent by considering \eqref{08210826} to state that
$(\forall\uppsi\in Y^{\ast})(T^{\dagger}(\uppsi)\in\lim\mf{F}_{\mr{ev}_{\uppsi}(\dagger(\mf{G}))})$ with respect to $X_{s}^{\ast}$ 
namely
$(\forall x\in X)(\forall\uppsi\in Y^{\ast})(\uppsi(Tx)=\lim_{\mf{F}_{\mr{ev}_{\uppsi}(\dagger(\mf{G}))}}\mr{ev}_{x})$ 
with respect to $\C$ 
and by employing \eqref{08210826} this is equivalent to the following one 
$(\forall x\in X)(\forall\uppsi\in Y^{\ast})(\uppsi(Tx)=\lim\mf{F}_{(\mr{ev}_{x}\circ\mr{ev}_{\uppsi}\circ\dagger)(\mf{G})})$
with respect to $\C$ 
but 
$\mr{ev}_{x}\circ\mr{ev}_{\uppsi}\circ\dagger=\uppsi\circ\mr{ev}_{x}$
so, it is equivalent to 
$(\forall x\in X)(\forall\uppsi\in Y^{\ast})(\uppsi(Tx)=\lim\mf{F}_{(\uppsi\circ\mr{ev}_{x})(\mf{G})})$ with respect to $\C$ 
which by \eqref{08210826} is equivalent to
\begin{equation*}
(\forall x\in X)(\forall\uppsi\in Y^{\ast})(\uppsi(Tx)=\lim_{\mf{F}_{\mr{ev}_{x}(\mf{G})}}\uppsi\text{ with respect to $\C$}). 
\end{equation*}
Now the above limit follows since \eqref{08211333} and the continuity of $\uppsi$ therefore, 
our second claim is proven thus, we have shown that the map $\dagger$ is continuous so, st.\eqref{08080010prest1} follows. 
\end{proof}
\begin{lemma}
\label{08080010}
$\mr{tls}_{\geq}$ is $\imath(\mr{ptls})-$quasi enriched, 
with $\imath$ here the forgetful functor from $\mr{ptls}$ to $\mr{set}$,
if we provide $\mf{L}(X,Y)$ with the topology of simple convergence for every $X,Y\in\mr{tls}_{\geq}$,
moreover the maps $X\mapsto X_{s}^{\ast}$ and $T\mapsto T^{\dagger}$, where 
$X_{s}^{\ast}$ is provided by the canonical structure of p.t.l.s., determine an element in 
$\mr{Fct}_{\mr{ptls}}(\mr{tls}_{\geq}^{op},\mr{tls}_{\geq})$\footnote{where we let  
$\mr{Fct}_{\mr{ptls}}(\mr{tls}_{\geq}^{op},\mr{tls}_{\geq})$ denote 
$\mr{Fct}_{\imath(\mr{ptls})}(\mr{tls}_{\geq}^{op},\mr{tls}_{\geq})$.}
such that $\mr{K}(X)^{\dagger}\subseteq\mr{K}(X_{s}^{\ast})$,
for all $X\in\mr{ptls}$. In other words if $X,Y\in\mr{ptls}$, then
\begin{enumerate}
\item
$\dagger\in\mf{L}(\mf{L}_{s}(X,Y),\mf{L}_{s}(Y_{s}^{\ast}, X_{s}^{\ast}))$ and $T^{\dagger}(\upomega)=\upomega\circ T$
for every $T\in\mf{L}(X,Y)$ and $\upomega\in Y^{\ast}$;
\label{08080010st1}
\item
if in addition $Z\in\mr{ptls}$, $S\in\mf{L}(Y,Z)$ and $T\in\mf{L}(X,Y)$, 
then $(S\circ T)^{\dagger}=T^{\dagger}\circ S^{\dagger}$;
\label{08080010st2}
\item
$\mr{P}(X,Y)^{\dagger}\subseteq\mr{P}(Y_{s}^{\ast},X_{s}^{\ast})$, i.e. 
$U^{\dagger}(\mf{P}_{Y})\subseteq\mf{P}_{X}$, for all $U\in\mr{P}(X,Y)$; 
\label{08080010st3}
\item
$\mr{K}(X)^{\dagger}\subseteq\mr{K}(X_{s}^{\ast})$.
\label{08080010st4}
\end{enumerate}
\end{lemma}
\begin{proof}
The first and second statements follow by Lemma \ref{08080010pre}, while the remaining ones are trivial.
\end{proof}
\begin{definition}
\label{08081848}
Let $\mc{A},\mc{B}\in\mr{tsa}$ define
\begin{equation*}
\begin{aligned}
\mr{Q}(\mc{A},\mc{B})&\coloneqq\{T\in\mr{P}(\mc{A},\mc{B})\,\vert\, T(\un)\leq\un\},
\\
\mf{Q}(\mc{A},\mc{B})&\coloneqq\{T\up_{\mr{Ef}(\mc{A})}^{\mr{Ef}(\mc{B})}\,\vert\,T\in\mr{Q}(\mc{A},\mc{B})\},
\end{aligned}
\end{equation*}
set $\mr{Q}(\mc{A})\coloneqq\mr{Q}(\mc{A},\mc{A})$ and $\mf{Q}(\mc{A})\coloneqq\mf{Q}(\mc{A},\mc{A})$. 
\end{definition}
Clearly $\mr{Q}(\mc{A})$ and $\mf{Q}(\mc{A})$ are subsemigroups with identity of $\mr{P}(\mc{A})$.
$\mr{Q}(\mc{A},\mc{B})$ is called the set of devices from $\mc{A}$ to $\mc{B}$ while
$\mf{Q}(\mc{A},\mc{B})$ is called the set of devices from $\mr{Ef}(\mc{A})$ to $\mr{Ef}(\mc{B})$.
Next we introduce our definition of propensity map
\begin{definition}
[\textbf{Propensity map}]
\label{08081546}
Let $\mc{A},\mc{B}\in\mr{tsa}$ define
\begin{equation*}
\boxed{
\begin{aligned}
&\mf{Z}(\mc{A},\mc{B})\coloneqq\{J\in\mr{P}(\mc{A}_{s}^{\ast},\mc{B}_{s}^{\ast})
\,\vert\,(\forall\upphi\in\mf{P}_{\mc{A}})(J(\upphi)(\un)\leq\upphi(\un))\},
\\
\\
\mf{b}_{\mc{A},\mc{B}}:&\mf{Z}(\mc{A},\mc{B})\times\mf{P}_{\mc{A}}^{\natural}\times\mr{Ef}(\mc{B})\to[0,1]
\qquad
(J,\upomega,e)\mapsto\frac{J(\upomega)(e)}{\upomega(\un)}.
\end{aligned}}
\end{equation*}
\end{definition}
Let $\mc{A},\mc{B}\in\mr{tsa}$, $J\in\mf{Z}(\mc{A},\mc{B})$, $\upomega\in\mf{P}_{\mc{A}}^{\natural}$ and 
$e\in\mr{Ef}(\mc{B})$. Thus $J(\upomega)\in\mf{P}_{\mc{B}}$ such that $J(\upomega)(e)\in[0,\upomega(\un)]$, 
indeed $J(\upomega)(e)\in[0,J(\upomega)(\un)]$
since $J(\upomega)\in\mf{P}_{\mc{B}}$, the remaining follows by the definition of $\mf{Z}(\mc{A},\mc{B})$.
Therefore $\mf{b}_{\mc{A},\mc{B}}$ is a well-defined map into $[0,1]$.
We call $\mf{Z}(\mc{A},\mc{B})$ the set of channels from $\mc{A}^{\ast}$ to $\mc{B}^{\ast}$, 
$\mf{P}_{\mc{A}}$ the set of statistical ensembles of $\mc{A}$, $\mr{Ef}(\mc{A})$ the set of effects of $\mc{A}$,
and $\mf{b}_{\mc{A},\mc{B}}$ the propensity map relative to $(\mc{A},\mc{B})$.
We will provide a complete physical interpretation of the above data in Def.\ref{08081839} and Postulate \ref{11030928}.
\begin{definition}
Set $\mf{Z}(\mc{A})\coloneqq\mf{Z}(\mc{A},\mc{A})$, $\mf{b}_{\mc{A}}\coloneqq\mf{b}_{\mc{A},\mc{A}}$ and define 
\begin{equation*}
\mf{p}_{\mc{A}}:\mf{P}_{\mc{A}}^{\natural}\times\mr{Ef}(\mc{A})\to[0,1]
\quad
(\upomega,e)\mapsto\mf{b}_{\mc{A}}(\mr{Id}_{\mc{A}_{s}^{\ast}},\upomega,e).
\end{equation*}
Set $p_{\mc{A}}\coloneqq\mf{p}_{\mc{A}}\up\mf{E}_{\mc{A}}\times\mr{Ef}(\mc{A})$. 
\end{definition}
$\mf{Z}(\mc{A})$ is a subsemigroup with identity of $\mr{P}(\mc{A}_{s}^{\ast})$ and $p_{\mc{A}}$ is the usual state-effect 
duality. Next we define operations and their actions on devices and channels.
\begin{definition}
\label{11010823}
Let $\mc{A}\in\mr{tsa}$, define
\begin{equation*}
\begin{cases}
\Uplambda_{\mc{A}}\coloneqq\{a\in\mc{A}\,\vert\,aa^{\ast}\leq\un\},
\\
\Uptheta_{\mc{A}}\coloneqq\{a\in\mc{A}\,\vert\,a^{\ast}a\leq\un\},
\\
\Upgamma_{\mc{A}}\coloneqq\Uptheta_{\mc{A}}\cap\Uplambda_{\mc{A}};
\end{cases}
\end{equation*}
called the set of operations on $\mc{A}$, the set of co-operations on $\mc{A}$ and the set of equioperations on $\mc{A}$
respectively. Next
\begin{equation*}
\begin{cases}
\upzeta_{\mc{A}}:\Uplambda_{\mc{A}}\to\mr{Q}(\mc{A}),\,a\mapsto\ep_{\mc{A}}(a),
\\
\upgamma_{\mc{A}}:\Uptheta_{\mc{A}}\to\mr{Q}(\mc{A}),\,a\mapsto\updelta_{\mc{A}}(a),
\\
\upzeta_{\mc{A}}^{\dagger}:\Uplambda_{\mc{A}}\to\mf{Z}(\mc{A}),\,a\mapsto\ep_{\mc{A}}(a)^{\dagger},
\\
\upgamma_{\mc{A}}^{\dagger}:\Uptheta_{\mc{A}}\to\mf{Z}(\mc{A}),\,a\mapsto\updelta_{\mc{A}}(a)^{\dagger};
\end{cases}
\end{equation*}
where $\dagger$ is the map of which in Lemma \ref{08080010}\eqref{08080010st1}. Furthermore define 
\begin{equation*}
\begin{cases}
\upzeta_{\mc{A},e}:\Uplambda_{\mc{A}}\to\mf{Q}(\mc{A}),\,a\mapsto\ep_{\mc{A}}(a)\up_{\mr{Ef}(\mc{A})}^{\mr{Ef}(\mc{A})};
\\
\\
\upgamma_{\mc{A},e}:\Uptheta_{\mc{A}}\to\mf{Q}(\mc{A}),\,a\mapsto\updelta_{\mc{A}}(a)\up_{\mr{Ef}(\mc{A})}^{\mr{Ef}(\mc{A})}.
\end{cases}
\end{equation*}
\end{definition}
\begin{lemma}
\label{08191955}
$\Uplambda_{\mc{A}}$  and $\Uptheta_{\mc{A}}$ are subsemigroups of $\mc{A}$, while $\Upgamma_{\mc{A}}$
is a semigroup with involution.
\end{lemma}
\begin{proof}
Let $a,b\in\Uptheta_{\mc{A}}$ thus, $(ab)^{\ast}ab=\updelta_{\mc{A}}(b)(a^{\ast}a)$, but $a^{\ast}a\leq\un$ while 
$\updelta_{\mc{A}}(b)$ is order preserving since it is positive by \eqref{08091352}, so  
$\updelta_{\mc{A}}(b)(a^{\ast}a)\leq\updelta_{\mc{A}}(b)(\un)=b^{\ast}b\leq\un$, so $(ab)^{\ast}ab\leq\un$ namely 
$ab\in\Uptheta_{\mc{A}}$ proving that $\Uptheta_{\mc{A}}$ is a semigroup. Next let $x,y\in\Uplambda_{\mc{A}}$ thus, 
$xy(xy)^{\ast}=\updelta_{\mc{A}}(x^{\ast})(yy^{\ast})\leq xx^{\ast}\leq\un$ proving that $\Uplambda_{\mc{A}}$ is a semigroup.
\end{proof}
\begin{convention}
By abuse of language we let $\mr{Q}(\mc{A})$, $\mf{Q}(\mc{A})$, $\mf{Z}(\mc{A})$, $\Uplambda_{\mc{A}}$,
$\Uptheta_{\mc{A}}$ and $\Upgamma_{\mc{A}}$ denote also the corresponding semigroups.
\end{convention}
\begin{remark}
\label{11030533}
Let $\mc{A}\in\mr{tsa}$ thus, $\mc{A}^{+}=(\mc{A}^{op})^{+}$, therefore we deduce that 
$\Uplambda_{\mc{A}}=\Uptheta_{\mc{A}^{op}}^{op}$ and 
$\upzeta_{\mc{A}}=\upgamma_{\mc{A}^{op}}$. These facts will be used mostly without
additional mention in order to eliminate redundancies in proofs and definitions.
\end{remark}
For our physical interpretation given in Def. \ref{08081839}, 
we require the algebraic information stated in the following two propositions
\begin{proposition}
\label{08090901}
Let $\circ$ be the map composition then any of the following set of data determines uniquely a category
\begin{enumerate}
\item
$\lr{\mr{Obj}(\mr{tsa})}{\{\mf{Z}(\mc{A},\mc{B})\,\vert\,\mc{A},\mc{B}\in\mr{tsa}\},\circ}$; 
\label{08090901st1}
\item
$\lr{\mr{Obj}(\mr{tsa})}{\{\mr{Q}(\mc{A},\mc{B})\,\vert\,\mc{A},\mc{B}\in\mr{tsa}\},\circ}$; 
\label{08090901st3}
\item
$\lr{\{\mr{Ef}(\mc{A})\,\vert\,\mc{A}\in\mr{tsa}\}}{\{\mf{Q}(\mc{A},\mc{B})\,\vert\,\mc{A},\mc{B}\in\mr{tsa}\},\circ}$. 
\label{08090901st2}
\end{enumerate}
\end{proposition}
\begin{proof}
St.\eqref{08090901st1} amounts to prove that $K\circ J\in\mf{Z}(\mc{A},\mc{C})$ 
for any $J\in\mf{Z}(\mc{A},\mc{B})$, $K\in\mf{Z}(\mc{B},\mc{C})$, which follows since $\mr{ptls}$ is a category and since 
for all $\phi\in\mf{P}_{\mc{A}}$ we have $K(J(\phi))(\un)\leq J(\phi)(\un)\leq\un$.
St.\eqref{08090901st3} follows since \eqref{08111513} while st.\eqref{08090901st2} by st.\eqref{08090901st3}. 
\end{proof}
\begin{proposition}
\label{08091056}
Let $\mc{A},\mc{B}\in\mr{tsa}$, then
\begin{enumerate}
\item
$\mr{K}(\mc{A})\subseteq\mr{Q}(\mc{A})$;
\label{08091056st1}
\item
$\upzeta_{\mc{A}}\in\mr{Mor_{smg}}(\Uplambda_{\mc{A}},\mr{Q}(\mc{A}))$; 
\label{11030548}
\item
$\upgamma_{\mc{A}}\in\mr{Mor_{smg}}(\Uptheta_{\mc{A}}^{op},\mr{Q}(\mc{A}))$; 
\label{11220848}
\item
$\upzeta_{\mc{A},e}\in\mr{Mor_{smg}}(\Uplambda_{\mc{A}},\mf{Q}(\mc{A}))$; 
\label{11030549}
\item
$\upgamma_{\mc{A},e}\in\mr{Mor_{smg}}(\Uptheta_{\mc{A}}^{op},\mf{Q}(\mc{A}))$; 
\label{08091056st2}
\item
$\mr{Q}(\mc{B},\mc{A})^{\dagger}\subseteq\mf{Z}(\mc{A},\mc{B})$;
\label{08091056st3}
\item
$\upzeta_{\mc{A}}^{\dagger}\in\mr{Mor_{smg}}(\Uplambda_{\mc{A}}^{op},\mf{Z}(\mc{A}))$;
\label{08091056last}
\item
$\upgamma_{\mc{A}}^{\dagger}\in\mr{Mor_{smg}}(\Uptheta_{\mc{A}},\mf{Z}(\mc{A}))$.
\label{08091056st4}
\end{enumerate}
\end{proposition}
\begin{proof}
St.\eqref{11030548} and st.\eqref{11220848} follow since \eqref{08091352}.
St.\eqref{08091056st2} is well-set, $\mf{Q}(\mc{A})$ being a semigroup since Prp.\ref{08090901}\eqref{08090901st2}, and 
it follows since $\upgamma_{\mc{A}}(\Uptheta_{\mc{A}})\subseteq\mr{Q}_{\mc{A}}$ and \eqref{08091352}. Similar proof shows
st.\eqref{11030549}.
St.\eqref{08091056st3} 
follows since Lemma \ref{08080010}\eqref{08080010st3} and since any element in $\mf{P}_{\mc{A}}$ is an order morphism 
being linear. St.\eqref{08091056st4} is well-set, $\mf{Z}(\mc{A})$ being a semigroup since 
Prp.\ref{08090901}\eqref{08090901st1}, and it follows since st.(\ref{11220848},\ref{08091056st3}) and 
Lemma \ref{08080010}\eqref{08080010st2}.
St.\eqref{08091056last} is well-set and it follows since st.(\ref{11030548},\ref{08091056st3}) and 
Lemma \ref{08080010}\eqref{08080010st2}.
\end{proof}
We have $(\cdot)^{\ast}\in\mr{Mor}_{\mr{tg}}(\mr{Aut}_{\mr{ptls}}(X),\mr{Aut}_{\mr{ptls}}(X_{s}^{\ast}))$ for any $X\in\mr{ptls}$, 
since Lemma \ref{08080010} so, 
$(\cdot)^{\ast}\in \mr{Mor}_{\mr{tg}}(\mr{Aut}_{\mr{tsa}}(\mc{A}),\mr{Aut}_{\mr{ptls}}(\mc{A}_{s}^{\ast}))$, 
for any $\mc{A}\in\mr{tsa}$. $\mr{Q}(\mc{A},\mc{B})$ is closed in $\mf{L}_{s}(\mc{A},\mc{B})$ 
since $\mr{P}(\mc{A},\mc{B})$ is so. $\Uptheta_{\mc{A}}$ is a subsemigroup of $\mc{A}$ and if the product in $\mc{A}$ is 
jointly continuous then $\Uptheta_{\mc{A}}$ is closed. $\updelta_{\mc{A}}^{\dagger}:\mc{A}\to\mr{P}(\mc{A}_{s}^{\ast})$ is a 
semigroup morphism since \eqref{08091352} and Lemma \ref{08080010}.
\begin{remark}
[\textbf{Ozawa channel map associated with a measuring process}]
\label{10311634}
\emph{Semiobservables and observables} on a Hilbert space $\mf{H}$ are defined in \cite[section 2]{oza1}, while 
\emph{discrete} semiobservables are defined in \cite[p.85]{oza1}. Here we note that an observable in this context is a 
spectral measure in the Banach space $\mf{H}$ in the sense of Dunford-Schwartz. 
That follows since \cite[Lemma 4.3]{schumb} in whose proof it is employed without proving it the following simple fact. 
If $p,q\in\Pr(\mc{A})$ such that $p+q\in\Pr(\mc{A})$, then $pq=\ze$. Indeed $p+q\in\Pr(\mc{A})$ implies 
$p+q=(p+q)^{2}=p^{2}+q^{2}+pq+qp=p+q+pq+qp$ so, $pq=-qp$ therefore $pqp=-qp$ but $pqp\in\mc{A}_{ob}$ since 
$p,q\in\mc{A}_{ob}$ by definition so, $-qp=-pq$ namely $qp=pq$, but then $pq=-qp$ implies $pq=-pq$ or $2pq=\ze$ and then 
$pq=\ze$.
\par
\emph{We can then associate with any observable $V$ on $\mf{H}$ with value space $(\R,\mf{B}(\R))$ a possibly unbounded 
selfadjoint operator $o^{V}$ in $\mf{H}$ defined as the one whose resolution of the identity equals $V$}. 
In particular since the spectrum of any selfadjoint operator equals the support of its resolution of the identity 
\cite[Prp.5.10]{schumb}, we obtain that \emph{$V$ is discrete if and only if the spectrum of $o^{V}$ is discrete}.
\par
According to \cite[Thm. 5.1]{oza1} or \cite[(5.7)]{oza2} with any \emph{measuring process} $\mf{x}$ \cite[Def.3.1]{oza1}
of a semiobservable $X$ on $\mf{H}$ with value space $(\Lambda,\mf{B})$ 
a $X$-compatible CP instrument $\mc{I}_{\mf{x}}$ on $\mf{L}(\mf{H})$ \cite[section 4]{oza1} remains associated and 
given by \cite[(5.2)]{oza1} or equivalently by \cite[(5.6) and (3.11)]{oza1}. 
Measuring processes of any observable $V$ might be referred as measuring processes of the corresponding operator $o^{V}$.
\par
Set $\mc{M}=\lr{\mf{L}(\mf{H})}{\sigma(\mf{L}(\mf{H}),\mf{L}(\mf{H})_{\ast})}$,
$\mc{N}=\lr{\mf{L}(\mf{H}\otimes\mf{K})}{\sigma(\mf{L}(\mf{H}\otimes\mf{K}),\mf{L}(\mf{H}\otimes\mf{K})_{\ast})}$,
define for every $b\in\mf{L}(\mf{K})$ the map $R_{b}^{\otimes}:\mc{M}\to\mc{N}$, $a\mapsto a\otimes b$, and
$\mf{I}_{\mf{x}}\coloneqq\mc{I}_{\mf{x}}^{\dagger}$. 
Thus, if $\mf{x}=\lr{\mf{K},Y}{\sigma,U}$, then since \cite[(5.2)]{oza1} we deduce that 
\begin{equation}
\label{10311837}
(\forall B\in\mf{B})(\mf{I}_{\mf{x}}(B)=(R_{Y(B)}^{\otimes})^{\dagger}\circ\upgamma_{\mc{N}}^{\dagger}(U)\circ E_{\sigma}^{\dagger});
\end{equation}
where 
\begin{equation*}
\begin{cases}
E_{\sigma}^{\dagger}:\mf{L}(\mf{H})_{\ast}\to\mf{L}(\mf{H}\otimes\mf{K})_{\ast},
\\
\omega_{\xi}\mapsto\omega_{\xi\otimes\sigma},\,\xi\text{ trace class operator on }\mf{H}.
\end{cases}
\end{equation*}
Notice that 
\begin{equation}
\label{08020326}
\mf{L}(\mf{K})\ni p=p^{\ast}p\Rightarrow R_{p}^{\otimes}
=
\upzeta_{\mc{N}}(\un_{\mf{H}}\otimes p)
\circ
R_{\un_{\mf{K}}}^{\otimes}
\end{equation}
therefore 
\begin{equation}
\label{10311837b}
(\forall B\in\mf{B})(
\mf{I}_{\mf{x}}(B)=
(R_{\un_{\mf{K}}}^{\otimes})^{\dagger}
\circ
\upzeta_{\mc{N}}^{\dagger}(\un_{\mf{H}}\otimes Y(B))
\circ
\upgamma_{\mc{N}}^{\dagger}(U)
\circ 
E_{\sigma}^{\dagger}).
\end{equation}
Next for every $\rho$ trace class operator on $\mf{H}$ we have by \eqref{10311837} 
\begin{equation}
\label{11191127}
(\forall a\in\mf{L}(\mf{H}))
(\mf{I}_{\mf{x}}(B)(\omega_{\rho})a=\mr{Tr}((\rho\otimes\sigma)U^{\ast}(a\otimes Y(B))U)).
\end{equation}
Furthermore for every $B\in\mf{B}$ define the endomorphism $\mc{Y}_{\mf{x}}(B)$ of the linear space of 
trace class operators on $\mf{H}$ such that for every $\rho$ trace class operator on $\mf{H}$ we have 
\begin{equation*}
\mc{Y}_{\mf{x}}(B)\rho\coloneqq E_{\mf{K}}(\ep_{\mc{N}}(U)(\rho\otimes\sigma)R_{Y(B)}^{\otimes}(\un_{\mf{H}}));
\end{equation*}
where $E_{\mf{K}}$ is the partial trace over $\mf{K}$ \cite[section 2]{oza1}. Thus, \cite[(5.3)]{oza1} yields
\begin{equation}
\label{11191108}
\mf{I}_{\mf{x}}(B)\omega_{\rho}=\omega_{\mc{Y}_{\mf{x}}(B)\rho}.
\end{equation}
We call \emph{$\mf{I}_{\mf{x}}$ the channel map associated with the measurement process $\mf{x}$}.
By \eqref{11191127} and \cite[(3.1)]{oza1} immediately we see that $\mc{I}_{\mf{x}}$ is $X$-compatible namely
\begin{equation}
\label{10311146}
(\forall B\in\mf{B})(\forall\uppsi\in\mf{L}(\mf{H})_{\ast})
(\mf{I}_{\mf{x}}(B)(\uppsi)(\un)=\uppsi(X(B))).
\end{equation}
\end{remark}
\begin{remark}
[Cycon-Hellwig conditional expectation]
\cite[Thm.3.2]{oza2} associates any couple $(\mc{I},\upphi)$ where $\mc{I}$ is an instrument and 
$\upphi$ is a normal state over the same von Neumann algebra, with a Cycon-Hellwig conditional expectation say 
$\mc{E}_{(\mc{I},\upphi)}$ such that 
\begin{equation*}
(\forall B\in\mf{B}(\Lambda))(\mc{E}_{(\mc{I},\upphi)}^{\dagger}(B)\upphi=\mc{I}^{\dagger}(B)\upphi).
\end{equation*}
In particular for every measuring process $\mf{x}$ of a semiobservable on $\mf{H}$ and every normal state $\upphi$ on 
$\mf{H}$, there exists a Cycon-Hellwig conditional expectation 
$\mc{E}_{(\mf{x},\upphi)}\coloneqq\mc{E}_{(\mc{I}_{\mf{x}},\upphi)}$ such that 
$\mc{E}_{(\mf{x},\upphi)}^{\dagger}(B)\upphi=\mc{I}_{\mf{x}}(B)\upphi$ for all $B\in\mf{B}$.
\end{remark}
Next we introduce the physical interpretation of the previous data.
In the following definition $\un$ is the unit of the unital algebra $\mc{A}$,
while $\un_{\mc{A}}$ is the identity map on $\mc{A}$ in agreement with section \ref{not1}.
\begin{definition}
[\textbf{Semantics}]
\label{08081839}
We call $(\mf{R},\mf{r},\mf{D},\mf{d},\mf{C},\mf{c},\mf{c}_{2},\mf{T},\mf{t},\mf{E},\mf{e},\mf{O},\mf{o})$ 
a semantics if it satisfies the following properties: For any $\mc{A},\mc{B},\mc{C}\in\mr{tsa}$, 
every $J\in\mf{Z}(\mc{A},\mc{B})$, $K\in\mf{Z}(\mc{B},\mc{C})$, 
$J_{1},J_{2}\in\mf{Z}(\mc{A},\mc{B})$ such that $J_{1}+J_{2}\in\mf{Z}(\mc{A},\mc{B})$ 
$T\in\mr{Q}(\mc{A},\mc{B})$, $S\in\mr{Q}(\mc{B},\mc{C})$, $T_{1},T_{2}\in\mr{Q}(\mc{A},\mc{B})$ such that 
$T_{1}+T_{2}\in\mr{Q}(\mc{A},\mc{B})$, $y\in\mc{A}_{ob}$, $e,f\in\mr{Ef}(\mc{A})$ such that 
$e+f\in\mr{Ef}(\mc{A})$, $\uppsi\in\mf{P}_{\mc{A}}$, $\upxi\in\mf{P}_{\mc{B}}$, $c,d\in\Uplambda_{\mc{A}}$, and 
$a,b\in\Uplambda_{\mc{A}}$ such that $a+b\in\Uplambda_{\mc{A}}$, $z\in\Upgamma_{\mc{A}}$, every Hilbert space $\mf{H}$, 
every semiobservable $X$ on $\mf{H}$ with value space $(\Omega,\mf{B})$, every measuring process $\mf{x}$ of $X$ and every 
$B\in\mf{B}$ we have that 
\begin{enumerate}
\item Operations
\begin{enumerate}
\item
$\mf{R}(c)=$ the operation of filtering by $\mf{r}(c)$, 
\label{08081839o}
\item
$\mf{r}(c)=$ executing the measurement of $\mf{O}(cc^{\ast})$ and then selecting the subensemble of those preparations 
such that the value $1$ is detected,
\label{12221404}
\item
$\mf{R}(d c)\equiv$ $\mf{R}(d)$ following $\mf{R}(c)$,
\item
$\mf{R}(a+b)\equiv$ $\mf{R}(a)$ in \textbf{concealed} alternative to $\mf{R}(b)$,
\item
$\mf{r}(X(B)^{\frac{1}{2}})=$ 
the selection of the subensemble of those preparations such that a value in $B$ is detected,
conditioned by the execution of the measurement of $\mf{O}(X)$.
\label{08081839mpe}
\end{enumerate}
\item Equioperations
\begin{enumerate}
\item
$\mf{r}(z^{\ast})=$ the reverse of $\mf{r}(z)$,
\end{enumerate}
\item Devices
\begin{enumerate}
\item
$\mf{D}(T)=$ the device $\mf{d}(T)$,
\label{22121423c}
\item
$\mf{D}(T)\equiv$ the reference frame transformation $\mf{d}(T)$,
\label{03041351}
\item
$\mf{d}(\un_{\mc{A}})=$ producing no variations,
\item
$\mf{D}(S\circ T)\equiv$ $\mf{D}(S)$ following $\mf{D}(T)$, 
\item
$\mf{D}(S\circ T)\equiv$ $\mf{D}(T)$ followed by $\mf{D}(S)$,
\item
$\mf{D}(T_{1}+T_{2})=$ $\mf{D}(T_{1})$ in \textbf{detected} alternative to $\mf{D}(T_{2})$, 
\item
$\mf{d}(\upzeta_{\mc{A}}(c))=$ implementing $\mf{R}(c)$.
\label{22121423d}
\end{enumerate}
\item Channels
\begin{enumerate}
\item 
$\mf{C}(J)=$ the channel $\mf{c}(J)$, 
\label{22121423a}
\item
$\mf{c}(J)=$ selecting the subensemble of those preparations that are transmitted by $\mf{c}_{2}(J)$,
\label{06281007}
\item
$\mf{c}_{2}(\mf{I}_{\mf{x}}(B))=$ executing the measurement of $\mf{O}(X)$ and then selecting the subensemble of those 
preparations such that a value in $B$ is detected,
\label{08081839mpc}
\item
$\mf{C}(K\circ J)\equiv$ $\mf{C}(K)$ following $\mf{C}(J)$,
\label{10261906}
\item
$\mf{C}(K\circ J)\equiv$ $\mf{C}(J)$ followed by $\mf{C}(K)$,
\item
$\mf{C}(J_{1}+J_{2})=$ $\mf{C}(J_{1})$ in \textbf{detected} alternative to $\mf{C}(J_{2})$, 
\item
$\mf{c}(T^{\dagger})\equiv$ induced by $\mf{D}(T)$;
\label{22121423b}
\end{enumerate}
\item Statistical ensembles
\begin{enumerate}
\item
$\mf{T}(\uppsi)=$ the statistical ensemble $\mf{t}(\uppsi)$, 
\item
$\mf{t}(J(\uppsi))\equiv$ resulting next $\mf{C}(J)$ applies to $\mf{T}(\uppsi)$;
\item
$\mf{t}(T^{\dagger}(\upxi))\equiv$ that is detected in any reference frame $\bullet$ as $\mf{T}(\upxi)$ is detected in the 
reference frame obtained by applying $\mf{D}(T)$ to $\bullet$,
\label{03041352}
\item
If $\uppsi\in\mf{P}_{\mc{A}}^{\natural}$, then $\mf{t}(\uppsi(\un)^{-1}\uppsi)=\mf{t}(\uppsi)$;
\end{enumerate}
\item Effects
\begin{enumerate}
\item 
$\mf{E}(e)\equiv$ the effect $\mf{e}(e)$,
\item
$\mf{e}(\ze)=$ of selecting no outputs,
\item
$\mf{e}(\un)=$ of selecting one output,
\item
$\mf{e}(cc^{\ast})=$ produced by $\mf{R}(c)$,
\label{08081839oe}
\item 
$\mf{E}(e+f)\equiv$ $\mf{E}(e)$ in detected alternative to $\mf{E}(f)$,
\item
$\mf{e}(T(e))\equiv$ resulting next $\mf{D}(T)$ applies to $\mf{E}(e)$;
\end{enumerate}
\item Observables and semiobservables
\begin{enumerate}
\item
$\mf{O}(X)=$ the semiobservable $\mf{o}(X)$,
\item
$\mf{O}(y)=$ the observable $\mf{o}(y)$,
\item
$\mf{o}(\un)=$ proportional to the number of preparations
\footnote{this is the interpretation of Haag and Kastler in \cite{haks}.},
\item
$\mf{o}(\un)\equiv$ strength\footnote{with the same meaning of strength of a beam discussed in \cite{dl}.},
\item
$\mf{o}(T(y))\equiv$ resulting next $\mf{D}(T)$ applies to $\mf{O}(y)$.
\end{enumerate}
\end{enumerate}
\end{definition}
\begin{postulate}
\label{11030928}
Let $(\mf{R},\mf{r},\mf{D},\mf{d},\mf{C},\mf{c},\mf{c}_{2},\mf{T},\mf{t},\mf{E},\mf{e},\mf{O},\mf{o})$ be a semantics,
$\mc{A},\mc{B}\in\mr{tsa}$, $\uppsi\in\mf{P}_{\mc{A}}$, $y\in\mc{A}_{ob}$, $J\in\mf{Z}(\mc{A},\mc{B})$, 
$\upomega\in\mf{P}_{\mc{A}}^{\natural}$ and $e\in\mr{Ef}(\mc{B})$. We postulate that
\begin{itemize}
\item
$\mf{b}_{\mc{A},\mc{B}}(J,\upomega,e)$
equals 
\textbf{the empirical representation of the propensity conditioned by $\mf{T}(\upomega)$ to detect $\mf{E}(e)$ when tested
on $\mf{T}(J(\upomega))$};
\end{itemize}
and 
\begin{itemize}
\item
$\uppsi(y)$ equals \textbf{the total value of $\mf{O}(y)$ in $\mf{T}(\uppsi)$}.
\end{itemize}
\end{postulate}
\begin{convention}
When using the above semantics we convein to replace 
``implementing the operation of filtering through the selection of'' by ``selecting''
and replace 
``produced by the operation of filtering through the selection of'' by ``produced by selecting''.
\end{convention}
\begin{definition}
\label{11061008}
Let $(\mf{R},\mf{r},\mf{D},\mf{d},\mf{C},\mf{c},\mf{c}_{2},\mf{T},\mf{t},\mf{E},\mf{e},\mf{O},\mf{o})$ be a semantics, 
$\mc{A}\in\mr{tsa}$, $y\in\mc{A}_{ob}$ and $\upomega\in\mf{P}_{\mc{A}}^{\natural}$. 
We let 
\begin{itemize}
\item $\upomega(\un)^{-1}\upomega(y)$ be \textbf{the expectation value of $\mf{O}(y)$ in $\mf{T}(\upomega)$}.
\end{itemize}
\end{definition}
\begin{remark}
\label{03111723}
Notice that Def.\ref{08081839}\eqref{03041352} interprets the trivial fact 
$T^{\dagger}(\upxi)=\upxi\circ T$ by understanding $T$ in $\upxi\circ T$ in terms of 
Def.\ref{08081839}\eqref{03041351}. In particular if $T^{\dagger}\vert_{\mf{P}_{\mc{B}}}^{\mf{P}_{\mc{A}}}$ admits a right inverse say 
$V$, for instance whenever $T$ admits a left inverse, then for every $\eta\in\mf{P}_{\mc{A}}$ we have $\eta=T^{\dagger}(V\eta)$ 
namely: The statistical ensemble $\mf{t}(\eta)$ equals the statistical ensemble that is detected in any reference frame 
$\bullet$ as the statistical ensemble $\mf{t}(V\eta)$ is detected in the reference frame obtained by applying the reference 
frame transformation $\mf{d}(T)$ to $\bullet$.
\end{remark}
\begin{remark}
Def.\ref{08081839}\eqref{08081839mpc} is consistent with the physical interpretation provided by Davies and Lewis in 
\cite[p.246]{dl} of the action of an instrument over a statistical ensemble, called in \cite[p.246]{dl} a beam of 
non-interacting copies of a system. A direct generalization of the position in \cite[p.246]{dl} is provided in 
Def.\ref{08081839}\eqref{06281007}. Our postulate \ref{11030928} regarding $\mf{b}_{\mc{A},\mc{B}}(J,\upomega,e)$
widely extends the positions in \cite[(5)p.850]{haks}, \cite[p.246]{dl}, and \cite[p.214]{ed} put in order of generality. 
Indeed in the very special case in which $\mc{A}$ is a $C^{\ast}$-algebra, $\mc{A}=\mc{B}$ and $e=\un$, then our 
$\mf{b}_{\mc{A},\mc{B}}(J,\upomega,e)$ equals the trasmission probability as defined in \cite[p.214]{ed} which is meant as a direct 
generalization of the probability furnished in \cite[p.246]{dl} in case $J=\mf{I}_{\mf{x}}(B)$ and also a generalization of the 
transmission probability provided in \cite[(5)p.850]{haks} where $J=\upzeta_{\mc{A}}^{\dagger}(c)$  and $c\in\Uplambda_{\mc{A}}$.
\end{remark}
\begin{remark}
[\textbf{Different semantics between the channels $\mf{I}_{\mf{x}}(B)$ and $\upzeta_{\mc{M}}^{\dagger}(X(B)^{\frac{1}{2}})$}]
\label{11221252}
Having the relationship between the channel $\mf{I}_{\mf{x}}(B)$ and the operation $X(B)$ 
already established in \eqref{10311146} and investigated in subsubsection \ref{10311154}, instead here 
by analyzing Def.\ref{08081839}\eqref{08081839mpe} and Def.\ref{08081839}\eqref{08081839mpc} we remark the difference 
between the channel $\mf{I}_{\mf{x}}(B)$ and the channel induced by the operation $X(B)$. 
More specifically and by letting $\mc{M}$ and $\mc{N}$ as defined in Rmk.\ref{10311634}, Def.\ref{08081839}\eqref{08081839mpe} 
conveys the fact that the channel $\upzeta_{\mc{M}}^{\dagger}(X(B)^{\frac{1}{2}})$ only selects (without execution of the measurement 
of $X$) the subensemble of those preparations in which a value in $B$ is detected, implicitly assuming that 
the measurement of the semiobservable $X$ has been previously executed. More precisely 
Def.\ref{08081839}(\ref{22121423a},\ref{22121423b},\ref{22121423c},\ref{22121423d},\ref{08081839o},\ref{08081839mpe})
\begin{itemize}
\item
$\mf{C}(\upzeta_{\mc{M}}^{\dagger}(X(B)^{\frac{1}{2}}))$
=
the channel induced by the device implementing the operation of filtering through 
the selection of the subensemble of those preparations such that a value in $B$ is detected,
conditioned by the execution of the measurement of the semiobservable $\mf{o}(X)$.
\end{itemize}
Instead Def.\ref{08081839}\eqref{08081839mpc} by which we obtain that 
\begin{itemize}
\item
$\mf{C}(\mf{I}_{\mf{x}}(B))$
=
the channel executing the measurement of the semiobservable $\mf{o}(X)$ and then 
selecting the subensemble of those preparations such that a value in $B$ is detected, 
\end{itemize}
takes into account the decomposition \eqref{10311837b}, specifically it conveys 
the fact that the channel $\mf{I}_{\mf{x}}(B)$ is composed by the channel $\upgamma_{\mc{N}}^{\dagger}(U)$ followed by the channel 
$\upzeta_{\mc{N}}^{\dagger}(\un_{\mf{H}}\otimes Y(B))$. 
The former is interpreted as executing the measurement of the semiobservable $X$ being the time translation of the system 
$I+II$ corresponding to the time interval during which the measurement of $X$ as the interaction between $I$ and $II$ takes 
place, where $I$ is the measured system while $II$ is the system of the measuring apparatus. Instead acccording to 
Def.\ref{08081839}\eqref{08081839mpe} and as above said the latter is interpreted as selecting the subensemble of those 
preparations such that a value in $B$ is detected conditioned by the execution of the measurement of the semiobservable $X$. 
The fact that in general the channel $\mf{I}_{\mf{x}}(B)$ differs from the channel $\upzeta_{\mc{M}}^{\dagger}(X(B)^{\frac{1}{2}})$ 
does not prevent the existence of special cases in which insted such channels equal, the emblematic example is that in which
$B$ is a singlet and $X$ is a discrete observable see Rmk.\ref{10301327}. Furthermore the difference of these two channels when 
referring to the same semiobservable $X$ does not prevent the existence of an observable $Y$ permitting to interpret the channel
$\upzeta_{\mc{M}}^{\dagger}(X(B)^{\frac{1}{2}})$ as executing a measurement of $Y$ see subsection \ref{10301329} and the comments in 
the following remark.
\end{remark}
\begin{remark}
[\textbf{Semantics of $\upzeta_{\mc{A}}(c)^{\dagger}$ and then a different semantics of $\upzeta_{\mc{A}}^{\dagger}(X(B)^{\frac{1}{2}})$}]
\label{22121413}
Saved what stated in Rmk.\ref{11221252} we have anyway what stated in subsection \ref{10301329} which justifies 
Def.\ref{08081839}\eqref{12221404}. More specifically we obtain by 
Def.\ref{08081839}(\ref{22121423a},\ref{22121423b},\ref{22121423c},\ref{22121423d},\ref{08081839o},\ref{12221404})
that
\begin{itemize}
\item
$\mf{C}(\upzeta_{\mc{A}}^{\dagger}(c))$
=
the channel induced by the device implementing the operation of filtering through the execution of the measurement of the 
observable $\mf{o}(cc^{\ast})$ and then the selection of the subensemble of those preparations such that the value $1$ is 
detected.
\end{itemize}
In particular 
\begin{itemize}
\item
$\mf{C}(\upzeta_{\mc{M}}^{\dagger}(X(B)^{\frac{1}{2}}))$
=
the channel induced by the device implementing the operation of filtering through the execution of the measurement of the 
observable $\mf{o}(X(B))$ and then the selection of the subensemble of those preparations such that the value $1$ is 
detected,
\end{itemize}
then compare the above with what stated in subsection \ref{10301329}.
\end{remark}
\subsection{The empirical representation of the propensity}
\begin{remark}
[\textbf{Propensity versus probability and the role of time}]
\label{05171147}
We retain the concept of propensity primitive and consider the value $\mf{b}_{\mc{A},\mc{B}}(J,\upomega,e)$, which is 
a frequency as we shall see in Rmk.\ref{08101005}, its experimentally testable representative. 
Therefore, the empirical characters of $\mf{b}_{\mc{A},\mc{B}}(J,\upomega,e)$ 
are ascribable to the representative rather than to the propensity itself.
This because the concept of frequency is related to and dependent by the concept of time.
By the very definition of frequency, performing trials implies at least an operative meaning of time labelling them,
and this results to supply the theory with a primitive concept of time.
However we do not assume any global nor primitive concept of time.
Rather as advanced in the introduction, in our framework the couple formed by a species $\mr{a}$ of dynamical patterns and 
a context $M$ in the context category source of $\mr{a}$, determines a collection of experimentally detectable trajectories
whose dynamics is implemented by the morphism part $\uptau_{\mr{a}(M)}$ of the dynamical functor 
acting over the morphisms of the corresponding dynamical category $\mr{G}_{M}^{\mr{a}}$. Thus we can read 
$\mr{Mor}_{\mr{G}_{M}^{\mr{a}}}$ as a type of proper time associated with the species $\mr{a}$ in the context $M$.
\end{remark}
\begin{remark}
[\textbf{The empirical representation of the propensity is a frequency}]
\label{08101005}
The empirical representation of the propensity is a frequency.
Indeed $\mf{b}_{\mc{A},\mc{B}}(J,\upomega,e)=\mf{b}_{\mc{A},\mc{B}}(\upzeta^{\dagger}(e^{1/2})\circ J,\upomega,\un)$
and $\mf{b}_{\mc{A},\mc{B}}(J,\upomega,\un)=\frac{J(\upomega)(\un)}{\upomega(\un)}$ which \emph{equals the ratio
of the total value of the observable proportional to the number of preparations in the statistical ensemble resulting next 
the channel $\mf{c}(J)$ applies to the statistical ensemble $\mf{t}(\upomega)$, over the total value of the 
observable proportional to the number of preparations in the statistical ensemble $\mf{t}(\upomega)$}. Furthermore we have 
that the expectation value of the observable proportional to the number of preparations in the statistical ensemble 
$\mf{t}(\upomega)$ equals $1$ as expected.
\end{remark}
\begin{remark}
[\textbf{The propensity map slightly generalizes the state-effect duality}]
\label{08101104}
Let $\mc{A}\in\mr{tsa}$, $J\in\mf{Z}(\mc{A},\mc{B})$, $\upomega\in\mf{P}_{\mc{A}}^{\natural}$ 
such that $J(\upomega)(\un)\neq 0$ and $e\in\mr{Ef}(\mc{B})$. We have 
\begin{equation*}
\mf{b}_{\mc{A},\mc{B}}(J,\upomega,e)\leq\mf{p}_{\mc{B}}(J(\upomega),e)=p_{\mc{B}}(\frac{J(\upomega)}{J(\upomega)(\un)},e),
\end{equation*}
in particular $\mf{p}_{\mc{B}}$ implements $p_{\mc{B}}$, while $\mf{b}_{\mc{A},\mc{B}}(J,\upomega,e)$ has no counterpart in 
terms of $p_{\mc{B}}$ unless $J\in\mr{Q}(\mc{B},\mc{A})^{\dagger}$ or $J(\upomega)(\un)=\upomega(\un)$ see below,
hence $\mf{b}_{\mc{A},\mc{B}}$ generalizes $p_{\mc{B}}$. Furthermore \eqref{11251447a}$\Rightarrow$\eqref{11251447b} 
and \eqref{11251447c}$\Leftrightarrow$\eqref{11251447d}$\Leftrightarrow$\eqref{11251447e}, 
while if $\mc{B}$ is the closure of the linear space generated by 
$\mr{Ef}(\mc{B})$,\footnote{for instance any von Neumann algebra provided with the norm topology since 
the spectral decomposition of every selfadjoint element by the spectral theorem,
and since any element is linear combination of its real and imaginary parts.
As a result it is true also for any $C^{\ast}$ algebra being isometric via the universal representation to a $C^{\ast}$ 
subalgebra of a suitable von Neumann algebra.}then \eqref{11251447a}$\Leftarrow$\eqref{11251447b}. 
Here
\begin{enumerate}
\item
$J\in\mr{Q}(\mc{B},\mc{A})^{\dagger}$,
\label{11251447a}
\item
$(\exists\,T\in\mr{Q}(\mc{B},\mc{A}))(\forall\upphi\in\mf{P}_{\mc{A}}^{\natural})(\forall e\in\mr{Ef}(\mc{B}))
(\mf{b}_{\mc{A},\mc{B}}(J,\upphi,e)=p_{\mc{A}}(\frac{\upphi}{\upphi(\un)},T(e)))$,
\label{11251447b}
\item
$J(\upomega)(\un)=\upomega(\un)$,
\label{11251447c}
\item
$(\forall e\in\mr{Ef}(\mc{B}))(\mf{b}_{\mc{A},\mc{B}}(J,\upomega,e)=p_{\mc{B}}(\frac{J(\upomega)}{J(\upomega)(\un)},e))$,
\label{11251447d}
\item
$(\exists e\in\mr{Ef}(\mc{B}))(J(\upomega)(e)\neq 0\wedge
\mf{b}_{\mc{A},\mc{B}}(J,\upomega,e)=p_{\mc{B}}(\frac{J(\upomega)}{J(\upomega)(\un)},e))$;
\label{11251447e}
\end{enumerate}
where item \eqref{11251447a} is well-set since Prp.\ref{08091056}\eqref{08091056st3}, while item \eqref{11251447b} 
since Prp.\ref{08090901}\eqref{08090901st2}.
\end{remark}
\begin{definition}
\label{12041214}
If $b$ is a possibly unbounded selfadjoint operator in a Hilbert space $H$, then we let $E^{b}$ denote the resolution of the 
identity of $b$. If $x,y\in H$, then $E_{x,y}^{b}:\mc{B}(\C)\to\C$, $A\mapsto\lr{x}{E^{b}(A)y}$ and let 
$E_{x}^{b}\coloneqq E_{x,x}^{b}$.
\end{definition}
\begin{definition}
\label{12041215}
Let $\psi\in\mf{E}_{\mc{A}}$ admitting a GNS construction $\lr{\mf{H},\pi}{\Omega}$ and $a\in\mc{A}_{ob}$.
Define $\upmu_{\psi}^{a,\pi}\coloneqq\omega_{\Omega}\circ E^{\pi(a)}$. 
\end{definition}
\begin{remark}
[\textbf{Expectation value and empirical representation of the propensity}]
\label{11261141}
Let $\mc{A}\in\mr{tsa}$, $\upomega\in\mf{P}_{\mc{A}}^{\natural}$ such that $\uppsi\coloneqq\upomega(\un)^{-1}\upomega$ 
admits a GNS construction $\lr{\mf{H}}{\uppi,\Upomega}$\footnote{for instance any $m^{\ast}$-convex algebra with a bounded 
approximate identity.}, $O\in\mc{A}_{ob}$, and $\omega_{\Upomega}$ be the vector state on 
$\mc{M}=\lr{\mf{L}(\mf{H})}{\sigma(\mf{L}(\mf{H}),\mf{L}(\mf{H})_{\ast})}$, induced by the unit vector $\Upomega$. Thus, 
$\upmu_{\uppsi}^{O,\uppi}$ is a probability measure whose support equals the spectrum of $\uppi(O)$ and it is such that 
\begin{equation}
\label{12051211}
\begin{aligned}
\frac{\upomega(O)}{\upomega(\un)}
&=
\int\lambda\,d\upmu_{\uppsi}^{O,\uppi}(\lambda);
\\
\upmu_{\uppsi}^{O,\uppi}(B)
&=
p_{\mc{M}}(\omega_{\Upomega},E^{\uppi(O)}(B)),\,\forall B\in\mf{B}(\R).
\end{aligned}
\end{equation}
\end{remark}
Now \eqref{12051211} justifies the following addition to Def.\ref{08081839}
\begin{definition}
\label{12041208}
Let $\omega\in\mf{P}_{\mc{A}}^{\natural}$, $a\in\mc{A}_{ob}$, $A\in\mc{B}(\R)$ and $\psi\in\mf{E}_{\mc{A}}$ admitting a GNS 
construction $\lr{\mf{H},\pi}{\Omega}$. Set
\begin{enumerate}
\item
$\mf{t}(\omega_{\Omega})=\mf{t}(\psi)$; 
\item
$\mf{O}(E^{\pi(a)})=\mf{O}(a)$.
\end{enumerate}
\end{definition}
Indeed \eqref{12051211} interpreted under the light of Def.\ref{08081839}\eqref{08081839mpe} and Def.\ref{12041208} implies the 
following compatibility between empirical representation of the propensity and expectation value:
\emph{The expectation value of the observable $\mf{o}(O)$ in the statistical ensemble $\mf{t}(\upomega)$
equals the integral of the identity map on $\R$ against the measure mapping any Borelian set $B$ of $\R$ into 
the empirical representation of the propensity conditioned by the statistical ensemble $\mf{t}(\upomega)$
to detect the effect produced by selecting the subensemble of those preparations such that a
value in $B$ is detected, conditioned by the execution of the measurement of the observable $\mf{o}(O)$,
when tested on the statistical ensemble $\mf{t}(\upomega)$}.
\subsection{The contravariant $\dagger$ and the covariant $\upzeta_{\mc{A}}$}
\begin{remark}
[\textbf{Semantics of $\dagger$ and of its contravariant nature}]
Def.\ref{08081839}, Postulate \ref{11030928} and Def.\ref{11061008} yield
\emph{
\begin{enumerate}
\item
The empirical representation of the propensity conditioned by the statistical ensemble $\mf{t}(\upomega)$
to detect the effect, resulting next the device $\mf{d}(T)$ applies to the effect $\mf{e}(e)$, 
when tested on the statistical ensemble $\mf{t}(J(\upomega))$
equals
the empirical representation of the propensity conditioned by the statistical ensemble $\mf{t}(\upomega)$
to detect the effect $\mf{e}(e)$ when tested on 
the statistical ensemble resulting next the channel induced by the device $\mf{d}(T)$ applies to 
the statistical ensemble $\mf{t}(J(\upomega))$.
\item
The expectation value in the statistical ensemble $\mf{t}(\uppsi)$ of the observable 
resulting next the device $\mf{d}(T)$ applies to the observable $\mf{o}(y)$
equals
the expectation value of the observable $\mf{o}(y)$ in the statistical ensemble, resulting next the channel induced by the 
device $\mf{d}(T)$ applies to the statistical ensemble $\mf{t}(\uppsi)$.
\item
The channel induced by the device $\mf{d}(T)$ following the channel induced by the device $\mf{d}(S)$
equals the channel induced by the device $\mf{d}(S)$ following the device $\mf{d}(T)$.
\end{enumerate}}
The last statement follows by Lemma \ref{08080010}\eqref{08080010st2}.
\end{remark}
\begin{remark}
[\textbf{$\Uplambda_{\mc{A}}$ and $\upzeta_{\mc{A}}$ versus $\Uptheta_{\mc{A}}$ and $\upgamma_{\mc{A}}$}]
\label{11021152}
Rmk. \ref{11030533} grants us to develop a semantics either for operations and their actions or for co-operations
and their actions, but not for both if we want to avoid redundancies. 
Next Prp. \ref{08091056}(\ref{11030548},\ref{11030549}) exhibits that the semigroup of operations acts on observables and 
effects, while Prp. \ref{08091056}\eqref{08091056st4} establishes that the semigroup of co-operations acts on statistical
ensembles. Now the present exposition takes for granted the point of view according to which observables are primitive 
entities and represented by selfadjoint elements of a $\mr{tsa}$ algebra $\mc{B}$, while statistical ensembles are derived 
ones and represented by continuous positive functionals on $\mc{B}$.
Because of this and since semantical definitions are preferrable referring to primitive entities rather than
derived ones, when we have no options to set up both in the same context, in Def. \ref{08081839} we had to provide a 
semantical definition to operations and to their action $\upzeta$ on observables and desert it to co-operations and to 
their action $\upgamma^{\dagger}$ on statistical ensembles.
\par
Let us outline some applications of the semantics when dealing with $\upzeta$. 
Let $\mc{A}\in\mr{tsa}$, $d,c\in\Uptheta_{\mc{A}}$ and $u,v\in\Uplambda_{\mc{A}}$ 
Thus, since $\upzeta_{\mc{A}}(vu)=\upzeta_{\mc{A}}(v)\circ\upzeta_{\mc{A}}(v)$ 
while $\upzeta_{\mc{A}}^{\dagger}(vu)=\upzeta_{\mc{A}}^{\dagger}(u)\circ\upzeta_{\mc{A}}^{\dagger}(v)$ we have:
\emph{
\begin{itemize}
\item
$\mf{D}(\upzeta_{\mc{A}}(u))=$ the device implementing the operation of filtering through $\mf{r}(u)$.
\item
The device implementing the operation of filtering through $\mf{r}(vu)$
equals
the device implementing 
the operation of filtering through $\mf{r}(v)$
after 
the operation of filtering through $\mf{r}(u)$
equals
the device implementing the operation of filtering through $\mf{r}(v)$
following 
the device implementing the operation of filtering through $\mf{r}(u)$.
\item
$\mf{D}(\upzeta_{\mc{A}}^{\dagger}(u))=$ the channel induced by the device implementing the operation of filtering through 
$\mf{r}(u)$.
\item
The channel induced by the device implementing the operation of filtering through $\mf{r}(vu)$
equals
the channel induced by the device implementing the operation of filtering through $\mf{r}(u)$
following 
the channel induced by the device implementing the operation of filtering through $\mf{r}(v)$.
\end{itemize}}
\end{remark}
\begin{remark}
[\textbf{$G$-action on $\mc{A}$ and $G^{op}$-action on $\mc{A}_{s}^{\ast}$}]
Let us analize an emblematic way of implementing a group $G$ as group of transformations on $\mc{A}$
and the group $G^{op}$, the opposite of $G$, as group of transformations on $\mc{A}_{s}^{\ast}$
with the additional property of possessing a semantics. 
Now since Prp.\ref{08091056} $\Uplambda_{\mc{A}}$ acts on $\mc{A}$ through the map $\upzeta_{\mc{A}}$ 
while $\Uplambda_{\mc{A}}^{op}$ acts on $\mc{A}_{s}^{\ast}$ through the map $\upzeta_{\mc{A}}^{\dagger}$ 
and the set of values of both these maps are provided with a semantics. 
Therefore, if we want a $G$ action $\uptau$ on $\mc{A}$ and the $G^{op}$ action $\uptau^{\dagger}$ on $\mc{A}_{s}^{\ast}$
both provided with semantics, then we can take a group morphism $V:G\to\mr{U}(\mc{A})$ and define 
\begin{equation}
\label{11021356}
\begin{aligned}
\uptau_{V}&\coloneqq\imath_{\mr{Q}(\mc{A})}^{\mf{L}(\mc{A})}\circ\upzeta_{\mc{A}}\circ
\imath_{\mr{U}(\mc{A})}^{\Uplambda_{\mc{A}}}\circ V:G\to\mf{L}(\mc{A});
\\
\uptau_{V}^{\dagger}&\coloneqq\imath_{\mf{Z}(\mc{A})}^{\mf{L}(\mc{A}_{s}^{\ast})}\circ\upzeta_{\mc{A}}^{\dagger}
\circ\imath_{\mr{U}(\mc{A})}^{\Uplambda_{\mc{A}}^{op}}\circ V\circ\imath_{G^{op}}^{G}:G^{op}\to\mf{L}(\mc{A}_{s}^{\ast}).
\end{aligned}
\end{equation}
Let us denote $\uptau_{V}$ simply by $\uptau$ thus, 
\emph{$\uptau$ is an action of $G$ while $\uptau^{\dagger}$ is an action of $G^{op}$}.
Next let us set\footnote{a more contextualized semantics will be developed in 
Def.\ref{01200837}, see specifically Def.\ref{01200837}\eqref{01200837g}.} 
\begin{itemize}
\item
$\mf{r}(V(g))=$ the $G$ transformation of magnitude $g$,
\end{itemize}
so for every $g\in G$ and $a\in\mc{A}_{ob}$ we have according to our semantics that \emph{$\mf{O}(\uptau(g)a)$ equals 
the observable resulting next the device implementing the operation of filtering through the $G$ transformation of 
magnitude $g$ applies to the observable $\mf{o}(a)$}.
\par
If in particular $\mc{A}$ acts on some Hilbert space $\mf{H}$ and $\rho$ is any trace class operator on $\mf{H}$, 
then $\uptau^{\dagger}(g)(\omega_{\rho})=\omega_{\upzeta_{\mc{A}}(V^{\ast}(g))(\rho )}$ for every $g\in G$, as a result for every 
$v\in\mf{H}$ we obtain 
\begin{equation}
\label{11021119}
\uptau^{\dagger}(g)(\omega_{v})=\omega_{V^{\ast}(g)v}.
\end{equation}
Had we selected $\Uptheta_{\mc{A}}$ instead of $\Uplambda_{\mc{A}}$ as set of entities to be provided with a semantics, 
so implicitly taking statistical ensembles primitive instead of observables see Rmk. \ref{11021152},
we would have employed $\upgamma_{\mc{A}}$ in place of $\upzeta_{\mc{A}}$ in \eqref{11021356} and obtained $\uptau$ as an 
action of $G^{op}$ and $\uptau^{\dagger}$ as an action of $G$. 
\par
If $G=G^{op}$, for instance when $G$ is commutative, then $\uptau$ and $\uptau^{\dagger}$ would be both actions of $G$. 
For example when $G=\R$ interpreted as the group of time translations, then it is customary for quantum systems admitting 
a Hamiltonian $H$ to take the $\R$-action $V$ as above to be generated by $ih^{-1}H$;
equivalently it is the $\R$-action $V^{\ast}$, actually an $\R^{op}$-action, assumed to be strongly 
continuous and to have $-i\hbar^{-1}H$ as infinitesimal generator. Namely defined
\begin{equation*}
\begin{cases}
\mc{D}\coloneqq\{y\in\mf{H}\,\vert\,\exists\,\lim_{t\to 0,t\neq 0}\frac{V^{\ast}(t)-\un}{t}y\},
\\
(\forall y\in\mc{D})(Hy\coloneqq i\hbar\lim_{t\to 0,t\neq 0}\frac{V^{\ast}(t)-\un}{t}y);
\end{cases}
\end{equation*}
it is easy to prove that $V^{\ast}(\R)\mc{D}\subseteq\mc{D}$ and that for every $y\in\mc{D}$ and every $t\in\R$ we obtain
the notable Schr\"{o}dinger equation 
\begin{equation}
\label{11021145}
\frac{du_{y}}{dt}(t)=-i\hbar^{-1}Hu_{y}(t),
\end{equation}
where $u_{y}:\R\to\mc{D}$, $t\mapsto V^{\ast}(t)y$ for every $y\in\mc{D}$.
It is worthwhile remarking that \eqref{11021145} and \eqref{11021119} represent the core of the traditional approach to 
time translation of vector states in quantum physics.
\end{remark}
\subsection{Coherence of the semantics}
\subsubsection{Compatibility between the semantics of the channel $\mf{I}_{\mf{x}}(B)$ and that of the effect $X(B)$}
\label{10311154}
Equality \eqref{10311146} ensures compatibility between Def.\ref{08081839}\eqref{08081839mpc} and 
Def.\ref{08081839}\eqref{08081839mpe}. More specifically let $X$ be a semiobservable on $\mf{H}$ with values in 
$(\Lambda,\mf{B})$, let $\mc{M}=\lr{\mf{L}(\mf{H})}{\sigma(\mf{L}(\mf{H}),\mf{L}(\mf{H})_{\ast})}$, 
let $\mf{x}$ be a measuring process of $X$, $\upphi\in\mf{P}_{\mc{M}}^{\natural}$ and $B\in\mf{B}$ we have 
\begin{equation}
\label{11030032}
\mf{b}_{\mc{M}}(\mf{I}_{\mf{x}}(B),\upphi,\un)
=
\mf{b}_{\mc{M}}(\mr{Id}_{\mc{M}_{s}^{\ast}},\upphi,X(B))
=
\mf{b}_{\mc{M}}\left(\upzeta_{\mc{M}}^{\dagger}(X(B)^{\frac{1}{2}}),\upphi,\un\right).
\end{equation}
By the equalities above we obtain: 
\emph{The empirical representation of the propensity conditioned by the statistical ensemble $\mf{t}(\upphi)$ 
to detect the effect of selecting one output when tested on the statistical ensemble resulting next the channel 
executing the measurement of the semiobservable $\mf{o}(X)$ and then selecting the subensemble of those preparations such that 
a value in $B$ is detected, applies to the statistical ensemble $\mf{t}(\upphi)$ 
equals
the empirical representation of the propensity conditioned by the statistical ensemble $\mf{t}(\upphi)$ to detect 
the effect produced by selecting the subensemble of those preparations such that 
a value in $B$ is detected, conditioned by the execution of the measurement of the semiobservable $\mf{o}(X)$,
when tested on the statistical ensemble $\mf{t}(\upphi)$ 
equals
the empirical representation of the propensity conditioned by the statistical ensemble $\mf{t}(\upphi)$ to detect 
the effect of selecting one output when tested on the statistical ensemble resulting next the channel, induced by the 
device selecting the subensemble of those preparations such that 
a value in $B$ is detected, conditioned by the execution of the measurement of the semiobservable $\mf{o}(X)$,
applies to the statistical ensemble $\mf{t}(\upphi)$}.
\par
Despite the second equality above and made exception for the case when $X$ is a discrete observable,
$\mf{x}$ is the von Neumann measuring process associated with $X$, and $B$ is a singlet (Rmk.\ref{11041328}), in general 
$\mf{I}_{\mf{x}}(B)$ differs from $\upzeta_{\mc{M}}^{\dagger}(X(B)^{\frac{1}{2}})$.
See subsection \ref{10301329} dealing with the task of finding a suitable measuring process for which the equality holds.
\subsubsection{Compatibility 2}
Def.\ref{08081839}\eqref{08081839o} and Def.\ref{08081839}\eqref{08081839oe} are compatibly defined.
More specifically if $\mc{A}\in\mr{tsa}$, $a\in\Uplambda_{\mc{A}}$ and $\upomega\in\mf{P}_{\mc{A}}^{\natural}$, then 
the trivial equality $\upzeta^{\dagger}(a)(\upomega)(\un)=\upomega(aa^{\ast})$ implies: 
\emph{The empirical representation of the propensity conditioned by the statistical ensemble $\mf{t}(\upomega)$ to detect 
the effect of selecting one output when tested on the statistical ensemble resulting next the channel,
induced by the device implementing the operation of filtering through $\mf{r}(a)$, applies to the statistical ensemble
$\mf{t}(\upomega)$
equals
the empirical representation of the propensity conditioned by the statistical ensemble $\mf{t}(\upomega)$ to detect 
the effect, produced by the operation of filtering through $\mf{r}(a)$, when tested on the statistical ensemble 
$\mf{t}(\upomega)$}.
\par
The \textbf{noncommutativity} of the observable algebra $\mc{A}$ is the paramount feature distinguishing quantum systems from the
classical ones whose observable algebra is commutative. It has remarkable empirical consequences: The first is   
the irreducible randomicity of the observable algebra and its consequences analyzed in subsection \ref{11211146}; the second one
is the generalized phenomenon of interference described in subsection \ref{08091625}. Instead in subsection \ref{01142016} 
we analyze the reconciliation of special relativity with quantum mechanics by interpreting locality as synonym of commutative 
subalgebras of $\mc{A}$. 
\par
In subsection \ref{10301329} we analyze the Ozawa channel $\mf{I}_{\mf{x}}(B)$, in subsection \ref{01192115} we outline discrete 
observables, finally in subsection \ref{01192117} we apply 
our formalism to provide the Wigner formula for a sequence of measurements of possibly continuous semiobservables in 
Schr\"{o}dinger picture instead of the usual Heisenberg picture. 
\subsection{Irreducible Randomicity of $\mc{A}$}
\label{11211146}
In this subsection we let $\mc{A}\in\mr{tsa}$.
\begin{definition}
[\textbf{Dispersion}]
\label{11201916}
Let $\omega\in\mf{P}_{\mc{A}}^{\natural}$ and $a\in\mc{A}_{ob}$, define $\mr{mv}(\omega,a)\coloneqq\omega(\un)^{-1}\omega(a)$ 
namely the mean value of $a$ in $\omega$ and 
\begin{equation*}
\mr{dis}(\omega,a)\coloneqq\mr{mv}\left(\omega,(a-\mr{mv}(\omega,a)\un)^{2}\right),
\end{equation*}
called \emph{dispersion of the observable $\mf{o}(a)$ in the statistical ensemble $\mf{t}(\omega)$}. 
\end{definition}
Notice that 
\begin{equation}
\label{11290928}
\begin{aligned}
(\forall\omega\in\mf{P}_{\mc{A}}^{\natural})(\forall a\in\mc{A}_{ob})(\mr{dis}(\omega,a)&\in\R_{\geq 0}),
\\
(\forall\omega\in\mf{P}_{\mc{A}}^{\natural})(\forall a\in\mc{A}_{ob})(\mr{dis}(\omega,a)&=\mr{dis}(\omega(\un)^{-1}\omega,a)),
\\
(\forall\psi\in\mf{E}_{\mc{A}})(\forall a\in\mc{A}_{ob})(\mr{dis}(\psi,a)&=\psi(a^{2})-\psi(a)^{2});
\end{aligned}
\end{equation}
where the first statement follows since $(a-\mr{mv}(\omega,a)\un)^{2}=b^{\ast}b$ with $b=a-\mr{mv}(\omega,a)\un$.
Thus, if $\psi$ is also admitting a GNS representation say $\lr{H}{\pi,\Omega}$ and there 
exists $\lambda\in\R$ such that $\pi(a)\Omega=\lambda\Omega$, then clearly $\mr{dis}(\psi,a)=0$.
Recall Def.\ref{12041214} and Def.\ref{12041215}. If $\psi\in\mf{E}_{\mc{A}}$ admits a GNS construction $\lr{\mf{H},\pi}{\Omega}$ 
and $a\in\mc{A}_{ob}$, then $\upmu_{\psi}^{a,\pi}=E_{\Omega}^{\pi(a)}$ and by the Borelian functional calculus we have 
\begin{equation}
\label{11291834}
\begin{aligned}
\mr{dis}(\psi,a)
&=
\psi((a-\psi(a)\un)^{2})
\\
&
=\omega_{\Omega}((\pi(a)-\psi(a)\un_{\mf{H}})^{2}) 
\\
&
=
\int(\lambda-\psi(a))^{2}\,d\upmu_{\psi}^{a,\pi}(\lambda).
\end{aligned}
\end{equation}
\begin{lemma}
[\textbf{Characterization of $\mr{dis}(\psi,a)=0$}]
\label{11291835}
Let $\psi\in\mf{E}_{\mc{A}}$ admitting a GNS construction $\lr{\mf{H},\pi}{\Omega}$ and $a\in\mc{A}_{ob}$.
Thus, 
$\eqref{11291835a}\Leftrightarrow\eqref{11291835b}\Leftrightarrow\eqref{11291835c}
\Leftrightarrow\eqref{11291835d}\Leftrightarrow\eqref{11291835e}$ 
where 
\begin{enumerate}
\item
$\mr{dis}(\psi,a)=0$; 
\label{11291835a}
\item
$\upmu_{\psi}^{a,\pi}(\complement_{\C}\{\psi(a)\})=0$;
\label{11291835b}
\item
$\upmu_{\psi}^{a,\pi}(\{\psi(a)\})=1$;
\label{11291835c}
\item
$E^{\pi(a)}(\{\psi(a)\})\Omega=\Omega$;
\label{11291835d}
\item
$\pi(a)\Omega=\psi(a)\Omega$.
\label{11291835e}
\end{enumerate}
\end{lemma}
\begin{proof}
$\eqref{11291835a}\Leftrightarrow\eqref{11291835b}$ follows since \eqref{11291834} and \cite[IV.34 Cor.1; IV.13 Thm.1]{IntBourb},
$\eqref{11291835b}\Leftrightarrow\eqref{11291835c}$ follows since $\upmu_{\psi}^{a,\pi}$ is a probability measure.
Next we claim to prove $\eqref{11291835b}\Leftrightarrow\eqref{11291835d}$. Clearly we have 
\begin{equation*}
\eqref{11291835b}\Leftrightarrow E^{\pi(a)}(\complement_{\C}\{\psi(a)\})\Omega=\ze 
\end{equation*}
moreover since $E^{\pi(a)}$ is additive and $E^{\pi(a)}(\R)=\un$ 
we obtain $\un=E^{\pi(a)}(\{\psi(a)\})+E^{\pi(a)}(\complement_{\C}\{\psi(a)\})$ so, 
$\Omega=E^{\pi(a)}(\{\psi(a)\})\Omega+E^{\pi(a)}(\complement_{\C}\{\psi(a)\})\Omega$
and then 
\begin{equation}
\label{12041005}
E^{\pi(a)}(\complement_{\C}\{\psi(a)\})\Omega=\ze\Leftrightarrow\eqref{11291835d}
\end{equation}
and our claim follows.
Finally we claim to prove $\eqref{11291835d}\Leftrightarrow\eqref{11291835e}$. 
Let $x\in\mf{H}$. If \eqref{11291835d} is satisfied, then $E^{\pi(a)}(\complement_{\C}\{\psi(a)\})\Omega=\ze$ by \eqref{12041005}
thus, $\complement_{\C}\{\psi(a)\}$ is a $E_{x,\Omega}^{\pi(a)}$-negligible set since \cite[III.15 (17) last inequality]{IntBourb} and 
\cite[IV.34 Prp.1]{IntBourb} therefore for every $E_{x,\Omega}^{\pi(a)}$-integrable map $f$ and $g$ such that $f(\psi(a))=g(\psi(a))$ 
we deduce by \cite[IV.33 end of the page]{IntBourb} that $\int f\,dE_{x,\Omega}^{\pi(a)}=\int g\,dE_{x,\Omega}^{\pi(a)}$ in particular 
\begin{equation*}
\int\lambda\,dE_{x,\Omega}^{\pi(a)}(\lambda)=\int g(\lambda)\,dE_{x,\Omega}^{\pi(a)}(\lambda),
\end{equation*} 
where $g(\psi(a))=\psi(a)$ and $g(\lambda)=0$ for every $\lambda\in\complement_{\C}\{\psi(a)\}$. Next $g$ is a step function so 
\begin{equation*}
\begin{aligned}
\int g(\lambda)\,dE_{x,\Omega}^{\pi(a)}(\lambda)=g(\psi(a))E_{x,\Omega}^{\pi(a)}(\{\psi(a)\})
&=\psi(a)\lr{x}{E^{\pi(a)}(\{\psi(a)\})\Omega}
\\
&=\psi(a)\lr{x}{\Omega}
=\lr{x}{\psi(a)\Omega};
\end{aligned}
\end{equation*}
while by the Borelian functional calculus 
\begin{equation*}
\int\lambda\,dE_{x,\Omega}^{\pi(a)}(\lambda)
=
\lr{x}{\pi(a)\Omega},
\end{equation*}
then $\lr{x}{\pi(a)\Omega}=\lr{x}{\psi(a)\Omega}$ for every $x\in\mf{H}$ namely $\pi(a)\Omega=\psi(a)\Omega$ id est
\eqref{11291835e}. The implication $\eqref{11291835e}\Rightarrow\eqref{11291835d}$ follows by a general result see for instance 
\cite[Prp.5.10(ii)]{schumb}.
\end{proof}
If $\omega\in\mf{P}_{\mc{A}}^{\natural}$ such that $\omega(\un)^{-1}\omega$ admits a GNS construction $\lr{\mf{H},\pi}{\Omega}$, then
for every $a\in\mc{A}_{ob}$ and $A\in\mc{B}(\R)$ we have by Def.\ref{12041208}
\begin{itemize}
\item
$\mu_{\omega(\un)^{-1}\omega}^{a,\pi}(A)$ equals the empirical representation of the propensity conditioned by the statistical ensemble 
$\mf{t}(\omega)$ to detect the effect, produced by selecting the subensemble of those preparations such that a value in $A$ is 
detected, conditioned by the execution of the measurement of the observable $\mf{o}(a)$, 
when tested on the statistical ensemble $\mf{t}(\omega)$. 
\end{itemize}
Thus, we can now state the following 
\begin{corollary}
[\textbf{Semantic characterization of $\mr{dis}(\omega,a)=0$ and $\mr{dis}(\omega,a)\neq 0$}]
\label{11301647}
Let $\omega\in\mf{P}_{\mc{A}}^{\natural}$ be such that $\omega(\un)^{-1}\omega$ admits a GNS construction $\lr{\mf{H},\pi}{\Omega}$, 
and let $a\in\mc{A}_{ob}$, then
\begin{enumerate}
\item
We have 
$\eqref{12041100a}\Leftrightarrow\eqref{12041100b}$
where 
\begin{enumerate}
\item
$\mr{dis}(\omega,a)=0$; 
\label{12041100a}
\item
it equals to $1$ the empirical representation of the propensity conditioned by the statistical ensemble 
$\mf{t}(\omega)$ to detect the effect, produced by selecting the subensemble of those preparations such that the value 
$\omega(\un)^{-1}\omega(a)$ is detected, conditioned by the execution of the measurement of the observable $\mf{o}(a)$, 
when tested on the statistical ensemble $\mf{t}(\omega)$.
\label{12041100b}
\end{enumerate}
\label{11301647I}
\item
We have 
$\eqref{11301647Ii}\Leftrightarrow\eqref{11301647Iii}\Leftrightarrow\eqref{11301647Iiii}\Leftrightarrow\eqref{11301647Iiv}$, 
where
\begin{enumerate}
\item
$\mr{dis}(\omega,a)\neq 0$;
\label{11301647Ii}
\item
it differs from $0$ the empirical representation of the propensity conditioned by the statistical ensemble 
$\mf{t}(\omega)$ to detect the effect, produced by selecting the subensemble of those preparations such that a value 
$\lambda\neq\omega(\un)^{-1}\omega(a)$ is detected, conditioned by the execution of the measurement of the observable 
$\mf{o}(a)$, when tested on the statistical ensemble $\mf{t}(\omega)$;
\label{11301647Iii}
\item
none of the value $\gamma\in\R$ is obtained surely by measuring the observable $\mf{o}(a)$ when 
tested on the statistical ensemble $\mf{t}(\omega)$;
\label{11301647Iiii}
\item
$(\forall\gamma\in\R)(\mu_{\omega(\un)^{-1}\omega}^{a,\pi}(\{\gamma\})\neq 1)$.
\label{11301647Iiv}
\end{enumerate}
\label{11301647II}
\end{enumerate}
\end{corollary}
\begin{proof}
Recall $\mr{dis}(\omega,a)=\mr{dis}(\omega(\un)^{-1}\omega,a)$ thus, st.\eqref{11301647I} follows since 
Lemma \ref{11291835}$(\eqref{11291835a}\Leftrightarrow\eqref{11291835c})$. 
Next $\eqref{11301647Ii}\Leftrightarrow\eqref{11301647Iii}$ follows since 
Lemma \ref{11291835}$(\eqref{11291835a}\Leftrightarrow\eqref{11291835b})$.
Next we claim to prove that $\eqref{11301647Ii}\Leftrightarrow\eqref{11301647Iiv}$. 
For the remaining of the present proof we let $\psi=\omega(\un)^{-1}\omega$ so, we have   
\begin{equation}
\label{11301818}
(\forall\lambda\in\complement_{\C}\{\psi(a)\})(\upmu_{\psi}^{a,\pi}(\{\lambda\})\neq 1),
\end{equation}
indeed if by absurdum we would have $\alpha\in\complement_{\C}\{\psi(a)\}$ such that $\upmu_{\psi}^{a,\pi}(\{\alpha\})=1$ then 
being $\upmu_{\psi}^{a,\pi}$ a probability measure so, in particular additive and with maximal value equal to $1$, 
we would obtain $\upmu_{\psi}^{a,\pi}(\complement_{\C}\{\alpha\})=0$ then, by the functional calculus that $\psi(a)=\alpha$ which is 
absurdum then \eqref{11301818} follows. Next Lemma \ref{11291835}$(\eqref{11291835a}\Leftrightarrow\eqref{11291835b})$ is 
equivalent to state that  
\begin{equation}
\label{11301817}
\mr{dis}(\psi,a)\neq 0\Leftrightarrow\upmu_{\psi}^{a,\pi}(\complement_{\C}\{\psi(a)\})\neq 0.
\end{equation}
while Lemma \ref{11291835}$(\eqref{11291835b}\Leftrightarrow\eqref{11291835c})$ and \eqref{11301818} imply 
\begin{equation}
\label{11301819}
\upmu_{\psi}^{a,\pi}(\complement_{\C}\{\psi(a)\})\neq 0
\Leftrightarrow
(\forall\gamma\in\R)(\upmu_{\psi}^{a,\pi}(\gamma)\neq 1).
\end{equation}
Our claim follows by \eqref{11301817} and \eqref{11301819}.
\end{proof}
\begin{convention}
\label{12031227}
Let $\omega\in\mf{P}_{\mc{A}}^{\natural}$ be such that $\omega(\un)^{-1}\omega$ admits a GNS construction $\lr{\mf{H},\pi}{\Omega}$.
We convein to denote any of the equivalent statements in Cor.\ref{11301647}\eqref{11301647I} by the following statement:
\emph{on the statistical ensemble $\mf{t}(\omega)$ the measurement of the observable $\mf{o}(a)$ is obtained surely with value 
$\omega(\un)^{-1}\omega(a)$}. 
We convein to denote any of the equivalent statements in Cor.\ref{11301647}\eqref{11301647II} by the following statement: 
\emph{on the statistical ensemble $\mf{t}(\omega)$ the measurement of the observable $\mf{o}(a)$ is obtained randomly}.
Furthermore more in general if $\upphi\in\mf{P}_{\mc{A}}^{\natural}$ and $a\in\mc{A}_{ob}$ we convein 
to denote the statement $\mr{dis}(\upphi,a)=0$ by the following statement: 
\emph{on the statistical ensemble $\mf{t}(\upphi)$ the measurement of the observable $\mf{o}(a)$ is obtained surely with value 
$\upphi(\un)^{-1}\upphi(a)$}; while to denote the statement $\mr{dis}(\upphi,a)\neq 0$ by the following statement: 
\emph{on the statistical ensemble $\mf{t}(\upphi)$ the measurement of the observable $\mf{o}(a)$ is obtained randomly}.
\end{convention}
\begin{remark}
\label{11301806}
Let $\omega\in\mf{P}_{\mc{A}}^{\natural}$ be such that $\omega(\un)^{-1}\omega$ admits a GNS construction $\lr{\mf{H},\pi}{\Omega}$, 
and let $a\in\mc{A}_{ob}$.
It is worth remarking that stating that on the statistical ensemble $\mf{t}(\omega)$ the measurement of the observable 
$\mf{o}(a)$ is obtained randomly is equivalent to state that it is different from $0$ the fraction of individual preparations of 
the statistical ensemble $\mf{t}(\omega)$ in which a value $\lambda\neq\omega(\un)^{-1}\omega(a)$ is obtained by measuring the 
observable $\mf{o}(a)$ when tested on the statistical ensemble $\mf{t}(\omega)$. Equivalently there exists a nonempty subensemble
$\mf{t}(\phi)$ of $\mf{t}(\omega)$ such that for each preparation $z$ belonging
to the statistical ensemble $\mf{t}(\phi)$ there is no way to know previously the measurement of the observable $\mf{o}(a)$ 
what value will be obtained on $z$ after the measurement of the observable $\mf{o}(a)$.
\end{remark}
Easily we check that $\mf{P}_{\mc{A}}^{\natural}$ is a convex subset of the affine space underlying the linear topological space 
$\mf{L}_{s}(\mc{A},\C)$, therefore we can provide the following 
\begin{definition}
Define $\mh{P}_{\mc{A}}^{\natural}$ to be the set of \emph{extremal} points of the convex set $\mf{P}_{\mc{A}}^{\natural}$.
\end{definition}
\begin{definition}
[\textbf{Irreducible randomicity of $\mc{A}$}]
Let $\omega\in\mf{P}_{\mc{A}}^{\natural}$, we call $\omega$ \emph{dispersion-free} iff 
$(\forall a\in\mc{A}_{ob})(\mr{dis}(\omega,a)=0)$. We call $\mc{A}$ to be \emph{irreducibly random} iff 
\begin{equation*}
(\forall\upphi\in\mh{P}_{\mc{A}}^{\natural})(\exists\,b\in\mc{A}_{ob})(\mr{dis}(\upphi,b)\neq 0),
\end{equation*}
namely for every $\upphi\in\mh{P}_{\mc{A}}^{\natural}$ we have that $\upphi$ is \textbf{not} dispersion-free. 
\end{definition}
\begin{definition}
We call $\mc{A}$ to be \emph{strongly \textbf{not} irreducibly random} iff 
\begin{equation*}
(\forall\upphi\in\mh{P}_{\mc{A}}^{\natural})(\forall\,b\in\mc{A}_{ob})(\mr{dis}(\upphi,b)=0).
\end{equation*}
\end{definition}
\begin{remark}
[The observable algebra of any classical physics theory is strongly not irreducibly random.]
In virtue of Convention \ref{12031227} we know that if $\omega\in\mf{P}_{\mc{A}}^{\natural}$ and 
$a\in\mc{A}_{ob}$, then $\mr{dis}(\omega,a)=0$ is equivalent to state that on the statistical ensemble $\mf{t}(\omega)$ the 
measurement of the observable $\mf{o}(a)$ is obtained surely with value $\omega(\un)^{-1}\omega(a)$. Next it is well-known that
the observable algebra of a classical physics theory whose phase space equals $X$, is the commutative unital $C^{\ast}$-algebra 
$\mc{B}^{X}$ obtained by adding the unity $\un$ to the $C^{\ast}$-algebra $\mc{C}_{c}(X)$ of compactly supported complex valued 
continous maps defined on the locally compact space $X$. Next 
\begin{equation*}
\mc{B}_{ob}^{X}=\{f\in\mc{C}_{c}(X)\,\vert\,(\forall x\in X)(f(x)\in\R)\}\cup\{\un\},
\end{equation*}
while $\mh{P}_{\mc{B}^{X}}^{\natural}$ equals the set of functionals (actually states) associated with the points of $X$ via the 
evaluation map, namely the set of the Dirac functionals centered in points of $X$ id est 
\begin{equation}
\label{12210102}
\mh{P}_{\mc{B}_{X}}^{\natural}=\{\delta_{x}\}_{x\in X},
\end{equation}
where $\delta_{x}(f)\coloneqq f(x)$ and $\delta_{x}(\un)=1$. Now clearly $\delta_{x}(f^{2})=f^{2}(x)=\delta_{x}(f)^{2}$ 
for every $f\in\mc{C}_{c}(X)$ and then
\begin{equation}
\label{12210103}
(\forall x\in X)(\forall f\in\mc{B}_{ob}^{X})(\mr{dis}(\delta_{x},f)=0).
\end{equation}
We then conclude by \eqref{12210102} and \eqref{12210103} that \emph{the observable algebra $\mc{B}_{X}$ of a classical physics 
theory whose phase space equals $X$ is strongly not irreducibly random}.
\end{remark}
Next we shall see under which conditions a noncommutative algebra $\mc{A}$ is irreducibly random. 
Let us start with the following 
\begin{definition}
Define $\mr{pr}(\mc{A})\coloneqq\{p\in\mc{A}\,\vert\,p=p^{2}\}$. We say that 
\emph{$\mc{A}$ has the property of absence of dispersion} iff there exist $q\in\mr{pr}(\mc{A})$, and $f\in\mc{C}([0,1],\mc{A})$
such that $f(t)\in\mr{pr}(\mc{A})$ for every $t\in[0,1]$, $f(0)=q$ and $f(1)=\un-q$.
\end{definition}
\begin{lemma}
\label{11261953}
Assume that $\mc{A}$ has the property of absence of dispersion and let $\psi\in\mf{P}_{\mc{A}}$. 
If for all $p\in\mr{pr}(\mc{A})$ we have $\psi(p)=\psi(p)^{2}$, then $\psi\notin\mf{E}_{\mc{A}}$.
\end{lemma}
\begin{proof}
Let $a\in\mc{A}$ so, $\psi(a)-\psi(a)^{2}=0$ is equivalent to $\psi(a)(\psi(a)-1)=0$ i.e. $\psi(a)=0$ or $\psi(a)=1$.
If for all $p=p^{2}\in\mc{A}$ we have $\psi(p)=\psi(p)^{2}$, then 
\begin{equation}
\label{11262041}
(\forall p=p^{2}\in\mc{A})(\psi(p)=0\vee\psi(p)=1).
\end{equation}
Next $\psi\circ f\in\mc{C}([0,1],\R)$ so \eqref{11262041} implies either $\psi\circ f$ is constantly equal to $0$
or $\psi\circ f$ is constantly equal to $1$. Assume by absurdum that $\psi(\un)=1$ then 
in the first case $0=\psi(q)=\psi(\un-q)=\psi(\un)-\psi(q)=1$ which is a contradiction; 
in the second case $1=\psi(q)=\psi(\un-q)=\psi(\un)-\psi(q)=0$ which is a contradiction.
So $\psi(\un)\neq 1$ by absurdum.
\end{proof}
\begin{corollary}
[$\mc{A}$ is irreducibly random $1$]
\label{11211958}
If $\mc{A}$ has the property of absence of dispersion, then $\mc{A}$ is irreducibly random.
\end{corollary}
\begin{proof}
If by absurdum $\omega\in\mh{P}_{\mc{A}}^{\natural}$ is dispersion-free, then by letting $\psi=\omega(\un)^{-1}\omega$ we would obtain
by \eqref{11290928} that $\psi(p)=\psi(p)^{2}$ for all $p\in\mr{pr}(\mc{A})$. Then Lemma \ref{11261953} would imply that 
$\psi\notin\mf{E}_{\mc{A}}$ which is absurdum thus, the statement follows by absurdum.
\end{proof}
If $\omega\in\mh{P}_{\mc{A}}^{\natural}$ is such that $\omega(\un)^{-1}\omega$ admits a GNS construction $\lr{\mf{H},\pi}{\Omega}$
and $a\in\mc{A}_{ob}$, then as a consequence of the above statement we have that for every subensebmle $\phi$ of $\omega$ such 
that $\phi(\un)\neq 0$ there exists $a\in\mc{A}$ such that on the statistical ensemble $\mf{t}(\phi)$ the measurement of the 
observable $\mf{o}(a)$ is obtained randomly. This statement is the precise version of what in \cite[pages 198-199]{vonNeum} is 
described as the ``violation of the principle of sufficient cause``.
Next we start dealing with indetermination relations by provoding the following
\begin{lemma}
[Cauchy-Schwarz inequality]
\label{12011735}
Let $A$ be a complex $\ast$-algebra and $\omega$ a complex linear functional on $A$ such that $\omega(z^{\ast}z)\geq 0$ for every
$z\in A$. Then $\omega$ is hermitian namely $\ov{\omega(z)}=\omega(z^{\ast})$ and for all $x,y\in A$ we have 
\begin{equation*}
\begin{aligned}
\ov{\omega(x^{\ast}y)}
&=
\omega(y^{\ast}x)
\\
|\omega(x^{\ast}y)|^{2}
&\leq
\omega(x^{\ast}x)
\omega(y^{\ast}y).
\end{aligned}
\end{equation*}
\end{lemma}
\begin{proof}
Define $f:A\times A\to\C$, $(x,y)\mapsto\omega(x^{\ast}y)$ thus, $f$ is a sesquilinear form moreover by hypothesis 
$f$ is hermitian since the polarization formula,  see \cite[V.2 Remarque]{EVT} and positive so, the equality in the statement 
follows; while the inequality follows since \cite[V.3 Prp.2]{EVT}. 
\end{proof}
Let $A_{ob}\coloneqq\{x\in A\,\vert\,x=x^{\ast}\}$. 
\begin{proposition}
\label{12011736}
Let $A$ be a complex $\ast$-algebra and $\psi$ a complex linear functional on $A$ such that $\psi(z^{\ast}z)\geq 0$ for every
$z\in A$. Then we have for every $a,b\in A_{ob}$ 
\begin{equation*}
\psi((a-\psi(a))^{2})\psi((b-\psi(b))^{2})\geq\frac{1}{4}|\psi([a,b])|^{2}
\end{equation*}
\end{proposition}
\begin{proof}
Since Lemma \ref{12011735} we know that $\psi$ is hermitian then with the position $x=a-\psi(a)$ and $y=b-\psi(b)$ we deduce 
by the inequality in Lemma \ref{12011735} that 
\begin{equation}
\label{12011736I}
\psi((a-\psi(a))^{2})\psi((b-\psi(b))^{2})
\geq
|\psi(xy)|^{2}.
\end{equation}
Next by letting $[a,b]=ic$ we have 
\begin{equation}
\label{12011736II}
\begin{aligned}
xy
=
\frac{xy+yx}{2}+\frac{xy-yx}{2}
&=
\frac{xy+yx}{2}+\frac{1}{2}[x,y]
\\
&=
\frac{xy+yx}{2}+\frac{[a,b]}{2}
=
\frac{xy+yx}{2}+i\frac{c}{2}
\end{aligned}
\end{equation}
but $xy+yx\in A_{ob}$ as well $c\in A_{ob}$ indeed $c=-i[a,b]$ so, $\psi(\frac{xy+yx}{2})\in\R$ and $\psi(\frac{c}{2})\in\R$ then
\begin{equation}
\label{12011736III}
|\psi(\frac{xy+yx}{2})+i\psi(\frac{c}{2})|^{2}
=
|\psi(\frac{xy+yx}{2})|^{2}+|\psi(\frac{c}{2})|^{2}
\geq
|\psi(\frac{c}{2})|^{2}.
\end{equation}
The statement follows since \eqref{12011736I},\eqref{12011736II},\eqref{12011736III}.
\end{proof}
\begin{theorem}
[\textbf{Indetermination Relations}]
\label{12011848}
Let $\upphi\in\mf{P}_{\mc{A}}^{\natural}$ and $a,b\in A_{ob}$. Then 
\begin{equation*}
\mr{dis}(\upphi,a)\mr{dis}(\upphi,b)
\geq
\left|\frac{\upphi([a,b])}{2\upphi(\un)}\right|^{2}.
\end{equation*}
\end{theorem}
\begin{proof}
The statement follows since the first equality in \eqref{11290928} and Prp.\ref{12011736}.
\end{proof}
\begin{corollary}
[\textbf{Randomicity of the measurement of $a$ and $b$}]
\label{11031210}
Let $\upphi\in\mf{P}_{\mc{A}}^{\natural}$ and $a,b\in A_{ob}$. If $k\coloneqq\upphi([a,b])\neq 0$, then 
\begin{equation*}
\mr{dis}(\upphi,a)\mr{dis}(\upphi,b)\geq|k(2\upphi(\un))^{-1}|^{2}>0.
\end{equation*}
In particular $\mr{dis}(\upphi,a)\neq 0$ and $\mr{dis}(\upphi,b)\neq 0$ namely on the statistical ensemble $\mf{t}(\upphi)$ 
the measurement of the observables $\mf{o}(a)$ and $\mf{o}(b)$ are obtained randomly.
\end{corollary}
\begin{proof}
By Thm.\ref{12011848}, Cor.\ref{11301647}\eqref{11301647II} and Convention \ref{12031227}.
\end{proof}
\begin{corollary}
[\textbf{Asymptotic Complementarity}]
\label{11031132}
Let $X$ be a set, $\mf{F}$ be a filter on $X$ so, by necessity $X$ is nonempty moreover let $Z\in\mf{F}$, 
$\mf{f}:X\to\mf{P}_{\mc{A}}^{\natural}$ and $\mf{a},\mf{b}:X\to A_{ob}$. If 
\begin{equation}
\label{12031308}
\inf_{x\in Z}\left|\frac{\mf{f}(x)([\mf{a}(x),\mf{b}(x)])}{\mf{f}(x)(\un)}\right|>0
\end{equation}
then,
\begin{equation*}
\lim_{x\in\mf{F}}\mr{dis}(\mf{f}(x),\mf{b}(x))=0\Rightarrow\lim_{x\in\mf{F}}\mr{dis}(\mf{f}(x),\mf{a}(x))=+\infty.
\end{equation*}
\end{corollary}
\begin{proof}
By Thm.\ref{12011848} and \eqref{12031308} we deduce that 
\begin{equation*}
\exists\,\mf{g}:Z\ni x\mapsto\frac{1}{\mr{dis}(\mf{f}(x),\mf{b}(x))}\in\widetilde{\R}
\end{equation*}
and that 
\begin{equation}
\label{12031152I}
(\forall x\in Z)
\left(\mr{dis}(\mf{f}(x),\mf{a}(x))
\geq
\left|\frac{\mf{f}(x)([\mf{a}(x),\mf{b}(x)])}{2\mf{f}(x)(\un)}\right|^{2}
\mf{g}(x)\right).
\end{equation}
Next the limit in the hypothesis implies
$\lim_{x\in\mr{Tr}_{Z}(\mf{F})}\mr{dis}(\mf{f}(x),\mf{b}(x))=0$ therefore
since the first statement in \eqref{11290928} we obtain
\begin{equation}
\label{12031152II}
\lim_{\mr{Tr}_{Z}(\mf{F})}\mf{g}=+\infty.
\end{equation}
By \eqref{12031152II} and \eqref{12031308} we obtain 
\begin{equation}
\label{12031152III}
\lim_{x\in\mr{Tr}_{Z}(\mf{F})}
\left|\frac{\upphi([\mf{a}(x),\mf{b}(x)])}{2\mf{f}(x)(\un)}\right|^{2}
\mf{g}(x)=+\infty;
\end{equation}
then since \eqref{12031152I} we obtain 
$\lim_{x\in\mr{Tr}_{Z}(\mf{F})}\mr{dis}(\mf{f}(x),\mf{a}(x))=+\infty$ therefore the statement follows.
\end{proof}
\begin{remark}
We can summarize the statement of Cor.\ref{11031132} by stating that: If the bounding hypothesis \eqref{12031308} is satisfied,
then the more \emph{precisely} the map of observables $\mf{b}$ is measured on the map of the statistical ensembles $\mf{f}$, 
the more \emph{randomly} the map of observables $\mf{a}$ is measured on the map of the statistical ensembles $\mf{f}$.
\end{remark}
\begin{definition}
$\mc{A}$ has the quasi-Heisenberg propery relative to $a,b\in\mc{A}_{ob}$ iff
\begin{equation*}
(\forall\phi\in\mh{P}_{\mc{A}}^{\natural})(\phi([a,b])\neq 0).
\end{equation*}
$\mc{A}$ has the Heisenberg propery relative to $a,b\in\mc{A}_{ob}$ iff there exists $h\in\complement_{\R}\{0\}$ such that 
$[a,b]=ih\un$.
\end{definition}
Clearly the Heisenberg propery relative to $a,b\in\mc{A}_{ob}$ implies the quasi-Heisenberg propery relative to 
$a,b\in\mc{A}_{ob}$.
\begin{corollary}
[\textbf{Quasi-Heisenberg property implies irreducible randomicity and Heisenberg property implies asymptotic complementarity}]
\label{12011849}
Assume that $\mc{A}$ has the quasi-Heisenberg property relative to $a,b\in\mc{A}_{ob}$. Then,
\begin{equation}
\label{12011849a}
(\forall\upphi\in\mh{P}_{\mc{A}}^{\natural})(\mr{dis}(\upphi,a)\mr{dis}(\upphi,b)>0).
\end{equation}
In particular for every $\upphi\in\mh{P}_{\mc{A}}^{\natural}$ we have $\mr{dis}(\upphi,a)\neq 0$ and $\mr{dis}(\upphi,b)\neq 0$;
id est $\omega\in\mh{P}_{\mc{A}}^{\natural}$ implies that on the statistical ensemble $\mf{t}(\omega)$ the measurement of 
the observables $\mf{o}(a)$ and $\mf{o}(b)$ are obtained randomly; equivalently $\mc{A}$ is \emph{irreducibly random}.
\par
Assume in addition $\mc{A}$ has the Heisenberg property relative to $a,b\in\mc{A}_{ob}$ and $h\in\complement_{\R}\{0\}$ such that 
$[a,b]=ih\un$. Then we have the following sharper estimation 
\begin{equation}
\label{12011849I}
(\forall\upphi\in\mf{P}_{\mc{A}}^{\natural})(\mr{dis}(\upphi,a)\mr{dis}(\upphi,b)\geq|h(2\upphi(\un))^{-1}|^{2}>0).
\end{equation}
Furthermore let $X$ be a set, $\mf{F}$ be a filter on $X$, and $\mf{f}:X\to\mf{E}_{\mc{A}}$. Thus,
\begin{equation}
\label{12011849II}
\lim_{x\in\mf{F}}\mr{dis}(\mf{f}(x),b)=0\Rightarrow\lim_{x\in\mf{F}}\mr{dis}(\mf{f}(x),a)=+\infty.
\end{equation}
\end{corollary}
\begin{proof}
\eqref{12011849a} follows since Thm.\ref{12011848}, \eqref{12011849I} follows since Cor.\ref{11031210}, 
\eqref{12011849II} follows since Cor.\ref{11031132}.
\end{proof}
Next let $\mf{H}$ be a complex Hilbert space. Define 
$\partial\mf{H}\coloneqq\{v\in\mf{H}\,\vert\,\|v\|=1\}$ and for every $v\in\partial\mf{H}$ let $P_{v}\in\Pr(\mf{H})$
such that $P_{v}w\coloneqq v\lr{v}{w}$. Finally let $\mf{L}_{s}(\mf{H})$ be the locally convex space whose underlying linear 
space equals $\mf{L}(\mf{H})$ and whose topology is the strong operator topology.
\begin{lemma}
\label{11261951}
The map $\partial\mf{H}\ni v\mapsto P_{v}\in\mf{L}_{s}(\mf{H})$ is continuous.
\end{lemma}
\begin{proof}
Let $\{v_{n}\}_{n\in\N}$ be a sequence in $\partial\mf{H}$ and $v\in\partial\mf{H}$ such that 
$\lim_{n\in\N}v_{n}=v$ in the topological space $\partial\mf{H}$ so, 
$\lim_{n\in\N}v_{n}=v$ in the topological space $\mf{H}$ namely $\lim_{n\in\N}(v_{n}-v)=\ze$.
Now let $w\in\mf{H}$ thus, 
\begin{equation*}
\begin{aligned}
(P_{v_{n}}-P_{v})w
&=
v_{n}\lr{v_{n}}{w}-v\lr{v}{w}
\\
&=
v_{n}\lr{v_{n}}{w}-v\lr{v_{n}}{w}
+v\lr{v_{n}}{w}-v\lr{v}{w}
\\
&=
(v_{n}-v)\lr{v_{n}}{w}
+v
\lr{v_{n}-v}{w}.
\end{aligned}
\end{equation*}
Next the sequence $n\mapsto\lr{v_{n}}{w}$ is bounded say by $M>0$ being convergent therefore, 
\begin{equation*}
\|(P_{v_{n}}-P_{v})w\|\leq\|(v_{n}-v)\|M+\|v\|\,\|v_{n}-v\|\,\|w\|,
\end{equation*}
thus $\lim_{n\in\N}\|(P_{v_{n}}-P_{v})w\|=0$ and the statement follows.
\end{proof}
\subsection{Discrete Semiobservables}
\label{01192115}
\begin{remark}
[\textbf{von Neumann channel associated with a discrete observable}]
\label{11041328}
If $A$ is a discrete observable (Rmk.\ref{10311634}) on $\mf{H}$ with value space $(\R,\mf{B}(\R))$, 
then according to \cite[(9.3)]{oza1} there exists $\mf{a}$ a measuring process of $A$
whose associated channel map $\mf{I}_{\mf{a}}$ (Rmk.\ref{10311634}) is the usual von-Neumann map, namely by letting
$\mc{M}\coloneqq\lr{\mf{L}(\mf{H})}{\sigma(\mf{L}(\mf{H}),\mf{L}(\mf{H})_{\ast})}$ we have 
\begin{equation}
\label{10301254}
(\forall B\in\mf{B}(\R))
\begin{cases}
B\cap\sigma(o^{A})\neq\emptyset\Rightarrow
\mf{I}_{\mf{a}}(B)=\sum_{\lambda\in B\cap\sigma(o^{A})}\upzeta_{\mc{M}}^{\dagger}(A(\{\lambda\})),
\\
B\cap\sigma(o^{A})=\emptyset\Rightarrow
\mf{I}_{\mf{a}}(B)=\ze;
\end{cases}
\end{equation}
sum converging in $\mf{L}_{s}(\mc{M}_{s}^{\ast})$. 
Note that $\mc{M}_{s}^{\ast}=\lr{\mf{L}(\mf{H})_{\ast}}{\sigma(\mf{L}(\mf{H})_{\ast},\mf{L}(\mf{H}))}$.
We call $\mf{a}$ \emph{the von Neumann measuring process associated with the discrete observable $A$},
and call $\mf{I}_{\mf{a}}$ \emph{the von Neumann channel map associated with the discrete observable $A$}.
\end{remark}
\begin{remark}
[\textbf{Weak repeatibility}]
By definition given in \cite[p.293]{ozacp} see also \cite[(R) p.247]{dl} 
and translated into our context we say that a channel valued 
$\sigma$-additive measure $F:\mf{B}\to\mf{Z}(\mc{A})$, where $(\Omega,\mf{B})$ is a Borel space and $\mc{A}\in\mr{tsa}$, 
is \emph{weakly repeatable} iff 
\begin{equation*}
(\forall B,C\in\mf{B})(\forall\upomega\in\mf{P}_{\mc{A}}^{\natural})
(\mf{b}_{\mc{A}}(F(B\cap C),\upomega,\un)=\mf{b}_{\mc{A}}(F(B)\circ F(C),\upomega,\un)).
\end{equation*}
Thus, \cite[Theorem 5.2 p.294]{ozacp} states what follows
\begin{corollary}
[\textbf{Weak repeatibility implies discreteness}]
\label{10281805}
Let $X$ be a semiobservable on a separable Hilbert space $\mf{H}$ and with value space $(\R,\mf{B}(\R))$, let $\mf{x}$
be a measuring process of $X$. If $\mf{I}_{\mf{x}}$ is weakly repeatable then $X$ is discrete.
\end{corollary}
\end{remark}
\begin{remark}
[\textbf{Projector valued conditional expectation induced by a discrete observable}]
\label{10211925}
Since \cite[Thm. 3.4]{tak1} any conditional expectation $\mc{E}$ \cite[Def. 3.3]{tak1} is positive and unit preserving 
thus, it belongs to $\mr{Q}(\mc{M},\mc{N})$, with $\mc{M}$ and $\mc{N}$ von Neumann algebras.
Next \cite[Thm.4.2]{davqos} shows that for every observable $a\in\mc{M}_{ob}$ with discrete spectrum, 
for instance a compact observable, there exists a unique normal conditional expectation $\mc{E}$ from $\mc{M}$ 
onto $\mc{N}$, the relative commutant in $\mc{M}$ of the set $E_{a}(\mc{B}(\R))$ where $E_{a}$ is the resolution of the 
identity of $a$, such that $\mc{E}^{\dagger}\omega$ is the statistical ensemble obtained after the observable $a$ has been 
measured over the statistical ensemble $\omega\in\mf{P}_{\mc{N}}$. Notice that by definition $\mc{E}$ is a projector, namely
$\mc{E}=\mc{E}\circ\imath_{\mc{N}}^{\mc{M}}\circ\mc{E}$ so, 
$\mc{E}^{\dagger}=\mc{E}^{\dagger}\circ(\imath_{\mc{N}}^{\mc{M}})^{\dagger}\circ\mc{E}^{\dagger}$ stating that a second measurement
of the same observable with discrete spectrum, does not change the state obtained by the first one.
We note that a conditional expectation $\mc{E}$ of $\mc{M}$ onto $\mc{N}$ w.r.t. $\upphi$ is defined in 
\cite[Def. 4.1 p.211]{tak2} and it exists if and only if $\mc{N}$ is globally invariant w.r.t. the modular group 
associated with $\upphi$, in such cases $\mc{E}$ is normal and uniquely determined by $\upphi$ \cite[Thm. 4.2 p.211]{tak2}.
More in general \cite[Thm.4.2 p.60]{davqos} allows to associate with $a$ a Cycon-Hellwig conditional expectation defined on 
the space of the discrete spectrum of $a$ provided with the discrete topology see for instance \eqref{10301254}, 
and just because the spectrum is discrete the Cycon-Hellwig conditional expectation associated with $a$ is projector 
valued (i.e. idempotent) which is equivalent to the repeatibility hypothesis. 
\end{remark}
\begin{remark}
[\textbf{Operation characterization of the von Neumann measuring process of a discrete observable}]
\label{10301327} 
Let $X$ be a \underline{discrete} observable on $\mf{H}$ with value space $(\R,\mf{B}(\R))$ and let $\mf{x}$ be 
the von Neumann measuring process associated with $X$ (Rmk. \ref{11041328}) thus,
\begin{equation*}
(\forall\lambda\in\sigma(o^{X}))(\mf{I}_{\mf{x}}(\{\lambda\})=\upzeta^{\dagger}(X(\{\lambda\}))).
\end{equation*}
\emph{The channel executing the measurement of the semiobservable $\mf{o}(X)$ and then selecting the subensemble
of those preparations such that the value $\{\lambda\}$ is detected
equals
the channel induced by the device selecting the subensemble of those preparations such that the value
$\{\lambda\}$ is detected, conditioned by the execution of the measurement of the semiobservable $\mf{o}(X)$.}
\end{remark}
\begin{definition}
[\textbf{Construction of discrete observables}]
\label{10311739}
Let $p$ be a family defined on $Z\subseteq\Z_{\geq}$ of orthogonal projectors in a Hilbert space $\mf{H}$
such that $p_{i}p_{j}=\ze$ if $i\neq j$ and $i,j\in Z$, 
and such that $\sum_{i\in Z}p_{i}=\mr{Id}_{\mf{H}}$ sum converging in $\mc{M}$, where 
$\mc{M}\coloneqq\lr{\mf{L}(\mf{H})}{\sigma(\mf{L}(\mf{H}),\mf{L}(\mf{H})_{\ast})}$,\footnote{Equivalently in weak operator 
topology since on the unit ball weak operator and sigma weak operator topologies coincide.}
and let $\lambda$ be a family defined on $Z$ 
of real numbers. We call $p$ a spectral map on $\mf{H}$ defined on $Z$ and $(p,\lambda)$ a spectral couple on $\mf{H}$ 
defined on $Z$. Define
\begin{equation*}
(\forall B\in\mf{B}(\R))
\begin{cases}
\overset{-1}{\lambda}(B)\neq\emptyset
\Rightarrow
E_{(p,\lambda)}(B)\coloneqq\sum_{i\in\overset{-1}{\lambda}(B)}p_{i},
\\
\overset{-1}{\lambda}(B)=\emptyset
\Rightarrow
E_{(p,\lambda)}(B)\coloneqq\ze;
\end{cases}
\end{equation*}
sum converging in $\mc{M}$. Thus, $E_{(p,\lambda)}$ is called the discrete observable associated with $(p,\lambda)$. 
Easily we see that $E_{(p,\lambda)}$ is a spectral measure and therefore, we can define 
$\lr{p}{\lambda}\coloneqq\int\imath\,dE_{(p,\lambda)}$ in the sense of the functional
calculus associated with any spectral measure in a Banach space. 
\end{definition}
\begin{remark}
Let $(p,\lambda)$ be a spectral couple on $\mf{H}$ defined on $Z$ thus, $\lr{p}{\lambda}$ is a possibly 
unbounded selfadjoint operator in $\mf{H}$ whose discrete spectrum equals $\lambda(Z)$, and it is such that 
\begin{equation*}
(\forall v\in\mr{Dom}(\lr{p}{\lambda}))(\lr{p}{\lambda}v=\sum_{i\in Z}\lambda_{i}p_{i}v).
\end{equation*}
\end{remark}
\begin{remark}
[\textbf{Detected versus concealed alternatives. How they combine}]
\label{11041528}
Let $(p,\lambda)$ be a spectral couple on a Hilbert space $\mf{H}$ defined on $Z$ and $\mf{x}$ be the von Neumann 
measuring process associated with the discrete observable associated with $(p,\lambda)$ 
(Rmk.\ref{11041328} and Def.\ref{10311739}).
Let $\mc{M}\coloneqq\lr{\mf{L}(\mf{H})}{\sigma(\mf{L}(\mf{H}),\mf{L}(\mf{H})_{\ast})}$ thus, $\mf{I}_{\mf{x}}$ being by 
definition the dual of an instrument, if $\{B_{i}\}_{i\in\Z}\subset\mf{B}(\R)$ is a family of mutually disjoint sets, 
then we have 
\begin{equation}
\label{11041528a}
\mf{I}_{\mf{x}}(\bigcup_{i\in\Z}B_{i})=\sum_{i\in\Z}\mf{I}_{\mf{x}}(B_{i}),
\end{equation}
sum converging in $\mf{L}_{s}(\mc{M}_{s}^{\ast})$, 
note that $\mc{M}_{s}^{\ast}=\lr{\mf{L}(\mf{H})_{\ast}}{\sigma(\mf{L}(\mf{H})_{\ast},\mf{L}(\mf{H}))}$,
while \eqref{10301254} and Def.\ref{10311739} yield
\begin{equation}
\label{11041528c}
\begin{aligned}
\beta\in\lambda(Z)&\Rightarrow\mf{I}_{\mf{x}}(\{\beta\})=\upzeta^{\dagger}(\sum_{i\in\overset{-1}{\lambda}(\{\beta\})}p_{i});
\\
\beta\notin\lambda(Z)&\Rightarrow\mf{I}_{\mf{x}}(\{\beta\})=\ze.
\end{aligned}
\end{equation}
In general we have
\begin{equation}
\label{11041528b}
\upzeta^{\dagger}(\sum_{i\in\overset{-1}{\lambda}(\{\beta\})}p_{i}))\neq
\sum_{i\in\overset{-1}{\lambda}(\{\beta\})}\upzeta^{\dagger}(p_{i}).
\end{equation}
The right-hand side of \eqref{11041528a} and \eqref{11041528b} are channels limit in $\mf{L}_{s}(\mc{M}_{s}^{\ast})$ of 
filters of detected alternatives of channels, while the left-hand side of \eqref{11041528b} is the 
channel induced by the device implementing an operation which is the weak operator topology limit of a filter of 
concealed alternatives of operations.
\par
The property of the above dual instrument $\mf{I}_{\mf{x}}$, of encoding \emph{detected} alternatives of channels as in 
\eqref{11041528a} as well that of encoding channels induced by the device implementing \emph{concealed} alternatives of 
operations as in \eqref{11041528c}, makes $\mf{I}_{\mf{x}}$ one of those maps of channels where the detected and concealed 
alternatives combine. In the next subsection we analize quantitatively the difference between concealed and detected 
alternatives in a wider context than that established in \eqref{11041528b}.
\end{remark}
\subsection{Concealed versus detected alternatives. Interference phenomenon versus decoherence}
\label{08091625}
In the present subsection we fix $\mc{A}\in\mr{tsa}$. If $I$ is a nonempty, define $\Delta(I)\coloneqq\{(i,i)\,\vert\,i\in I\}$ 
called the diagonal of the set $I\times I$. Let us start with the following
\begin{definition}
[\textbf{Interference functional}]
\label{01191907}
Let $I$ be a nonempty and $a:I\to\mc{A}$. If $I$ is infinite, then in addition assume that $\sum_{(i,j)\in I\times I}a_{i}ba_{j}^{\ast}$
converges in $\mc{A}$ for every $b\in\mc{A}$. Define 
\begin{equation*}
\begin{aligned}
\mr{Int}(\cdot,a)\in\mf{L}(\mc{A}_{s}^{\ast}),\,\omega&\mapsto\mr{Int}(\omega,a),
\\
\mr{Int}(\omega,a):\mc{A}\ni b&\mapsto\sum_{(i,j)\in\complement_{I\times I}\Delta(I)}\omega(a_{i}ba_{j}^{\ast}).
\end{aligned}
\end{equation*}
$\mr{Int}(\omega,a)$ is called the \emph{interference functional associated with $\omega$ and $a$}.
\end{definition}
\begin{remark}
\label{01191906}
Under the assumption of Def.\ref{01191907} assumen in addition that $a_{i}a_{j}^{\ast}=\ze$ for every 
$(i,j)\in\complement_{I\times I}\Delta(I)$. Thus for every $\omega\in\mc{A}^{\ast}$ and $b\in\mc{A}$ we have 
$b\in a(I)^{\prime}\Rightarrow\mr{Int}(\omega,a)(b)=0$ equivalently
\begin{equation*}
\mr{Int}(\omega,a)(b)\neq 0\Rightarrow b\notin a(I)^{\prime}.
\end{equation*}
\end{remark}
\begin{theorem}
[\textbf{Concealed versus detected alternatives}]
\label{01182128}
Let $I$ be a nonempty and $a:I\to\Uplambda_{\mc{A}}$ 
such that $a_{i}a_{j}^{\ast}=\ze$ for every $(i,j)\in\complement_{I\times I}\Delta(I)$. 
If $I$ is infinite, then in addition assume that $\mc{A}$ is complete,
$\sum_{i\in I}a_{i}$ and $\sum_{(i,j)\in I\times I}a_{i}ba_{j}^{\ast}$ converge in $\mc{A}$ for every $b\in\mc{A}$.
Thus, for every $b\in\mc{A}$ we have 
\begin{equation*}
\upzeta(\sum_{i\in I}a_{i})(b)
=
\sum_{i\in I}\upzeta(a_{i})(b)
+\sum_{(i,j)\in\complement_{I\times I}\Delta(I)}a_{i}ba_{j}^{\ast}.
\end{equation*}
In particular if $\sum_{i\in I}a_{i}\in\Uplambda_{\mc{A}}$, then for every $\omega\in\mf{P}_{\mc{A}}$ we obtain
\begin{equation*}
\upzeta^{\dagger}(\sum_{i\in I}a_{i})(\omega)
=
\sum_{i\in I}\upzeta^{\dagger}(a_{i})(\omega)+\mr{Int}(\omega,a),
\end{equation*}
where the sum $\sum_{i\in I}\upzeta^{\dagger}(a_{i})(\omega)$ converges in $\lr{\mc{A}^{\ast}}{\sigma(\mc{A}^{\ast},\mc{A})}$.
\end{theorem}
\begin{proof}
If $I$ is finite the statement is easily provable by direct calculation so, let us assume that $I$ is infinite.
Next for every $b\in\mc{A}$ we have 
\begin{equation*}
\begin{aligned}
\upzeta(\sum_{i\in I}a_{i})(b)
&=
(\sum_{i\in I}a_{i})b(\sum_{i\in I}a_{i}^{\ast})
\\
&=
(\sum_{i\in I}a_{i}b)\sum_{i\in I}a_{i}^{\ast}
\\
&=
\sum_{j\in I}(\sum_{i\in I}a_{i}b)a_{j}^{\ast}
\\
&=
\sum_{j\in I}\sum_{i\in I}a_{i}ba_{j}^{\ast}
\\
&=
\sum_{(i,j)\in I\times I}a_{i}ba_{j}^{\ast}
\\
&=
\sum_{(i,j)\in\Delta(I)}a_{i}ba_{j}^{\ast}
+
\sum_{(i,j)\in\complement_{I\times I}\Delta(I)}a_{i}ba_{j}^{\ast};
\end{aligned}
\end{equation*}
where the first equality follows since the continuity of the involution on $\mc{A}$, the second, third and fourth equalities 
follow by the separate continuity of the product in $\mc{A}$, next since the hypothesis that 
$\sum_{(i,j)\in I\times I}a_{i}ba_{j}^{\ast}$ converges in $\mc{A}$ and that $\mc{A}$ is complete we have by 
\cite[III.39 Thm.2]{BourGT} the fifth equality, the sixth equality follows by \cite[III.39 Thm.2]{BourGT}.
\end{proof}
\begin{remark}
Thm.\ref{01182128} is the most general result relating the device implementing a concealed alternative of operations with the 
detected alternative of the devices implementing the same operations.
More precisely it states that the difference between the statistical ensemble resulting next the channel induced by the device 
implementing the \emph{concealed} alternative of the operations $a(I)$, applies to $\mf{T}(\omega)$ and the statistical ensemble 
resulting next the \emph{detected} alternative of the channels induced by the devices implementing the operations $a(I)$, 
applies to $\mf{T}(\omega)$ equals the \emph{interference} functional associated with $\omega$ and $a$. 
Therefore for every $b\in\mc{A}_{ob}$ according to Rmk.\ref{01191906} if the expectation value of $\mf{O}(b)$
in these two resulting statistical ensembles differ, then necessarily $b$ does \emph{not} commute with some operation in $a(I)$.
\end{remark}
What just above established proves that the channel induced by a concealed alternative of a set of operations $a(I)$ and the 
detected alternative of the channels induced by $a(I)$ are not always different, to the contary when applied to $\omega$ they 
equal whenever the value of the interference functional associated with $\omega$ and $a$ is zero. Next we are interested on any 
channel that via composition to the right maps the former channel to the latter. 
\par
More precisely let us now call decoherence with respect to a family $a$ of operations any channel such that for every 
$\omega\in\mf{P}_{\mc{A}}$ it transforms the statistical ensemble resulting next the channel induced by the device implementing 
the \emph{concealed} alternative of the operations $a(I)$, applies to $\mf{T}(\omega)$ into the the statistical ensemble 
resulting next the \emph{detected} alternative of the channels induced by the devices implementing the operations $a(I)$, 
applies to $\mf{T}(\omega)$, namely 
\begin{definition}
[\textbf{Decoherence w.r.t. a family of operations}]
Assume the hypothesis of Thm.\ref{01182128} in addition let $\mc{B}\in\mr{tsa}$ and $J\in\mf{Z}(\mc{A},\mc{B})$. We say that 
\emph{$J$ is a decoherence with respect to the family $a$ of operations} iff 
\begin{equation*}
J\circ\upzeta^{\dagger}(\sum_{i\in I}a_{i})
=
J\circ\sum_{i\in I}\upzeta^{\dagger}(a_{i}).
\end{equation*}
\end{definition}
\begin{remark}
Usually in literature decoherence deals with a channel that does not preserve pure states. Here we are interested to a specific
kind of channels related to $a$ that do not preserve pure states, more precisely we can state that if $J$ is a decoherence with 
respect to the family $a$ of operations and $J$ does not map nonpure states into pure ones, then 
$J\circ\upzeta^{\dagger}(\sum_{i\in I}a_{i})$ does not preserve pure states since the channel $\sum_{i\in I}\upzeta^{\dagger}(a_{i})$
maps every pure state into a nonpure state.
\end{remark}
Next we can state the following easy 
\begin{corollary}
[\textbf{Characterization of decoherence with respect to the family $a$ of operations}]
\label{01182221p}
Assume the hypothesis of Thm.\ref{01182128} in addition let $\mc{B}\in\mr{tsa}$ and $J\in\mf{Z}(\mc{A},\mc{B})$. Then,
\begin{equation*}
J\circ\mr{Int}(\cdot,a)=\ze
\Leftrightarrow
J\circ\upzeta^{\dagger}(\sum_{i\in I}a_{i})
=
J\circ\sum_{i\in I}\upzeta^{\dagger}(a_{i}).
\end{equation*}
\end{corollary}
\begin{proof}
The statement follows directly since Thm.\ref{01182128}.
\end{proof}
\begin{remark}
[Collapse channels]
Under the notation in Rmk.\ref{10311634} assume that $X$ and $Y$ are discrete in the sense of Rmk.\ref{11041328} 
moreover assume that $\sigma$ is a vector state. Then we obtain as a result of Cor.\ref{01182221p} and \eqref{10311837} that 
$(R_{Y(B)}^{\otimes})^{\dagger}$ is a \emph{decoherence} with respect to the family $\{P_{i}\otimes P_{\eta_{i}}\}_{i}$ of operations 
where $P_{i}=\sum_{j}P_{\xi_{ij}}$, $P_{v}$ is the one-dimensional projection associated with vector $v$, and 
where $\xi_{ij}$ and $\eta_{i}$ are as in \cite[section 9]{oza1} corresponding to the above $\mf{x}$. 
Paraphrasing von Neumann, $(R_{Y(B)}^{\otimes})^{\dagger}$ will be called a \emph{collapse channel}.
Since the decomposition \eqref{08020326} and since $\upzeta_{\mc{N}}^{\dagger}(\un_{\mf{H}}\otimes Y(B))$ preserves the coherence
and purity of the functionals, we deduce that the decoherence nature of $(R_{Y(B)}^{\otimes})^{\dagger}$ is ascribible only to the 
factor $(R_{\un_{\mf{K}}}^{\otimes})^{\dagger}$ which equals $(\imath_{\mc{M}}^{\mc{N}})^{\dagger}$. Thus, we have that the 
decoherence collapse channel $(R_{Y(B)}^{\otimes})^{\dagger}$ factors into the composition of the coherence preserving and purity 
preserving channel $\upzeta_{\mc{N}}^{\dagger}(\un_{\mf{H}}\otimes Y(B))$ - which is the gnoseological (non-physical) channel 
consisting in selecting the subensemble of those preparations in which a result belonging to $B$ is detected when measuring 
$Y$ - and the gnoseological but decoherence channel $(\imath_{\mc{M}}^{\mc{N}})^{\dagger}$ which 
restricts any ensemble of preparations of $\mc{N}$ into an ensemble of preparations of $\mc{M}$. 
\emph{Therefore in conclusion the decoherence nature of the collapse channel $(R_{Y(B)}^{\otimes})^{\dagger}$ is attributable 
exclusively to the gnoseological then non-physical channel that restricts any ensemble of preparations of $\mc{N}$ into an 
ensemble of preparations of $\mc{M}$.} 
\end{remark}
Next we apply the above result to state the following
\begin{corollary}
[The channel $(\imath_{a(I)^{\prime}}^{\mc{A}})^{\dagger}$ is a decoherence w.r.t. $a$]
\label{01182221}
Under the hypothesis of Thm.\ref{01182128} we obtain
\begin{equation*}
(\imath_{a(I)^{\prime}}^{\mc{A}})^{\dagger}\circ\upzeta^{\dagger}(\sum_{i\in I}a_{i})
=
(\imath_{a(I)^{\prime}}^{\mc{A}})^{\dagger}\circ\sum_{i\in I}\upzeta^{\dagger}(a_{i}).
\end{equation*}
\end{corollary}
\begin{proof}
We have $\mr{Int}(\omega,a)\circ\imath_{a(I)^{\prime}}^{\mc{A}}=\ze$ for every $\omega\in\mc{A}^{\ast}$ thus, the statement follows 
since Cor.\ref{01182221p}.
\end{proof}
Thus, the restriction to the commutant of $a(I)$ of the statistical ensemble resulting next the channel induced by the device 
implementing the \emph{concealed} alternative of the operations $a(I)$, applies to $\mf{T}(\omega)$ 
equals the restriction to the commutant of $a(I)$ of the statistical ensemble resulting next the \emph{detected} alternative of 
the channels induced by the devices implementing the operations $a(I)$, applies to $\mf{T}(\omega)$. 
\par
A similar result for a different but closely related concept of decoherence will be obtained in Thm.\ref{01051727} regarding the
concealed alternative of the product of two families of operations and then restricting to the commutant of the range of one of 
the families.
Next we pass to analyze in terms of probability the case in which the two channels differ
\begin{corollary}
[\textbf{Generalized interference phenomenon}]
\label{01181148}
Under the hypothesis of Thm.\ref{01182128} we have for every $\omega\in\mf{P}_{\mc{A}}^{\natural}$ and every $c\in\Uplambda_{\mc{A}}$ 
that 
\begin{equation*}
\begin{aligned}
\mf{b}_{\mc{A}}\bigl(\upzeta^{\dagger}(c)\circ\upzeta^{\dagger}(\sum_{i\in I}a_{i}),\omega,\un\bigr)
&=
\mf{b}_{\mc{A}}\bigl(\upzeta^{\dagger}(c)\circ\sum_{i\in I}\upzeta^{\dagger}(a_{i}),\omega,\un\bigr)
+
\frac{\mr{Int}(\omega,a)(cc^{\ast})}{\omega(\un)}
\\
&=
\sum_{i\in I}\mf{b}_{\mc{A}}\left(\upzeta^{\dagger}(c)\circ\upzeta^{\dagger}(a_{i}),\omega,\un\right)
+
\frac{\mr{Int}(\omega,a)(cc^{\ast})}{\omega(\un)}.
\end{aligned}
\end{equation*}
\end{corollary}
\begin{proof}
Since Thm.\ref{01182128}. 
\end{proof}
\begin{corollary}
\label{12281837}
Under the hypothesis of Thm.\ref{01182128} we have \eqref{12281837a}$\Leftrightarrow$\eqref{12281837b}, where 
\begin{equation}
\label{12281837a}
\mr{Int}(\omega,a)(cc^{\ast})\neq 0,
\end{equation}
and 
\begin{equation}
\label{12281837b}
\mf{b}_{\mc{A}}\bigl(\upzeta^{\dagger}(c)\circ\upzeta^{\dagger}(\sum_{i\in I}a_{i}),\omega,\un\bigr)
\neq
\sum_{i\in I}\mf{b}_{\mc{A}}\left(\upzeta^{\dagger}(c)\circ\upzeta^{\dagger}(a_{i}),\omega,\un\right).
\end{equation}
\end{corollary}
\begin{proof}
Since Cor.\ref{01181148}.
\end{proof}
\begin{definition}
Under the hypothesis of Thm.\ref{01182128} we say that the \emph{generalized interference phenomenon occurs} iff the two 
equivalent statements of Cor.\ref{12281837} are satisfied.
\end{definition}
In order to interpret this result let us introduce the following 
\begin{convention}
\label{01151241}
Let $c\in\Uplambda_{\mc{A}}$ and $\omega\in\mf{P}_{\mc{A}}^{\dagger}$ thus, we let the first statement below be a shortcut of the 
second one, where
\begin{itemize}
\item
\emph{the detection conditioned by $\mf{T}(\omega)$ of the action of $\mf{R}(c)$ over}  
\item
\emph{the empirical representation of the propensity conditioned by $\mf{T}(\omega)$ to detect the effect of selecting one output
when tested on the statistical ensemble resulting next the channel induced by the device implementing the operation of filtering
through $\mf{r}(c)$ applies to}.
\end{itemize}
\end{convention}
Therefore, Cor.\ref{01181148} establishes that 
\begin{remark}
[If the generalized interference phenomenon occurs, then the system cannot be classical]
\label{01192039}
The detection conditioned by $\mf{T}(\omega)$ of the action of $\mf{R}(c)$ over the statistical ensemble resulting next the 
channel induced by the device selecting the \emph{concealed} alternative of the operations in $a(I)$, applies to 
$\mf{T}(\omega)$ differs of the amount 
\begin{equation*}
\frac{\mr{Int}(\omega,a)(cc^{\ast})}{\omega(\un)}
\end{equation*}
from the the detection conditioned by $\mf{T}(\omega)$ of the action of $\mf{R}(c)$ over the \emph{detected} alternative of the 
channels induced by the devices selecting the operations in $a(I)$, applies to $\mf{T}(\omega)$.
Next according to Rmk.\ref{01191906} if the generalized interference phenomenon occurs, 
then necessarily the effect $cc^{\ast}$ does \emph{not} commute with some operation in $a(I)$, in other words
\begin{equation*}
\mr{Int}(\omega,a)(cc^{\ast})\neq 0\Rightarrow cc^{\ast}\notin a(I)^{\prime}.
\end{equation*}
As a result if the generalized interference phenomenon occurs, then the observable algebra of the system in which it occurs
is noncommutative therefore the system cannot be classical.
\end{remark}
Next we specialize Cor.\ref{01181148} to the case $I=\{1,2\}$ to explain the interference phenomenon occurring in the 
emblematic double slit experiment. This in order to elucidate in this very simple case the concepts of concealed and detected 
alternative and how the nonzero difference between them implies the noncommutative nature of the observable algebra $\mc{A}$. 
Let $a_{1},a_{2},c\in\mc{A}$, set
\begin{equation*}
\begin{cases}
\lr{a_{1}}{c,a_{2}}\coloneqq a_{1}\upzeta_{\mc{A}}(c)(\un)a_{2}^{\ast};
\\
\mr{Int}(a_{1},a_{2},c)\coloneqq\lr{a_{1}}{c,a_{2}}+\lr{a_{1}}{c,a_{2}}^{\ast}.
\end{cases}
\end{equation*} 
It is worthwhile remarking that by assuming $a_{1}a_{2}^{\ast}=\ze$ we have that if 
$cc^{\ast}\in\{a_{1}\}^{\prime}\cup\{a_{2}\}^{\prime}$, then $\mr{Int}(a_{1},a_{2},c)=\ze$, equivalently 
$\mr{Int}(a_{1},a_{2},c)\neq\ze$ implies that 
$cc^{\ast}\in\complement_{\mc{A}}(\{a_{1}\}^{\prime})\cap\complement_{\mc{A}}(\{a_{2}\}^{\prime})$.
Let $a_{1},a_{2},c\in\Uplambda_{\mc{A}}$ such that $a_{1}+a_{2}\in\Uplambda_{\mc{A}}$ thus by Thm.\ref{01182128}, 
\begin{equation*}
\upzeta_{\mc{A}}((a_{1}+a_{2})c)(\un)=\sum_{i=1}^{2}\upzeta_{\mc{A}}(a_{i}c)(\un)+\mr{Int}(a_{1},a_{2},c).
\end{equation*}
Therefore, namely by Cor.\ref{01181148}, 
\begin{equation}
\label{11211153} 
\begin{aligned}
\mf{b}_{\mc{A}}\left(\upzeta_{\mc{A}}^{\dagger}(c)\circ\upzeta_{\mc{A}}^{\dagger}(a_{1}+a_{2}),\upomega,\un\right)
&=
\mf{b}_{\mc{A}}\left(\upzeta_{\mc{A}}^{\dagger}(c)\circ\sum_{i=1}^{2}\upzeta_{\mc{A}}^{\dagger}(a_{i}),\upomega,\un\right)
+\frac{\upomega(\mr{Int}(a_{1},a_{2},c))}{\upomega(\un)}
\\
&=
\sum_{i=1}^{2}\mf{b}_{\mc{A}}\left(\upzeta_{\mc{A}}^{\dagger}(c)\circ\upzeta_{\mc{A}}^{\dagger}(a_{i}),\upomega,\un\right)
+\frac{\upomega(\mr{Int}(a_{1},a_{2},c))}{\upomega(\un)}.
\end{aligned}
\end{equation}
The first equality above yields:
\emph{The empirical representation of the propensity conditioned by the statistical ensemble 
$\mf{t}(\upomega)$ to detect the effect of selecting one output when tested on the statistical ensemble resulting next $Z$
following $X$ applies to the statistical ensemble $\mf{t}(\upomega)$, differs of 
the amount $\upomega(\un)^{-1}\upomega(\mr{Int}(a_{1},a_{2},c))$ from the empirical representation of the propensity 
conditioned by the statistical ensemble $\mf{t}(\upomega)$ to detect the effect of selecting one output when tested on 
the statistical ensemble resulting next $Z$ following $Y$ applies to the statistical ensemble $\mf{t}(\upomega)$. Here
\begin{itemize}
\item
$Z=$ the channel induced by the device implementing the operation of filtering through $\mf{r}(c)$; 
\item
$X=$ the channel induced by the device implementing the operation of filtering through $\mf{r}(a_{1})$ 
in \underline{concealed} alternative to the operation of filtering through $\mf{r}(a_{2})$;
\item
$Y=$ the channel induced by the device implementing the operation of filtering through $\mf{r}(a_{1})$ 
in \underline{detected} alternative to the channel induced by the device implementing the operation of filtering 
through $\mf{r}(a_{2})$.
\end{itemize}}
Now let us put into play time translation. In order to do this let $V:\R\to\mr{U}(\mc{A})$ be a group morphism, 
let $\uptau_{V}$ as in \eqref{11021356} simply denoted as $\uptau$, and let $t_{i}\in\R$ with $i\in\{0,1,2,3\}$ 
thus, by letting $x(t_{0},t_{1},t_{2})\coloneqq V(t_{1}-t_{0})xV(t_{2}-t_{1})$ for every $x\in\mc{A}$ 
and by taking into account that since Lemma \ref{08191955} we have 
$x\in\Uplambda_{\mc{A}}\Rightarrow x(t_{0},t_{1},t_{2})\in\Uplambda_{\mc{A}}$ 
and $y\in\Uplambda_{\mc{A}}\Rightarrow yV(t_{3}-t_{2})\in\Uplambda_{\mc{A}}$, we obtain
\begin{multline*}
\uptau^{\dagger}(t_{3}-t_{2})\circ\upzeta_{\mc{A}}^{\dagger}(c)\circ\uptau^{\dagger}(t_{2}-t_{1})
\circ\upzeta_{\mc{A}}^{\dagger}(a_{1}+a_{2})\circ\uptau^{\dagger}(t_{1}-t_{0})
=\\
\upzeta_{\mc{A}}^{\dagger}((a_{1}(t_{0},t_{1},t_{2})+a_{2}(t_{0},t_{1},t_{2}))cV(t_{3}-t_{2}))
=\\
\upzeta_{\mc{A}}^{\dagger}(cV(t_{3}-t_{2}))
\circ
\upzeta_{\mc{A}}^{\dagger}(a_{1}(t_{0},t_{1},t_{2})+a_{2}(t_{0},t_{1},t_{2})).
\end{multline*}
Thus, by the first equality in \eqref{11211153} and taking $t_{0}<t_{1}<t_{2}<t_{3}$ we deduce that 
\begin{multline*}
\mf{b}_{\mc{A}}\left(\uptau^{\dagger}(t_{3}-t_{2})\circ\upzeta_{\mc{A}}^{\dagger}(c)\circ\uptau^{\dagger}(t_{2}-t_{1})
\circ\upzeta_{\mc{A}}^{\dagger}(a_{1}+a_{2})\circ\uptau^{\dagger}(t_{1}-t_{0}),\upomega,\un\right)
=
\\
\mf{b}_{\mc{A}}\left(\upzeta_{\mc{A}}^{\dagger}(cV(t_{3}-t_{2}))\circ\sum_{i=1}^{2}\upzeta_{\mc{A}}^{\dagger}(a_{i}(t_{0},t_{1},t_{2})),
\upomega,\un\right)
+\frac{\upomega(\mr{Int}(a_{1}(t_{0},t_{1},t_{2}),a_{2}(t_{0},t_{1},t_{2}),cV(t_{3}-t_{2})))}{\upomega(\un)};
\end{multline*}
therefore we obtain 
\begin{multline}
\label{11211153time} 
\mf{b}_{\mc{A}}\left(\uptau^{\dagger}(t_{3}-t_{2})\circ\upzeta_{\mc{A}}^{\dagger}(c)\circ\uptau^{\dagger}(t_{2}-t_{1})
\circ\upzeta_{\mc{A}}^{\dagger}(a_{1}+a_{2})\circ\uptau^{\dagger}(t_{1}-t_{0}),\upomega,\un\right)
=
\\
\mf{b}_{\mc{A}}\left(
\uptau^{\dagger}(t_{3}-t_{2})\circ\upzeta_{\mc{A}}^{\dagger}(c)\circ\uptau^{\dagger}(t_{2}-t_{1})\circ
\left(\sum_{i=1}^{2}\upzeta_{\mc{A}}^{\dagger}(a_{i})\right)\circ\uptau^{\dagger}(t_{1}-t_{0}),
\upomega,\un\right)+
\\
\frac{\upomega(\mr{Int}(a_{1}(t_{0},t_{1},t_{2}),a_{2}(t_{0},t_{1},t_{2}),cV(t_{3}-t_{2})))}{\upomega(\un)}.
\end{multline}
Finally if $\mc{A}$ is a von Neumann algebra acting on a Hilbert space $\mf{H}$ and $a_{1},a_{2}\in\mr{Pr}(\mc{A})$ such 
that $a_{1}a_{2}=\ze$ and $a_{3}\coloneqq\un-(a_{1}+a_{2})\neq\ze$, 
then the analysis in \eqref{11211153} 
can be equivalently obtained by constructing two suitable discrete observables one describing 
concealed alternatives the other the detected ones. More exactly let $(a,\lambda)$ and $(a,\mu)$ be two spectral couples on
$\mf{H}$ defined on $\{1,2,3\}$ such that $\lambda_{1}=\lambda_{2}=1$, $\lambda_{3}=0$, while $\mu_{i}=i$ with $i\in\{1,2\}$
and $\mu_{3}=0$. 
Let $\mf{x}$ and $\mf{y}$ be the von Neumann measuring processes associated with the discrete observables associated 
with $(a,\lambda)$ and $(a,\mu)$ respectively. Thus \eqref{11041528a} and \eqref{11041528c} yield 
$\mf{I}_{\mf{z}}(\{1,2\})=\sum_{i=1}^{2}\mf{I}_{\mf{z}}(\{i\})$ with $\mf{z}\in\{\mf{x},\mf{y}\}$, and 
$\mf{I}_{\mf{x}}(\{1\})=\upzeta_{\mc{A}}^{\dagger}(a_{1}+a_{2})$ and $\mf{I}_{\mf{x}}(\{2\})=\ze$; while
$\mf{I}_{\mf{y}}(\{i\})=\upzeta_{\mc{A}}^{\dagger}(a_{i})$ with $i\in\{1,2\}$. Therefore, 
\begin{equation*}
\begin{cases}
\mf{I}_{\mf{x}}(\{1,2\})=\upzeta_{\mc{A}}^{\dagger}(a_{1}+a_{2}),
\\
\mf{I}_{\mf{y}}(\{1,2\})=\sum_{i=1}^{2}\upzeta_{\mc{A}}^{\dagger}(a_{i});
\end{cases}
\end{equation*}
and then \eqref{11211153} would read as follows
\begin{equation}
\label{11211154} 
\mf{b}_{\mc{A}}\left(\upzeta_{\mc{A}}^{\dagger}(c)\circ\mf{I}_{\mf{x}}(\{1,2\}),\upomega,\un\right)
=
\mf{b}_{\mc{A}}\left(\upzeta_{\mc{A}}^{\dagger}(c)\circ\mf{I}_{\mf{y}}(\{1,2\}),\upomega,\un\right)
+\frac{\upomega(\mr{Int}(a_{1},a_{2},c))}{\upomega(\un)}.
\end{equation}
\emph{In conclusion \eqref{11211153time}, or the simplified atemporal versions \eqref{11211153} and \eqref{11211154}, 
are what we mean to be the interference phenomenon.}
\subsection{When $\upzeta_{\mc{M}}^{\dagger}(X(B))$ equals $\mf{I}_{\mf{a}}(K)$ with $X$ an observable and $\mf{a}$ a measuring 
process of a suitable observable}
\label{10301329} 
Let $X$ be an observable on $\mf{H}$ with value space $(\R,\mf{B}(\R))$, thus the conclusion of Rmk. \ref{10301327} and 
its natural generalization to a Borelian set of $\R$ might fail to be true for any $\mf{x}$ measuring process of $X$,
if $X$ is not discrete. In other words for any $\mf{x}$ non discrete measuring process of $X$, there might exist a Borelian 
$C$ of $\R$ such that $\upzeta^{\dagger}(X(C))\neq\mf{I}_{\mf{x}}(C)$. However let $B\in\mf{B}(\R)$, define $\lambda:\{0,1\}\to\R$, 
$\lambda(0)=0$ and $\lambda(1)=1$; and define $p:\{0,1\}\to\Pr(\mf{H})$, $p(0)=\un-X(B)$ and $p(1)=X(B)$, then 
$E_{(p,\lambda)}(\{0\})=\un-X(B)$, $E_{(p,\lambda)}(\{1\})=X(B)$, $E_{(p,\lambda)}(D)=\ze$ for every $D\in\mf{B}(\R)$ such that 
$D\subseteq\complement_{\R}\{0,1\}$, and $\lr{p}{\lambda}=X(B)$ see Def.\ref{10311739}. Next let $\mf{a}$ be 
the von Neumann measuring process associated with the discrete observable $E_{(p,\lambda)}$ see Rmk. \ref{11041328} thus, 
by letting $\mc{M}=\lr{\mf{L}(\mf{H})}{\sigma(\mf{L}(\mf{H}),\mf{L}(\mf{H})_{\ast})}$ we have 
\begin{equation*}
\begin{aligned}
\mf{I}_{\mf{a}}(\{1\})
&=
\upzeta_{\mc{M}}^{\dagger}(X(B));
\\
\mf{I}_{\mf{a}}(\{0\})
&=
\upzeta_{\mc{M}}^{\dagger}(\un-X(B));
\\
(\forall D\in\mf{B}(\R))(D\subseteq\complement_{\R}\{0,1\}\Rightarrow\mf{I}_{\mf{a}}(D)
&=
\ze).
\end{aligned}
\end{equation*}
Notice that by Def.\ref{08081839}\eqref{08081839mpc}, Def.\ref{12041208} and since $\lr{p}{\lambda}=X(B)$ we have 
\begin{itemize}
\item
$\mf{C}(\mf{I}_{\mf{a}}(\{1\}))=$ the channel executing the measurement of the observable $\mf{o}(X(B))$ and then
selecting the subensemble of those preparations such that a value in $\{1\}$ is detected.
\end{itemize}
The first of the above equalities shows that is always possible to construct an observable (actually a discrete one) 
with value space $(\R,\mf{B}(\R))$ and a measuring process $\mf{a}$ of it such that there exists a a Borelian set $K$
(in our construction $K=\{1\}$) such that the following equality holds true
\begin{equation*}
\upzeta_{\mc{M}}^{\dagger}(X(B))=\mf{I}_{\mf{a}}(K).
\end{equation*}
\subsection{Locality}
\label{01142016}
Let us start with the following 
\begin{lemma}
\label{01150014}
Let $\alpha\in\mr{Q}(\mc{A})$, $\beta\in\mr{End}(\mc{A})$, 
$b\in\Uplambda_{\mc{A}}$ and $\omega\in\mf{P}_{\mc{A}}^{\natural}$, then
\begin{equation*}
\alpha^{\dagger}\circ\upzeta^{\dagger}(b)\circ\beta^{\dagger}
=
\alpha^{\dagger}\circ\beta^{\dagger}\circ\upzeta^{\dagger}(\beta(b)).
\end{equation*}
\end{lemma}
\begin{proof}
$\alpha^{\dagger}\circ\upzeta^{\dagger}(b)\circ\beta^{\dagger}=(\beta\circ\upzeta(b)\circ\alpha)^{\dagger}$
but for every $a\in\mc{A}$ we have 
\begin{equation*}
(\beta\circ\upzeta(b)\circ\alpha)(a)
=
\beta(b\alpha(a)b^{\ast})
=
\beta(b)\beta(\alpha(a))\beta(b)^{\ast}
=
(\upzeta(\beta(b))\circ\beta\circ\alpha)(a),
\end{equation*}
thus,
\begin{equation*}
(\beta\circ\upzeta(b)\circ\alpha)^{\dagger}
=
\alpha^{\dagger}\circ\beta^{\dagger}\circ\upzeta^{\dagger}(\beta(b)),
\end{equation*}
and the statement follows.
\end{proof}
By taking into account Convention \ref{01151142} we state the following 
\begin{lemma}
[\textbf{From Schr\"{o}dinger to Heisenberg picture of the composition of the action of $b$ with a sequence of aqft measuring 
processes}]
\label{01150124}
Let $n\in\Z_{\geq 2}$, $s:[0,n]\cap\Z\to\R_{\geq 0}$ be such that $s_{0}=0$ and $s_{k}>s_{k-1}$ 
for every $k\in[1,n]\cap\Z$, $w:[1,n-1]\cap\Z\to\Uplambda_{\mc{A}}$, $b\in\Uplambda_{\mc{A}}$ and $\uptau:\R\to\mr{Aut}(\mc{A})$ 
be a group action. Thus,
\begin{equation*}
\upzeta^{\dagger}(b)\circ\bigl(\prod_{k=n}^{2}\uptau^{\dagger}(s_{k}-s_{k-1})\circ\upzeta^{\dagger}(w_{k-1})\bigr)
\circ\uptau^{\dagger}(s_{1}-s_{0})
=
\uptau^{\dagger}(s_{n})\circ\upzeta^{\dagger}(\uptau(s_{n})b)\circ\prod_{k=n-1}^{1}\upzeta^{\dagger}(\uptau(s_{k})w_{k}).
\end{equation*}
\end{lemma}
\begin{proof}
Let $w_{n}=b$ then,
\begin{equation*}
\begin{aligned}
\upzeta^{\dagger}(b)\circ\bigl(\prod_{k=n}^{2}\uptau^{\dagger}(s_{k}-s_{k-1})\circ\upzeta^{\dagger}(w_{k-1})\bigr)
\circ\uptau^{\dagger}(s_{1}-s_{0})
&=
\prod_{k=n}^{1}\upzeta^{\dagger}(w_{k})\circ\uptau^{\dagger}(s_{k}-s_{k-1})
\\
&=
\uptau^{\dagger}(s_{n})\circ\prod_{k=n}^{1}\uptau^{\dagger}(-s_{k})\circ\upzeta^{\dagger}(w_{k})\circ\uptau^{\dagger}(s_{k})
\\
&=
\uptau^{\dagger}(s_{n})\circ\prod_{k=n}^{1}\upzeta^{\dagger}(\uptau(s_{k})w_{k})
\\
&=
\uptau^{\dagger}(s_{n})\circ\upzeta^{\dagger}(\uptau(s_{n})b)\circ\prod_{k=n-1}^{1}\upzeta^{\dagger}(\uptau(s_{k})w_{k}),
\end{aligned}
\end{equation*}
where the second equality follows by Lemma \ref{11180953}, the third one by Lemma \ref{01150014}, and the statement follows.
\end{proof}
\begin{corollary}
[\textbf{From Schr\"{o}dinger to Heisenberg picture of the detection of the action of $b$ over the statistical ensemble 
resulting next a sequence of aqft measuring processes is performed}]
\label{011509017}
Let $n\in\Z_{\geq 2}$, $s:[0,n]\cap\Z\to\R_{\geq 0}$ be such that $s_{0}=0$ and $s_{k}>s_{k-1}$ 
for every $k\in[1,n]\cap\Z$, $w:[1,n-1]\cap\Z\to\Uplambda_{\mc{A}}$, $b\in\Uplambda_{\mc{A}}$ and $\uptau:\R\to\mr{Aut}(\mc{A})$ 
be a group action. Thus for every $\omega\in\mf{P}_{\mc{A}}^{\dagger}$ we obtain
\begin{multline*}
\left(\upzeta^{\dagger}(b)\circ\bigl(\prod_{k=n}^{2}\uptau^{\dagger}(s_{k}-s_{k-1})\circ\upzeta^{\dagger}(w_{k-1})\bigr)
\circ\uptau^{\dagger}(s_{1}-s_{0})\right)(\omega)(\un)
=\\
\left(\upzeta^{\dagger}(\uptau(s_{n})b)\circ\upzeta^{\dagger}\bigl(\prod_{k=1}^{n-1}(\uptau(s_{k})w_{k})\bigr)\right)(\omega)(\un).
\end{multline*}
\end{corollary}
\begin{proof}
By Lemma \ref{01150124} and Prp.\ref{08091056}\eqref{08091056last}.
\end{proof}
\begin{theorem}
[\textbf{Undetectability via $b$ of the causal disjoint $d_{i}$ in the entanglement implemented by $\sum_{i}a_{i}d_{i}$}]
\label{01051727}
Let $\omega\in\mf{P}_{\mc{A}}^{\natural}$, $n\in\Z_{\geq 1}$, $d:[1,n]\cap\Z\to\Uplambda_{\mc{A}}$ and 
$a:[1,n]\cap\Z\to\Uplambda_{\mc{A}}$.
If 
\begin{enumerate}
\item
$\{a_{i}\}_{i\in[1,n]\cap\Z}\cup\{b\}\subset\{d_{i}\}_{i\in[1,n]\cap\Z}^{\prime}$,
\label{01051727a} 
\item
$d_{i}d_{i}^{\ast}=\un$, for all $i\in[1,n]\cap\Z$, 
\label{01051727b}
\end{enumerate}
and in addition 
\begin{itemize}
\item
$(\forall i,j\in[1,n]\cap\Z)(i\neq j\Rightarrow d_{i}d_{j}^{\ast}=\ze)$,
\end{itemize}
then
\begin{equation*}
\left(\upzeta^{\dagger}(b)\circ\upzeta^{\dagger}(\sum_{i=1}^{n}a_{i}d_{i})\right)(\omega)
(\un)
=
\sum_{i=1}^{n}\left(\upzeta^{\dagger}(b)\circ\upzeta^{\dagger}(a_{i})\right)(\omega)(\un).
\end{equation*}
\end{theorem}
\begin{proof}
We have
\begin{equation}
\label{01131657}
\left(\upzeta^{\dagger}(b)\circ\upzeta^{\dagger}(\sum_{i=1}^{n}a_{i}d_{i})\right)(\omega)(\un)
=
\omega(\upxi);
\end{equation}
where 
\begin{equation*}
\upxi\coloneqq(\sum_{i=1}^{n}a_{i}d_{i})bb^{\ast}(\sum_{i=1}^{n}d_{i}^{\ast}a_{i}^{\ast}).
\end{equation*}
Next
\begin{equation}
\label{01131658}
\upxi=
\sum_{i\in\Delta([1,n]\cap\Z)}a_{i}d_{i}bb^{\ast}d_{i}^{\ast}a_{i}^{\ast}
+
\sum_{(i,j)\in\complement_{([1,n]\cap\Z)^{2}}\Delta([1,n]\cap\Z)}a_{i}d_{i}bb^{\ast}d_{j}^{\ast}a_{j}^{\ast}.
\end{equation}
Now by hypothesis \eqref{01051727a} and hypothesis $\bullet$ we obtain for every $i,j\in[1,n]\cap\Z$
\begin{equation*}
\begin{aligned}
a_{i}d_{i}bb^{\ast}d_{i}^{\ast}a_{i}^{\ast}
&=
a_{i}bd_{i}d_{i}^{\ast}b^{\ast}a_{i}^{\ast}
\\
a_{i}d_{i}bb^{\ast}d_{j}^{\ast}a_{j}^{\ast}
&=
a_{i}bd_{i}d_{j}^{\ast}b^{\ast}a_{j}^{\ast}
=\ze,\,i\neq j,
\end{aligned}
\end{equation*}
then the statement follows since \eqref{01131657}, \eqref{01131658} and hypothesis \eqref{01051727b}.
\end{proof}
The statement of the above theorem clearly holds for more general but unphysical hypothesis for instance for 
$\mc{A}^{\ast}$, $\mc{A}$ and $\ep$ replacing $\mf{P}_{\mc{A}}^{\natural}$, $\Uplambda_{\mc{A}}$ and $\upzeta$ respectively.
\begin{remark}
\label{01201126}
Notice that we obtain $\upzeta^{\dagger}(\sum_{i=1}^{n}a_{i}d_{i}))(\omega)$ as the state in \cite[(13)]{dpl} by letting $n=2$, 
$a_{r}=F_{r}W_{j_{r},F}^{\ast}$, $d_{r}=E_{r}W_{r,E}^{\ast}$ and $\omega=\omega_{v}$ with $v=2^{-\frac{1}{2}}(\Psi+\Phi)$ where 
$\Psi$ is a unit vector in $E_{1}F_{1}\mf{H}$, while $\Phi$ is a unit vector in $E_{2}F_{2}\mf{H}$.
In particular Thm.\ref{01051727} is satisfied for every $b\in\{E_{r},W_{r,E}^{\ast}\}_{r\in\{1,2\}}^{\prime}$ and it represents a wide
generalization of the statement of invariance established after \cite[(13)]{dpl}.
\end{remark}
Let us analyze in more details the physical content of Thm.\ref{01051727}. 
The statistical ensemble $\upzeta^{\dagger}(\sum_{i=1}^{n}a_{i}d_{i})(\omega)$ is meant to generalize the usual concept of 
entanglement introduced in tensor product setting to the quantum field theory setting where physical states are obtained by 
applying creation operations. 
Commutativity here is intended as a necessary condition of causal disjointness which in algebraic quantum field theory setting 
is expressed by the request that if $O$ and $Q$ are open bounded subsets of $\mathbb{M}^{4}$ such that $O\subseteq Q^{\prime}$, 
then $\mc{B}(O)\subseteq\mc{B}(Q)^{\prime}$ where $Q^{\prime}$ is the causal complement of $Q$. In the same setting we recall also 
that by additivity we have $\mc{B}(O_{1})\mc{B}(O_{2})\subset\mc{B}(O_{1}\cup O_{2})$ with $O_{i}$ open bounded subsets of 
$\mathbb{M}^{4}$. 
Therefore, we obtain that Thm.\ref{01051727} states what follows:
\begin{itemize}
\item
\emph{If $\mf{R}(b)$ and every operation in the range of the map $a$ is causally disjoint from every operation of 
the range of the map $d$ and the additional technical hypotheses \eqref{01051727b} and $\bullet$ are satisfied, then 
the detection conditioned by $\mf{T}(\omega)$ of the action of $\mf{R}(b)$ over the statistical ensemble resulting next the 
channel, induced by the device selecting the \emph{concealed} alternative of the 
\underline{product of the operations in the range of $a$ with those in the range of $d$}, 
applies to $\mf{T}(\omega)$ equals the detection conditioned by $\mf{T}(\omega)$ of the action 
of $\mf{R}(b)$ over the statistical ensemble resulting next the \emph{detected} alternative of the channels induced by the 
devices selecting \underline{the operations in the range of $a$} applies to $\mf{T}(\omega)$.}
\end{itemize}
We remark two facts:
\begin{enumerate}
\item
\emph{Via the operation $b$ a concealed alternative is measured as a detected alternative.} 
The detection conditioned by $\mf{T}(\omega)$ of the action of $\mf{R}(b)$ over a statistical ensemble induced by a 
\emph{concealed} alternative of operations equals the detection conditioned by $\mf{T}(\omega)$ of the action of $\mf{R}(b)$ 
over a statistical ensemble induced by a \emph{detected} alternative of channels.
This result is expected since it is a variation of the general result stated already in Cor.\ref{01182221} and it relies on 
the commutativity between $b$ and $d([1,n]\cap\Z)$.
\label{fact1}
\item
\emph{Undetectability via the operation $b$ of the causally disjoint operation $d$ in the product of operations}. 
In case $n=1$ we obtain: If both the operations $b$ and $a$ are causally disjoint from the operation $d$ and the additional 
technical hypothesis \eqref{01051727b} is satisfied, then the detection conditioned by $\mf{T}(\omega)$ of the action of 
$\mf{R}(b)$ over the statistical ensemble resulting next the channel induced by the product of the operations $a$ and $d$, 
applies to $\mf{T}(\omega)$ equals the detection conditioned by $\mf{T}(\omega)$ of the action of $\mf{R}(b)$ over the 
statistical ensemble resulting next the channel induced by the operation $a$, applies to $\mf{T}(\omega)$. 
\label{fact2}
\end{enumerate}
As an important particular case of the latter case we have that 
\begin{itemize}
\item
\emph{Undetectability via the operation $b$ of the causally disjoint operation $d$}. In case $n=1$ and $a$ is the constant map 
equal to $\un$ we obtain: If the operation $b$ is causally disjoint from the operation $d$ and the additional technical 
hypothesis \eqref{01051727b} is satisfied, then the detection conditioned by $\mf{T}(\omega)$ of the action of $\mf{R}(b)$ 
over the statistical ensemble resulting next the channel induced by the operation $d$, applies to $\mf{T}(\omega)$ equals 
the detection conditioned by $\mf{T}(\omega)$ of the action of $\mf{R}(b)$ over $\mf{T}(\omega)$.
\end{itemize}
In what follows we apply the next definition to the cases $\mc{C}$ is the algebra underlying either $\mr{Q}(\mc{A})$ or 
$\mf{Z}(\mc{A})$.
\begin{definition}
Let $\mc{C}$ be a non necessarily unital algebra and $A\subseteq\mc{C}$, define 
$A^{\prime}\coloneqq\{g\in\mc{C}\,\vert\,(\forall a\in A)(g\circ a=a\circ g)\}$.
\end{definition}
Thus the last statement widely generalizes \cite[(10)]{dpl}, while the above fact \eqref{fact2} namely Thm.\ref{01051727}
for the case $n=1$ is largely extended and reinforced by the following 
\begin{theorem}
[\textbf{Undetectability via $\mf{z}$ of the causal disjoint $\mf{y}$ in the product $\mf{x}\circ\mf{y}$}]
\label{01131920}
Let $\mf{x},\mf{y},\mf{z}\in\mr{Q}(\mc{A})$ and $\omega\in\mf{P}_{\mc{A}}^{\natural}$.
If 
\begin{enumerate}
\item
$\{\mf{x},\mf{z}\}\in\{\mf{y}\}^{\prime}$, 
\item
$(\forall\mf{a}\in\{\mf{y}\}^{\prime})\bigl((\mf{y}^{\dagger}(\omega)-\omega)(\mf{a}(\un))=\ze\bigr)$;
\end{enumerate}
then,
\begin{equation*}
(\mf{z}^{\dagger}\circ(\mf{x}\circ\mf{y})^{\dagger})(\omega)(\un)
=
(\mf{z}^{\dagger}\circ\mf{x}^{\dagger})(\omega)(\un).
\end{equation*}
\end{theorem}
\begin{proof}
We have 
\begin{equation*}
\begin{aligned}
(\mf{z}^{\dagger}\circ(\mf{x}\circ\mf{y})^{\dagger})(\omega)
&=
(\mf{z}^{\dagger}\circ(\mf{y}\circ\mf{x})^{\dagger})(\omega)
\\
&=
(\mf{z}^{\dagger}\circ\mf{x}^{\dagger}\circ\mf{y}^{\dagger})(\omega)
\\
&=
((\mf{x}\circ\mf{z})^{\dagger}\circ\mf{y}^{\dagger})(\omega)
\\
&=
(\mf{y}^{\dagger}(\omega))\circ(\mf{x}\circ\mf{z})
\end{aligned}
\end{equation*}
thus,
\begin{equation*}
(\mf{z}^{\dagger}\circ(\mf{x}\circ\mf{y})^{\dagger})(\omega)(\un)
=
(\mf{y}^{\dagger}(\omega))((\mf{x}\circ\mf{z})(\un)),
\end{equation*}
and the statement follows since $\mf{x}\circ\mf{z}\in\{\mf{y}\}^{\prime}$.
\end{proof}
A similar but weaker result holds for $\mf{Z}(\mc{A})$ more exactly we have the following 
\begin{proposition}
\label{01141820}
Let $\mf{x},\mf{y},\mf{z}\in\mf{Z}(\mc{A})$ and $\omega\in\mf{P}_{\mc{A}}^{\natural}$.
If 
\begin{enumerate}
\item
$\{\mf{x},\mf{z}\}\in\{\mf{y}\}^{\prime}$, 
\item
$(\forall\mf{a}\in\{\mf{y}\}^{\prime})(\mf{a}(\mf{y}(\omega)-\omega)(\un)=\ze)$;
\end{enumerate}
then,
\begin{equation*}
(\mf{z}\circ\mf{y}\circ\mf{x})(\omega)(\un)
=
(\mf{z}\circ\mf{x})(\omega)(\un).
\end{equation*}
\end{proposition}
\begin{proof}
$\mf{z}\circ\mf{y}\circ\mf{x}=\mf{z}\circ\mf{x}\circ\mf{y}$ since $\mf{x}\in\{\mf{y}\}^{\prime}$ then, the statement follows since
$\mf{z}\circ\mf{x}\in\{\mf{y}\}^{\prime}$.
\end{proof}
Next if instead we want to detect the join action of two maps of operations $a$ and $d$, not the entanglement but just 
the conposition of the operations, we have to use special operations $b$ see for instance the following 
\begin{proposition}
[Detection of the composition of operations]
\label{01051914}
Let $\omega\in\mf{P}_{\mc{A}}^{\natural}$, $n\in\Z_{\geq 1}$, $d:[1,n]\cap\Z\to\Uplambda_{\mc{A}}$, $a:[1,n]\cap\Z\to\Uplambda_{\mc{A}}$ 
and $k\in[1,n]\cap\Z$. If $a_{k}d_{k}b=a_{k}d_{k}$ and $a_{i}d_{i}b=\ze$ for every $i\in\complement_{[1,n]\cap\Z}\{k\}$, then 
\begin{equation*}
\left(\upzeta^{\dagger}(b)\circ\upzeta^{\dagger}(\sum_{i=1}^{n}a_{i}d_{i})\right)(\omega)(\un)
=
\upzeta^{\dagger}(a_{k}d_{k})(\omega)(\un).
\end{equation*}
\end{proposition}
\begin{remark}
For instance with the data provided in Rmk.\ref{01201126} and if the set $\{E_{1},W_{r}^{\ast}\}_{r\in\{1,2\}}$ is localized in 
$O_{E}$ and the set $\{F_{1},W_{j_{r}}^{\ast}\}_{r\in\{1,2\}}$ is localized in $O_{F}$ where $O_{F}\subseteq O_{E}^{\prime}$ so, 
$\mc{B}(O_{F})\subseteq\mc{B}(O_{E})^{\prime}$, we have as a result of Thm.\ref{01051727} that the action of the map 
$r\mapsto E_{r}W_{r}^{\ast}$, whose operations are then localized in $O_{E}$,
is \emph{undetectable} by any operation $b$ which is causally disjoint from $O_{E}$ as for instance by the operations 
$\{F_{r}\}_{r\in\{1,2\}}$ localized in $O_{F}$. However as a result of Prp.\ref{01051914} the joint action 
$r\mapsto E_{r}W_{r}^{\ast}F_{r}W_{j_{r}}^{\ast}$ of the maps of operations $r\mapsto E_{r}W_{r}^{\ast}$ and $r\mapsto F_{r}W_{j_{r}}^{\ast}$ 
can be detected for instance via the operation $E_{r}F_{r}$ with $r\in\{1,2\}$ which by additivity is instead localized in 
$O_{E}\cup O_{F}$ but in general neither in $O_{E}$ nor in $O_{F}$. 
\end{remark}
\begin{remark}
The additivity property ensures that if $a_{i}$ is localized in $O_{i}$, then $a_{1}a_{2}$ is localized in $O_{1}\cup O_{2}$.
Now let us see this in a concrete example. In \cite{815039} it is designed an experiment to test the generalized Bell 
inequalities via composition of operations localized in causally disjoint sets. The generalized Bell inequalities 
when violated would imply that the involved statistical ensemble is an entangled one, and that the observable algebra of the 
quantum sysyem under investigation is a noncommutative one thus, ruling out any attempt to explain via a hidden variables theory
the spacelike correlations detected. In the experiment described in \cite{815039} we deal with a physical system (two photons)
whose statistical ensembles are sets of samples that consist of $4$-plets $(t_{a},I_{a},t_{b},I_{b})$. Here $I_{a}$ and $I_{b}$ are 
photonic signals emerging $I_{a}$ from a polarizer posited in Alice station localized in $O_{a}$ and the other signal $I_{b}$ from 
another polarizer placed in Bob station localized $O_{b}$. Here $O_{a}$ and $O_{b}$ are two open subsets of $\mathbb{M}^{4}$ 
such that $O_{a}\subseteq O_{b}^{\prime}$: they are causally disjoint. Furthermore for every $c\in\{a,b\}$, $t_{c}$ is a time-tag 
registered by an atomic clock $C_{c}$ localized in $O_{c}$ and detecting the instant of time in which $I_{c}$ happens 
\cite[Fig.2]{815039}. The two clocks $C_{a}$ and $C_{b}$ are mainteined \emph{syncronized} by the emission of laser pulses $670$ 
nm wavelength and $3$ ns width \cite[p.5041]{815039} ensuring that $|t_{a}-t_{b}|$ the absolute value of the difference of the 
time-tags remains less or equal to a determined value $\ep$ small enough to allow that for any two samples 
$(t_{a},I_{a},t_{b},I_{b})$ and $(r_{a},P_{a},r_{b},P_{b})$ the time-delay $\min\{|t_{a}-r_{a}|,\,|t_{b}-r_{b}|\}$ is strictly higher 
than $\ep$. Now fundamental here is the following fact: 
\begin{itemize}
\item
The physical information in the sample $(t_{a},I_{a},t_{b},I_{b})$ is conveyed by the couple $(I_{a},I_{b})$, while the time-tags 
$t_{a}$ and $t_{b}$ are necessary to reconstruct the sample whenever $I_{d}$ with $d\in\{a,b\}$, results not detectable when 
performing a suitable operation, for instance when an operation localized in $O_{c}$ is performed, with $c\in\{a,b\}$ such that 
$c\neq d$ therefore, the sample information available for the action of this operation is only the portion relative to $c$ 
namely $(t_{c},I_{c})$. The reconstruction of the whole sample when required\footnote{if for instance we are interested only 
on testing operations localized in $O_{c}$, then by Thm.\ref{01051727} the partial information $I_{c}$ suffices.} is then possible
only \emph{a posteriori} by using the a priori data $(t_{c},I_{c})$ and the a posteriori data $(t_{d},I_{d})$ available in the 
other region $O_{d}$ via the action of an operation localized in $O_{d}$ and then reconstructing uniquely the sample by employing 
the request $|t_{a}-t_{b}|\leq\ep$. 
\end{itemize}
Therefore, with the help of time-tags is consistently possible via the conjugation $\upzeta^{\dagger}$ to perform the action of an
operation $o_{a}$ localized in $O_{a}$ followed by the action of an operation $o_{b}$ localized in $O_{b}$ and clearly obtain that
this action equals that of the operation $o_{a}o_{b}$ which, at difference of the single operations $o_{a}$ and $o_{b}$, holds 
access to the entire $4$-plet forming the sample $(t_{a},I_{a},t_{b},I_{b})$, it is localized neither in $O_{a}$ nor in $O_{b}$
but it is localized in $O_{a}\cup O_{b}$.
\end{remark}
\par
Finally we state the following
\begin{corollary}
\label{01051906}
Assume the hypotheses of Thm.\ref{01051727} and let $k\in[1,n]\cap\Z$.
If $\upzeta^{\dagger}(b)\left(\upzeta^{\dagger}(a_{k})(\omega)\right)=\upzeta^{\dagger}(a_{k})(\omega)$ for instance if $a_{k}b=a_{k}$;
and $\upzeta^{\dagger}(b)\left(\upzeta^{\dagger}(a_{i})(\omega)\right)=\ze$ for every $i\in\complement_{[1,n]\cap\Z}\{k\}$
for instance if $a_{i}b=\ze$ for every $i\in\complement_{[1,n]\cap\Z}\{k\}$,
then 
\begin{equation*}
\left(\upzeta^{\dagger}(b)\circ\upzeta^{\dagger}(\sum_{i=1}^{n}a_{i}d_{i})\right)(\omega)(\un)
=
\upzeta^{\dagger}(a_{k})(\omega)(\un).
\end{equation*}
\end{corollary}
\begin{proof}
From Thm.\ref{01051727}.
\end{proof}
\subsection{Applications}
\label{01192117}
Typically the Wigner formula for a sequence of measurements of discrete observables and its generalization to 
continuous observables is provided in Heisenberg picture, see for instance \cite[(W2) p.5597]{oza3} 
for discrete observables and \cite[(97), (87) and (32)]{oza3} for continuous observables.
However by employing the semantics developed in Def.\ref{08081839} and with the help of only \cite[(5.3)]{oza1} in the
form given in \eqref{11191108}, we judge that the Wigner formula is more intuitive and technically much simpler to prove 
in Schr\"{o}dinger picture than in Heisenberg picture as performed in \cite{oza3}.
\par
In Thm.\ref{11191154} and Rmk.\ref{11191905} respectively we prove and physically interpret the generalized Wigner formula 
for a sequence of measuring processes of semiobservables, in particular continuous observables.
In Cor.\ref{10211852} and Rmk.\ref{11181512} respectively we prove independently from the above result
and physically interpret the generalized Wigner formula for a sequence of von Neumann measuring processes associated with 
discrete observables.
\par
Incidentally our results are in terms of a statistical ensemble, rather than a state, obtained after the action of a 
channel; the normalization reappearing in virtue of Postulate \ref{11030928} whenever we are interested to calculate 
probabilities as at the end of Rmk.\ref{11191905} and Rmk.\ref{11181512}. 
\begin{convention}
\label{01151142}
Let $X$ be a semigroup, $n\in\Z_{\geq 1}$ and $s:[1,n]\cap\Z\to X$. If $n=1$, then $\prod_{k=n}^{1}s_{k}=s_{1}$; 
if $n\in\Z_{\geq 2}$, then $\prod_{k=n}^{1}s_{k}=\prod_{k=1}^{n}s_{p(k)}$ where $p:[1,n]\cap\Z\to[1,n]\cap\Z$ 
is such that $p(1)=n$ and $p(k+1)=p(k)-1$ for every $k\in[1,n-1]\cap\Z$. 
\end{convention}
Let us start with the following trivial result: 
\begin{lemma}
\label{11180953}
Let $\mc{A}\in\mr{tsa}$, $n\in\Z_{\geq 1}$, $s:[0,n]\cap\Z\to\R_{\geq 0}$ be such that $s_{0}=0$ and $s_{k}>s_{k-1}$ 
for every $k\in[1,n]\cap\Z$, $J:[1,n]\cap\Z\to\mf{Z}(\mc{A})$ and $\uptau:\R\to\mr{Aut}(\mc{A})$ be a group action. Thus,
\begin{equation*}
\prod_{k=n}^{1}J_{k}\circ\uptau^{\dagger}(s_{k}-s_{k-1})
=
\uptau^{\dagger}(s_{n})\circ\prod_{k=n}^{1}\uptau^{\dagger}(-s_{k})\circ J_{k}\circ\uptau^{\dagger}(s_{k}).
\end{equation*}
\end{lemma}
\begin{convention}
If $\mf{H}$ is a Hilbert space, then we let $\mr{tc}(\mf{H})$ denote the linear space of trace class operators on $\mf{H}$
and $\mr{tc}^{+}(\mf{H})$ denote the set of positive trace class operators on $\mf{H}$.
\end{convention}
\begin{definition}
Let $\mf{H}$ be a Hilbert space and $a\in\Uplambda_{\mf{L}(\mf{H})}$,
define $\upeta_{\mf{H}}(a)\coloneqq\upzeta_{\mf{L}(\mf{H})}(a)\up_{\mr{tc}(\mf{H})}^{\mr{tc}(\mf{H})}$. 
\end{definition}
We set some standard notation about normal functionals. Let $\mc{M}$ be a von Neumann algebra acting on a 
Hilbert space $\mf{H}$, $\rho$ be a trace class operator acting on $\mf{H}$, then we let $\omega_{\rho}^{\mc{M}}$, 
or simply $\omega_{\rho}$ whenever it will not cause confusion, be the following normal functional $a\mapsto\mr{Tr}(\rho a)$ 
on $\mc{M}$. 
\begin{lemma}
\label{11191157}
Let $\mc{M}$ be a von Neumann algebra acting on a Hilert space $\mf{H}$, $a\in\Upgamma_{\mc{M}}$,
$\mf{x}$ be a measuring process of a semiobservable on $\mf{H}$ with value space $(\Omega,\mf{B})$.
Thus, for every $B\in\mf{B}$ and every $\rho\in\mr{tc}(\mf{H})$ we have 
\begin{equation*}
(\upzeta_{\mc{M}}(a)\circ\mc{I}_{\mf{x}}(B)\circ\upzeta_{\mc{M}}(a^{\ast}))^{\dagger}\omega_{\rho}
=
\omega_{\upeta_{\mf{H}}(a)\circ\mc{Y}_{\mf{x}}(B)\circ\upeta_{\mf{H}}(a^{\ast})\rho}.
\end{equation*}
\end{lemma}
\begin{proof}
Let $x\in\mc{M}$ so,
\begin{equation*}
\begin{aligned}
((\upzeta_{\mc{M}}(a)\circ\mc{I}_{\mf{x}}(B)\circ\upzeta_{\mc{M}}(a^{\ast}))^{\dagger}\omega_{\rho})(x)
&=
\mr{Tr}(\rho(\upzeta_{\mc{M}}(a)\circ\mc{I}_{\mf{x}}(B)\circ\upzeta_{\mc{M}}(a^{\ast}))x)
\\
&=
\mr{Tr}(\rho a(\mc{I}_{\mf{x}}(B)\circ\upzeta_{\mc{M}}(a^{\ast}))(x)a^{\ast})
\\
&=
\mr{Tr}((a^{\ast}\rho a)(\mc{I}_{\mf{x}}(B)\circ\upzeta_{\mc{M}}(a^{\ast}))x)
\\
&=
\mr{Tr}\left((\upeta_{\mf{H}}(a^{\ast})\rho)\mc{I}_{\mf{x}}(B)(\upzeta_{\mc{M}}(a^{\ast})x)\right)
\\
&=
(\upzeta_{\mc{M}}^{\dagger}(a^{\ast})\circ\mf{I}_{\mf{x}}(B))(\omega_{\upeta_{\mf{H}}(a^{\ast})\rho})(x)
\\
&=
\upzeta_{\mc{M}}^{\dagger}(a^{\ast})(\omega_{(\mc{Y}_{\mf{x}}(B)\circ\upeta_{\mf{H}}(a^{\ast}))\rho})(x)
\\
&=
\mr{Tr}((\mc{Y}_{\mf{x}}(B)\circ\upeta_{\mf{H}}(a^{\ast}))(\rho)\upzeta_{\mc{M}}(a^{\ast})x)
\\
&=
\mr{Tr}((\mc{Y}_{\mf{x}}(B)\circ\upeta_{\mf{H}}(a^{\ast}))(\rho)a^{\ast}xa)
\\
&=
\mr{Tr}(a(\mc{Y}_{\mf{x}}(B)\circ\upeta_{\mf{H}}(a^{\ast}))(\rho)a^{\ast}x)
\\
&=
\mr{Tr}((\upeta_{\mf{H}}(a)\circ\mc{Y}_{\mf{x}}(B)\circ\upeta_{\mf{H}}(a^{\ast}))(\rho)x);
\end{aligned}
\end{equation*}
where the sixth equality follows by \eqref{11191108}.
\end{proof}
\begin{theorem}
[\textbf{Generalized\footnote{because $\rho$ might not be positive.}Wigner formula in Schr\"{o}dinger picture. 
Semiobservables}]
\label{11191154}
Let $\mc{M}$ be a von Neumann algebra acting on a Hilbert space $\mf{H}$, $\rho\in\mr{tc}(\mf{H})$,
let $V:\R\to\mr{U}(\mc{M})$ be a group action, let $n\in\Z_{\geq 1}$, $s:[0,n]\cap\Z\to\R_{\geq 0}$ be such that $s_{0}=0$ 
and $s_{k}>s_{k-1}$ for every $k\in[1,n]\cap\Z$.
For every $k\in[1,n]\cap\Z$ let $X_{k}$ be a semiobservable on $\mf{H}$ with value space $(\Omega_{k},\mf{B}_{k})$, 
$\mf{x}_{k}$ be a measuring process of $X_{k}$ and $B_{k}\in\mf{B}_{k}$.
Let $\uptau:\R\to\mr{Aut}^{\ast}(\mc{M})$ be the group action so defined $\uptau(t)\coloneqq\upzeta_{\mc{M}}(V(t))$ 
for every $t\in\R$. Thus, 
\begin{equation*}
\left(\prod_{k=n}^{1}\mf{I}_{\mf{x}_{k}}(B_{k})\circ\uptau^{\dagger}(s_{k}-s_{k-1})\right)\omega_{\rho}
=
\omega_{\left(\prod_{k=n}^{1}\upeta_{\mf{H}}(V(s_{k}))\circ\mc{Y}_{\mf{x}_{k}}(B_{k})\circ\upeta_{\mf{H}}(V(-s_{k}))\right)\rho}\circ\uptau(s_{n}).
\end{equation*}
\end{theorem}
\begin{proof}
By taking $J_{k}=\mf{I}_{\mf{x}_{k}}(B_{k})$ and recalling that by definition $\mf{I}_{\mf{x}_{k}}=\mc{I}_{\mf{x}_{k}}^{\dagger}$
we obtain by Lemma \ref{11180953} and by the fact that $\dagger$ is contravariant that 
\begin{equation*}
\prod_{k=n}^{1}\mf{I}_{\mf{x}_{k}}(B_{k})\circ\uptau^{\dagger}(s_{k}-s_{k-1})
=
\uptau^{\dagger}(s_{n})\circ\prod_{k=n}^{1}
(\upzeta_{\mc{M}}(V(s_{k}))\circ\mc{I}_{\mf{x}_{k}}(B_{k})\circ\upzeta_{\mc{M}}(V(-s_{k})))^{\dagger}.
\end{equation*}
Thus the statement follows since Lemma \ref{11191157}.
\end{proof}
\begin{remark}
[\textbf{Interpretation of the Wigner formula. Semiobservables.}]
\label{11191905}
In addition to the hypothesis of Thm.\ref{11191154} assume that $\rho\in\mr{tc}^{+}(\mf{H})$ and set $\mf{r}(V(t))=$ the 
time translation of magnitude $t$.
Thus, Def.\ref{08081839} yields:
\par
\emph{$\mf{T}(\left(\prod_{k=n}^{1}\mf{I}_{\mf{x}_{k}}(B_{k})\circ\uptau^{\dagger}(s_{k}-s_{k-1})\right)\omega_{\rho})=$
the statistical ensemble resulting next 
the channel executing the measurement of the semiobservable $\mf{o}(X_{n})$ and then selecting the subensemble
of those preparations such that a value in $B_{n}$ is detected,
following the channel induced by the device implementing the operation of filtering through
the time translation of magnitude $s_{n}-s_{n-1}$;
following the channel executing the measurement of the semiobservable $\mf{o}(X_{n-1})$ and then selecting the subensemble
of those preparations such that a value in $B_{n-1}$ is detected,
following the channel induced by the device implementing the operation of filtering through
the time translation of magnitude $s_{n-1}-s_{n-2}$; following $\dots\dots$
the channel executing the measurement of the semiobservable $\mf{o}(X_{1})$ and then selecting the subensemble
of those preparations such that a value in $B_{1}$ is detected,
following the channel induced by the device implementing the operation of filtering through
the time translation of magnitude $s_{1}$;
applies to the statistical ensemble $\mf{t}(\omega_{\rho})$}.
\par
Assume that $\mr{Tr}(\rho)\neq 0$ thus, Thm.\ref{11191154} and Postulate \ref{11030928} establish that
\begin{equation*}
\mr{Tr}(\rho)^{-1}
\mr{Tr}
\left(\left(\prod_{k=n}^{1}\upeta_{\mf{H}}(V(s_{k}))\circ\mc{Y}_{\mf{x}_{k}}(B_{k})\circ\upeta_{\mf{H}}(V(-s_{k}))\right)\rho\right)
\end{equation*}
\emph{equals the empirical representation of the propensity conditioned by the statistical ensemble 
$\mf{t}(\omega_{\rho})$ to detect the effect of selecting one output when tested on 
$\mf{T}(\left(\prod_{k=n}^{1}\mf{I}_{\mf{x}_{k}}(B_{k})\circ\uptau^{\dagger}(s_{k}-s_{k-1})\right)\omega_{\rho})$}.
\end{remark}
We might apply the above results to the case of discrete observables and the von Neumann measuring processes associated 
with them. However we prefer to derive directly the Wigner formula for a sequence of measurements of discrete observables.
\begin{lemma}
\label{10211832}
Let $\mc{N}$ be a von Neumann algebra acting on a Hilbert space $\mf{H}$, $\rho\in\mc{N}\cap\mr{tc}(\mf{H})$, 
let $t_{i}\in\mr{U}(\mc{N})$ and $e_{i}\in\mc{N}$ for every $i\in[1,n]$ with $n\in\Z_{\geq 1}$. 
Thus by letting $t_{0}=\un$ the unit of $\mc{N}$ and $u_{k}\coloneqq t_{k-1}^{\ast}t_{k}$ for every $k\in[1,n]$, we obtain
\begin{equation*}
\left(\prod_{k=n}^{1}\ep_{\mc{N}}^{\dagger}(e_{k})\circ\ep_{\mc{N}}^{\dagger}(u_{k})\right)\omega_{\rho}
=
\ep_{\mc{N}}^{\dagger}(t_{n})(\omega_{\left(\prod_{k=n}^{1}\ep_{\mc{N}}(\ep_{\mc{N}}(t_{k})e_{k}^{\ast})\right)\rho}).
\end{equation*}
\end{lemma}
\begin{proof}
In this proof $\ep$ stands for $\ep_{\mc{N}}$. Since the first relation in \eqref{08091352} we have
\begin{equation*}
\prod_{k=n}^{1}\ep(\ep(t_{k})e_{k}^{\ast})
=
\ep\left(\prod_{k=n}^{1}t_{k}e_{k}^{\ast}t_{k}^{\ast}\right);
\end{equation*}
next 
\begin{equation*}
\prod_{k=n}^{1}t_{k}e_{k}^{\ast}t_{k}^{\ast}
=
t_{n}\prod_{k=n}^{1}e_{k}^{\ast}u_{k}^{\ast};
\end{equation*}
therefore 
\begin{equation*}
\begin{aligned}
\omega_{\left(\prod_{k=n}^{1}\ep(\ep(t_{k})e_{k}^{\ast})\right)\rho}
&=
\omega_{\ep\left(t_{n}\prod_{k=n}^{1}e_{k}^{\ast}u_{k}^{\ast}\right)\rho}
\\
&=
\ep^{\dagger}\left(\left(\prod_{k=1}^{n}u_{k}e_{k}\right)t_{n}^{\ast}\right)\omega_{\rho}
\\
&=
\left(\ep^{\dagger}(t_{n}^{\ast})\circ\prod_{k=n}^{1}\ep^{\dagger}(e_{k})\circ\ep^{\dagger}(u_{k})\right)\omega_{\rho};
\end{aligned}
\end{equation*}
where the second equality follows since
\begin{equation*}
\omega_{\ep(v)\rho}=\ep^{\dagger}(v^{\ast})\omega_{\rho};
\end{equation*}
while the third one since $\ep^{\dagger}$ is contravariant.
\end{proof}
\begin{theorem}
\label{10211852pre}
Let $\mc{N}$ be a von Neumann algebra acting on a Hilbert space $\mf{H}$, $\rho\in\mc{N}\cap\mr{tc}(\mf{H})$, 
$V:\R\to\mr{U}(\mc{N})$ be a group action, let $n\in\Z_{\geq 1}$, $s:[0,n]\cap\Z\to\R_{\geq 0}$ be such that $s_{0}=0$ 
and $s_{k}>s_{k-1}$ for every $k\in[1,n]\cap\Z$, and $e:[1,n]\cap\Z\to\Uplambda_{\mc{N}}$. 
Let $\uptau:\R\to\mr{Aut}^{\ast}(\mc{N})$ be the group action so defined $\uptau(t)\coloneqq\upzeta_{\mc{N}}(V(t))$ 
for every $t\in\R$ thus, 
\begin{equation*}
\left(\prod_{k=n}^{1}\upzeta_{\mc{N}}^{\dagger}(e_{k})\circ\uptau^{\dagger}(s_{k}-s_{k-1})\right)\omega_{\rho}
=
\omega_{\left(\prod_{k=n}^{1}\upzeta_{\mc{N}}(\uptau(s_{k})e_{k}^{\ast})\right)\rho}\circ\uptau(s_{n}).
\end{equation*}
\end{theorem}
\begin{proof}
Since Lemma \ref{10211832} applied for $t_{k}=V(s_{k})$ for every $k\in[1,n]\cap\Z$.
\end{proof}
\begin{corollary}
[\textbf{Generalized\footnote{because $\rho$ might not be positive.}Wigner formula in Schr\"{o}dinger picture. 
Channels induced by operations}]
\label{10211852}
Let $\mc{M}$ be a von Neumann algebra acting on a Hilbert space $\mf{H}$, $\rho\in\mr{tc}(\mf{H})$, 
$V:\R\to\mr{U}(\mc{M})$ be a group action, let $n\in\Z_{\geq 1}$, $s:[0,n]\cap\Z\to\R_{\geq 0}$ be such that $s_{0}=0$ 
and $s_{k}>s_{k-1}$ for every $k\in[1,n]\cap\Z$, and $e:[1,n]\cap\Z\to\Uplambda_{\mc{M}}$. 
Let $\uptau:\R\to\mr{Aut}^{\ast}(\mc{M})$ be the group action so defined $\uptau(t)\coloneqq\upzeta_{\mc{M}}(V(t))$ 
for every $t\in\R$ thus, 
\begin{equation*}
\left(\prod_{k=n}^{1}\upzeta_{\mc{M}}^{\dagger}(e_{k})\circ\uptau^{\dagger}(s_{k}-s_{k-1})\right)\omega_{\rho}
=
\omega_{\left(\prod_{k=n}^{1}\upeta_{\mf{H}}(\uptau(s_{k})e_{k}^{\ast})\right)\rho}\circ\uptau(s_{n}).
\end{equation*}
\end{corollary}
\begin{proof}
Apply Thm.\ref{10211852pre} to the von Neumann algebra $\mc{N}=\mf{L}(\mf{H})$ and to the group action 
$\imath_{\mr{U}(\mc{M})}^{\mr{U}(\mc{N})}\circ V$, then our statement follows by restricting at $\mc{M}$ the equality of 
normal functionals on $\mc{N}$ so obtained.
\end{proof}
\begin{remark}
[\textbf{Interpretation of the Wigner formula. Discrete observables.}]
\label{11181512}
For every $k\in[1,n]\cap\Z$, let $o_{k}\in\mc{M}_{ob}$ be with discrete spectrum, for instance compact, let $\sigma(o_{k})$
be the spectrum of $o_{k}$ and $E_{k}$ be the resolution of the identity of $o_{k}$.
Next let $k\in[1,n]\cap\Z$, $\mf{x}_{k}$ be the von Neumann measuring process associated with the discrete observable 
$E_{k}$ and $\mf{I}_{k}$ be the von Neumann channel map associated with the discrete observable $E_{k}$ namely 
$\mf{I}_{k}=\mf{I}_{\mf{x}_{k}}$ (Rmk.\ref{11041328}) thus, by \eqref{10301254} we deduce 
$\mf{I}_{k}(\{\lambda\})=\upzeta_{\mc{M}}^{\dagger}(E_{k}(\{\lambda\}))$ for every $\lambda\in\sigma(o_{k})$. 
Thus, by letting $\lambda_{i}\in\sigma(o_{i})$, $i\in[1,n]\cap\Z$ we obtain by Cor.\ref{10211852} for every 
$\rho\in\mr{tc}^{+}(\mf{H})$
\begin{equation}
\label{10212036}
\left(\prod_{k=n}^{1}\mf{I}_{k}(\{\lambda_{k}\})\circ\uptau^{\dagger}(s_{k}-s_{k-1})\right)\omega_{\rho}
=
\omega_{\left(\prod_{k=n}^{1}\upeta_{\mf{H}}(\uptau(s_{k})E_{k}(\{\lambda_{k}\}))\right)\rho}\circ\uptau(s_{n}).
\end{equation}
Assume that $\mr{Tr}(\rho)\neq 0$ thus, the above equality and Postulate \ref{11030928} establish that
\begin{equation*}
\mr{Tr}(\rho)^{-1}\mr{Tr}\left(\left(\prod_{k=n}^{1}\upeta_{\mf{H}}(\uptau(s_{k})E_{k}(\{\lambda_{k}\}))\right)\rho\right)
\end{equation*}
\emph{equals the empirical representation of the propensity conditioned by the statistical ensemble $\mf{t}(\omega_{\rho})$
to detect the effect of selecting one output when tested on 
$\mf{T}(\left(\prod_{k=n}^{1}\mf{I}_{k}(\{\lambda_{k}\})\circ\uptau^{\dagger}(s_{k}-s_{k-1})\right)\omega_{\rho})$}.
\end{remark}
\section{Species of dynamical patterns and equiformity principle}
\label{11301644}
We introduce in \textbf{Def. \ref{09081303}} 
the concept of dynamical pattern and its transformations,
the building block of all the constructions of this paper.
In Rmk. \ref{09090657} we make explicit the definition and show in
Cor. \ref{09100952} that dynamical patterns form a category $\mf{dp}$.
In Def. \ref{09081302}, Rmk. \ref{09081318} and Cor. \ref{09101202} 
we introduce, explain and organize in a category the concept of preordered dynamical pattern,
employed to address the dual of a dynamical pattern via the construction in 
Cor. \ref{10111742} of a contravariant functor.
In Thm. \ref{12312222} we prove the existence of a functor from $\mf{dp}$
to the category $\mf{Chdv}$ of channels and devices introduced in 
Def. \ref{12312135} and Cor. \ref{31122146}.
$\mf{Chdv}$ is essential in order to extract empirical information from $\mf{dp}$.
We consider a species of dynamical patterns contextualized in a category $\mf{D}$
to be a functor from $\mf{D}$ to $\mf{Chdv}$.
In Def. \ref{01162119} and Def. \ref{10201408}
we define collections of trajectories associated with any species
which encode the dynamical information of the species.
Experimental settings are introduced in Def. \ref{01161844},
while in Def. \ref{01161922} we define a link between experimental settings,
an auxiliary concept in order to express the fundamental equiformity principle in Prp. \ref{01162038}.
In Def. \ref{01200837} we present the physical interpretation of the data.
Lemma \ref{01011521} prepares to the main result of this section and one of the entire 
paper namely \textbf{Thm. \ref{10081910}}.
There we prove that given any connector $\mr{T}$, 
a natural transformation between species, 
then for any experimental setting $\mf{Q}$ of its target species 
and any function $s$ associated with $\mf{Q}$ and $\mr{T}$,
there exists an experimental setting $\mr{T}[\mf{Q},s]$ of its source species, 
such that $\mr{T}$ is a link from $\mf{Q}$ to $\mr{T}[\mf{Q},s]$, 
so inducing an equiformity principle.
With the price of coarsening the equiformity principle to the standard 
experimental setting of the source species we can eliminate the degeneration
in $s$ Thm. \ref{01181342}.
Finally we establish in the second and third main result
that vertical composition of connectors is contravariantly
represented as charge composition \textbf{Cor. \ref{10151636}}
and that horizontal composition is represented as charge transfer \textbf{Cor. \ref{11200910}}.
We recall that $\mr{tsa}$ and $\mr{ptls}$ are $\mr{top}-$quasi enriched categories.
\begin{definition}
\label{09081303}
The set of dynamical patterns $\mf{dp}$ is defined to be the set of the couples 
$\lr{\mr{W}}{\upeta}$ where $\mr{W}$ is a $\mr{top}-$quasi enriched 
category and $\upeta\in\mr{Fct}_{\mr{top}}(\mr{W},\mr{tsa})$.
Let $\mf{A}=\lr{\mr{W}}{\upeta}\in\mf{dp}$, then we denote $\mr{W}$ by $\mr{G}_{\mf{A}}$, $\mr{Mor}_{\mr{W}}(x,y)$ by 
$\mr{G}_{\mf{A}}(x,y)$, for all $x,y\in\mr{W}$, $\upeta$ by $\upsigma_{\mf{A}}$,
while the object and morphism maps $\upeta_{o}$ and $\upeta_{m}$ by $\mc{A}_{\mf{A}}$ and $\uptau_{\mf{A}}$ respectively.
The elements of $\mf{dp}$ are called dynamical patterns.
Let $\mf{A},\mf{B},\mf{C}\in\mf{dp}$, define $\un_{\mf{A}}\coloneqq(\un_{\mr{G}_{\mf{A}}},\un_{\upsigma_{\mf{A}}})$ and 
\begin{equation}
\label{09101001}
\mr{Mor}_{\mf{dp}}(\mf{A},\mf{B})
\coloneqq\coprod_{f\in\mr{Fct}_{\mr{top}}(\mr{G}_{\mf{B}},\mr{G}_{\mf{A}})}
\mr{Mor}_{\mr{Fct}(\mr{G}_{\mf{B}},\mr{tsa})}(\upsigma_{\mf{A}}\circ f,\upsigma_{\mf{B}}),
\end{equation}
and 
\begin{equation}
\label{09091136}
\begin{aligned}
(\circ):\mr{Mor}_{\mf{dp}}(\mf{B},\mf{C})&\times \mr{Mor}_{\mf{dp}}(\mf{A},\mf{B})\to \mr{Mor}_{\mf{dp}}(\mf{A},\mf{C}),
\\
(g,S)&\circ(f,T)\coloneqq(f\circ g,S\circ(T\ast \un_{g})).
\end{aligned}
\end{equation}
We call $\mf{A}$ a $\mc{U}-$type dynamical pattern or $\mc{U}-$type $\mf{dp}$
if $\mf{A}\in\mf{dp}$, $\mr{Obj}(\mr{G}_{\mf{A}})\simeq A$ and $A\subseteq\mc{U}$.
\end{definition}
Often we call $\mr{G}_{\mf{A}}$ the dynamical category of $\mf{A}$
and $\upsigma_{\mf{A}}$, as well by abuse of language $\uptau_{\mf{A}}$,
the dynamical functor of $\mf{A}$.
This definition nontrivially extends the category of dynamical systems,
where we specialize to the subcategory of those dynamical patterns
whose dynamical category is the groupoid associated with a topological group,
and thus the dynamical functor reduces to a representation of the 
topological group in terms of $\ast-$automorphisms of a $\ast-$topological algebra,
while morphisms are couples formed by a continuous group morphism and an
equivariant map between the respective representations.
\begin{remark}
\label{09090657a}
If $\mf{A}\in\mf{dp}$ then according to the notation in Def. \ref{09081303} we have
$\mf{A}=\lr{\mr{G}_{\mf{A}}}{\upsigma_{\mf{A}}}=
\lr{\left(\mr{Obj}(\mr{G}_{\mf{A}}),\{\mr{G}_{\mf{A}}(x,y)\}_{x,y\in \mr{Obj}(\mr{G}_{\mf{A}})}\right)}
{(\mc{A}_{\mf{A}},\uptau_{\mf{A}})}$. 
If in addition $\mf{A}$ is of $\mc{U}-$type, then $\mr{G}_{\mf{A}}$ 
is a $\mc{U}-$type category in particular it is an object of $\mr{cat}$.
\end{remark}
Next we provide a decodification of Def. \ref{09081303}.
\begin{remark}
\label{09090657}
Let $\mr{Mor}_{\mr{tsa}}(X,Y)$ be provided with the topology of pointwise convergence for all $X,Y\in\mr{tsa}$. 
Thus 
$\lr{\left(\mr{Obj}(\mr{G}_{\mf{A}}),\{\mr{G}_{\mf{A}}(x,y)\}_{x,y\in \mr{Obj}(\mr{G}_{\mf{A}})}\right)}{(\mc{A}_{\mf{A}},\uptau_{\mf{A}})}
\in\mf{dp}$ 
iff 
\begin{enumerate}
\item
$\mr{G}_{\mf{A}}$ is a $\mc{U}-$category 
such that $\mr{G}_{\mf{A}}(x,y)$ is a topological space and the morphism composition 
$\circ:\mr{G}_{\mf{A}}(y,z)\times\mr{G}_{\mf{A}}(x,y)\to\mr{G}_{\mf{A}}(x,z)$ is a separately continuous map, for all 
$x,y,z\in \mr{Obj}(\mr{G}_{\mf{A}})$;
\item
$\mc{A}_{\mf{A}}:\mr{Obj}(\mr{G}_{\mf{A}})\to \mr{Obj}(\mr{tsa})$;
\item
$\uptau_{\mf{A}}:\mr{Mor}_{\mr{G}_{\mf{A}}}\to \mr{Mor}_{\mr{tsa}}$
such that
$\uptau_{\mf{A}}^{y,z}:\mr{G}_{\mf{A}}(y,z)\to 
\mr{Mor}_{\mr{tsa}}\left(\mc{A}_{\mf{A}}(y),\mc{A}_{\mf{A}}(z)\right)$ 
is a continuous map, for all
$y,z\in \mr{Obj}(\mr{G}_{\mf{A}})$;
\item
$\uptau_{\mf{A}}(g\circ h)=\uptau_{\mf{A}}(g)\circ\uptau_{\mf{A}}(h)$, 
and $\uptau_{\mf{A}}(\un_{x})=\un_{\mc{A}_{\mf{A}}(x)}$,
for all $x,y,z\in \mr{Obj}(\mr{G}_{\mf{A}})$, $g\in\mr{G}_{\mf{A}}(y,z)$
and $h\in\mr{G}_{\mf{A}}(x,y)$.
\end{enumerate}
Let $\mf{A},\mf{B}\in\mf{dp}$, thus $(f,T)\in \mr{Mor}_{\mf{dp}}(\mf{A},\mf{B})$ iff
\begin{enumerate}
\item
$f=(f_{o},f_{m})$ such that 
$f_{o}:\mr{Obj}(\mr{G}_{\mf{B}})\to \mr{Obj}(\mr{G}_{\mf{A}})$
and
$f_{m}:\mr{Mor}_{\mr{G}_{\mf{B}}}\to \mr{Mor}_{\mr{G}_{\mf{A}}}$;
\item
for all $y,z\in \mr{Obj}(\mr{G}_{\mf{B}})$ 
\begin{enumerate}
\item
$f_{m}^{y,z}:\mr{G}_{\mf{B}}(y,z)\to\mr{G}_{\mf{A}}(f_{o}(y),f_{o}(z))$ is a continuous map;
\item
$f_{m}(g\circ h)=f_{m}(g)\circ f_{m}(h)$ and $f_{m}(\un_{x})=\un_{f_{o}(x)}$,
for all $x\in \mr{Obj}(\mr{G}_{\mf{B}})$, $g\in\mr{G}_{\mf{B}}(y,z)$ and $h\in\mr{G}_{\mf{B}}(x,y)$;
\item
$T\in\prod_{x\in \mr{Obj}(\mr{G}_{\mf{B}})}
\mr{Mor}_{\mr{tsa}}\left(\mc{A}_{\mf{A}}(f_{o}(x)),\mc{A}_{\mf{B}}(x)\right)$ 
such that for all 
$g\in\mr{G}_{\mf{B}}(y,z)$ 
we have that the following diagram in $\mr{tsa}$ is commutative
\begin{equation*}
\xymatrix{
\mc{A}_{\mf{A}}(f_{o}(y))
\ar[rr]^{T(y)}
\ar[dd]_{\uptau_{\mf{A}}(f_{m}(g))}
&&
\mc{A}_{\mf{B}}(y)
\ar[dd]^{\uptau_{\mf{B}}(g)}
\\
&&
\\
\mc{A}_{\mf{A}}(f_{o}(z))
\ar[rr]_{T(z)}
&&
\mc{A}_{\mf{B}}(z)}
\end{equation*}
\end{enumerate}
\end{enumerate}
\end{remark}
In order to prove in Cor. \ref{09100952}
that dynamical patterns and their transformations
form a category we need the following 
\begin{lemma}
\label{09101139}
The composition in \eqref{09091136} is a well-defined associative map such that 
$(f,T)\circ\un_{\mf{A}}=\un_{\mf{B}}\circ(f,T)=(f,T)$, 
for all $\mf{A},\mf{B}\in\mf{dp}$ and $(f,T)\in\mr{Mor}_{\mf{dp}}(\mf{A},\mf{B})$.
\end{lemma}
\begin{proof}
Let $\mf{A},\mf{B},\mf{C},\mf{D}\in\mf{dp}$, 
$(f,T)\in \mr{Mor}_{\mf{dp}}(\mf{A},\mf{B})$, $(g,S)\in \mr{Mor}_{\mf{dp}}(\mf{B},\mf{C})$, 
and $(h,V)\in \mr{Mor}_{\mf{dp}}(\mf{C},\mf{D})$.
Then $f\in\mr{Fct}_{\mr{top}}(\mr{G}_{\mf{B}},\mr{G}_{\mf{A}})$, 
$T\in \mr{Mor}_{\mr{Fct}(\mr{G}_{\mf{B}},\mr{tsa})}(\upsigma_{\mf{A}}\circ f,\upsigma_{\mf{B}})$,
and $g\in\mr{Fct}_{\mr{top}}(\mr{G}_{\mf{C}},\mr{G}_{\mf{B}})$, 
$S\in \mr{Mor}_{\mr{Fct}(\mr{G}_{\mf{C}},\mr{tsa})}(\upsigma_{\mf{B}}\circ g,\upsigma_{\mf{C}})$.
Thus $f\circ g\in\mr{Fct}_{\mr{top}}(\mr{G}_{\mf{C}},\mr{G}_{\mf{A}})$, and 
$T\ast \un_{g}\in \mr{Mor}_{\mr{Fct}(\mr{G}_{\mf{C}},\mr{tsa})}(\upsigma_{\mf{A}}\circ f\circ g,\upsigma_{\mf{B}}\circ g)$
hence $S\circ(T\ast\un_{g})\in \mr{Mor}_{\mr{Fct}(\mr{G}_{\mf{C}},\mr{tsa})}(\upsigma_{\mf{A}}\circ f\circ g,\upsigma_{\mf{C}})$ 
which prove that the composition in \eqref{09091136} is well-defined.
Next $((h,V)\circ(g,S))\circ(f,T)=(fgh,V\circ(S\ast \un_{h})\circ(T\ast \un_{gh}))$, 
while $(h,V)\circ((g,S)\circ(f,T))=\left(fgh,V\circ\left((S\circ(T\ast \un_{g}))\ast \un_{h}\right)\right)$.
Next we have 
\begin{equation*}
\begin{aligned}
(S\ast \un_{h})\circ(T\ast \un_{gh})
&=
(S\ast \un_{h})\circ\left((T\ast \un_{g})\ast \un_{h}\right)
\\
&=(S\circ(T\ast \un_{g}))\ast(\un_{h}\circ \un_{h})=(S\circ(T\ast \un_{g}))\ast \un_{h},
\end{aligned}
\end{equation*}
where in the first equality we used 
$\un_{gh}=\un_{g}\ast \un_{h}$ and the associativity of $\ast$, the second equality follows since
\eqref{bor135}; thus the composition in \eqref{09091136} is associative.
Next $(f,T)\circ\un_{\mf{A}}=(\un_{\mr{G}_{\mf{A}}}\circ f, T\circ(\un_{\upsigma_{\mf{A}}}\ast\un_{f}))$, and
$\un_{\mf{B}}\circ(f,T)=(f\circ \un_{\mr{G}_{\mf{B}}},\un_{\upsigma_{\mf{B}}}\circ(T\ast\un_{\un_{\mr{G}_{\mf{B}}}}))$.
Moreover 
$\un_{\upsigma_{\mf{A}}}\ast\un_{f}=\un_{\upsigma_{\mf{A}}\circ f}$ so $T\circ(\un_{\upsigma_{\mf{A}}}\ast\un_{f})=T$,
and 
$\un_{\upsigma_{\mf{B}}}\circ(T\ast\un_{\un_{\mr{G}_{\mf{B}}}})=T\ast\un_{\un_{\mr{G}_{\mf{B}}}}=T\circ(\un_{\mr{G}_{\mf{B}}})_{o}=T$ 
since \eqref{20061403}, therefore $(f,T)\circ\un_{\mf{A}}=\un_{\mf{B}}\circ(f,T)=(f,T)$.
\end{proof}
\begin{corollary}
\label{09100952}
There exists a unique $\mc{U}_{0}-$type category $\mf{dp}$ 
such that $\mr{Obj}(\mf{dp})$ is the set of the $\mf{A}\in\mf{dp}$ for which $\mf{A}$ is a $\mc{U}-$type dynamical pattern;
for any $\mf{A},\mf{B}\in \mr{Obj}(\mf{dp})$ the set of morphisms of $\mf{dp}$ 
from $\mf{A}$ to $\mf{B}$ is the set in \eqref{09101001}, 
and the law of composition of morphisms of $\mf{dp}$ is the map in \eqref{09091136}.
$\un_{\mf{D}}$ defined in Def. \ref{09081303} is the unit morphism 
in $\mf{dp}$ relative to $\mf{D}$, for all $\mf{D}\in \mr{Obj}(\mf{dp})$.
In particular $\mf{dp}$ is an object of $\mr{Cat}$.
\end{corollary}
\begin{proof}
The existence and uniqueness of the category $\mf{dp}$ follows since Lemma \ref{09101139}.
Let $\mf{A},\mf{B}\in\mr{Obj}(\mf{dp})$ then
$\mr{G}_{\mf{A}}$ and $\mr{G}_{\mf{B}}$ are $\mc{U}-$type categories 
hence $\mc{U}_{0}-$small categories, similarly $\mr{tsa}$ is a $\mc{U}_{0}-$small category,
therefore by Prp. \ref{12311604}\eqref{12311604st3}
\begin{equation*}
\begin{cases}
\upsigma_{\mf{A}},\upsigma_{\mf{B}}\in\mc{U}_{0},
\\
\mr{Obj}(\mr{Fct}_{\mr{top}}(\mr{G}_{\mf{B}},\mr{G}_{\mf{A}}))\in\mc{U}_{0},
\\
\mr{Mor}_{\mr{Fct}(\mr{G}_{\mf{B}},\mr{tsa})}(\upsigma_{\mf{A}}\circ f,\upsigma_{\mf{B}})\in\mc{U}_{0},
\forall 
f\in\mr{Fct}_{\mr{top}}(\mr{G}_{\mf{B}},\mr{G}_{\mf{A}});
\end{cases}
\end{equation*}
which together 
\cite[Def. 1.1.1.(vi,ix,x)]{ks} and \cite[Def. 1.1.1.(iii,v)]{ks} 
yield $\mf{A},\mf{B}\in\mc{U}_{0}$
and $\mr{Mor}_{\mf{dp}}(\mf{A},\mf{B})\in\mc{U}_{0}$ respectively, 
which proves that $\mf{dp}$ is a $\mc{U}_{0}-$type category.
\end{proof}
\begin{convention}
By abuse of language we denoted with the same symbol $\mf{dp}$ the category uniquely determined in Cor.\ref{09100952}
and the set in Def.\ref{09081303}. We convein to distinguish the two cases explicitly namely 
$\mf{A}\in\mr{Obj}(\mf{dp})$ stands for $\mf{A}$ is an object of the category $\mf{dp}$ in Cor.\ref{09100952}, while 
$\mf{B}\in\mf{dp}$ stands for $\mf{B}$ is an element of the set defined in Def.\ref{09081303}.
\end{convention}
\begin{definition}
\label{11301526}
Let $\mf{dp}_{\star}$ be the full subcategory of $\mf{dp}$ whose object set is the subset of the $\mf{A}\in\mr{Obj}(\mf{dp})$ such 
that $\mr{G}_{\mf{A}}$ is a groupoid whose inversion map is continuous.
\end{definition}
Note that if $\mf{A}\in\mr{Obj}(\mf{dp}_{\star})$ then $\uptau_{\mf{A}}(g^{-1})=\uptau_{\mf{A}}(g)^{-1}$ for all $g\in\mr{G}_{\mf{A}}$.
\begin{definition}
\label{09081302}
The set of preordered dynamical patterns $\mf{pdp}$ is defined to be the set of the couples $\lr{\mr{W}}{\upeta}$ where
$\mr{W}$ is a $\mr{top}-$quasi enriched category and $\upeta\in\mr{Fct}_{\mr{top}}(\mr{W},\mr{ptls})$.
Let $\mr{A}=\lr{\mr{W}}{\upeta}\in\mf{pdp}$, then we denote $\mr{W}$ by $\mr{G}_{\mr{A}}$, $\mr{Mor}_{\mr{W}}(x,y)$ by 
$\mr{G}_{\mr{A}}(x,y)$, for all $x,y\in\mr{W}$, $\upeta$ by $\upsigma_{\mr{A}}$,
while the object and morphism maps $\upeta_{o}$ and $\upeta_{m}$ by $\mr{X}_{\mr{A}}$ and $\uptau_{\mr{A}}$ respectively.
The elements of $\mf{pdp}$ are called preordered dynamical patterns. 
Let $\mr{A},\mr{B},\mr{C}\in\mf{pdp}$, define $\un_{\mr{A}}\coloneqq(\un_{\mr{G}_{\mr{A}}},\un_{\upsigma_{\mr{A}}})$ and
\begin{equation}
\label{09081302a}
\mr{Mor}_{\mf{pdp}}(\mr{A},\mr{B})\coloneqq\coprod_{f\in\mr{Fct}_{\mr{top}}(\mr{G}_{\mr{A}},\mr{G}_{\mr{B}})}
\mr{Mor}_{\mr{Fct}(\mr{G}_{\mr{A}},\mr{ptls})}(\upsigma_{\mr{A}},\upsigma_{\mr{B}}\circ f),
\end{equation}
and 
\begin{equation}
\label{09081302b}
\begin{aligned}
(\circ):\mr{Mor}_{\mf{pdp}}(\mr{B},\mr{C})&\times \mr{Mor}_{\mf{pdp}}(\mr{A},\mr{B})\to \mr{Mor}_{\mf{pdp}}(\mr{A},\mr{C}),
\\
(g,S)&\circ(f,T)\coloneqq(g\circ f,(S\ast \un_{f})\circ T).
\end{aligned}
\end{equation}
We call $\mr{A}$ a $\mc{U}-$type preordered dynamical pattern or $\mc{U}-$type $\mf{pdp}$
if $\mr{A}\in\mf{pdp}$, $\mr{Obj}(\mr{G}_{\mr{A}})\simeq A$ and $A\subseteq\mc{U}$.
\end{definition}
\begin{remark}
\label{09081318a}
If $\mr{A}\in\mf{pdp}$ then according to the notation in the previous definition we have
$\mr{A}=\lr{\mr{G}_{\mr{A}}}{\upsigma_{\mr{A}}}=
\lr{\left(\mr{Obj}(\mr{G}_{\mr{A}}),\{\mr{G}_{\mr{A}}(x,y)\}_{x,y\in \mr{Obj}(\mr{G}_{\mr{A}})}\right)}
{(\mr{X}_{\mr{A}},\uptau_{\mr{A}})}$. 
If in addition $\mr{A}$ is of $\mc{U}-$type, then $\mr{G}_{\mr{A}}$ 
is a $\mc{U}-$type category in particular it is an object of $\mr{cat}$.
\end{remark}
\begin{remark}
\label{09081318}
Let $\mr{Mor}_{\mr{ptls}}(X,Y)$ be provided with the topology of pointwise convergence for all $X,Y\in\mr{ptls}$. 
Thus $\lr{\left(\mr{Obj}(\mr{G}_{\mr{A}}),\{\mr{G}_{\mr{A}}(x,y)\}_{x,y\in \mr{Obj}(\mr{G}_{\mr{A}})}\right)}
{(\mr{X}_{\mr{A}},\uptau_{\mr{A}})}\in\mf{pdp}$ iff 
\begin{enumerate}
\item
$\mr{G}_{\mr{A}}$ is a $\mc{U}-$category such that
$\mr{G}_{\mr{A}}(x,y)$ is a topological space and the morphism composition 
$\circ:\mr{G}_{\mr{A}}(y,z)\times\mr{G}_{\mr{A}}(x,y)\to\mr{G}_{\mr{A}}(x,z)$ 
is a separately continuous map, for all 
$x,y,z\in \mr{Obj}(\mr{G}_{\mr{A}})$;
\item
$\mr{X}_{\mr{A}}:\mr{Obj}(\mr{G}_{\mr{A}})\to \mr{Obj}(\mr{ptls})$;
\item
$\uptau_{\mr{A}}:\mr{Mor}_{\mr{G}_{\mr{A}}}\to \mr{Mor}_{\mr{ptls}}$
such that 
$\uptau_{\mr{A}}^{y,z}:\mr{G}_{\mr{A}}(y,z)\to 
\mr{Mor}_{\mr{ptls}}\left(\mr{X}_{\mr{A}}(y),\mr{X}_{\mr{A}}(z)\right)$ 
is a continuous map, for all
$y,z\in \mr{Obj}(\mr{G}_{\mr{A}})$;
\item
$\uptau_{\mr{A}}(g\circ h)=\uptau_{\mr{A}}(g)\circ\uptau_{\mr{A}}(h)$, 
and $\uptau_{\mr{A}}(\un_{x})=\un_{\mr{X}_{\mr{A}}(x)}$,
for all $x,y,z\in \mr{Obj}(\mr{G}_{\mr{A}})$, $g\in\mr{G}_{\mr{A}}(y,z)$
and $h\in\mr{G}_{\mr{A}}(x,y)$.
\end{enumerate}
Let $\mr{A},\mr{B}\in\mf{pdp}$, thus $(f,T)\in \mr{Mor}_{\mf{pdp}}(\mr{A},\mr{B})$ iff
\begin{enumerate}
\item
$f=(f_{o},f_{m})$ such that 
$f_{o}:\mr{Obj}(\mr{G}_{\mr{A}})\to \mr{Obj}(\mr{G}_{\mr{B}})$
and
$f_{m}:\mr{Mor}_{\mr{G}_{\mr{A}}}\to \mr{Mor}_{\mr{G}_{\mr{B}}}$;
\item
for all $y,z\in \mr{Obj}(\mr{G}_{\mr{A}})$ 
\begin{enumerate}
\item
$f_{m}^{y,z}:\mr{G}_{\mr{A}}(y,z)\to\mr{G}_{\mr{B}}(f_{o}(y),f_{o}(z))$ is a continuous map;
\item
$f_{m}(g\circ h)=f_{m}(g)\circ f_{m}(h)$ and $f_{m}(\un_{x})=\un_{f_{o}(x)}$,
for all $x\in \mr{Obj}(\mr{G}_{\mr{A}})$, $g\in\mr{G}_{\mr{A}}(y,z)$ 
and $h\in\mr{G}_{\mr{A}}(x,y)$;
\item
$T\in\prod_{x\in \mr{Obj}(\mr{G}_{\mr{A}})}
\mr{Mor}_{\mr{ptls}}\left(\mr{X}_{\mr{A}}(x),\mr{X}_{\mr{B}}(f_{o}(x))\right)$ 
such that for all 
$g\in\mr{G}_{\mr{A}}(y,z)$ we have that the following diagram in $\mr{ptls}$ is commutative
\begin{equation*}
\xymatrix{
\mr{X}_{\mr{A}}(y)
\ar[rr]^{T(y)}
\ar[dd]_{\uptau_{\mr{A}}(g)}
&&
\mr{X}_{\mr{B}}(f_{o}(y))
\ar[dd]^{\uptau_{\mr{B}}(f_{m}(g))}
\\
&&
\\
\mr{X}_{\mr{A}}(z)
\ar[rr]_{T(z)}
&&
\mr{X}_{\mr{B}}(f_{o}(z))
}
\end{equation*}
\end{enumerate}
\end{enumerate}
\end{remark}
Under the same line of reasoning used in the proof of Lemma \ref{09101139} we show that
\begin{lemma}
\label{09101200}
The composition in \eqref{09081302b} is a well-defined associative map such that 
$(f,T)\circ\un_{\mr{A}}=\un_{\mr{B}}\circ(f,T)=(f,T)$, 
for all $\mr{A},\mr{B}\in\mf{pdp}$ and $(f,T)\in \mr{Mor}_{\mf{pdp}}(\mr{A},\mr{B})$.
\end{lemma}
\begin{corollary}
\label{09101202}
There exists a unique $\mc{U}_{0}-$type category $\mf{pdp}$ such that 
$\mr{Obj}(\mf{pdp})$ is the set of the $\mr{A}\in\mf{pdp}$ for which $\mr{A}$ is a $\mc{U}-$type preordered dynamical pattern,
for any $\mr{A},\mr{B}\in\mr{Obj}(\mf{pdp})$ the set of morphisms of $\mf{pdp}$ from $\mr{A}$ to $\mr{B}$ is the set in 
\eqref{09081302a}, and the law of composition of morphisms of $\mf{pdp}$ is the map in \eqref{09081302b}.
$\un_{\mr{D}}$ defined in Def. \ref{09081302} is the unit morphism in $\mf{pdp}$ relative to $\mr{D}$, for all 
$\mr{D}\in \mr{Obj}(\mf{pdp})$. In particular $\mf{pdp}$ is an object of $\mr{Cat}$.
\end{corollary}
\begin{proof}
Since Lemma \ref{09101200} and then follows the line of reasoning present in the proof of Cor. \ref{09100952},
by considering the $\mc{U}-$category $\mr{ptls}$ instead of $\mr{tsa}$.
\end{proof}
\begin{convention}
By abuse of language we denoted with the same symbol $\mf{pdp}$ the category uniquely determined in Cor.\ref{09101202}
and the set in Def.\ref{09081302}. We convein to distinguish the two cases explicitly namely 
$\mr{A}\in\mr{Obj}(\mf{pdp})$ stands for $\mr{A}$ is an object of the category $\mf{pdp}$ in Cor.\ref{09101202}, while 
$\mr{B}\in\mf{pdp}$ stands for $\mr{B}$ is an element of the set defined in Def.\ref{09081302}.
\end{convention}
\begin{definition}
Let $\mf{pdp}_{\star}$ be the full subcategory of $\mf{pdp}$ whose object set is the subset of the $\mr{A}\in\mr{Obj}(\mf{pdp})$ 
such that $\mr{G}_{\mr{A}}$ is a groupoid whose inversion map is continuous.
\end{definition}
Note that if $\mr{A}\in\mf{pdp}_{\star}$ then $\uptau_{\mr{A}}(g^{-1})=\uptau_{\mr{A}}(g)^{-1}$ for all $g\in\mr{G}_{\mr{A}}$.
\begin{definition}
\label{12051932}
Let $\mf{A}\in\mf{dp}$, 
define 
$\mf{A}^{\dagger}\coloneqq\lr{\mr{G}_{\mf{A}}^{op}}{\upsigma_{\mf{A}}^{\dagger}}$,
where 
$\upsigma_{\mf{A}}^{\dagger}\coloneqq(\mc{A}_{\mf{A}}^{\ast},\uptau_{\mf{A}}^{\dagger})$, with
\begin{equation*}
\begin{aligned}
\mc{A}_{\mf{A}}^{\ast}:\mr{Obj}(\mr{G}_{\mf{A}})&\to \mr{Obj}(\mr{ptls}) & x&\mapsto(\mc{A}_{\mf{A}}(x))^{\ast},
\\
\uptau_{\mf{A}}^{\dagger}:\mr{Mor}_{\mr{G}_{\mf{A}}^{op}}&\to \mr{Mor}_{\mr{ptls}}
\\
\uptau_{\mf{A}}^{\dagger}:\mr{Mor}_{\mr{G}_{\mf{A}}^{op}}(x,y)&\to 
\mr{Mor}_{\mr{ptls}}(\mc{A}_{\mf{A}}^{\ast}(x),\mc{A}_{\mf{A}}^{\ast}(y))
&
g&\mapsto(\uptau_{\mf{A}}(g))^{\dagger},\,\forall x,y\in\mr{G}_{\mf{A}}.
\end{aligned}
\end{equation*}
Let $\mf{B}\in\mf{dp}$ such that $\mr{G}_{\mf{B}}$ is a groupoid whose inversion map is continuous. Define 
$\mf{B}^{\ast}\coloneqq\lr{\mr{G}_{\mf{B}}}{\upsigma_{\mf{B}}^{\ast}}$,
where $\upsigma_{\mf{B}}^{\ast}\coloneqq(\mc{A}_{\mf{B}}^{\ast},\uptau_{\mf{B}}^{\ast})$, with
$\uptau_{\mf{B}}^{\ast}:\mr{Mor}_{\mr{G}_{\mf{B}}}\to \mr{Mor}_{\mr{ptls}}$
such that 
\begin{equation*}
\uptau_{\mf{B}}^{\ast}:\mr{Mor}_{\mr{G}_{\mf{B}}}(x,y)\to 
\mr{Mor}_{\mr{ptls}}(\mc{A}_{\mf{B}}^{\ast}(x),\mc{A}_{\mf{B}}^{\ast}(y))
\quad
g\mapsto(\uptau_{\mf{B}}(g^{-1}))^{\dagger},\,\forall x,y\in\mr{G}_{\mf{B}}.
\end{equation*}
\end{definition}
The previous definition as well the next one are well set 
since \eqref{08161036} and Lemma \ref{08080010}\eqref{08080010st2}.
\begin{definition}
\label{09101557}
Let $D$ be a category, $\mr{a},\mr{b}\in\mr{Fct}(D,\mr{tsa})$ and 
$T\in\prod_{d\in D}\mr{Mor}_{\mr{tsa}}(\mr{a}(d),\mr{b}(d))$, then
define $T^{\dagger}\in\prod_{d\in D}\mr{Mor}_{\mr{ptls}}(\mr{b}(d)^{\ast},\mr{a}(d)^{\ast})$ such that 
$T^{\dagger}(e)\coloneqq(T(e))^{\dagger}$, for all $e\in D$.
\end{definition}
\begin{theorem}
\label{09101517}
Let $\mf{A},\mf{B},\mf{C}\in\mf{dp}$, $\mf{D}\in\mf{dp}_{\star}$, $(f,T)\in \mr{Mor}_{\mf{dp}}(\mf{A},\mf{B})$ and 
$(g,S)\in \mr{Mor}_{\mf{dp}}(\mf{B},\mf{C})$. 
Thus, $\mf{A}^{\dagger}\in\mf{pdp}$ and it is of $\mc{U}-$type if it is so $\mf{A}$, 
and $\mf{D}^{\ast}$ is an object of $\mf{pdp}_{\star}$. 
Furthermore, $(f,T^{\dagger})\in \mr{Mor}_{\mf{pdp}}(\mf{B}^{\dagger},\mf{A}^{\dagger})$
and if we set $(f,T)^{\dagger}\coloneqq(f,T^{\dagger})$, then $((g,S)\circ(f,T))^{\dagger}=(f,T)^{\dagger}\circ(g,S)^{\dagger}$.
\end{theorem}
\begin{proof}
Let $\mf{A},\mf{B}\in\mf{dp}$, $\mf{D}\in\mf{dp}_{\star}$, and $(f,T)\in \mr{Mor}_{\mf{dp}}(\mf{A},\mf{B})$. 
$\uptau_{\mf{A}}^{\dagger}$ and $\uptau_{\mf{D}}^{\ast}$ are continuous maps since Lemma \ref{08080010}\eqref{08080010st1}.
Next $\uptau_{\mf{A}}^{\dagger}(h\circ^{op}g)=\uptau_{\mf{A}}^{\dagger}(h)\circ\uptau_{\mf{A}}^{\dagger}(g)$, 
and $\uptau_{\mf{D}}^{\ast}(l\circ m)=\uptau_{\mf{D}}^{\ast}(l)\circ\uptau_{\mf{D}}^{\ast}(m)$, 
since \eqref{08161036} and Lemma \ref{08080010}\eqref{08080010st2}, where $(\circ)^{op}$ is the composition 
of morphisms of $\mr{G}_{\mf{A}}^{op}$; hence $\mf{A}^{\dagger},\mf{D}^{\ast}\in\mf{pdp}$.
Next clearly $f\in\mr{Fct}(\mr{G_{\mf{B}}}^{op},\mr{G_{\mf{A}}}^{op})$, moreover since the diagram in Rmk. \ref{09090657} and
\eqref{08161036} and Lemma \ref{08080010}\eqref{08080010st2}, we deduce that the following diagram in $\mr{ptls}$
is commutative for all $g\in\mr{G}_{\mf{B}}^{op}(z,y)$ 
\begin{equation*}
\xymatrix{
\mc{A}_{\mf{A}}^{\ast}(f_{o}(y))
&&
\mc{A}_{\mf{B}}^{\ast}(y)
\ar[ll]_{T^{\dagger}(y)}
\\
&&
\\
\mc{A}_{\mf{A}}^{\ast}(f_{o}(z))
\ar[uu]^{\uptau_{\mf{A}}^{\dagger}(f_{m}(g))}
&&
\mc{A}_{\mf{B}}^{\ast}(z)
\ar[uu]_{\uptau_{\mf{B}}^{\dagger}(g)}
\ar[ll]^{T^{\dagger}(z)}}
\end{equation*}
Therefore $(f,T^{\dagger})\in Mr_{\mf{pdp}}(\mf{B}^{\dagger},\mf{A}^{\dagger})$ since the diagram in 
Rmk. \ref{09081318}. The equality in the statement follows since Lemma \ref{08080010}\eqref{08080010st2},
\eqref{09091136} and \eqref{09081302b}.
\end{proof}
\begin{corollary}
\label{10111742}
The maps $\mf{A}\mapsto\mf{A}^{\dagger}$, $(f,T)\mapsto(f,T^{\dagger})$ and 
respectively $\mf{A}\mapsto\mf{A}^{\ast}$ and $(f,T)\mapsto(f,T^{\dagger})$,  
determine uniquely an element in $\mr{Fct}(\mf{dp}^{op},\mf{pdp})$ 
and respectively in $\mr{Fct}(\mf{dp}_{\star}^{op},\mf{pdp}_{\star})$. 
\end{corollary}
\begin{proof}
By Thm. \ref{09101517} and Cor. \ref{09101202}.
\end{proof}
\begin{definition}
\label{01011605}
Let $\mr{ptsa}$ be the category constructed in Prp.\ref{08090901}\eqref{08090901st3},
and let $\mf{p}\in\mr{Fct}(\mr{tsa},\mr{ptsa})$ the forgetful functor. 
\end{definition}
\begin{definition}
\label{12312135}
Let $\mf{A},\mf{B},\mf{C}\in\mf{dp}$, define 
\begin{equation}
\label{12312137}
\begin{split}
\mr{Mor}_{\mf{Chdv}}&(\mf{A},\mf{B})\coloneqq
\\
&\coprod_{f\in\mr{Fct}_{\mr{top}}(\mr{G}_{\mf{B}},\mr{G}_{\mf{A}})}
\mr{Mor}_{\mr{Fct}(\mr{G}_{\mf{B}}^{op},\mr{ptls})}
(\upsigma_{\mf{B}}^{\dagger},\upsigma_{\mf{A}}^{\dagger}\circ f)
\times
\mr{Mor}_{\mr{Fct}(\mr{G}_{\mf{B}},\mr{ptsa})}
(\mf{p}\circ\upsigma_{\mf{A}}\circ f,\mf{p}\circ\upsigma_{\mf{B}})
\end{split}
\end{equation}
and 
\begin{equation}
\label{12312113}
\begin{aligned}
(\circ):\mr{Mor}_{\mf{Chdv}}(\mf{B},\mf{C})&\times \mr{Mor}_{\mf{Chdv}}(\mf{A},\mf{B})\to \mr{Mor}_{\mf{Chdv}}(\mf{A},\mf{C}),
\\
(g,L,S)&\circ(f,H,T)
\coloneqq
(f\circ g,(H\ast\un_{g})\circ L,S\circ(T\ast \un_{g})).
\end{aligned}
\end{equation}
\end{definition}
\begin{corollary}
\label{31122146}
There exists a unique $\mc{U}_{0}-$type category $\mf{Chdv}$ such that $\mr{Obj}(\mf{Chdv})=\mr{Obj}(\mf{dp})$, 
for any $\mf{A},\mf{B}\in\mr{Obj}(\mf{Chdv})$ the set of morphisms
of $\mf{Chdv}$ from $\mf{A}$ to $\mf{B}$ is the set in \eqref{12312137}
and the law of composition of morphisms of $\mf{Chdv}$ is the map
in \eqref{12312113},
moreover $(\un_{\mr{G}_{\mf{A}}},\un_{\upsigma_{\mf{A}}^{\dagger}},\un_{\mf{p}\circ\upsigma_{\mf{A}}})$
is the unit morphism in $\mf{Chdv}$ relative to $\mf{A}$.
In particular $\mf{Chdv}$ is an object of $\mr{Cat}$.
\end{corollary}
\begin{proof}
By Cor. \ref{09100952} and Cor. \ref{09101202}, and since 
$\mr{Obj}(\mr{Fct}(A,B))=\mr{Obj}(\mr{Fct}(A^{op},B^{op}))$ for all categories $A,B$.
\end{proof}
\begin{theorem}
\label{12312222}
The maps $\mf{A}\mapsto\mf{A}$ and $(f,T)\mapsto(f,T^{\dagger},\un_{\mf{p}}\ast T)$ 
determine uniquely an element $\ps{\Uppsi}\in\mr{Fct}(\mf{dp},\mf{Chdv})$.
\end{theorem}
\begin{proof}
Since Cor. \ref{10111742}.
\end{proof}
\begin{convention}
\label{11211940}
For any map $U$ on a set $X$ and at values in $\mr{Mor}_{\mf{Chdv}}$ 
let $U_{j}$ and $U_{1}^{o},U_{1}^{m}$, 
$j\in\{1,2,3\}$ be maps on $X$ such that 
$U(M)=(U_{1}(M),U_{2}(M),U_{3}(M))$ and $U_{1}(M)=(U_{1}^{o}(M),U_{1}^{m}(M))$, 
namely $U_{1}^{o}(M)=(U_{1}(M))_{o}$ is the object map 
and $U_{1}^{m}(M)=(U_{1}(M))_{m}$ is the morphism map 
respectively of the functor $U_{1}(M)$ for all $M\in X$.
Moreover $U_{3}^{\dagger}(M)(\cdot)\coloneqq U_{3}(M)(\cdot)^{\dagger}$,
where we recall that $U_{3}(M)$ and $U_{2}(M)$ are natural transformations.
If $U=\mr{a}_{m}$ with $\mr{a}$ any functor at values in $\mf{Chdv}$, then we convein to denote
$(\mr{a}_{m})_{1}^{o}$ and $(\mr{a}_{m})_{1}^{m}$ 
simply by
$\mr{a}_{1}^{o}$ and $\mr{a}_{1}^{m}$,
while
$(\mr{a}_{m})_{2}$ and $(\mr{a}_{m})_{3}$
by 
$\mr{a}_{2}$ and $\mr{a}_{3}$
respectively.
We apply to objects of $\mf{Chdv}$ the notation in Def. \ref{09081303} for objects of $\mf{dp}$.
\end{convention}
Next we shall introduce the concept of trajectory associated with any 
species of dynamical patterns, namely a
functor from a category $\mf{D}$ said context category and $\mf{Chdv}$,
and any context. 
\begin{definition}
[\textbf{Trajectories}]
\label{01162119}
Let $\mf{D}$ be a category, $\mr{a}\in\mr{Fct}(\mf{D},\mf{Chdv})$ and $M\in\mf{D}$
define
\begin{equation*}
\mf{f}^{\mr{a},M}
\in
\prod_{(x,y)\in\mr{G}_{\mr{a}(M)}\times\mr{G}_{\mr{a}(M)}}
\mr{Mor}_{\mr{set}}(
\mf{P}_{\mc{A}_{\mr{a}(M)}(y)}\times\mc{A}_{\mr{a}(M)}(x)_{ob},
\R^{\mr{G}_{\mr{a}(M)}(x,y)}),
\end{equation*}
such that for all $x,y\in\mr{G}_{\mr{a}(M)}$, 
$(\uppsi,A)\in
\mf{P}_{\mc{A}_{\mr{a}(M)}(y)}\times\mc{A}_{\mr{a}(M)}(x)_{ob}$
and $g\in\mr{G}_{\mr{a}(M)}(x,y)$
\begin{equation*}
\mf{f}_{(\uppsi,A)}^{\mr{a},M,x,y}
(g)
\doteqdot
\uppsi(\uptau_{\mr{a}(M)}(g)A),
\end{equation*}
set
\begin{equation*}
\mr{Tr}^{\mr{a}}(M,x,y)\coloneqq\{\mf{f}_{(\uppsi,A)}^{\mr{a},M,x,y}\,\vert\,
(\uppsi,A)\in\mf{P}_{\mc{A}_{\mr{a}(M)}(y)}\times\mc{A}_{\mr{a}(M)}(x)_{ob}\}.
\end{equation*}
Let $\mr{b}\in\mr{Fct}(\mf{D},\mf{Chdv})$, $N\in\mf{D}$,
$t\in\mr{G}_{\mr{a}(M)}$, and $u,v\in\mr{G}_{\mr{b}(N)}$, define 
\begin{equation*}
\begin{aligned}
\mf{d}^{\mr{a},M;\mr{b},N;t,u,v}:\,
\mf{Z}(\mc{A}_{\mr{a}(M)}(t),\mc{A}_{\mr{b}(N)}(v))&
\times\mf{P}_{\mc{A}_{\mr{a}(M)}(t)}^{\natural}\times\mr{Ef}(\mc{A}_{\mr{b}(N)}(u))
\to[0,1]^{\mr{G}_{\mr{b}(N)}(u,v)},
\\
(J,\upomega,e)&\mapsto\left(l\mapsto\frac{J(\upomega)(\uptau_{\mr{b}(N)}(l)e)}{\upomega(\un)}\right).
\end{aligned}
\end{equation*}
We let $\mf{d}^{\mr{b},N;u,v}$ denote $\mf{d}^{\mr{a},M;\mr{b},N;t,u,v}$ whenever 
it will be clear by the context the functor $\mr{a}$, and $M$ and $t$ involved.
\end{definition}
Notice that 
$\upomega(\un)\mf{d}_{(J,\upomega,e)}^{\mr{b},N;u,v}(l)=\mf{f}_{(J(\upomega),e)}^{\mr{b},N;u,v}(l)$.
Although the precise physical interpretation will follow 
after Def. \ref{01200837}, we can say that 
$\mf{f}_{(\uppsi,A)}^{\mr{a},M,x,y}$ is roughly
the trajectory 
mapping any morphism $g$ of the dynamical category of the 
dynamical pattern $\mr{a}(M)$
- 
associated with the context $M$
and implemented by the species $\mr{a}$
-
into the measure in the statistical ensemble $\uppsi$
of the observable 
$\uptau_{\mr{a}(M)}(g)A$.
Next we give a variant of the above definition suitable to 
better reveal the observable dependency.
\begin{definition}
[Observable trajectories]
\label{10201408}
Let $\mf{D}$ be a category, $\mr{a}\in\mr{Fct}(\mf{D},\mf{Chdv})$ and $M\in\mf{D}$
define
\begin{equation*}
\begin{aligned}
\mf{t}^{\mr{a},M}
\in
\prod_{(x,y)\in\mr{G}_{\mr{a}(M)}\times\mr{G}_{\mr{a}(M)}}
\mr{Mor}_{\mr{set}}(\mf{P}_{\mc{A}_{\mr{a}(M)}(y)},
&\mr{Mor}_{\mr{set}}(\mc{A}_{\mr{a}(M)}(x)_{ob},\R^{\mr{G}_{\mr{a}(M)}(x,y)})),
\\
\mf{t}^{\mr{a},M,x,y,\uppsi}(A)
&\coloneqq
\mf{f}_{(\uppsi,A)}^{\mr{a},M,x,y}.
\end{aligned}
\end{equation*}
Set
\begin{equation*}
\mc{O}^{\mr{a}}(M)\coloneqq\{\mf{t}^{\mr{a},M,x,y,\uppsi}
\,\vert\,
x,y\in\mr{G}_{\mr{a}(M)},
\uppsi\in\mf{P}_{\mc{A}_{\mr{a}(M)}(y)}\}.
\end{equation*}
\end{definition}
$\mf{f}^{\mr{a},M,x,y}$ is defined over all the couples of statistical ensembles 
and observables, however when dealing with a precise experimental setting one 
encounters subsets of those couples which are in general equivariant 
with respect to geometrical and dynamical transformations implemented by the species $\mr{a}$.
Thus we introduce the following 
\begin{definition}
[Experimental settings]
\label{01161844}
Let $\mf{D}$ be a category and $\mr{a}\in\mr{Fct}(\mf{D},\mf{Chdv})$, 
let $Exp(\mr{a})$ be called the set of experimental settings 
associated with $\mr{a}$, and defined to be the 
set of the $\mf{E}=(\mf{S},\mr{R})$ such that for all $M\in\mf{D}$ the following holds
\begin{enumerate}
\item 
$\mr{R}_{M}$ is a subcategory of $\mr{G}_{\mr{a}(M)}$;
\label{01161844st1}
\item 
$\mf{S}_{M}\in\prod_{t\in\mr{R}_{M}}\mc{P}(\mf{P}_{\mc{A}_{\mr{a}(M)}(t)})$; 
\label{01161844st2}
\item
for all $N\in\mf{D}$ and $\phi\in \mr{Mor}_{\mf{D}}(M,N)$
\begin{enumerate}
\item
$\mr{a}_{1}(\phi)\in\mr{Fct}_{\mr{top}}(\mr{R}_{N},\mr{R}_{M})$,
\label{01161844st3a}
\item
$\mr{a}_{3}^{\dagger}(\phi)(u)\mf{S}_{N}(u)
\subseteq
\mf{S}_{M}(\mr{a}_{1}^{o}(\phi)u)$, for all $u\in\mr{R}_{N}$;
\label{01161844st3b}
\end{enumerate}
\label{01161844st3}
\item
$\uptau_{\mr{a}(M)}^{\dagger}(g)\mf{S}_{M}(z)
\subseteq
\mf{S}_{M}(y)$,
for all $y,z\in\mr{R}_{M}$ and $g\in\mr{R}_{M}(y,z)$.
\label{01161844st4}
\end{enumerate}
If in addition \eqref{01161844st3b} holds by replacing $\mr{a}_{3}^{\dagger}$ by $\mr{a}_{2}$,
then $(\mf{S},\mr{R})$ is said to be complete.
Moreover let $Exp^{\ast}(\mr{a})$ be the subset of the $(\mf{S},\mr{R})$ in $Exp(\mr{a})$ s.t. 
for all $M\in\mf{D}$, $z\in\mr{R}_{M}$ and $a\in\mc{A}_{\mr{a}(M)}(z)$
\begin{equation}
\label{01161908}
\upzeta_{\mc{A}_{\mr{a}(M)}(z)}^{\dagger}(a)\mf{S}_{M}(z)\subseteq\mf{S}_{M}(z).
\end{equation}
\end{definition}
\begin{definition}
\label{01161912}
Let $\mf{D}$ be a category and $\mr{a}\in\mr{Fct}(\mf{D},\mf{Chdv})$, 
define $\mf{P}^{\mr{a}}$ and $\mr{G}^{\mr{a}}$ be maps on 
$\mr{Obj}(\mf{D})$ such that 
$\mr{G}_{M}^{\mr{a}}\coloneqq\mr{G}_{\mr{a}(M)}$ for all $M\in\mf{D}$ 
and 
$\mf{P}_{M}^{\mr{a}}(t)\coloneqq\mf{P}_{\mc{A}_{\mr{a}(M)}(t)}$ 
for all $t\in\mr{G}_{\mr{a}(M)}$.
It is easy to see that $(\mf{P}^{\mr{a}},\mr{G}^{\mr{a}})\in Exp^{\ast}(\mr{a})$ and it is complete, we call 
$(\mf{P}^{\mr{a}},\mr{G}^{\mr{a}})$ the standard experimental setting associated with $\mr{a}$. 
\end{definition}
\begin{remark}
\label{11211156}
Let $\mr{S},\mr{T}\in \mr{Mor}_{\mr{Fct}(\mf{D},\mf{Chdv})}$
such that $d(\mr{T})=c(\mr{S})$, $M\in\mf{D}$
and let $\mr{c}=c(\mr{T})$.
Thus by Conv. \ref{11211940}
\begin{equation*}
\begin{aligned}
(\mr{T}\circ\mr{S})(M)
&=
((\mr{T}\circ\mr{S})_{1}(M),
(\mr{T}\circ\mr{S})_{2}(M),
(\mr{T}\circ\mr{S})_{3}(M))
\\
&=
(\mr{S}_{1}(M)\circ\mr{T}_{1}(M),
(\mr{S}_{2}(M)\ast\un_{\mr{T}_{1}(M)})
\circ
\mr{T}_{2}(M),
\mr{T}_{3}(M)\circ
(\mr{S}_{3}(M)\ast\un_{\mr{T}_{1}(M)})),
\end{aligned}
\end{equation*}
where the second equality arises 
since $(\mr{T}\circ\mr{S})(M)=\mr{T}(M)\circ\mr{S}(M)$
and by the definition of morphism composition in $\mf{Chdv}$.
In particular by \eqref{20061403} 
we obtain for all $y\in \mr{Obj}(\mr{G}_{M}^{\mr{c}})$
\begin{equation*}
(\mr{T}\circ\mr{S})_{3}(M)(y)
=
\mr{T}_{3}(M)(y)\circ\mr{S}_{3}(M)(\mr{T}_{1}^{o}(M)y),
\end{equation*}
and
$(\mr{T}\circ\mr{S})_{3}^{\dagger}(M)(y)
=
((\mr{T}\circ\mr{S})_{3}(M)(y))^{\dagger}$.
\end{remark}
\begin{remark}[\textbf{Equiformity Principle and Natural Transformations}]
\label{12140558}
It is worthwhile remarking that we introduce links (Def. \ref{01161922})
\emph{only} with the intent of enlightening those specific properties possessed 
by any \textbf{natural transformation} $\mr{T}$ between functors valued in $\mf{Chdv}$ 
(Thm. \ref{01181342}\eqref{01181342st1} and Thm. \ref{10081910}\eqref{10081910st2})
that guarantee that the Equiformity Principle (Prp. \ref{01162038}) holds true for $\mr{T}$
(Thm. \ref{01181342}\eqref{01181342st2} and Thm. \ref{10081910}\eqref{10081910st5}).
The Equiformity Principle for $\mr{T}$ communicates 
in the physical terms of statistical ensembles, observables and devices 
what the first diagram in Lemma \ref{01011521} encodes in the mathematical terms of category theory.
\end{remark}
\begin{definition}
[Auxiliary concept of link]
\label{01161922}
Let $\mf{D}$ be a category, $\mr{a},\mr{b}\in\mr{Fct}(\mf{D},\mf{Chdv})$, 
and $(\mf{S}^{\mr{a}},\mr{R}^{\mr{a}})\in Exp(\mr{a})$,
$(\mf{S}^{\mr{b}},\mr{R}^{\mr{b}})\in Exp(\mr{b})$. 
We define $\mr{T}$ to be a link from $(\mf{S}^{\mr{b}},\mr{R}^{\mr{b}})$ 
to $(\mf{S}^{\mr{a}},\mr{R}^{\mr{a}})$ 
if $\mr{T}\in\prod_{O\in\mf{D}}\mr{Mor}_{\mf{Chdv}}(\mr{a}(O),\mr{b}(O))$ such that 
for all $M,N\in\mf{D}$, $\phi\in \mr{Mor}_{\mf{D}}(M,N)$, $y,z\in\mr{R}_{N}^{\mr{b}}$ 
and $g\in\mr{R}_{N}^{\mr{b}}(y,z)$ we have
\begin{enumerate}
\item
$\mr{T}_{1}(N)\in\mr{Fct}_{\mr{top}}(\mr{R}_{N}^{\mr{b}},\mr{R}_{N}^{\mr{a}})$;
\label{01161922st1}
\item
$\mr{T}_{3}^{\dagger}(N)(z)\mf{S}_{N}^{\mr{b}}(z)\subseteq\mf{S}_{N}^{\mr{a}}(\mr{T}_{1}^{o}(N)z)$;
\label{01161922st2}
\item
$\mr{T}_{3}(N)(z)\circ\uptau_{\mr{a}(N)}(\mr{T}_{1}^{m}(N)g)\circ\mr{a}_{3}(\phi)(\mr{T}_{1}^{o}(N)y)
=
\mr{b}_{3}(\phi)(z)\circ\uptau_{\mr{b}(M)}(\mr{b}_{1}^{m}(\phi)g)\circ\mr{T}_{3}(M)(\mr{b}_{1}^{o}(\phi)y)$;
\label{01161922st3}
\item
$\mr{T}_{2}(M)(\mr{b}_{1}^{o}(\phi)y)\circ\uptau_{\mr{b}(M)}^{\dagger}(\mr{b}_{1}^{m}(\phi)g)\circ\mr{b}_{2}(\phi)(z)
=
\mr{a}_{2}(\phi)(\mr{T}_{1}^{o}(N)y)\circ\uptau_{\mr{a}(N)}^{\dagger}(\mr{T}_{1}^{m}(N)g)\circ\mr{T}_{2}(N)(z)$.
\label{01161922st4}
\end{enumerate}
If in addition \eqref{01161922st2} holds by replacing $\mr{T}_{3}^{\dagger}$ by $\mr{T}_{2}$, then $\mr{T}$
is said to be complete.
\end{definition}
\begin{definition}
[Reduction]
Let $\mf{D}$ be a category, 
$\mr{a},\mr{b}\in\mr{Fct}(\mf{D},\mf{Chdv})$, and 
$\mf{E}^{\mr{a}}=(\mf{S}^{\mr{a}},\mr{R}^{\mr{a}})\in Exp(\mr{a})$ and 
$\mf{E}^{\mr{b}}=(\mf{S}^{\mr{b}},\mr{R}^{\mr{b}})\in Exp(\mr{b})$ and $N\in\mf{D}$.
We call $\mr{T}$ be a \emph{reduction from $\mf{E}^{\mr{a}}$ to $\mf{E}^{\mr{b}}$ in $N$} iff
$\mr{T}$ is a link from $\mf{E}^{\mr{b}}$ to $\mf{E}^{\mr{a}}$, $\mr{R}_{N}^{\mr{b}}$ is a $\mr{top}$-quasi enriched subcategory
of $\mr{R}_{N}^{\mr{a}}$ and $\mr{T}_{1}(N)$ is the injection functor.
\end{definition}
\begin{convention}
For any category $\mf{D}$, 
$\mr{a},\mr{b}\in\mr{Fct}(\mf{D},\mf{Chdv})$, 
and $\mr{T}\in\prod_{O\in\mf{D}}\mr{Mor}_{\mf{Chdv}}(\mr{a}(O),\mr{b}(O))$ 
by abuse of language we generalize the notation for natural transformations
so that $d(\mr{T})=\mr{a}$ and $c(\mr{T})=\mr{b}$ and call them
domain and codomain of $\mr{T}$ respectively.
\end{convention}
Let us introduce the interpretation of the data above defined.
\begin{definition}
\label{01200837}
We call $(\mf{M},\mf{s},\mf{u})$ a semantics for $\mf{Chdv}$, 
if $\mf{M}=(\mf{R},\mf{r},\mf{D},\mf{d},\mf{C},\mf{c},\mf{T},\mf{t},\mf{E},\mf{e},\mf{O},\mf{o})$ 
is a semantics, see Def. \ref{08081839}, and for any category $\mf{D}$, any
$\mr{a},\mr{b},\mr{c}\in\mr{Fct}(\mf{D},\mf{Chdv})$, 
$\mr{T}\in\prod_{O\in\mf{D}}\mr{Mor}_{\mf{Chdv}}(\mr{a}(O),\mr{b}(O))$,
$\mr{S}\in\prod_{O\in\mf{D}}\mr{Mor}_{\mf{Chdv}}(\mr{b}(O),\mr{c}(O))$ 
and 
$\mf{E}=(\mf{S},\mr{R})\in Exp(\mr{b})$,
we have for all $M,N,O\in \mr{Obj}(\mf{D})$, $\phi\in \mr{Mor}_{\mf{D}}(M,N)$, 
$\psi\in \mr{Mor}_{\mf{D}}(N,O)$, $x,y\in\mr{G}_{\mr{a}}^{N}$,
$h,g\in\mr{G}_{\mr{a}}^{N}(x,y)$, $u\in \mr{Mor}_{\mr{G}_{\mr{a}}^{N}}$,
$n\in\mr{G}_{\mr{b}}^{M}$, $z\in\mr{G}_{\mr{b}}^{M}$, 
$w\in \mr{Mor}_{\mr{G}_{\mr{b}}^{M}}$
and $r\in\mr{R}_{M}$
\begin{enumerate}
\item
$\mf{u}(\un_{N})\equiv$ producing no variations; 
\item
$\mf{s}(x)\equiv$ the region $\mf{u}(x)$;
\item
$\mf{s}(g)\equiv$ the action $\mf{u}(g)$ mapping $\mf{s}(x)$ into $\mf{s}(y)$;
\item
$\mf{s}(hg)\equiv$ $\mf{s}(h)$ after $\mf{s}(g)$;
\item
$\mf{s}(\phi)\equiv$ the geometric action $\mf{u}(\phi)$;
\item
$\mf{s}(\psi\circ\phi)\equiv$ $\mf{s}(\psi)$ after $\mf{s}(\phi)$;
\item
$\mf{s}(\mr{a})\equiv$ the species $\mf{u}(\mr{a})$;
\item
$\mf{s}(d(\mr{a}))\equiv$ the enviroment domain of $\mf{s}(\mr{a})$;
\item
$\mf{s}(M)\equiv$ the context $\mf{u}(M)$;
\item
$\mf{u}(c(\phi))\equiv$ projection of $\mf{s}(d(\phi))$ through $\mf{s}(\phi)$;
\item
$\mf{u}(c(\phi))\equiv$ obtained by transforming $\mf{s}(d(\phi))$ through $\mf{s}(\phi)$;
\item
$\mf{s}(\mr{a}(M))\equiv$ dynamical pattern associated with $\mf{s}(M)$ and implemented by $\mf{s}(\mr{a})$;
\item
$\mf{s}(\mr{a}_{1}(\phi))\equiv$ the probe $\mf{u}(\mr{a}_{1}(\phi))$; 
\item
$\mf{u}(\mr{a}_{1}(\phi))\equiv$ assembled by $\mf{s}(\mr{a})$ and implementing $\mf{s}(\phi)$;
\item
$\mf{u}(\mr{a}_{1}^{o}(\phi)x)\equiv$ resulting next $\mf{s}(\mr{a}_{1}(\phi))$ applies to $\mf{s}(x)$
\item
$\mf{u}(\mr{a}_{1}^{m}(\phi)(g))\equiv$ that operates in any reference frame $\bullet$ as $\mf{s}(g)$ 
operates in the reference frame obtained by applying $\mf{D}(\mr{a}_{3}(\phi)(x))$ to $\bullet$;
\label{03141328}
\item
$\mf{u}(\mr{a}_{1}^{m}(\phi)(u))\equiv$ resulting next $\mf{s}(\mr{a}_{1}(\phi))$ applies to $\mf{s}(u)$;
\item
$\mf{u}(\mr{a}_{2}(\phi))\equiv\mf{u}(\mr{a}_{1}(\phi))$;
\item
$\mf{c}(\mr{a}_{2}(\phi)(x))\equiv$ placed in $\mf{s}(x)$, $\mf{u}(\mr{a}_{2}(\phi))$;
\item
$\mf{u}(\mr{a}_{3}(\phi))\equiv\mf{u}(\mr{a}_{1}(\phi))$; 
\item
$\mf{d}(\mr{a}_{3}(\phi)(x))\equiv$ placed in $\mf{s}(x)$, $\mf{u}(\mr{a}_{3}(\phi))$;
\item
$\mf{s}(\mr{T})\equiv$ the connector $\mf{u}(\mr{T})$;
\item
$\mf{u}(\mr{T})\equiv$ adding the charge $\mr{T}$\footnote{The reason of 
such a choice is related to 
Thm. \ref{10081910}\eqref{10081910st1}.};
\item
$\mf{u}(\mr{T})\equiv$ from $\mf{s}(d(\mr{T}))$ to $\mf{s}(c(\mr{T}))$;
\item
$\mf{s}(\mr{T})\equiv$ the sector $\mf{u}(\mr{T})$,
if $d(\mr{T})=c(\mr{T})$;
\item
$\mf{u}(\mr{T})\equiv$ of $\mf{s}(d(\mr{T}))$,
if $d(\mr{T})=c(\mr{T})$;
\item
$\mf{u}(\un_{\mr{a}})\equiv$ canonically associated with $\mf{s}(\mr{a})$;
\item
$\mf{s}(\mr{T}\circ\mr{S})\equiv$ $\mf{s}(\mr{T})$ following $\mf{s}(\mr{S})$;
\item
$\mf{s}(\mr{T}_{1}(M))\equiv$ the probe $\mf{u}(\mr{T}_{1}(M))$;
\item
$\mf{u}(\mr{T}_{1}(M))\equiv$ assembled by $\mf{s}(\mr{T})$ and transporting into $\mf{s}(\mr{a}(M))$; 
\item
$\mf{u}(\mr{T}_{1}(M)^{o}(z))\equiv$ resulting next $\mf{s}(\mr{T}_{1}(M))$ applies to $\mf{s}(z)$; 
\item
$\mf{u}(\mr{T}_{1}(M)^{m}(w))\equiv$ that, via $\mf{s}(\mr{T}_{1}(M))$, 
operates in $\mf{s}(\mr{a}(M))$ as $\mf{s}(w)$ operates in $\mf{s}(\mr{b}(M))$;
\item
$\mf{u}(\mr{T}_{1}(M)^{m}(w))\equiv$ resulting next $\mf{s}(\mr{T}_{1}(M))$ applies to $\mf{s}(w)$; 
\item
$\mf{u}(\mr{T}_{2}(M))\equiv\mf{u}(\mr{T}_{1}(M))$;
\item
$\mf{c}(\mr{T}_{2}(M)(z))\equiv$ placed in $\mf{s}(z)$, $\mf{s}(\mr{T}_{2}(M))$;
\item
$\mf{u}(\mr{T}_{3}^{\dagger}(M))\equiv\mf{u}(\mr{T}_{1}(M))$; 
\item
$\mf{c}(\mr{T}_{3}^{\dagger}(M)(z))\equiv$ placed in $\mf{s}(z)$, $\mf{u}(\mr{T}_{3}^{\dagger}(M))$;
\item
$\mf{u}(\mr{T}_{3}(M))\equiv$ assembled by $\mf{s}(\mr{T})$ and transporting into $\mf{s}(\mr{b}(M))$; 
\item
$\mf{d}(\mr{T}_{3}(M)(n))\equiv$ placed in $\mf{s}(n)$, $\mf{u}(\mr{T}_{3}(M))$;
\item
$\mf{d}(\uptau_{\mr{a}(N)}(g))\equiv$ assembled by $\mf{s}(\mr{a})$ and implementing $\mf{s}(g)$ in $\mf{s}(N)$;
\label{01200837g}
\item
$\mf{d}(\uptau_{\mr{a}(N)}(g))\equiv$ assembled by $\mf{s}(\mr{a})$ and implementing in $\mf{s}(N)$ $\mf{s}(g)$;
\item
$\mf{s}(\mf{S})\equiv$ the bundle $\mf{S}$ of empirical sectors of $\mf{s}(\mr{b})$;
\item
$\mf{s}(\mr{R})\equiv$ the bundle $\mr{R}$ of dynamical categories of $\mf{s}(\mr{b})$;
\item
$\mf{s}(\mf{S}_{M}(r))\equiv$ the fiber $\mf{u}(\mf{S}_{M}(r))$;
\item
$\mf{u}(\mf{S}_{M}(r))\equiv$ in $\mf{s}(r)$ of $\mf{s}(M)$ of $\mf{s}(\mf{S})$;
\item
$\mf{s}(\mr{R}_{M})\equiv$ the fiber $\mf{u}(\mr{R}_{M})$;
\item
$\mf{u}(\mr{R}_{M})\equiv$ in $\mf{s}(M)$ of $\mf{s}(\mr{R})$;
\item
$\mf{s}(\mf{E})\equiv$ the experimental setting $\mf{u}(\mf{E})$;
\item
$\mf{u}(\mf{E})\equiv$ determined by $\mf{s}(\mf{S}_{Q}(a))$ and $\mf{s}(\mr{R}_{Q})$,
for all $Q\in\mf{D}$ and $a\in\mr{R}_{Q}$.
\end{enumerate}
\end{definition}
\begin{remark}
\label{03141425}
Def.\ref{01200837}\eqref{03141328} is consistent with the first equality in Prp.\ref{01162038}, 
in the special case where $\mr{a}=\mr{b}$ and $\mr{T}$ is the identity, 
see also the right commutative subdiagram in the first diagram of Lemma \ref{01011521}.
In particular if $\phi$ is an isomorphism then for any $h\in\mr{G}_{\mr{a}(N)}(\mr{a}_{1}^{o}(\phi^{-1})x,\mr{a}_{1}^{o}(\phi^{-1})y)$
we have by applying Def.\ref{01200837}\eqref{03141328} to $g=\mr{a}_{1}^{m}(\phi^{-1})(h)$
that $\mf{u}(h)\equiv$ that operates in any reference frame $\bullet$ as $\mf{s}(\mr{a}_{1}^{m}(\phi^{-1})(h))$ operates in the 
reference frame obtained by applying $\mf{D}(\mr{a}_{3}(\phi)x)$ 
to $\bullet$.
\end{remark}
\begin{definition}
[Active and reference frame interpretations]
\label{26091109}
Under the notation in Def. \ref{01200837},
let us define
\begin{enumerate}
\item 
Active interpretation.
\begin{enumerate}
\item
$\mf{u}(\mr{a}_{1}^{m}(\phi)(u))\equiv$ resulting next $\mf{s}(\mr{a}_{1}(\phi))$ applies to $\mf{s}(u)$;
\item
$\mf{u}(c(\phi))\equiv$ obtained by transforming $\mf{s}(d(\phi))$ through $\mf{s}(\phi)$.
\end{enumerate}
\label{09261133a}
\item 
Reference frame interpretation.
\begin{enumerate}
\item
$\mf{u}(\mr{a}_{1}^{m}(\phi)(g))\equiv$ that operates in any reference frame $\bullet$ as $\mf{s}(g)$ 
operates in the reference frame obtained by applying $\mf{D}(\mr{a}_{3}(\phi)(x))$ to $\bullet$;
\item
$\mf{u}(c(\phi))\equiv$ projection of $\mf{s}(d(\phi))$ through $\mf{s}(\phi)$.
\end{enumerate}
\label{09261133p}
\end{enumerate}
\end{definition}
\begin{convention}
\label{01201542}
For the remaining of the paper let $(\mf{M},\mf{s},\mf{u})$ be a fixed semantics for $\mf{Chdv}$ where
$\mf{M}=(\mf{R},\mf{r},\mf{D},\mf{d},\mf{C},\mf{c},\mf{T},\mf{t},\mf{E},\mf{e},\mf{O},\mf{o})$. 
\end{convention}
\begin{definition}
[Collections of trajectories]
\label{03032202}
Let $\mf{D}$ be a category, $\mr{a},\mr{b}\in\mr{Fct}(\mf{D},\mf{Chdv})$, 
$\mf{E}^{\mr{a}}=(\mf{S}^{\mr{a}},\mr{R}^{\mr{a}})\in Exp(\mr{a})$,
$\mf{E}^{\mr{b}}=(\mf{S}^{\mr{b}},\mr{R}^{\mr{b}})\in Exp(\mr{b})$, 
$\mr{T}$ be a link from $\mf{E}^{\mr{b}}$ to $\mf{E}^{\mr{a}}$,
$M,N\in\mf{D}$, $\phi\in \mr{Mor}_{\mf{D}}(M,N)$ and $y,z\in\mr{R}_{N}^{\mr{b}}$. 
Define  
\begin{equation*}
\begin{aligned}
\mr{Tr}(\mr{b};M;\phi,\mr{T},y,z\,\vert\,\mf{S}^{\mr{b}})
&\coloneqq
\{
\mf{f}_{(\upomega,B)}^{\mr{b},M,\mr{b}_{1}^{o}(\phi)y,\mr{b}_{1}^{o}(\phi)z}
\,
\vert
\,
\upomega\in\mr{b}_{3}^{\dagger}(\phi)(z)\mf{S}_{N}^{\mr{b}}(z),
B\in\mr{T}_{3}(M)(\mr{b}_{1}^{o}(\phi)y)\mc{A}_{\mr{a}(M)}((\mr{T}_{1}^{o}(M)\circ\mr{b}_{1}^{o}(\phi))y)_{ob}
\};
\\
\mr{Tr}(\mr{a};N;\mr{T},\phi,y,z\,\vert\,\mf{S}^{\mr{b}})
&\coloneqq
\{
\mf{f}_{(\upomega,B)}^{\mr{a},N,\mr{T}_{1}^{o}(N)y,\mr{T}_{1}^{o}(N)z}
\,\vert\,
\upomega\in\mr{T}_{3}^{\dagger}(N)(z)\mf{S}_{N}^{\mr{b}}(z),
B\in\mr{a}_{3}(\phi)(\mr{T}_{1}^{o}(N)y)\mc{A}_{\mr{a}(M)}((\mr{T}_{1}^{o}(M)\circ\mr{b}_{1}^{o}(\phi))y)_{ob}
\}.
\end{aligned}
\end{equation*}
Furthermore define 
\begin{equation*}
\begin{aligned}
\mc{Tr}(\mr{b};M;\phi,\mr{T},y,z\,\vert\,\mf{E}^{\mr{b}})
&\coloneqq
(\mr{b}_{1}^{m}(\phi)\up\mr{R}_{N}^{\mr{b}}(y,z))^{\dagger}
\mr{Tr}(\mr{b};M;\phi,\mr{T},y,z\,\vert\,\mf{S}^{\mr{b}});
\\
\mc{Tr}(\mr{a};N;\mr{T},\phi,y,z\,\vert\,\mf{E}^{\mr{b}})
&\coloneqq
(\mr{T}_{1}^{m}(N)\up\mr{R}_{N}^{\mr{b}}(y,z))^{\dagger}
\mr{Tr}(\mr{a};N;\mr{T},\phi,y,z\,\vert\,\mf{S}^{\mr{b}}).
\end{aligned}
\end{equation*}
\end{definition}
Let us call
\begin{itemize}
\item
$\mc{Tr}(\mr{b};M;\phi,\mr{T},y,z\,\vert\,\mf{E}^{\mr{b}})$ the collection of trajectories associated 
with the species $\mr{b}$, the context $M$, the argument transformed by $\phi$, 
the state initial conditions transformed by $\phi$ and determined by $z$,
the observable initial conditions transformed by $\mr{T}$ and determined by $y$; 
\item
$\mc{Tr}(\mr{a};N;\mr{T},\phi,y,z\,\vert\,\mf{E}^{\mr{b}})$ the collection of trajectories associated 
with the species $\mr{a}$, the context $N$, the argument transformed by $\mr{T}$, 
the state initial conditions transformed by $\mr{T}$ and determined by $z$, 
the observable initial conditions transformed by $\phi$ and determined by $y$. 
\end{itemize}
Recall that according to Def.\ref{08081839}\eqref{03041351} observable initial conditions transformations can be interpreted as
reference frame transformations.
\par
With in mind Rmk. \ref{12140558} we can now state the following
\begin{proposition}
[\textbf{Equiformity Principle Invariant Form}]
\label{01162038}
Let $\mf{D}$ be a category, 
$\mr{a},\mr{b}\in\mr{Fct}(\mf{D},\mf{Chdv})$, and 
$\mf{E}^{\mr{a}}=(\mf{S}^{\mr{a}},\mr{R}^{\mr{a}})\in Exp(\mr{a})$,
$\mf{E}^{\mr{b}}=(\mf{S}^{\mr{b}},\mr{R}^{\mr{b}})\in Exp(\mr{b})$, 
and $\mr{T}$ be a link from $\mf{E}^{\mr{b}}$ to $\mf{E}^{\mr{a}}$.
Thus for all $M,N\in\mf{D}$, $\phi\in \mr{Mor}_{\mf{D}}(M,N)$, 
$y,z\in\mr{R}_{N}^{\mr{b}}$ 
and $g\in\mr{R}_{N}^{\mr{b}}(y,z)$ we have
for all 
$\uppsi\in\mf{S}_{N}^{\mr{b}}(z)$ and 
$A\in\mc{A}_{\mr{a}(M)}((\mr{T}_{1}^{o}(M)\circ\mr{b}_{1}^{o}(\phi))y)_{ob}$ 
\begin{equation*}
\boxed{\boxed{\mf{f}_{(\mr{b}_{3}^{\dagger}(\phi)(z)\uppsi,\mr{T}_{3}(M)(\mr{b}_{1}^{o}(\phi)y)A)}^{\mr{b},M,\mr{b}_{1}^{o}(\phi)y,\mr{b}_{1}^{o}(\phi)z}
(\mr{b}_{1}^{m}(\phi)g)
=
\mf{f}_{(\mr{T}_{3}^{\dagger}(N)(z)\uppsi,\mr{a}_{3}(\phi)(\mr{T}_{1}^{o}(N)y)A)}^{\mr{a},N,\mr{T}_{1}^{o}(N)y,\mr{T}_{1}^{o}(N)z}
(\mr{T}_{1}^{m}(N)g).}}
\end{equation*}
In particular 
\begin{equation*}
\mc{Tr}(\mr{b};M;\phi,\mr{T},y,z\,\vert\,\mf{E}^{\mr{b}})
=
\mc{Tr}(\mr{a};N;\mr{T},\phi,y,z\,\vert\,\mf{E}^{\mr{b}}).
\end{equation*}
\end{proposition}
\begin{proof}
Straightforward consequence of Def. \ref{01161922}\eqref{01161922st3}.
\end{proof} 
The above equality reads as follows
\begin{remark}
[Equiformity principle for collections of trajectories]
The collection of trajectories associated with the species $\mr{b}$, the context $M$, the argument transformed by $\phi$, 
the state initial conditions transformed by $\phi$ and determined by $z$, 
the observable initial conditions transformed by $\mr{T}$ and determined by $y$ 
equals
the collection of trajectories associated with the species $\mr{a}$, the context $N$, the argument transformed by $\mr{T}$, 
the state initial conditions transformed by $\mr{T}$ and determined by $z$, 
the observable initial conditions transformed by $\phi$ and determined by $y$.
\end{remark}
\begin{remark}
Notice that the second equality in Prp.\ref{01162038} can be put in the following form
\begin{equation*}
(\mr{b}_{1}^{m}(\phi)\up\mr{R}_{N}^{\mr{b}}(y,z))^{\dagger}
\mr{Tr}(\mr{b};M;\phi,\mr{T},y,z\,\vert\,\mf{S}^{\mr{b}})
=
\mc{Tr}(\mr{a};N;\mr{T},\phi,y,z\,\vert\,\mf{E}^{\mr{b}}).
\end{equation*}
\end{remark}
By employing Def. \ref{08081839} and Def. \ref{01200837} and by multiplying the equality in Prp.\ref{01162038}
by $\uppsi(\un)^{-1}$ we obtain
\begin{proposition}
[\textbf{Interpretation of the equiformity principle}]
\label{12011304}
Under the hypothesis of Prp. \ref{01162038} and by assuming that $\uppsi(\un)\neq 0$ we have that:
\par
The expectation value in the statistical ensemble 
$\mf{t}(\mr{T}_{3}^{\dagger}(N)(z)\uppsi)$
of the observable resulting next the device 
$\mf{d}(\uptau_{\mr{a}(N)}(\mr{T}_{1}^{m}(N)g))$
applies to the observable
$\mf{o}(\mr{a}_{3}(\phi)(\mr{T}_{1}^{o}(N)y)A)$,
equals
the expectation value in the statistical ensemble 
$\mf{t}(\mr{b}_{3}^{\dagger}(\phi)(z)\uppsi)$
of the observable resulting next the device 
$\mf{d}(\uptau_{\mr{b}(M)}(\mr{b}_{1}^{m}(\phi)g))$
applies to the observable
$\mf{o}(\mr{T}_{3}(M)(\mr{b}_{1}^{o}(\phi)y)A)$.
Here
\begin{itemize}
\item
$\mf{t}(\mr{T}_{3}^{\dagger}(N)(z)\uppsi)\equiv$
resulting next the channel placed in the region $\mf{u}(z)$, 
assembled by the connector adding the charge $\mr{T}$
and transporting into $\mf{s}(\mr{a}(N))$,
applies to the statistical ensemble $\mf{t}(\uppsi)$; 
\item
$\mf{t}(\mr{T}_{3}^{\dagger}(N)(z)\uppsi)\equiv$
resulting next the channel induced by the device 
placed in the region $\mf{u}(z)$, 
assembled by the connector adding the charge $\mr{T}$
and transporting into $\mf{s}(\mr{b}(N))$,
applies to the statistical ensemble $\mf{t}(\uppsi)$; 
\item
$\mf{d}(\uptau_{\mr{a}(N)}(\mr{T}_{1}^{m}(N)g))\equiv$
assembled by the species $\mf{u}(\mr{a})$ and implementing
in the context $\mf{u}(N)$ 
the action $\mf{u}(\mr{T}_{1}^{m}(N)g)$
mapping the region $\mf{u}(\mr{T}_{1}^{o}(N)y)$
into the region $\mf{u}(\mr{T}_{1}^{o}(N)z)$;
\item
$\mf{u}(\mr{T}_{1}^{m}(N)g)\equiv$
resulting next the probe,
assembled by the connector adding the charge $\mr{T}$
and transporting into $\mf{s}(\mr{a}(N))$,
applies to the action $\mf{u}(g)$
mapping the region $\mf{u}(y)$
into the region $\mf{u}(z)$;
\item
$\mf{o}(\mr{a}_{3}(\phi)(\mr{T}_{1}^{o}(N)y)A)\equiv$
resulting next 
the device placed in the region $\mf{u}(\mr{T}_{1}^{o}(N)y)$
assembled by the species $\mf{u}(\mr{a})$ and implementing the 
geometric action $\mf{u}(\phi)$,
applies to the observable $\mf{o}(A)$;
\item
$\mf{u}(\mr{T}_{1}^{o}(N)y)\equiv$
resulting next the probe, 
assembled by the connector adding the charge $\mr{T}$
and transporting into $\mf{s}(\mr{a}(N))$,
applies to the region $\mf{u}(y)$;
\item
$\mf{t}(\mr{b}_{3}^{\dagger}(\phi)(z)\uppsi)\equiv$
resulting next the channel induced by 
the device placed in the region $\mf{u}(z)$,
assembled by the species $\mf{u}(\mr{b})$
and implementing the geometric action $\mf{u}(\phi)$,
applies to the statistical ensemble $\mf{t}(\uppsi)$; 
\item
$\mf{d}(\uptau_{\mr{b}(M)}(\mr{b}_{1}^{m}(\phi)g))\equiv$
assembled by the species $\mf{u}(\mr{b})$ and implementing
in the context $\mf{u}(M)$ 
the action $\mf{u}(\mr{b}_{1}^{m}(\phi)g)$
mapping the region $\mf{u}(\mr{b}_{1}^{o}(\phi)y)$
into the region $\mf{u}(\mr{b}_{1}^{o}(\phi)z)$;
\item
$\mf{u}(\mr{b}_{1}^{m}(\phi)g)\equiv$
resulting next the probe, 
assembled by the species $\mf{u}(\mr{b})$
and implementing the geometric action $\mf{u}(\phi)$,
applies to the action $\mf{u}(g)$
mapping the region $\mf{u}(y)$
into the region $\mf{u}(z)$;
\item
$\mf{o}(\mr{T}_{3}(M)(\mr{b}_{1}^{o}(\phi)y)A)\equiv$
resulting next the device placed in the region $\mf{u}(\mr{b}_{1}^{o}(\phi)y)$,
assembled by the connector adding the charge $\mr{T}$
and transporting into $\mf{s}(\mr{b}(M))$,
applies to the observable $\mf{o}(A)$;
\item
$\mf{s}(\mr{a}(N))\equiv$
dynamical pattern associated with the context $\mf{u}(N)$ and implemented by the species $\mf{u}(\mr{a})$;
\item
$\mf{s}(\mr{b}(M))\equiv$
dynamical pattern associated with the context $\mf{u}(M)$ and implemented by the species $\mf{u}(\mr{b})$.
\end{itemize}
\end{proposition}
\begin{definition}
$\intercal$ is the conjugate over the second place of variability
of a map of maps.
\end{definition}
Next Prp. \ref{10201503} and Prp. \ref{01162038pr} follow by Prp. \ref{01162038}. 
\begin{proposition}
[Equiformity Principle Equivariant Form]
\label{10201503}
Under the hypothesis of Prp. \ref{01162038} we have 
\begin{multline*}
(\mr{a}_{3}^{\dagger}(\phi)(\mr{T}_{1}^{o}(N)y)
\circ
(\mr{T}_{1}^{m}(N)\up\mr{R}_{N}^{\mr{b}}(y,z))^{\intercal})\,
\mf{t}^{\mr{a},N,\mr{T}_{1}^{o}(N)y,\mr{T}_{1}^{o}(N)z,\mr{T}_{3}^{\dagger}(N)(z)\uppsi}
=\\
(\mr{T}_{3}^{\dagger}(M)(\mr{b}_{1}^{o}(\phi)y)
\circ
(\mr{b}_{1}^{m}(\phi)
\up\mr{R}_{N}^{\mr{b}}(y,z))^{\intercal})\,
\mf{t}^{\mr{b},M,\mr{b}_{1}^{o}(\phi)y,\mr{b}_{1}^{o}(\phi)z,\mr{b}_{3}^{\dagger}(\phi)(z)\uppsi}.
\end{multline*}
\end{proposition}
\begin{proposition}
[\textbf{Equiformity Principle for Propensities}]
\label{01162038pr}
In addition to the hypothesis of Prp. \ref{01162038} let $\uppsi(\un)\neq 0$ and
$e\in\mr{Ef}(\mc{A}_{\mr{a}(M)}((\mr{T}_{1}^{o}(M)\circ\mr{b}_{1}^{o}(\phi))y))$ thus,
\begin{equation*}
\mf{d}_{(\mr{T}_{3}^{\dagger}(N)(z),\uppsi,
\mr{a}_{3}(\phi)(\mr{T}_{1}^{o}(N)y)e)}^{\mr{a},N,\mr{T}_{1}^{o}(N)y,\mr{T}_{1}^{o}(N)z}
(\mr{T}_{1}^{m}(N)g)
=
\mf{d}_{(\mr{b}_{3}^{\dagger}(\phi)(z),\uppsi,\mr{T}_{3}(M)(\mr{b}_{1}^{o}(\phi)y)e)}^{\mr{b},M,
\mr{b}_{1}^{o}(\phi)y,\mr{b}_{1}^{o}(\phi)z}
(\mr{b}_{1}^{m}(\phi)g).
\end{equation*}
\end{proposition}
The next is a preparatory result to Thm. \ref{10081910}
\begin{lemma}
\label{01011521}
Let $\mf{D}$ be a category, $\mr{a},\mr{b}\in\mr{Fct}(\mf{D},\mf{Chdv})$, 
and $\mr{T}\in \mr{Mor}_{\mr{Fct}(\mf{D},\mf{Chdv})}(\mr{a},\mr{b})$
then for all $M,N\in\mc{D}$, $\phi\in \mr{Mor}_{\mf{D}}(M,N)$,
$y,z\in\mr{G}_{\mr{b}(N)}$ and $g\in\mr{G}_{\mr{b}(N)}(y,z)$ 
the following is a commutative diagram in $\mr{ptsa}$ 
\begin{equation*}
\xymatrix{
\mc{A}_{\mr{a}(M)}((\mr{T}_{1}^{o}(M)\circ\mr{b}_{1}^{o}(\phi))y)
\ar[rrrr]^{\mr{T}_{3}(M)(\mr{b}_{1}^{o}(\phi)y)}
\ar[rdd]^{\uptau_{\mr{a}(M)}((\mr{T}_{1}^{m}(M)\circ\mr{b}_{1}^{m}(\phi))g)}
\ar[dddddd]^{\mr{a}_{3}(\phi)(\mr{T}_{1}^{o}(N)y)}
& & & &
\mc{A}_{\mr{b}(M)}(\mr{b}_{1}^{o}(\phi)y)
\ar[ldd]_{\uptau_{\mr{b}(M)}(\mr{b}_{1}^{m}(\phi)g)}
\ar[dddddd]^{\mr{b}_{3}(\phi)y}
\\
&&&&
\\
&
\mc{A}_{\mr{a}(M)}((\mr{T}_{1}^{o}(M)\circ\mr{b}_{1}^{o}(\phi))z)
\ar[rr]^{\mr{T}_{3}(M)(\mr{b}_{1}^{o}(\phi)z)}
\ar[dd]^{\mr{a}_{3}(\phi)(\mr{T}_{1}^{o}(N)z)}
& &
\mc{A}_{\mr{b}(M)}(\mr{b}_{1}^{o}(\phi)z)
\ar[dd]^{\mr{b}_{3}(\phi)z}
&
\\
&&&&
\\
&
\mc{A}_{\mr{a}(N)}(\mr{T}_{1}^{o}(N)z)
\ar[rr]^{\mr{T}_{3}(N)z}
& &
\mc{A}_{\mr{b}(N)}(z)
&
\\
&&&&
\\
\mc{A}_{\mr{a}(N)}(\mr{T}_{1}^{o}(N)y)
\ar[rrrr]_{\mr{T}_{3}(N)y}
\ar[ruu]_{\uptau_{\mr{a}(N)}(\mr{T}_{1}^{m}(N)g)}
& & & &
\mc{A}_{\mr{b}(N)}(y)
\ar[luu]^{\uptau_{\mr{b}(N)}(g)}
}
\end{equation*}
and the following is a commutative diagram in $\mr{ptls}$ 
\begin{equation*}
\xymatrix{
\mc{A}^{\ast}_{\mr{a}(M)}((\mr{T}_{1}^{o}(M)\circ\mr{b}_{1}^{o}(\phi))y)
& & & &
\mc{A}^{\ast}_{\mr{b}(M)}(\mr{b}_{1}^{o}(\phi)y)
\ar[llll]_{\mr{T}_{2}(M)(\mr{b}_{1}^{o}(\phi)y)}
\\
&&&&
\\
&
\mc{A}^{\ast}_{\mr{a}(M)}((\mr{T}_{1}^{o}(M)\circ\mr{b}_{1}^{o}(\phi))z)
\ar[luu]_{\uptau^{\dagger}_{\mr{a}(M)}((\mr{T}_{1}^{m}(M)\circ\mr{b}_{1}^{m}(\phi))g)}
& &
\mc{A}^{\ast}_{\mr{b}(M)}(\mr{b}_{1}^{o}(\phi)z)
\ar[ll]_{\mr{T}_{2}(M)(\mr{b}_{1}^{o}(\phi)z)}
\ar[ruu]^{\uptau^{\dagger}_{\mr{b}(M)}(\mr{b}_{1}^{m}(\phi)g)}
&
\\
&&&&
\\
&
\mc{A}^{\ast}_{\mr{a}(N)}(\mr{T}_{1}^{o}(N)z)
\ar[uu]_{\mr{a}_{2}(\phi)(\mr{T}_{1}^{o}(N)z)}
\ar[ldd]^{\uptau^{\dagger}_{\mr{a}(N)}(\mr{T}_{1}^{m}(N)g)}
& &
\mc{A}^{\ast}_{\mr{b}(N)}(z)
\ar[ll]_{\mr{T}_{2}(N)z}
\ar[uu]_{\mr{b}_{2}(\phi)z}
\ar[rdd]_{\uptau^{\dagger}_{\mr{b}(N)}(g)}
&
\\
&&&&
\\
\mc{A}^{\ast}_{\mr{a}(N)}(\mr{T}_{1}^{o}(N)y)
\ar[uuuuuu]_{\mr{a}_{2}(\phi)(\mr{T}_{1}^{o}(N)y)}
& & & &
\mc{A}^{\ast}_{\mr{b}(N)}(y)
\ar[llll]^{\mr{T}_{2}(N)y} 
\ar[uuuuuu]_{\mr{b}_{2}(\phi)y}
}
\end{equation*}
\end{lemma}
\begin{proof}
Since $\mr{T}$ is a natural transformation from $\mr{a}$ to $\mr{b}$, 
the following is a commutative diagram in $\mf{Chdv}$ for all 
$M,N\in\mf{D}$ and $\phi\in \mr{Mor}_{\mf{D}}(M,N)$
\begin{equation}
\label{01011237}
\xymatrix{
\mr{a}(M)\ar[dd]^{\mr{a}(\phi)}\ar[rr]^{\mr{T}(M)}&&\mr{b}(M)\ar[dd]^{\mr{b}(\phi)}
\\
&&
\\
\mr{a}(N)\ar[rr]_{\mr{T}(N)}&&\mr{b}(N).
}
\end{equation}
Next let us call square $c$ the central subsquare in the first diagram of the statement, 
and square $u$, $d$, $l$ and $r$ the up, down, left and right subsquare of the first 
diagram in the statement respectively; moreover let $\star$ denote the commutativity 
of the diagram in Rmk. \ref{09090657}.
The squares $u$ and $d$ are commutative since $\star$ and the arrows 
$\mr{a}(M)\xrightarrow{\mr{T}(M)}\mr{b}(M)$ and $\mr{a}(N)\xrightarrow{\mr{T}(N)}\mr{b}(N)$. 
The square $l$ is commutative since $\star$, the arrow $\mr{a}(M)\xrightarrow{\mr{a}(\phi)}\mr{a}(N)$ 
and the $1^{th}$ component of the commutativity of the diagram \eqref{01011237}, i.e.
\begin{equation}
\label{11210858}
\mr{T}_{1}(M)\circ\mr{b}_{1}(\phi)=\mr{a}_{1}(\phi)\circ\mr{T}_{1}(N).
\end{equation}
The square $r$ is commutative since $\star$ and the arrow $\mr{b}(M)\xrightarrow{\mr{b}(\phi)}\mr{b}(N)$,
finally the square $c$ is commutative since 
the $3^{th}$ component of the commutativity of the diagram
in \eqref{01011237}, i.e. 
$\mr{b}_{3}(\phi)\circ(\mr{T}_{3}(M)\ast\un_{\mr{b}_{1}(\phi)}) 
=
\mr{T}_{3}(N)\circ(\mr{a}_{3}(\phi)\ast\un_{\mr{T}_{1}(N)})$.
By denoting with $\natural$ the commutativity of the diagram in Rmk. \ref{09081318}
we obtain the commutativity of the second diagram in the statement just by following the same 
line of reasoning used for the first one, by replacing $\star$ by $\natural$ and 
by taking the $2^{th}$ component of the commutativity of the diagram in \eqref{01011237}
instead of the $3^{th}$ one.
\end{proof}
In what follows see also Rmk.\ref{03141425} about an interpretation of $\mr{b}_{1}^{m}(\phi)g$ with the obvious replacements.
\begin{proposition}
[Invariance under reference frame transformations]
\label{03151821}
Under the hypothesis of Lemma \ref{01011521} assume in addition that $\phi$ is invertible thus, 
for every $\omega\in\mf{P}_{\mc{A}_{\mr{b}(N)}(z)}$ and $a\in\mc{A}_{\mr{b}(N)}(y)_{ob}$ we have 
\begin{equation*}
\mf{f}_{(\omega,a)}^{\mr{b},N,y,z}(g)=
\mf{f}_{(\mr{b}_{3}^{\dagger}(\phi)(z)\omega,\mr{b}_{3}(\phi^{-1})(\mr{b}_{1}^{o}(\phi)y)a)}^{\mr{b},M,\mr{b}_{1}^{o}(\phi)y,\mr{b}_{1}^{o}(\phi)z)}
(\mr{b}_{1}^{m}(\phi)g).
\end{equation*}
In particular if $\mf{E}^{\mr{b}}$ is the standard experimental setting associated with $\mr{b}$, then 
\begin{equation*}
\begin{aligned}
\mc{Tr}(\mr{b};N;\un_{\mr{b}\rightarrow\mr{b}},\un_{N\rightarrow N},y,z\,\vert\,\mf{E}^{\mr{b}})
&=
\mr{b}_{1}^{m}(\phi)^{\dagger}
\mc{Tr}(\mr{b};M;\un_{\mr{b}\rightarrow\mr{b}},\phi^{-1},\mr{b}_{1}^{o}(\phi)y,\mr{b}_{1}^{o}(\phi)z\,\vert\,\mf{E}^{\mr{b}})
\\
&=
\mc{Tr}(\mr{b};M;\phi,\un_{\mr{b}\rightarrow\mr{b}},y,z\,\vert\,\mf{E}^{\mr{b}}).
\end{aligned}
\end{equation*}
\end{proposition}
\begin{proof}
$\mr{b}(\phi^{-1})=\mr{b}(\phi)^{-1}$ in the category $\mf{Chdv}$ since $\mr{b}$ is a functor, 
thus, $(\mr{b}_{3}(\phi)y)\circ(\mr{b}_{3}(\phi^{-1})(\mr{b}_{1}^{o}(\phi)y))=\un_{\mc{A}_{\mr{b}(N)}(y)}$ then the first equality in
the statement follows since the right subdiagram of the first diagram in Lemma \ref{01011521}. The second equality follows
since the first one.
\end{proof}
\begin{remark}
The above result states that given the reference frame transformation $\mr{b}_{3}(\phi^{-1})(\mr{b}_{1}^{o}(\phi)y)$
we can rearrange in the transformed reference frame the same experiments performed in the untransformed reference frame.
\end{remark}
In what follows see also Rmk.\ref{03141425} about an interpretation of $\mr{b}_{1}^{m}(\phi^{-1})h$.
\begin{proposition}
\label{11171604}
Under the hypothesis of Lemma \ref{01011521} we have that
\begin{enumerate}
\item
if $\phi$ is invertible, then for all 
$h\in\mr{G}_{\mr{b}(M)}(\mr{b}_{1}^{o}(\phi^{-1})y,\mr{b}_{1}^{o}(\phi^{-1})z)$
\begin{equation*}
\uptau_{\mr{b}(N)}(\mr{b}_{1}^{m}(\phi^{-1})h) 
\circ
(\mr{b}_{3}(\phi)y)
=
(\mr{b}_{3}(\phi)z)
\circ
\uptau_{\mr{b}(M)}(h);
\end{equation*}
\label{11171604st1}
\item
if 
$\mr{G}_{\mr{b}(N)}(y,z)=\mr{G}_{\mr{a}(N)}(\mr{T}_{1}^{o}(N)y,\mr{T}_{1}^{o}(N)z)$
and 
$\mr{T}_{1}^{m}(N)\up\mr{G}_{\mr{b}(N)}(y,z)=\un_{\mr{G}_{\mr{b}(N)}(y,z)}$,
then
\begin{equation*}
\uptau_{\mr{b}(N)}(g)
\circ
(\mr{T}_{3}(N)y)
=
(\mr{T}_{3}(N)z)
\circ
\uptau_{\mr{a}(N)}(g).
\end{equation*}
\label{11171604st2}
\end{enumerate}
\end{proposition}
\begin{proof}
$\mr{b}(\phi^{-1})=\mr{b}(\phi)^{-1}$ in the category $\mf{Chdv}$ since $\mr{b}$ is a functor, 
in particular $\mr{b}_{1}^{o}(\phi^{-1})=\mr{b}_{1}^{o}(\phi)^{-1}$ and $\mr{b}_{1}^{m}(\phi^{-1})=\mr{b}_{1}^{m}(\phi)^{-1}$. 
Then sts. \eqref{11171604st1} and \eqref{11171604st2} 
follow since the right and the bottom squares
of the first commutative diagram in Lemma \ref{01011521}
respectively. 
\end{proof}
Let us comment the above result.
In \cite{fv} given a theory namely a functor say $\mc{B}$
from the category of globally hyperbolic spacetimes $\mr{Loc}$ 
to $\mr{Phys}$ one obtains that: 
\begin{enumerate}
\item
The relative Cauchy evolution $rce_{M}^{\mc{B}}$ for any object $M$ of $\mr{Loc}$
is a geometric object:
\emph{It depends uniquely by the morphism map of $\mc{B}$},
being a composition of four its factorizations see \cite[p.16]{fv}.
\item
$rce^{\mc{B}}$ is equivariant under action of the morphisms of $\mr{Loc}$ implemented by $\mc{B}$
\cite[Prp. 3.7(eq. 3.6)]{fv}.
\item
Given a theory $\mc{A}$ and a suitable natural transformation $\zeta$ from $\mc{A}$ to $\mc{B}$,
$rce^{\mc{A}}$ and $rce^{\mc{B}}$ are $\zeta-$related 
\cite[Prp. 3.8]{fv}.
\end{enumerate}
Instead in our approach the concept of dynamics is structured in a wholly general sense
in the object set of $\mf{dp}$, hence in that of the target category $\mf{Chdv}$,
rather than constructed via the morphism set of the source category.
Thus for a fixed but general functor $\mr{b}$ from a category $\mf{D}$ to the category $\mf{Chdv}$
one obtains that:
\begin{enumerate}
\item 
The dynamics $\uptau_{\mr{b}(M)}$ for any object $M$ of $\mf{D}$ is \emph{not} a geometric object.
$\uptau_{\mr{b}(M)}$ does \textbf{not} present any other relationship with $\mr{b}_{3}$
other than the symmetries appearing in the first diagram in Lemma \ref{01011521}.
In particular $\uptau_{\mr{b}(M)}$ does \textbf{not} factorizes through $\mr{b}_{3}$.
\label{04251333a}
\item
$\uptau_{\mr{b}(\cdot)}$ is equivariant under action of the morphisms of $\mf{D}$
implemented by $\mr{b}_{3}$ as shown in the right square of the first diagram in Lemma \ref{01011521}.
As a particular result Prp. \ref{11171604}\eqref{11171604st1} follows 
which can be considered a generalization of \cite[Prp. 3.7 (eq. 3.6)]{fv}.
\label{04251333b}
\item
Given a functor $\mr{a}$ from $\mf{D}$ to $\mf{Chdv}$
and a natural transformation $\mr{T}$ from $\mr{a}$ to $\mr{b}$,
$\uptau_{\mr{a}(\cdot)}$ and $\uptau_{\mr{b}(\cdot)}$ 
are $\mr{T}-$related as shown in the botton square of the first diagram in Lemma \ref{01011521}.
As a particular result Prp. \ref{11171604}\eqref{11171604st2} follows
which can be considered a generalization of \cite[Prp. 3.8]{fv}.
\label{04251333c}
\item
But the fundamental point is that Prp. \ref{11171604}\eqref{11171604st1} 
and Prp. \ref{11171604}\eqref{11171604st2} are not uncorrelated symmetries, 
rather they emerge as very specific outcomes of the unique general \textbf{equiformity principle of $\mr{T}$}
Thm. \ref{10081910}\eqref{10081910st5}.
This principle follows by the \emph{entire first diagram in Lemma \ref{01011521}}
which has not correspondence in \cite{fv}.
\label{04251333d}
\end{enumerate}
In the next definition we introduce the component required 
in order to state Thm. \ref{10081910}.
\begin{definition}
\label{10081523}
Let $\mf{D}$ be a category, $\mr{a},\mr{b}\in\mr{Fct}(\mf{D},\mf{Chdv})$,
$\mf{E}=(\mc{S},\mr{R})\in Exp(\mr{b})$ and 
$\mr{T}\in \mr{Mor}_{\mr{Fct}(\mf{D},\mf{Chdv})}(\mr{a},\mr{b})$.
Let 
\begin{equation*}
\Upgamma(\mf{E},\mr{T})
\coloneqq
\{s\in\prod_{M\in\mf{D}}
\mr{Mor}_{\mr{set}}(\mr{T}_{1}^{o}(M)(\mr{Obj}(\mr{R}_{M})),\mr{Obj}(\mr{R}_{M}))
\,\vert\,
(\ref{10111706},\ref{10111708})
\},
\end{equation*}
where
\begin{enumerate}
\item
$s_{M}\circ\mr{T}_{1}^{o}(M)=\un_{c(s_{M})}$
for all $M\in\mf{D}$;
\label{10111706}
\item
$\mr{b}_{1}^{o}(\phi)\circ s_{N}
=
s_{M}\circ\mr{a}_{1}^{o}(\phi)$,
for all $M,N\in\mf{D}$ and $\phi\in \mr{Mor}_{\mf{D}}(M,N)$.
\label{10111708}
\end{enumerate}
Next let $s\in\Upgamma(\mf{E},\mr{T})$ define
$\mr{T}[\mr{R},s]$ to be the map on $\mr{Obj}(\mf{D})$ mapping any $M\in\mf{D}$ to the 
unique subcategory $\mr{T}[\mr{R},s]_{M}$ of $\mr{G}_{M}^{\mr{a}}$ such that 
\begin{equation*}
\begin{aligned}
\mr{Obj}(\mr{T}[\mr{R},s]_{M})&=\mr{T}_{1}^{o}(M)(\mr{Obj}(\mr{R}_{M})),
\\
\mr{Mor}_{\mr{T}[\mr{R},s]_{M}}(u,t)&=\mr{T}_{1}^{m}(M)(\mr{Mor}_{\mr{R}_{M}}(s_{M}(u),s_{M}(t))),
\\
\forall u,t&\in \mr{Obj}(\mr{T}[\mr{R},s]_{M}).
\end{aligned}
\end{equation*}
Moreover define 
\begin{equation*}
\mr{T}[\mc{S},s]
\in
\prod_{M\in\mf{D}}
\prod_{t\in\mr{T}[\mr{R},s]_{M}}
\mc{P}(\mf{P}_{M}^{\mr{a}}(t)),
\end{equation*}
such that for all $M\in\mf{D}$ and $t\in\mr{T}[\mr{R},s]_{M}$
\begin{equation*}
\mr{T}[\mc{S},s]_{M}(t)
\coloneqq
\mr{T}_{3}^{\dagger}(M)(s_{M}(t))(\mc{S}_{M}(s_{M}(t))).
\end{equation*}
Finally set 
$\mr{T}[\mf{E},s]\coloneqq(\mr{T}[\mc{S},s],\mr{T}[\mr{R},s])$.
\end{definition}
\begin{remark}
Assume the notation of Def. \ref{10081523}, 
since \eqref{11210858} and $\mf{E}\in Exp(\mr{b})$ we deduce that
$\mr{a}_{1}^{o}(\phi)\mr{T}_{1}^{o}(N)(\mr{Obj}(\mr{R}_{N}))
\subseteq
\mr{T}_{1}^{o}(M)(\mr{Obj}(\mr{R}_{M}))$
rendering Def. \ref{10081523}\eqref{10111708} consistent.
\end{remark}
\begin{remark}
\label{10201649}
Assume the notation of Def. \ref{10081523}, 
and that $\mr{T}_{1}^{o}(M)$ is injective for all $M\in\mf{D}$,
then since \eqref{11210858} we have that
$M\mapsto\mr{T}_{1}^{o}(M)^{-1}\up\mr{T}_{1}^{o}(M)(\mr{Obj}(\mr{R}_{M}))\in\Upgamma(\mf{E},\mr{T})$,
called the projection associated with $\mf{E}$ and $\mr{T}$. 
\end{remark}
The next is the first main result of this paper.
For any connector $\mr{T}$, namely a natural transformation 
between two species contextualized on the same category,
any experimental setting $\mf{Q}$ of the target species of $\mr{T}$ 
and any $s\in\Upgamma(\mf{Q},\mr{T})$
we prove that the above defined $\mr{T}[\mf{Q},s]$
is an experimental setting of the source species of $\mr{T}$ 
and that $\mr{T}$ is a link from $\mf{Q}$ to $\mr{T}[\mf{Q},s]$,
consequently with $\mr{T}$ remains associated a collection of 
equiformity principles.
\begin{theorem}
[\textbf{The Equiformity Principle for a Natural Transformation $1$}]
\label{10081910}
Let $\mf{D}$ be a category, $\mr{a},\mr{b}\in\mr{Fct}(\mf{D},\mf{Chdv})$,
$\mf{E}\in Exp(\mr{b})$ 
and 
$\mr{T}\in \mr{Mor}_{\mr{Fct}(\mf{D},\mf{Chdv})}(\mr{a},\mr{b})$
such that $\Upgamma(\mf{E},\mr{T})\neq\emptyset$\footnote{for instance 
whenever $\mr{T}_{1}^{o}(M)$ is injective for all $M\in\mf{D}$
see Rmk. \ref{10201649}.}.
Thus for all $s\in\Upgamma(\mf{E},\mr{T})$
\begin{enumerate}
\item
$\mr{T}[\mf{E},s]\in Exp(\mr{a})$;
\label{10081910st1}
\item
$\mr{T}$ is a link from $\mf{E}$ to $\mr{T}[\mf{E},s]$;
\label{10081910st2}
\item
if either $\mr{T}_{3}^{\dagger}=\mr{T}_{2}$ or $\mr{a}$ factorizes through $\ps{\Uppsi}$
then $\mr{T}[\mf{E},s]$ is complete;
\label{10081910st3}
\item
if $\mf{E}\in Exp^{\ast}(\mr{b})$, then $\mr{T}[\mf{E},s]\in Exp^{\ast}(\mr{a})$;
\label{10081910st4}
\item
The Equiformity Principle holds true for the natural transformation $\mr{T}$:
The statements of Prp. \ref{01162038}, Prp. \ref{10201503} and Prp. \ref{01162038pr}
hold true by replacing 
$\mf{E}^{\mr{a}}$ by $\mr{T}[\mf{E},s]$ 
and $\mf{E}^{\mr{b}}$ by $\mf{E}$.
\label{10081910st5}
\end{enumerate}
\end{theorem}
\begin{proof}
In what follows whenever we mention to show items of Def. \ref{01161844} 
and Def. \ref{01161922} 
we mean to prove the corresponding statements for $\mr{T}[\mf{E},s]$ and $\mr{T}$
respectively.
Def. \ref{01161844}\eqref{01161844st3a} follows since \eqref{11210858}
and Def. \ref{10081523}\eqref{10111708}.
By the external square in the first diagram in Lemma \ref{01011521} and
since Lemma \ref{08080010}\eqref{08080010st2} we deduce for all 
$M,N\in\mf{D}$, $\phi\in \mr{Mor}_{\mf{D}}(M,N)$ and $y\in\mr{G}_{N}^{\mr{b}}$
that
\begin{equation*}
\mr{a}_{3}^{\dagger}(\phi)(\mr{T}_{1}^{o}(N)y)\circ\mr{T}_{3}^{\dagger}(N)(y)
=
\mr{T}_{3}^{\dagger}(M)(\mr{b}_{1}^{o}(\phi)y)\circ\mr{b}_{3}^{\dagger}(\phi)(y).
\end{equation*}
Hence by letting $\mf{E}=(\mc{S},\mr{R})$ 
and using Def. \ref{10081523}\eqref{10111706}
we have for all $t\in\mr{T}[\mr{R},s]_{N}$
\begin{equation*}
\mr{a}_{3}^{\dagger}(\phi)(t)\circ\mr{T}_{3}^{\dagger}(N)(s_{N}(t))
=
\mr{T}_{3}^{\dagger}(M)(\mr{b}_{1}^{o}(\phi)s_{N}(t))\circ\mr{b}_{3}^{\dagger}(\phi)(s_{N}(t)),
\end{equation*}
therefore
\begin{equation*}
\begin{aligned}
\mr{a}_{3}^{\dagger}(\phi)(t)
\mr{T}[\mc{S},s]_{N}(t)
&=
(\mr{a}_{3}^{\dagger}(\phi)(t)\circ\mr{T}_{3}^{\dagger}(N)(s_{N}(t)))
\mc{S}_{N}(s_{N}(t))
\\
&=
(\mr{T}_{3}^{\dagger}(M)(\mr{b}_{1}^{o}(\phi)s_{N}(t))\circ\mr{b}_{3}^{\dagger}(\phi)(s_{N}(t)))
\mc{S}_{N}(s_{N}(t))
\\
&\subseteq
\mr{T}_{3}^{\dagger}(M)(\mr{b}_{1}^{o}(\phi)s_{N}(t))
\mc{S}_{M}(\mr{b}_{1}^{o}(\phi)s_{N}(t))
\\
&=
\mr{T}_{3}^{\dagger}(M)(s_{M}(\mr{a}_{1}^{o}(\phi)t))
\mc{S}_{M}(s_{M}(\mr{a}_{1}^{o}(\phi)t))
=
\mr{T}[\mc{S},s]_{M}(\mr{a}_{1}^{o}(\phi)t),
\end{aligned}
\end{equation*}
which proves Def. \ref{01161844}\eqref{01161844st3b}.
Since Rmk. \ref{09090657} and Lemma \ref{08080010}\eqref{08080010st2}
we have for any $y,z\in\mr{G}_{M}^{\mr{b}}$ and $g\in\mr{G}_{M}^{\mr{b}}(y,z)$
that 
\begin{equation*}
\uptau_{\mr{a}(M)}^{\dagger}(\mr{T}_{1}^{m}(M)g)\circ\mr{T}_{3}^{\dagger}(M)(z)
=
\mr{T}_{3}^{\dagger}(M)(y)\circ\uptau_{\mr{b}(M)}^{\dagger}(g).
\end{equation*}
Let $u,t\in\mr{T}[\mr{R},s]_{M}$ and $h\in\mr{T}[\mr{R},s]_{M}(u,t)$,
hence there exists $g\in\mr{R}_{M}(s_{M}(u),s_{M}(t))$ such that 
$h=\mr{T}_{1}^{m}(M)g$, therefore
\begin{equation*}
\begin{aligned}
\uptau_{\mr{a}(M)}^{\dagger}(h)\mr{T}[\mc{S},s]_{M}(t)
&=
(\uptau_{\mr{a}(M)}^{\dagger}(\mr{T}_{1}^{m}(M)g)\circ\mr{T}_{3}^{\dagger}(M)(s_{M}(t)))
\mc{S}_{M}(s_{M}(t))
\\
&=
(\mr{T}_{3}^{\dagger}(M)(s_{M}(u))
\circ\uptau_{\mr{b}(M)}^{\dagger}(g))
\mc{S}_{M}(s_{M}(t))
\\
&\subseteq
\mr{T}_{3}^{\dagger}(M)(s_{M}(u))
\mc{S}_{M}(s_{M}(u))
=
\mr{T}[\mc{S},s]_{M}(u),
\end{aligned}
\end{equation*}
so Def. \ref{01161844}\eqref{01161844st4} and st.\eqref{10081910st1} follows.
Next Def. \ref{01161922} follows since construction of $\mr{T}[\mf{E}]$ and Lemma \ref{01011521},
and st.\eqref{10081910st2} follows.
If $\mr{a}$ factorizes through the functor $\ps{\Uppsi}$ then $\mr{a}_{2}=\mr{a}_{3}^{\dagger}$ 
and the second option in st.\eqref{10081910st3} follows,
while the first one can be proven in the same way of st.\eqref{10081910st1} 
by using instead the second diagram in Lemma \ref{01011521}.  
For all $M\in\mf{D}$, $t\in\mr{T}[\mr{R},s]_{M}$ and $a\in\mc{A}_{\mr{a}(M)}(t)$ we have by recalling \eqref{11181224}
\begin{equation*}
\updelta^{\dagger}(a)\circ\mr{T}_{3}^{\dagger}(M)(s_{M}(t)) 
=
\mr{T}_{3}^{\dagger}(M)(s_{M}(t)) 
\circ
\updelta^{\dagger}(\mr{T}_{3}(M)(s_{M}(t))a)
\end{equation*}
hence st.\eqref{10081910st4} follows since st.\eqref{10081910st1}.
St.\eqref{10081910st5} follows since st.\eqref{10081910st1}
and Prps. \ref{01162038}, Prp. \ref{10201503} and \ref{01162038pr}.
\end{proof}
If we call charge from $\mr{b}$ to $\mr{a}$ any map 
from a possibly subset of $Exp(\mr{b})$ to $Exp(\mr{a})$, we obtain that
Thm. \ref{10081910}\eqref{10081910st1}
establishes that $\mr{T}[\cdot,s]$ is a charge from $\mr{b}$ to $\mr{a}$.
In case of species of dynamical systems we can see that the subset extends to the whole $Exp(\mr{b})$. 
We can avoid the use of the map $\Upgamma$ in the following case.
\begin{theorem}
[The Equiformity Principle for a Natural Transformation $2$]
\label{01181342}
Let $\mf{D}$ be a category, $\mr{a},\mr{b}\in\mr{Fct}(\mf{D},\mf{Chdv})$, 
$\mf{E}\in Exp(\mr{b})$
and $\mr{T}\in \mr{Mor}_{\mr{Fct}(\mf{D},\mf{Chdv})}(\mr{a},\mr{b})$.
Thus 
\begin{enumerate}
\item
$\mr{T}$ is a link from $\mf{E}$ to $(\mf{P}^{\mr{a}},\mr{G}^{\mr{a}})$;
\label{01181342st1}
\item
The Equiformity Principle holds true for the natural transformation $\mr{T}$:
The statements of Prp. \ref{01162038}, Prp. \ref{10201503} and Prp. \ref{01162038pr}.
hold true by replacing 
$\mf{E}^{\mr{a}}$ by $(\mf{P}^{\mr{a}},\mr{G}^{\mr{a}})$ 
and $\mf{E}^{\mr{b}}$ by $\mf{E}$.
\label{01181342st2}
\end{enumerate}
\end{theorem}
\begin{proof}
St.\eqref{01181342st1} follows since Lemma \ref{01011521},
st.\eqref{01181342st2} follows since st.\eqref{01181342st1}, 
Prp. \ref{01162038}, Prp. \ref{10201503} and Prp. \ref{01162038pr}.
\end{proof}
Thm. \ref{10081910}\eqref{10081910st1} establishes 
that $\mr{T}[\cdot,s]$ is a charge from the target to the source of $\mr{T}$,
therefore it is natural to expect that under suitable conditions, 
the vertical composition of connectors 
is contravariantly represented as composition of charges.
This is proven in the following second main result.
\begin{corollary}
[\textbf{Charge composition of connectors}]
\label{10151636}
Let $\mf{D}$ be a category, $\mr{a},\mr{b},\mr{c}\in\mr{Fct}(\mf{D},\mf{Chdv})$,
$\mf{Q}\in Exp(\mr{c})$,
$\mr{T}\in \mr{Mor}_{\mr{Fct}(\mf{D},\mf{Chdv})}(\mr{b},\mr{c})$,
$\mr{S}\in \mr{Mor}_{\mr{Fct}(\mf{D},\mf{Chdv})}(\mr{a},\mr{b})$
and
$r\in\Upgamma(\mf{Q},\mr{T})$.
Thus 
\begin{equation*}
r
\circ
\Upgamma(\mr{T}[\mf{Q},r],\mr{S})
\subseteq
\Upgamma(\mf{Q},\mr{T}\circ\mr{S}),
\end{equation*}
where the composition is defined pointwise.
Moreover if $\mr{T}_{1}^{o}(M)$ is injective for all $M\in\mf{D}$,
and $r$ is the projection associated with $\mf{Q}$ and $\mr{T}$,
Rmk. \ref{10201649},
then for any $q\in\Upgamma(\mr{T}[\mf{Q},r],\mr{S})$ we obtain 
\begin{equation*}
(\mr{T}\circ\mr{S})[\mf{Q},r\circ q]
=
\mr{S}[\mr{T}[\mf{Q},r],q].
\end{equation*}
\end{corollary}
\begin{proof}
The inclusion in the statement is well-set since Thm. \ref{10081910}, 
then it follows by a simple application of the definitions involved.
The equality in the statement is well-set since the inclusion.
The part concerning the object sets in the equality
follows by direct application of the involved 
definitions, let us prove the part concerning the statistical ensemble spaces.
Let $\mf{Q}=(\mf{S},\mr{R})$, by using Rmk. \ref{11211156} 
we obtain for all $M\in\mf{D}$
\begin{equation*}
\begin{aligned}
(\mr{T}\circ\mr{S})[\mf{S},r\circ q]_{M}(t)
&=
(\mr{T}\circ\mr{S})_{3}^{\dagger}(M)((rq)_{M}(t))
\mf{S}_{M}((rq)_{M}(t))
\\
&=
(\mr{T}_{3}(M)((rq)_{M}(t))
\circ
\mr{S}_{3}(M)(\mr{T}_{1}^{o}(rq)_{M}(t)))^{\dagger}
\mf{S}_{M}((rq)_{M}(t))
\\
&=
(\mr{S}_{3}^{\dagger}(M)(q_{M}(t))
\circ
\mr{T}_{3}^{\dagger}(M)((rq)_{M}(t)))
\mf{S}_{M}((rq)_{M}(t)).
\end{aligned}
\end{equation*}
Next 
\begin{equation*}
\begin{aligned}
\mr{S}[\mr{T}[\mf{S},r],q]_{M}(t) 
&=
\mr{S}_{3}^{\dagger}(M)(q_{M}(t))
\mr{T}[\mf{S},r]_{M}(q_{M}(t))
\\
&=
(\mr{S}_{3}^{\dagger}(M)(q_{M}(t))
\circ
\mr{T}_{3}^{\dagger}(M)((rq)_{M}(t)))
\mf{S}_{M}((rq)_{M}(t)).
\end{aligned}
\end{equation*}
\end{proof}
Next we prepare for another important consequence of Thm. \ref{10081910}
stated in Cor. \ref{11200910} until which we let 
$\mf{G},\mf{D}$ be categories,
$\mr{x},\mr{y}\in\mr{Fct}(\mf{G},\mf{D})$,
$\mr{a},\mr{b}\in\mr{Fct}(\mf{D},\mf{Chdv})$,
$\mr{L}\in \mr{Mor}_{\mr{Fct}(\mf{G},\mf{D})}(\mr{x},\mr{y})$
and
$\mr{T}\in \mr{Mor}_{\mr{Fct}(\mf{D},\mf{Chdv})}(\mr{a},\mr{b})$.
\begin{definition}
\label{11200906}
For any $\mf{Q}=(\mf{S},\mr{R})\in Exp(\mr{b})$ and 
$s\in\Upgamma(\mf{Q},\mr{T})$ 
define
\begin{equation*}
\Upxi(\mr{T},\mr{L},s)
\coloneqq
\{
r\in\prod_{G\in\mf{G}}\mr{Mor}_{\mr{set}}
\bigl((\mr{a}_{1}^{o}(\mr{L}(G))
\circ
\mr{T}_{1}^{o}(\mr{y}(G)))
\mr{Obj}(\mr{R}_{\mr{y}(G)}),
\mr{Obj}(\mr{R}_{\mr{y}(G)})\bigr)
\,\vert\,
(\ref{1123105a},\ref{1123105b})
\text{ hold true}\},
\end{equation*}
where 
\begin{enumerate}
\item
$r_{G}\circ\mr{a}_{1}^{o}(\mr{L}(G))
\circ
(\mr{T}_{1}^{o}(\mr{y}(G))(\mr{Obj}(\mr{R}_{\mr{y}(G)}))
\hookrightarrow
\mr{G}_{\mr{y}(G)}^{\mr{a}})
=s_{\mr{y}(G)}$,
for all $G\in\mf{G}$;
\label{1123105a}
\item
$\mr{b}_{1}^{o}(\mr{L}(G))\circ r_{G}=s_{\mr{x}(G)}$,
for all $G\in\mf{G}$.
\label{1123105b}
\end{enumerate}
Moreover set 
$\mr{y}[\mf{Q}]\coloneqq(\mf{S}\circ\mr{y}^{o},\mr{R}\circ\mr{y}^{o})$.
\end{definition}
\begin{remark}
\label{11201224}
For all $G\in\mf{G}$ by applying the general 
definition of Godement product (horizontal composition)
we have
$(\mr{T}\ast\mr{L})(G)
=
\mr{T}(\mr{y}(G))\circ\mr{a}(\mr{L}(G))
=
\mr{b}(\mr{L}(G))\circ\mr{T}(\mr{x}(G))
$.
Hence by the definition of morphism composition in $\mf{Chdv}$
we obatin
$(\mr{T}\ast\mr{L})_{1}(G)
=
\mr{a}_{1}(\mr{L}(G))
\circ
\mr{T}_{1}(\mr{y}(G))
=
\mr{T}_{1}(\mr{x}(G))
\circ
\mr{b}_{1}(\mr{L}(G))$,
and 
$(\mr{T}\ast\mr{L})_{3}(G)
=
\mr{T}_{3}(\mr{y}(G))
\circ
(\mr{a}_{3}(\mr{L}(G))\ast\un_{\mr{T}_{1}(\mr{y}(G))})$,
where we recall \eqref{20061403}.
\end{remark}
\begin{lemma}
\label{11200907}
For any $\mf{Q}\in Exp(\mr{b})$ and 
$s\in\Upgamma(\mf{Q},\mr{T})$ we have that
\begin{enumerate}
\item
$\mr{y}[\mf{Q}]\in Exp(\mr{b}\circ\mr{y})$;
\label{11200907st1}
\item
$\Upxi(\mr{T},\mr{L},s)\subset\Upgamma(y[\mf{Q}],\mr{T}\ast\mr{L})$.
\label{11200907st2}
\end{enumerate}
\end{lemma}
\begin{proof}
St. \eqref{11200907st1} 
follows by the functoriality of $\mr{y}$ and the definition
of experimental settings,
st. \eqref{11200907st2} follows since Rmk. \ref{11201224}
and Def. \ref{11200906}\eqref{1123105a}.
\end{proof}
Often we call $\mr{y}[\mf{Q}]$ the pullback of $\mf{Q}$ through $\mr{y}$.
Next we define an order relation in the collection of the experimental settings.
\begin{definition}
\label{11201700}
Let $\leq_{\mr{a}}$ or simply $\leq$ be the relation on $Exp(\mr{a})$
defined such that for all 
$(\mf{S},\mr{R}),(\mf{S}',\mr{R}')\in Exp(\mr{a})$
we have
$(\mf{S},\mr{R})\leq(\mf{S}',\mr{R}')$ 
iff 
$\mr{R}_{M}$ is a subcategory of $\mr{R}_{M}'$
and
$\mf{S}_{M}(t)\subseteq\mf{S}_{M}'(t)$
for all $t\in\mr{R}_{M}$
and $M\in\mf{D}$.
\end{definition}
In addition to the contravariant representation of the vertical 
composition of connectors in terms of composition of charges, 
if $\mr{T}_{1}^{o}(M)$ is injective for all $M\in\mf{D}$,
we have also a representation of the horizontal composition
of a connector $\mr{T}$ with a natural transformation $\mr{L}$
in terms of what we call charge transfer. 
Said $\mr{x}$ and $\mr{y}$ the source and target of $\mr{L}$ respectively,
the charge transfer roughly asserts that
the charge $(\mr{T}\ast\mr{L})[\cdot,r]$ for $r\in\Upxi(\mr{T},\mr{L},s)$ 
maps the pullback through $\mr{y}$ of any experimental setting 
$\mf{Q}$ of the target species of $\mr{T}$ into an experimental setting 
which is included in the pullback through $\mr{x}$ of the experimental setting
obtained by mapping $\mf{Q}$ through the charge $\mr{T}[\cdot,s]$,
where $s$ is the projection associated with $\mf{Q}$ and $\mr{T}$.
The following is the third main result of this paper.
\begin{corollary}
[\textbf{Charge transfer}]
\label{11200910}
If $\mr{T}_{1}^{o}(M)$ is injective for all $M\in\mf{D}$,
then for any $\mf{Q}\in Exp(\mr{b})$
said $s$ the projection associated with $\mf{Q}$ and $\mr{T}$
Rmk. \ref{10201649},
we obtain for all $r\in\Upxi(\mr{T},\mr{L},s)$ that
\begin{equation*}
(\mr{T}\ast\mr{L})[\mr{y}[\mf{Q}],r]
\leq
\mr{x}[\mr{T}[\mf{Q},s]].
\end{equation*}
\end{corollary}
\begin{proof}
Let $\mr{K}$ denote 
$(\mr{T}\ast\mr{L})[\mr{R}\circ\mr{y}^{o},r]$,
$\mf{X}$ denote
$(\mr{T}\ast\mr{L})[\mf{S}\circ\mr{y}^{o},r]$,
and let $\mf{Q}=(\mf{S},\mr{R})$. 
Thus
\begin{equation}
\begin{aligned}
\label{11231714}
(\mr{T}\ast\mr{L})[\mr{y}[\mf{Q}],r]
&=
(\mf{X},\mr{K}),
\\
\mr{x}[\mr{T}[\mf{Q},s]]
&=
(\mr{T}[\mf{S},s]\circ\mr{x}^{o},
\mr{T}[\mr{R},s]\circ\mr{x}^{o}).
\end{aligned}
\end{equation}
Moreover for any $G\in\mf{G}$
\begin{equation*}
\begin{aligned}
\mr{Obj}(\mr{K}_{G})
&=
(\mr{T}\ast\mr{L})_{1}^{o}(G)
(\mr{Obj}(\mr{R}_{\mr{y}(G)}))
\\
&=
(\mr{T}_{1}^{o}(\mr{x}(G))
\circ\mr{b}_{1}^{o}(\mr{L}(G)))
\mr{Obj}(\mr{R}_{\mr{y}(G)})
\\
&\subseteq
\mr{T}_{1}^{o}(\mr{x}(G))
\mr{Obj}(\mr{R}_{\mr{x}(G)})
\\
&=
\mr{Obj}((\mr{T}[\mr{R},s]\circ\mr{x}^{o})(G)),
\end{aligned}
\end{equation*}
where in the second equality we use Rmk. \ref{11201224},
the inclusion arises by $\mf{Q}\in Exp(\mr{b})$ 
and $\mr{L}(G)\in \mr{Mor}_{\mf{D}}(\mr{x}(G),\mr{y}(G))$.
Let $u,t\in \mr{Obj}(\mr{K}_{G})$ thus by Rmk. \ref{11201224}
\begin{equation*}
\begin{aligned}
\mr{Mor}_{\mr{K}_{G}}(u,t)
&=
(\mr{T}\ast\mr{L})_{1}^{m}(G)
\mr{Mor}_{\mr{R}_{\mr{y}(G)}}(r_{G}u,r_{G}t)
\\
&=
(\mr{T}_{1}^{m}(\mr{x}(G))
\circ\mr{b}_{1}^{m}(\mr{L}(G)))
\mr{Mor}_{\mr{R}_{\mr{y}(G)}}(r_{G}u,r_{G}t)
\\
&\subseteq
\mr{T}_{1}^{m}(\mr{x}(G))
\mr{Mor}_{\mr{R}_{\mr{x}(G)}}
((\mr{b}_{1}^{o}(\mr{L}(G))\circ r_{G})u,
(\mr{b}_{1}^{o}(\mr{L}(G))\circ r_{G})t)
\\
&=
\mr{T}_{1}^{m}(\mr{x}(G))
\mr{Mor}_{\mr{R}_{\mr{x}(G)}}
(s_{\mr{x}(G)}u,s_{\mr{x}(G)}t)
\\
&=
\mr{Mor}_{(\mr{T}[\mr{R},s]\circ\mr{x}^{o})(G)}(u,t),
\end{aligned}
\end{equation*}
where in the inclusion we used $\mf{Q}\in Exp(\mr{b})$
and in the subsequent equality the hypothesis on $r$.
Thus 
\begin{equation}
\label{11201805} 
\mr{K}_{G} 
\text{ is a subcategory of } 
(\mr{T}[\mr{R},s]\circ\mr{x}^{o})(G).
\end{equation}
Next 
\begin{equation*}
\begin{aligned}
\mf{X}_{G}(t)
&=
(\mr{T}\ast\mr{L})_{3}^{\dagger}(G)(r_{G}t)\mf{S}_{\mr{y}(G)}(r_{G}t)
\\
&=
\bigl((\mr{T}\ast\mr{L})_{3}(G)(r_{G}t)\bigr)^{\dagger}\mf{S}_{\mr{y}(G)}(r_{G}t)
\\
&=
\bigl(
\mr{T}_{3}(\mr{y}(G))(r_{G}t)
\circ
\mr{a}_{3}(\mr{L}(G))(\mr{T}_{1}^{o}(\mr{y}(G))r_{G}t)
\bigr)^{\dagger}
\mf{S}_{\mr{y}(G)}(r_{G}t)
\\
&=
\bigl(
\mr{a}_{3}^{\dagger}(\mr{L}(G))(\mr{T}_{1}^{o}(\mr{y}(G))r_{G}t)
\circ
\mr{T}_{3}^{\dagger}(\mr{y}(G))(r_{G}t)
\bigr)
\mf{S}_{\mr{y}(G)}(r_{G}t)
\\
&\subseteq
\mr{a}_{3}^{\dagger}(\mr{L}(G))(\mr{T}_{1}^{o}(\mr{y}(G))r_{G}(t))
\mr{T}[\mf{S},s]_{\mr{y}(G)}
(\mr{T}_{1}^{o}(\mr{y}(G))r_{G}(t))
\\
&\subseteq
\mr{T}[\mf{S},s]_{\mr{x}(G)}
((\mr{a}_{1}^{o}(\mr{L}(G))
\circ
\mr{T}_{1}^{o}(\mr{y}(G)))
r_{G}(t))
\\
&=
\mr{T}[\mf{S},s]_{\mr{x}(G)}
((\mr{T}_{1}^{o}(\mr{x}(G))
\circ
\mr{b}_{1}^{o}(\mr{L}(G)))
r_{G}(t))
\\
&=
\mr{T}[\mf{S},s]_{\mr{x}(G)}
((\mr{T}_{1}^{o}(\mr{x}(G))
\circ
s_{\mr{x}(G)})t)
\\
&=
\mr{T}[\mf{S},s]_{\mr{x}(G)}(t)
\\
&=
(\mr{T}[\mf{S},s]\circ\mr{x}^{o})(G)(t),
\end{aligned}
\end{equation*}
where we used Rmk. \ref{11201224} in the $3$th and $5$th equalities,
the $1$th inclusion arises by 
Thm. \ref{10081910}\eqref{10081910st2} and 
Def. \ref{01161922}\eqref{01161922st2},
the $2$th inclusion by Thm. \ref{10081910}\eqref{10081910st1} and 
Def. \ref{01161844}\eqref{01161844st3b},
the $6$th equality by hypothesis on $r$
and the $7$th equality by hypothesis on $s$.
Hence
\begin{equation}
\label{11212016}
\mf{X}_{G}(t)\subseteq(\mr{T}[\mf{S},s]\circ\mr{x}^{o})(G)(t).
\end{equation}
Finally \eqref{11231714},
\eqref{11201805} and \eqref{11212016} yield the statement.
\end{proof}
\section{The $2-$category $2-\mf{dp}$}
\label{11301655}
In section \ref{11301655a} we introduce the general 
structures utilized in section \ref{11301655b}
in order to construct the $2-$category $2-\mf{dp}$
and the fibered category of connectors.
Finally in section \ref{11301655c} we physically intepret the
introduced structures. 
\subsection{General definitions}
\label{11301655a}
For any category $\mr{C}$ we will define the $2-$category $2-\mr{C}$ 
Def. \ref{01181755} as 
a subcategory of $\mr{Cat}$ with the same object set and 
any morphism category $2-\mr{C}(\mr{A},\mr{B})$ 
roughly formed by factorizable functors through $\mr{C}$
as objects, with morphisms  
horizontal composition of natural transformations 
between the corresponding factors
\textbf{Def. \ref{09011927}}.
In section \ref{11301655b} we shall apply the construction to the case $\mr{C}=\mf{dp}$.
A standard construction applied to our category 
$2-\mr{C}(\mr{A},\mr{B})$ allows to provide
$\mr{Mor}_{2-\mr{C}(\mr{A},\mr{B})}(f,g)$
with the structure of a category 
for $f,g\in 2-\mr{C}(\mr{A},\mr{B})$
Prp. \ref{09021341}.
These categories are employed to introduce the fibered category 
of connectors \textbf{Def. \ref{10031410}}.
\begin{definition}
\label{09011927}
Let $\mr{A},\mr{B},\mr{C}$ be categories, 
define $2-\mr{C}(\mr{A},\mr{B})$ the subcategory of $\mr{Fct}(\mr{A},\mr{B})$ 
such that 
\begin{equation*}
\mr{Obj}(2-\mr{C}(\mr{A},\mr{B}))
\coloneqq
\{\Upphi\circ\Uppsi\mid\Uppsi\in\mr{Fct}(\mr{A},\mr{C}),\Upphi\in\mr{Fct}(\mr{C},\mr{B})\}, 
\end{equation*}
where $\circ$ is the standard composition of functors, while for any $f,g\in 2-\mr{C}(\mr{A},\mr{B})$, 
\begin{multline*}
\mr{Mor}_{2-\mr{C}(\mr{A},\mr{B})}(f,g)\coloneqq
\{\upbeta\ast\upalpha\mid\upbeta\in \mr{Mor}_{\mr{Fct}(\mr{C},\mr{B})}(\Upphi_{1},\Upphi_{2}),
\upalpha\in \mr{Mor}_{\mr{Fct}(\mr{A},\mr{C})}(\Uppsi_{1},\Uppsi_{2});\\
\Uppsi_{i}\in\mr{Fct}(\mr{A},\mr{C}),\Upphi_{i}\in\mr{Fct}(\mr{C},\mr{B}), i\in\{1,2\};
\Upphi_{1}\circ\Uppsi_{1}=f, \Upphi_{2}\circ\Uppsi_{2}=g
\}.
\end{multline*}
\end{definition}
The definition of $\mr{Mor}_{2-\mr{C}}(\mr{A},\mr{B})$ is well-done indeed the morphism composition law in 
$\mr{Mor}_{2-\mr{C}}(\mr{A},\mr{B})$ inherited by the one in $\mr{Fct}(\mr{A},\mr{B})$ is an internal operation by using
\cite[Prp. $1.3.5$]{bor}. 
Since \cite[Def. 1.2.5]{ks} 
the set of morphisms of any category can be provided with the structure
of category, then in the special case of the category $2-\mr{C}(\mr{A},\mr{B})$
we have that
$\mr{Mor}_{2-\mr{C}(\mr{A},\mr{B})}(f,g)$ is a subcategory 
of the category 
$\mr{Mor}_{2-\mr{C}(\mr{A},\mr{B})}$
that does not provide 
$2-\mr{C}(\mr{A},\mr{B})$ the structure of $2-$category. 
However we will prove in Prp. \ref{09021339} 
the existence of a partial product between the 
morphism sets of these categories. 
Next \cite[Def. 1.2.5]{ks} applied to our category $2-\mr{C}(\mr{A},\mr{B})$
reads as follows
\begin{definition}
\label{09011944}
Let $\mr{A},\mr{B},\mr{C}$ be categories, $f,g\in 2-\mr{C}(\mr{A},\mr{B})$ 
and $\delta,\ep,\eta\in \mr{Mor}_{2-\mr{C}(\mr{A},\mr{B})}(f,g)$,
define
\begin{multline*}
(f,g)\lr{\delta}{\ep}
\coloneqq\{(u,v)\in
\mr{Mor}_{2-\mr{C}(\mr{A},\mr{B})}(f,f)
\times
\mr{Mor}_{2-\mr{C}(\mr{A},\mr{B})}(g,g)
\,\vert\,
\ep\circ u=v\circ\delta\},
\end{multline*}
moreover define
\begin{equation*}
(\circ):(f,g)\lr{\ep}{\eta}\times(f,g)\lr{\delta}{\ep}\to(f,g)\lr{\delta}{\eta}
\end{equation*}
such that 
$(w,s)\circ(u,v)\coloneqq(w\circ u,s\circ v)$.
\end{definition}
Easy to prove is then the following
\begin{proposition}
\label{09021341}
Let $\mr{A},\mr{B},\mr{C}$ be categories and $f,g\in 2-\mr{C}(\mr{A},\mr{B})$. 
There exists a unique category whose object set is $\mr{Mor}_{2-\mr{C}(\mr{A},\mr{B})}(f,g)$, 
while $(f,g)\lr{\delta}{\ep}$ is the class
of the morphisms from $\delta$ to $\ep$, for all $\delta,\ep\in \mr{Mor}_{2-\mr{C}(\mr{A},\mr{B})}(f,g)$, 
and the morphism composition is the one defined in Def. \ref{09011944}.
\end{proposition}
\begin{definition}
\label{09021122}
Let $\mr{A},\mr{B},\mr{C}$ be categories, 
$f,g,h,k\in 2-\mr{C}(\mr{A},\mr{B})$, $\alpha,\alpha'\in \mr{Mor}_{2-\mr{C}(\mr{A},\mr{B})}(f,g)$
and $\beta,\beta'\in \mr{Mor}_{2-\mr{C}(\mr{A},\mr{B})}(h,k)$, define 
\begin{equation*}
D_{h,k;f,g}^{\beta,\beta';\alpha,\alpha'}
\coloneqq
\{((w,z);(u,v))\in(h,k)\lr{\beta}{\beta'}\times(f,g)\lr{\alpha}{\alpha'}\,\vert\, h\circ v=w\circ g\}.
\end{equation*}
Moreover define the map 
\begin{equation*}
\star:D_{h,k;f,g}^{\beta,\beta';\alpha,\alpha'}\ni((w,z),(u,v))\mapsto(h\circ u,z\circ g).
\end{equation*}
\end{definition}
\begin{proposition}
\label{09021339}
Let $\mr{A},\mr{B},\mr{C}$ be categories, 
$f,g,h,k\in 2-\mr{C}(\mr{A},\mr{B})$, $\alpha,\alpha'\in \mr{Mor}_{2-\mr{C}(\mr{A},\mr{B})}(f,g)$
and $\beta,\beta'\in \mr{Mor}_{2-\mr{C}(\mr{A},\mr{B})}(h,k)$, then
$\star:D_{h,k;f,g}^{\beta,\beta';\alpha,\alpha'}\to(h\circ f,k\circ g)\lr{\beta\ast\alpha}{\beta'\ast\alpha'}$.
\end{proposition}
\begin{proof}
Let $((w,z);(u,v))\in D_{h,k;f,g}^{\beta,\beta';\alpha,\alpha'}$ and $M\in\mr{A}$ then
\begin{equation*}
\begin{aligned}
(\beta'\ast\alpha')_{M}\circ h(u_{M})&=\beta'_{g(M)}\circ h(\alpha_{M}'\circ u_{M})\\
&=\beta'_{g(M)}\circ h(v_{M}\circ\alpha_{M})\\
&=\beta'_{g(M)}\circ w_{g(M)}\circ h(\alpha_{M})\\
&=z_{g(M)}\circ\beta_{g(M)}\circ h(\alpha_{M})=z_{g(M)}\circ(\beta\ast\alpha)_{M},
\end{aligned}
\end{equation*}
where the second equality comes since $(u,v)\in(f,g)\lr{\alpha}{\alpha'}$, 
the third one by $h(v_{M})=w_{g(M)}$, the fourth one by 
$(w,z)\in(h,k)\lr{\beta}{\beta'}$.
\end{proof}
\begin{definition}
\label{17270730}
Given any object $\mr{C}$ of $\mr{Cat}$, 
let $2-\mr{C}$ be the subcategory of $\mr{Cat}$, whose object set is $\mr{Obj}(\mr{Cat})$, and 
$\mr{Mor}_{2-\mr{C}}(\mr{A},\mr{B})\coloneqq 2-\mr{C}(\mr{A},\mr{B})$, for any $\mr{A},\mr{B}\in\mr{Cat}$.
\end{definition}
\begin{proposition}
\label{01181755}
$2-\mr{C}$ is a $2-$category such that 
$\mr{Mor}_{\mr{Mor}_{2-\mr{C}}(\mr{A},\mr{B})}(f,g)$
is a category for all
$f,g\in \mr{Mor}_{2-\mr{C}}(\mr{A},\mr{B})$
and $\mr{A},\mr{B}\in2-\mr{C}$.
\end{proposition}
\begin{proof}
The second morphism composition law $\ast$ in it is an 
internal law since \cite[Prp. $1.3.4$]{bor}, then the 
first sentence of the statement follows
since $\mr{Cat}$ is a $2-$category,
the second one follows since
Prp. \ref{09021341}. 
\end{proof}
\subsection{Fibered category of connectors}
\label{11301655b}
Here we apply the results of section \ref{11301655a}
to the case $\mr{C}=\mf{dp}$.
Since Cor. \ref{09100952} $\mf{dp}$ 
is an object of $\mr{Cat}$.
Hence $2-\mf{dp}$ is well-set according Def. \ref{17270730}, 
and it is a $2-$category since Prp. \ref{01181755},
thus we can set the following
\begin{definition}
[Fibered category of species]
\label{10031349}
Define $\mf{Sp}$ the fibered category over $2-\mf{dp}$ such that 
for all $\mf{D}\in 2-\mf{dp}$
\begin{equation*}
\mf{Sp}(\mf{D})=2-\mf{dp}(\mf{D},\mf{Chdv}),
\end{equation*}
moreover set
\begin{equation*}
\mf{Sp}_{\ast}\coloneqq\{(\mr{a},\mr{b})\in\mf{Sp}\times\mf{Sp}
\,\vert\, 
d(\mr{a})=d(\mr{b})
\}.
\end{equation*}
\end{definition}
Notice that since Thm. \ref{12312222} we have that 
$\ps{\Uppsi}\circ\mr{Fct}(\mf{D},\mf{dp})\subset
\mf{Sp}(\mf{D})$.
\begin{definition}
[Fibered category of connectors]
\label{10031410}
Let $\mf{Cnt}$ be the fibered category over $\mf{Sp}_{\ast}$ such that for all 
$(\mr{a},\mr{b})\in\mf{Sp}_{\ast}$ 
\begin{equation*}
\mf{Cnt}(\mr{a},\mr{b})\coloneqq
\mr{Mor}_{\mf{Sp}(d(\mr{a}))}(\mr{a},\mr{b}),
\end{equation*}
with the category structure described in Prp. \ref{09021341},
called the category of connectors from $\mf{s}(\mr{a})$ to $\mf{s}(\mr{b})$.
Define for all $\mr{a}\in\mf{Sp}$ the full subcategory 
$\mf{Cnt}_{\mr{a}}$ of $\mf{Cnt}$ such that 
$\mr{Obj}(\mf{Cnt}_{\mr{a}})=\{\mr{T}\in\mf{Cnt}\,\vert\,d(\mr{T})=\mr{a}\}$,
and let $\mf{Sct}$ be the fibered category over $\mf{Sp}$ such that for all $\mr{a}\in\mf{Sp}$
we have that 
$\mf{Sct}(\mr{a})=\mf{Cnt}(\mr{a},\mr{a})$,
called the category of sectors of $\mf{s}(\mr{a})$.
\end{definition}
\begin{definition}
[Experimental settings generated by $\mf{Cnt}$]
\label{10051613}
Let $\mf{ex}$ be the fibered set over $\mf{Cnt}$ such that for all $\mr{T}\in\mf{Cnt}$
\begin{equation*}
\mf{ex}(\mr{T})=\{\mr{T}[\mf{Q},t]\,\vert\,\mf{Q}\in Exp(c(\mr{T})),
t\in\Upgamma(\mf{Q},\mr{T})\},
\end{equation*}
called the collection of experimental settings generated by $\mf{s}(\mr{T})$,
well-set since $\mf{ex}(\mr{T})\subseteq Exp(d(\mr{T}))$ by Thm. \ref{10081910}.
\end{definition}
\begin{remark}
\label{10151948}
Since for any $\mr{a}\in\mf{Sp}$ and any $\mf{Q}=(\mf{S},\mr{R})\in Exp(\mr{a})$
the map on $\mr{Obj}(d(\mr{a}))$, 
assigning to any object $M$ the identity map 
on $\mr{R}_{M}$, belongs to $\Upgamma(\mf{Q},\un_{\mr{a}})$ we have that that
$\mf{ex}(\un_{\mr{a}})=Exp(\mr{a})$, i.e. the collection of experimental settings
generated by the identity morphism of any species is the maximal one, namely
equals the whole set of experimental settings of that species.
\end{remark}
Next we use Cor. \ref{10151636} 
in order to relate any morphism between connectors
to the experimental settings they generate.
\begin{proposition}
\label{10151140}
Let $(\mr{a},\mr{b})\in\mf{Sp}_{\ast}$, $\mr{S},\mr{T}\in\mf{Cnt}(\mr{a},\mr{b})$,
then
\begin{enumerate}
\item
$\mr{Mor}_{\mf{Cnt}(\mr{a},\mr{b})}(\mr{S},\mr{T})=
\{(u,v)\in\mf{Sct}(\mr{a})\times\mf{Sct}(\mr{b})
\,\vert\,
\mr{T}\circ u=v\circ\mr{S}\}$.
\label{10151140st1}
\item
Let $\mf{Q}\in Exp(\mr{b})$ and 
$(u,v)\in \mr{Mor}_{\mf{Cnt}(\mr{a},\mr{b})}(\mr{S},\mr{T})$.
If there exist 
$r\in\Upgamma(\mf{Q},\mr{T})$,
$q\in\Upgamma(\mr{T}[\mf{Q},r],u)$
and
$l\in\Upgamma(\mf{Q},v)$,
$n\in\Upgamma(v[\mf{Q},l],\mr{S})$
such that 
$r\circ q=l\circ n$,
then 
\begin{equation*}
u[\mr{T}[\mf{Q},r],q]
=
\mr{S}[v[\mf{Q},l],n].
\end{equation*}
\label{10151140st2}
\end{enumerate}
\end{proposition}
\begin{proof}
St. \eqref{10151140st1} 
is trivial, st. \eqref{10151140st2} follows 
by st. \eqref{10151140st1} and Cor. \ref{10151636}. 
\end{proof}
\begin{remark}
Notice that $u[\mr{T}[\mf{Q},r],q]\in\mf{ex}(\mr{S})$,
namely any morphism from the connector 
$\mr{S}$ to the connector $\mr{T}$, induces a map
from a subset of $\mf{ex}(\mr{T})$ to $\mf{ex}(\mr{S})$,
representing the direct empirical meaning of morphisms between connectors.
Even the existence of an isomorphism between connectors does not imply
the equality of the collections of experimental settings they generate.
This consideration makes unfeasible the attempt to generalize in our context the 
procedure of passing to quotient with respect to unitary sectors
performed in the DHR analysis of algebraic covariant sectors.
Only by choosing suitable species $\mr{a}$ 
and restricting to a specific subcategory say $\Updelta(\mr{a})$ of $\mf{Sct}(\mr{a})$
we obtain $\mf{ex}(\mr{S})=\mf{ex}(\mr{T})$
for any $\mr{S},\mr{T}\in\Updelta(\mr{a})$ such that 
$\mr{Inv}_{\Updelta(\mr{a})}(\mr{S},\mr{T})\neq\varnothing$.
The reason inherits by the fact that in our framework different connectors, or even sectors,
in general generate quite different experimental settings. 
Because of this rather than seeking for an equivalence of 
the category of some known mathematical structure
with a distinct subcategory of equivalence classes of sectors
of some suitable species;
it is preferable to establish an equivalence between 
a suitable $2-$category and $2-\mf{dp}$.
This in line with our original purpose of showing that 
there are properties of invariance concerning \emph{diverse} species, 
an emblematic case is the equiformity principle between 
a classical and a quantum species,
we shall return to this point.
\end{remark}
Next we present some easy to prove algebraic properties of $\mf{Cnt}$.
\begin{proposition}
\label{10052123}
\begin{enumerate}
\item
$(\forall\mr{T},\mr{S}\in\mf{Cnt})
((\mr{T},\mr{S})\in Dom(\circ)\Rightarrow\mr{T}\circ\mr{S}\in\mf{Cnt})$;
\item
$(\forall\mr{T}\in\mf{Cnt})
(\forall\mr{L}\in 2-cell(2-\mf{dp}))
((\mr{T},\mr{L})\in Dom(\ast)\Rightarrow\mr{T}\ast\mr{L}\in\mf{Cnt})$;
\item
$\mr{Mor}_{\mf{Sp}(\mf{Chdv})}\ast\mf{Cnt}\subset\mf{Cnt}$.
\end{enumerate}
\end{proposition}
\subsection{Physical posit}
\label{11301655c}
Here we describe the physical interpretation of $2-\mf{dp}$,
and in particular that of $\mf{Cnt}$. 
Thm. \ref{10081910}\eqref{10081910st5} permits the following
\begin{definition}
[Equiformity principle of $\mf{s}(\mr{T})$
and interpretations]
\label{10041214}
Let 
\begin{equation*}
\mc{E},
\mc{E}_{\mf{s}}
\in
\prod_{\mr{T}\in\mf{Cnt}}
\prod_{\mf{Q}\in Exp(c(\mr{T}))}
\prod_{t\in\Upgamma(\mf{Q},\mr{T})}
\mr{set},
\end{equation*}
such that 
\begin{equation*}
\begin{aligned}
\mc{E}(\mr{T})(\mf{Q})(t)
&\coloneqq
\{
\text{equalities in Prp. \ref{01162038} for }
\mf{E}^{\mr{a}}=\mr{T}[\mf{Q},t],
\mf{E}^{\mr{b}}=\mf{Q}
\},
\\
\mc{E}_{\mf{s}}(\mr{T})(\mf{Q})(t)
&\coloneqq
\{\mf{s}(x)
\,\vert\,
x\in\mc{E}(\mr{T})(\mf{Q})(t) 
\}.
\end{aligned}
\end{equation*}
Let 
\begin{equation*}
\begin{aligned}
\mc{X}_{\mf{s}}
&\in
\prod_{\mr{T}\in\mf{Cnt}}
\prod_{\mf{Q}\in\mf{ex}(\mr{T})}
\mr{set},
\\
\mc{X}_{\mf{s}}(\mr{T})(\mf{Q})
&\coloneqq
\mf{s}(\mf{Q}).
\end{aligned}
\end{equation*}
Finally for any $\mr{a}\in\mf{Sp}$ set 
$\mc{Z}(\mr{a})\coloneqq\mc{E}\up\mf{Sct}(\mr{a})$,
$\mc{Z}_{\mf{s}}(\mr{a})\coloneqq\mc{E}_{\mf{s}}\up\mf{Sct}(\mr{a})$.
\end{definition}
Essentially $\mc{E}_{\mf{s}}(\mr{T})(\mf{Q})(t)$ 
is a collection of sentences of the sort stated in 
Prp. \ref{12011304}.
The next reflects our physical interpretation of $2-\mf{dp}$
which by matter-of-fact represents a paradigm.
\begin{posit}
[Physical reality is $2-\mf{dp}$]
\label{10031521}
Our posit consists in what follows
\begin{enumerate}
\item
$2-\mf{dp}$ is the structure of the physical reality;
\item
$\mf{Cnt}$ as a structure 
and as a collection is the only empirically testable part of $2-\mf{dp}$;
\label{07141358}
\item
each connector of $\mf{Cnt}$ admits an empirical representation 
and a physical interpretation;
\item
for any $\mr{T}\in\mf{Cnt}$ 
\begin{enumerate}
\item
$(\mc{E}(\mr{T}),\mf{ex}(\mr{T}))$
is the empirical representation of $\mf{s}(\mr{T})$,
\label{10031521rps}
\item
$(\mc{E}_{\mf{s}}(\mr{T}),\mc{X}_{\mf{s}}(\mr{T}))$
is the physical interpretation of $\mf{s}(\mr{T})$;
\end{enumerate}
\item
The couple of results:
charge composition of connectors Cor. \ref{10151636}, 
applied to elements of $\mf{Cnt}$,  
and 
charge transfer Cor. \ref{11200910}, 
applied to elements of $\mf{Cnt}$ and to $2-$cells of $2-\mf{dp}$,
is 
the empirical representation of 
the algebraic structure of $\mf{Cnt}$ described in Prp. \ref{10052123},
whose physical interpretation is in terms of the semantics $(\mf{M},\mf{s},\mf{u})$.
\end{enumerate}
\end{posit}
\begin{remark}
[Empirical representation and physical interpretation of species]
Since Posit \ref{10031521} it follows that for each $\mr{a}\in\mf{Sp}$ 
\begin{enumerate}
\item
$(\mc{Z}(\mr{a}),\mf{ex}(\un_{\mr{a}}))$
is the empirical representation of $\mf{s}(\mr{a})$;
\item
$(\mc{Z}_{\mf{s}}(\mr{a}),\mc{X}_{\mf{s}}(\un_{\mr{a}}))$
is the physical interpretation of $\mf{s}(\mr{a})$.
\end{enumerate}
We recall that 
$\mf{ex}(\un_{\mr{a}}))=Exp(\mr{a})$.
\end{remark}
\begin{convention}
Let $\mr{T}\in\mf{Cnt}$ and $\mr{a}\in\mf{Sp}$,
we convein to call 
\begin{enumerate}
\item
$\mc{E}(\mr{T})$ the equiformity principle of $\mf{s}(\mr{T})$;
\item
$\mc{E}_{\mf{s}}(\mr{T})$ the interpretation of $\mc{E}(\mr{T})$; 
\item
$\mc{E}(\mr{T})(\mf{Q})(s)$ the equiformity principle of $\mf{s}(\mr{T})$
evaluated on $\mf{s}(\mr{T}[\mf{Q},s])$ and $\mf{s}(\mf{Q})$;
\item
$\mc{E}_{\mf{s}}(\mr{T})(\mf{Q})(s)$ the interpretation of $\mc{E}(\mr{T})(\mf{Q})(s)$ 
\item
$\mc{X}_{\mf{s}}(\mr{T})$ the interpretation of $\mf{ex}(\mr{T})$;
\item
$\mc{Z}(\mr{a})$ the equiformity principle of $\mf{s}(\mr{a})$;
\item
$\mc{Z}_{\mf{s}}(\mr{a})$ the interpretation of $\mc{Z}(\mr{a})$.
\end{enumerate}
\end{convention}
\begin{remark}
We have that
\begin{enumerate}
\item
the empirical representation of $\mf{s}(\mr{T})$
consists in the couple of 
the equiformity principle of $\mf{s}(\mr{T})$
and the collection of the experimental settings of $\mf{s}(\mr{T})$;
\item
the equiformity principle of $\mf{s}(\mr{a})$
is the map mapping any sector $\mr{T}$ of $\mf{s}(\mr{a})$
to the equiformity principle of $\mf{s}(\mr{T})$;
\item
the interpretation of the equiformity principle of $\mf{s}(\mr{a})$
is the map mapping any sector $\mr{T}$ of $\mf{s}(\mr{a})$
to the interpretation of the equiformity principle of $\mf{s}(\mr{T})$;
\item
$\mc{X}_{\mf{s}}(\un_{\mr{a}})$ is the interpretation of $Exp(\mr{a})$. 
\end{enumerate}
\end{remark}
Cor. \ref{10151636} permits the following
\begin{definition}
[Experimental settings generated by connectors and a vacuum connector]
\label{10052028}
Let $\uppi\in\mf{Cnt}$, $\mr{a}\in\mf{Sp}$,
$\mf{Q}\in Exp(c(\uppi))$ 
and $r\in\Upgamma(\mf{Q},\uppi)$
set
\begin{equation*}
\mf{Ex}(\mr{a},\uppi,\mf{Q},r)
\coloneqq
\{
(\uppi\circ\mr{S})
[\mf{Q},r\circ q]
\,\vert\,
\mr{S}\in\mf{Cnt}_{\mr{a}},c(\mr{S})=d(\uppi),
q\in\Upgamma(\uppi[\mf{Q},r],\mr{S})
\}.
\end{equation*}
\end{definition}
Notice that 
$\mf{Ex}(\mr{a},\uppi,\mf{Q},r)\subseteq Exp(\mr{a})$.
\begin{corollary}
\label{10052040}
Let $\mr{a}\in\mf{Sp}$, 
$\uppi\in\mf{Cnt}$, 
$\mf{Q}\in Exp(c(\uppi))$
and $r\in\Upgamma(\mf{Q},\uppi)$.
Thus
\begin{equation*}
\mf{Ex}(\mr{a},\uppi,\mf{Q},r)
=
\{
\mr{S}[\uppi[\mf{Q},r],q]
\,\vert\,
\mr{S}\in\mf{Cnt}_{\mr{a}},c(\mr{S})=d(\uppi),
q\in\Upgamma(\uppi[\mf{Q},r],\mr{S})
\},
\end{equation*}
in particular if $\uppi\in\mf{Cnt}_{\mr{a}}$, then 
\begin{equation*}
\mf{Ex}(\mr{a},\uppi,\mf{Q},r)
=
\{
\mr{S}[\uppi[\mf{Q},r],q]
\,\vert\,
\mr{S}\in\mf{Sct}(\mr{a}),
q\in\Upgamma(\uppi[\mf{Q},r],\mr{S})
\}.
\end{equation*}
\end{corollary}
\begin{proof}
Since Cor. \ref{10151636}.
\end{proof}
\section{Gravity Species}
\label{09291055}
In section \ref{10222056} we provide, in the fourth main result of the paper,
the construction of the $n-$dimensional classical gravity species $\mr{a}^{n}$ 
and then we establish that
the equivalence principle of general relativity
emerges as equiformity principle of the connector 
canonically associated with $\mr{a}^{n}$.
By using $\mr{a}^{n}$ as a model in section \ref{10181450}
we define the collection of $n-$dimensional gravity species 
in particular diverse subcollections of quantum gravity species.
The equiformity principle of a connector from $\mr{a}^{n}$ to a strict quantum gravity species
will provide a quantum realization of the velocity of maximal integral curves of 
complete vector fields on spacetimes.
Applied to Robertson-Walker spacetimes 
we establish that the Hubble parameter, the acceleration 
of the scale function and new constraints for its positivity
over a subset of the range of the galactic time of a geodesic $\alpha$, 
are expressed in terms of a quantum realization of the velocity of $\alpha$.
As a remarkable result we obtain that 
the existence of a connector from $\mr{a}^{4}$ 
to a $4-$dimensional strict quantum gravity species 
satisfying these constraints implies the positivity of the acceleration
and renders the dark energy hypothesis inessential.
\begin{notation}
\label{11041637}
Here manifold means second countable finite dimensional smooth manifold.
For any manifold $M$ let $\mc{A}(M)\coloneqq\mc{C}^{\infty}(M,\C)$ the complex $\ast-$algebra of 
smooth maps on $M$ at values in $\C$ with $\C$ provided by the 
pullback of the standard smooth structure on $\R^{2}$.
Any open set $X$ of $M$ is tacitly considered as a submanifold of $M$, so $\mc{A}(X)$ makes sense. 
We identify $\mc{A}(M)_{ob}$ with the real linear space underlying $\mf{F}(M)\coloneqq\mc{C}^{\infty}(M)$, 
the space of real valued smooth maps on $M$. For any $p\in M$ let 
$\updelta_{p}^{M}$ also denoted by $\omega_{p}$ 
be the Dirac distribution on $p$, namely the map on $\mc{A}(M)$ such that $\updelta_{p}^{M}(f)=f(p)$.
Any open set of $M$ is tacitly considered as an open submanifold of $M$.
Let $\mf{T}_{s}^{r}(M)$ denote the (real) linear space of $(r,s)$ tensor fields on $M$ with 
$(r,s)\in\{0,\dots,n\}^{2}$, $n=dim(M)$, 
Let $\xi$ be any coordinate system on $U\subset M$,
$A\in\mf{T}_{s}^{r}(M)$, $i\in\{0,\dots,n-1\}^{r}$ and $j\in\{0,\dots,n-1\}^{s}$.
Then let 
$x_{j}^{\xi}$ and $\partial_{\xi}^{j}$ 
denote the $j$th coordinate function and vector field of $\xi$ 
respectively, 
and $A_{j_{1},\dots,j_{s}}^{i_{1},\dots,i_{r};\xi}\in\mc{C}^{\infty}(U)$ 
denote the $(i_{1},\dots,i_{r};j_{1},\dots,j_{s})$ component of $A$ relative to $\xi$, 
where the case $(r,s)=(0,0)$ has to be understood as $A\up U$;
only when the coordinate system $\xi$ involved is known, we often remove the symbol $\xi$.
Let $\mf{X}(M)$ denote the linear space of vector fields on $M$, often we use the
standard identification 
$\mf{X}(M)$ with the derivations on $M$, i.e. 
$W'(h)\in\mf{F}(M)$ such that $W'(h)(q)\coloneqq W(q)(h)$, 
for all $W\in\mf{X}(M)$, $h\in\mf{F}(M)$ and $q\in M$, \cite[p. 12]{one}. 
Often we use also the identification $\mf{X}(M)=\mf{T}_{0}^{1}(M)$. 
Let $M,N$ be manifolds, $\phi:M\to N$ smooth, then
$d\phi:TM\to TN$ such that $d\phi(v)(h)\coloneqq v(h\circ\phi)$, for all $v\in TM$
and $h\in\mf{F}(N)$, \cite[p. 9]{one}.
If $A\in\mf{T}_{s}^{0}(N)$ with $s\in\N$, then $\phi^{\ast}A\in\mf{T}_{s}^{0}(M)$ 
is the pullback of $A$ by $\phi$ defined in \cite[p. 42, Def. 8]{one}.
For any manifold $M$ and $V,W\in\mf{X}(M)$ we have that
$[V,W]\in\mf{X}(M)$ such that $[V,W](p)(f)=V(p)(W(f))-W(p)(V(f))$
for all $f\in\mf{F}(M)$ and $p\in M$, \cite[p. 13]{one}. 
If $\phi:M\to N$ is a smooth map, $A\in\mf{X}(M)$ and $B\in\mf{X}(N)$, 
then by following the definition given in 
\cite[p.182]{lee} or \cite[Def.1.20]{one} we say that $A$ and $B$ 
are $\phi-$related if $d(\phi)\circ A=B\circ\phi$,
if $\phi$ is a diffeomorphism, then for any vector field $A$ on $M$ 
we let $d\phi A$ denote the unique vector field on $N$ $\phi-$related
to $A$, namely $d\phi A\coloneqq (d\phi)\circ A\circ\phi^{-1}$,
\cite[pg. 14]{one}. 
For any complete smooth vector field $V$ on $M$ let $\vartheta^{V}:\R\to Diff(M)$ the flow of $V$, 
hence $\vartheta^{V}(\tau)(p)=\alpha_{p}^{V}(\tau)$, with $p\in M$ and $\tau\in\R$ 
and $\alpha_{p}^{V}$ the unique inextendible integral curve of $V$ such that $\alpha(0)=p$.
For any smooth map $\phi:M\to N$ where $N$ is a manifold, we let $\phi^{\ast}$ denote the usual pullback,
so in particular
$\phi^{\ast}:\mc{C}^{\infty}(N)\to\mc{C}^{\infty}(M)$ such that $f\mapsto f\circ\phi$.
In order to avoid conflict with the notation used in section \ref{not1} for general invertible
morphisms in a category, we convein that $\phi^{\ast}$ 
always refers to the pullback of $\phi$ whenever $\phi$ is a smooth map between manifolds. 
If the contrary is not asserted, let $\pi$ denote the projection map from $TM$ to $M$.
We adopt for semi-Riemannian geometry the conventions in \cite{one},
in particular let $\mc{M}=(M,g)$ be a semi-Riemannian manifold then we 
let $\lr{\cdot}{\cdot}_{\mc{M}}$ or simply $\lr{\cdot}{\cdot}$ denote $g$ and
whenever it does not cause confusion let $D^{\mc{M}}$ be 
the Levi-Civita connection of $(M,g)$ \cite[Thm. 3.11]{one}.
We let $\mf{X}(\mc{M})$ denote $\mf{X}(M)$, let $g_{\mc{M}}$ denote $g$
and let $\phi:\mc{M}\to\mc{N}$ denote $\phi:M\to N$
for any semi-Riemannian manifold $\mc{N}=(N,h)$.
We say $U\in\mf{X}(\mc{M})$ to be geodesic if $D_{U}^{\mc{M}}U=0$. 
Notice that the integral curves of a geodesic vector field are geodesic curves.
$v\in T_{p}M$ is timelike if $g_{p}(v,v)<0$, while $U$ is timelike if 
$\lr{U}{U}_{\mc{M}}<\ze$.
If $X,Y\in\mf{X}(\mc{M})$, then
$X\perp Y$ means $\lr{X}{Y}_{\mc{M}}=\ze$, 
for any $U\in\mf{X}(\mc{M})$
let $U^{\perp}$ the set of all $Z\in\mf{X}(\mc{M})$
such that $X\perp U$.
Let $R^{\mc{M}}$ and $Ric^{\mc{M}}$ denote the Riemannian curvature and Ricci tensor of $\mc{M}$
respectively \cite[Def. 3.35, Def. 3.53]{one}, 
Assume $\mc{M}$ and $\mc{N}$ be spacetimes see 
\cite[p. 163]{one} except the restriction of dimension equal to $4$.
If $\mf{T}^{\mc{M}}$ and $\mf{T}^{\mc{N}}$ are the timelike
smooth vector fields of $\mc{M}$ and $\mc{N}$ respectively
determining the orientations of $\mc{M}$ and $\mc{N}$ according
to \cite[Lemma 5.32]{one}, then we say that 
$\phi$ preserves the orientation, iff $\mf{T}^{\mc{M}}$ and $\mf{T}^{\mc{N}}$ 
are $\phi-$related.
$S^{\mc{M}}$ is the scalar curvature of $\mc{M}$, \cite[Def. 3.53]{one},
and $G^{\mc{M}}$ be the Einstein gravitational tensor \cite[Def. 12.1]{one}.
Let $T^{\mc{M}}$ be the stress-energy tensor associated with $\mc{M}$
namely such that the Einstein equation $G^{\mc{M}}=8\uppi T^{\mc{M}}$ holds.
According \cite[Def. 12.4]{one} we have that $(U,\uprho,p)$ is a perfect 
fluid on $\mc{N}$ if $U\in\mf{X}(N)$ timelike future-pointing unit,
$\uprho,p\in\mf{F}(N)$, and for all $X,Y\in U^{\perp}$ we have 
$T^{\mc{N}}(U,U)=\uprho$,
$T^{\mc{N}}(X,U)=T^{\mc{N}}(U,X)=\ze$,
and 
$T^{\mc{N}}(X,Y)=p\lr{X}{Y}_{\mc{N}}$.
Let $n\in\N$ and $\{\ep_{\mu}\}_{\mu=0}^{n}$ be the standard basis in $\R^{n+1}$,
we sometime let $x^{\mu}=x(\mu)$ for $x\in\R^{n+1}$,
let $\Pr_{\mu}:\R^{n+1}\to\R$ be $x\mapsto x(\mu)$ and
$\imath_{\mu}:\R\to\R^{n+1}$ be such that 
$\Pr_{\nu}\circ\imath_{\mu}=\un_{\R}\delta_{\mu,\nu}$,
for all $\mu,\nu\in\{0,\dots,n\}$.
For any linear operator $L$ on $\R^{n+1}$ let $\mr{m}(L)$ be the matrix of $L$
w.r.t. the canonical basis of $\R^{n+1}$, i.e.  
$\sum_{\nu=0}^{n}\mr{m}(L)_{\mu\nu}\ep_{\mu}=L(\ep_{\nu})$,
for all $\mu,\nu\in\{0,\dots,n\}$.
Let $\R_{1}^{4}$ be the Minkowski $4-$space, while for any Minkowski 
spacetime $M$ let $Lor(M)$ be the set of the Lorentz coordinate systems in $M$,
see \cite[p.163-164, Def. 8 and 11]{one}.
Let $\mc{P}^{+}=\R^{4}\rtimes_{\eta}SL(2,\C)$,
where $\eta$ is the standard action of $SL(2,\C)$ on $\R^{4}$, 
and let $j_{1}$ and $j_{2}$ be the canonical injection of $\R^{4}$ and $SL(2,\C)$ into $\mc{P}^{+}$
respectively, see \cite{silterr} and reference therein.
Let $\uprho$ be the action of $\mc{P}^{+}$ on $\R^{4}$ such that 
$\uprho(x,\Lambda)y\coloneqq x+\eta(\Lambda)y$, for any $\Lambda\in SL(2,\C)$ and $x,y\in\R^{4}$.
Let $\mr{LS}$ denote the category of real linear spaces.
\end{notation}
\subsection{$n-$dimensional classical gravity species}
\label{10222056}
After preparatory lemmata,
the first result of this section is 
\textbf{Thm. \ref{11151519}}
where we prove that suitable $\phi$ preserve some preperties
of $\phi-$related vector fields of 
semi-Riemannian manifolds, such as the property of being geodesic
and, in case of spacetimes, of being the support field of a perfect fluid.
In Def. \ref{09251618} we introduce the category $\mr{vf}_{0}$
whose object set is the collection of couples of 
manifolds and complete smooth vector fields on them,
with smooth open maps relating the vector fields as morphisms.
In Def. \ref{11051643} we provide the definition of the category 
$\mr{St}_{n}$ of $n-$dimensional spacetimes and observer fields,
naturally embedded into $\mr{vf}_{0}$.
In Def. \ref{12051029} we define the collection of the relevant topologies 
on the algebras $\mc{A}(M)$ for all manifolds $M$, employed 
in Prp. \ref{09201647} to state that 
with any object of $\mr{vf}_{0}$ remains associated 
a dynamical pattern, namely an object of the category $\mf{dp}$,
whose dynamical category we construct in Prp. \ref{09201646}.
Prp. \ref{09201647} is the essential step toward \textbf{Thm. \ref{09201707}},
the main result of this section, 
where we construct a species on $\mr{vf}_{0}$,
namely a functor from $\mr{vf}_{0}$ to the category $\mf{Chdv}$.
As a consequence we obtain in Cor. \ref{01191431}
the $n-$dimensional classical gravity species 
$\mr{a}^{n}\in\mf{Sp}(\mr{St}_{n})$. 
Then we state in \textbf{Cor. \ref{01201634}}
that the equivalence principle of general relativity
is the equiformity principle of the connector 
canonically associated with $\mr{a}^{n}$.
\begin{lemma}
\label{11141553}
Let $M_{1},M_{2}$ be manifolds, $\phi:M_{1}\to M_{2}$ smooth,
and $V_{i},W_{i}\in\mf{X}(M_{i})$ for any $i\in\{1,2\}$.
If $V_{1}$ and $V_{2}$ as well $W_{1}$ and $W_{2}$ are $\phi-$related 
then 
$[V_{1},W_{1}]$ and $[V_{2},W_{2}]$ are $\phi-$related. 
\end{lemma}
\begin{proof}
Let $p\in M_{1}$ and $f\in\mf{F}(M_{2})$, then
\begin{equation*}
\begin{aligned}
([V_{2},W_{2}]\circ\phi)(p)(f)
&=(V_{2}\circ\phi)(p)(W_{2}(f))-(W_{2}\circ\phi)(p)(V_{2}(f))\\
&=(d\phi\circ V_{1})(p)(W_{2}(f))-(d\phi\circ W_{1})(p)(V_{2}(f))\\
&=V_{1}(p)(W_{2}(f)\circ\phi)-W_{1}(p)(V_{2}(f)\circ\phi)\\
&=V_{1}(p)((W_{2}\circ\phi)(\cdot)(f))-W_{1}(p)((V_{2}\circ\phi)(\cdot)(f))\\
&=V_{1}(p)((d\phi\circ W_{1})(\cdot)(f))-W_{1}(p)((d\phi\circ V_{1})(\cdot)(f))\\
&=V_{1}(p)(W_{1}(f\circ\phi))-W_{1}(p)(V_{1}(f\circ\phi))\\
&=[V_{1},W_{1}](p)(f\circ\phi)=(d\phi\circ[V_{1},W_{1}])(p)(f).
\end{aligned}
\end{equation*}
\end{proof}
\begin{proposition}
\label{11151106}
Let $M,N$ be manifolds, $\phi:M\to N$ smooth, $A\in\mf{T}_{s}^{0}(N)$ 
with $s\in\N$,
$X_{r}\in\mf{X}(M)$, $Y_{r}\in\mf{X}(N)$ such that $X_{r}$ and $Y_{r}$
are $\phi-$related for any $r\in\{1,\dots,s\}$.
Then 
$(\phi^{\ast}A)(X_{1},\dots,X_{s})=A(Y_{1},\dots,Y_{s})\circ\phi$.
Let $B\in\mf{T}_{s}^{0}(M)$ and assume that 
$\phi$ has dense range. 
If for any $(X_{1},\dots,X_{s})\in\mf{X}(M)^{s}$
there exists $(Y_{1},\dots,Y_{s})\in\mf{X}(N)^{s}$, 
such that 
$B(X_{1},\dots,X_{s})=A(Y_{1},\dots,Y_{s})\circ\phi$
and
$X_{r}$ and $Y_{r}$ are $\phi-$related for any $r\in\{1,\dots,s\}$,
then $B=\phi^{\ast}A$.
\end{proposition}
\begin{proof}
Let $p\in M$ then
\begin{equation}
\begin{aligned}
(\phi^{\ast}A)(X_{1},\dots,X_{s})(p)
&=
(\phi^{\ast}A)(p)(X_{1}(p),\dots,X_{s}(p))
\\
&=
A(\phi(p))((d\phi\circ X_{1})(p),\dots,(d\phi\circ X_{s})(p))
\\
&=
A(\phi(p))((Y_{1}\circ\phi)(p),\dots,(Y_{s}\circ\phi)(p))\\
&=
(A(Y_{1},\dots,Y_{s})\circ\phi)(p).
\end{aligned}
\end{equation}
The second sentence of the statement follows by the first one and since 
$\mf{F}(M)\subset\mc{C}(M)$. 
\end{proof}
Till Thm. \ref{11151519} let
$\mc{M},\mc{N}$ be semi-Riemannian manifolds,
$m=dim(\mc{M})$, $n=dim(\mc{N})$
and $\phi:\mc{M}\to\mc{N}$ smooth such that $g_{\mc{M}}=\phi^{\ast}g_{\mc{N}}$.
\begin{lemma}
\label{11151232}
Let $U,V,Y\in\mf{X}(\mc{M})$ 
and 
$\mc{U},\mc{V},\mc{Y}\in\mf{X}(\mc{N})$, 
such that 
$U$ and $\mc{U}$, 
$V$ and $\mc{V}$, 
and 
$Y$ and $\mc{Y}$ are $\phi-$related
respectively.
Thus $\lr{D_{U}^{\mc{M}}V}{Y}_{\mc{M}}=
\lr{D_{\mc{U}}^{\mc{N}}\mc{V}}{\mc{Y}}_{\mc{N}}\circ\phi$.
\end{lemma}
\begin{proof}
Since Prp. \ref{11151106} applied to $A=g_{\mc{N}}$ we have 
$\lr{V}{Y}_{\mc{M}}=\lr{\mc{V}}{\mc{Y}}_{\mc{N}}\circ\phi$,
and 
$\lr{U}{V}_{\mc{M}}=\lr{\mc{U}}{\mc{V}}_{\mc{N}}\circ\phi$,
while since Prp. \ref{11151106} and Lemma \ref{11141553}
we obtain
$\lr{U}{[V,Y]}_{\mc{M}}=\lr{\mc{U}}{[\mc{V},\mc{Y}]}_{\mc{N}}\circ\phi$,
and the statement follows since the Koszul formula \cite[p. 61 Thm. 11]{one}
applied to $\lr{\cdot}{\cdot}_{\mc{M}}$ and $\lr{\cdot}{\cdot}_{\mc{N}}$.
\end{proof}
\begin{lemma}
\label{11031600}
Let $X,Y,Z,K\in\mf{X}(\mc{M})$ 
and 
$\mc{X},\mc{Y},\mc{Z},\mc{K}\in\mf{X}(\mc{N})$, 
such that 
$X$ and $\mc{X}$, 
$Y$ and $\mc{Y}$, 
$Z$ and $\mc{Z}$ 
and
$K$ and $\mc{K}$ 
are $\phi-$related
respectively.
Then 
$\lr{R^{\mc{M}}_{XY}Z}{K}_{\mc{M}}
=
\lr{R^{\mc{N}}_{\mc{XY}}\mc{Z}}{\mc{K}}_{\mc{N}}\circ\phi$.
\end{lemma}
\begin{proof}
Since Lemma \ref{11151232} and \cite[Lemma 3.35]{one}.
\end{proof}
\begin{lemma}
\label{11031637}
Let $\{E_{k}\}_{k=1}^{m}$ be a frame field on $\mc{M}$. 
If the range of $\phi$ is dense and
$\{\mc{E}_{k}\}_{k=1}^{m}$ is a $m-$tupla 
of smooth vector fields on $\mc{N}$ such that $E_{k}$ and $\mc{E}_{k}$ are $\phi-$related
for all $k\in\{1,\dots,m\}$,
then 
$\{\mc{E}_{k}\}_{k=1}^{m}$ is an orthogonal system of unit vector fields on $\mc{N}$ such that 
$\lr{E_{i}}{E_{k}}_{\mc{M}}=\lr{\mc{E}_{i}}{\mc{E}_{k}}_{\mc{N}}\circ\phi$.
\end{lemma}
\begin{proof}
Since Prp. \ref{11151106} we obtain the equality, the first sentence of the statement 
then follows by this equality and by the fact that the range of $\phi$ is dense.
\end{proof}
\begin{definition}
Let the property $\star(\phi)$ be the following proposition
$(\forall X\in\mf{X}(\mc{M}))(\exists\,\mc{X}\in\mf{X}(\mc{N}))(X,\mc{X}\text{ are $\phi-$related})$.
While let the property $\ast(\phi)$ be the following proposition
$(\forall\mc{X}\in\mf{X}(\mc{N}))(\exists\,X\in\mf{X}(\mc{M}))
(X,\mc{X}\text{ are $\phi-$related})$.
\end{definition}
Clearly the properties $\star(\phi)$ and $\ast(\phi)$ 
are satisfied in case $\phi$ is a diffeomorphism. 
\begin{lemma}
\label{11031638}
Let the property $\star(\phi)$ be satisfied, let the range of $\phi$ be dense,
and let $m=n$.
Thus for all $X,Y\in\mf{X}(\mc{M})$
and
$\mc{X},\mc{Y}\in\mf{X}(\mc{N})$ such that 
$X,\mc{X}$ and $Y,\mc{Y}$ are $\phi-$related respectively,
we have that 
$Ric^{\mc{M}}(X,Y)=Ric^{\mc{N}}(\mc{X},\mc{Y})\circ\phi$.
\end{lemma}
\begin{proof}
Let $\{E_{k}\}_{k=1}^{m}$ be a frame field on $\mc{M}$
and let $\{\mc{E}_{k}\}_{k=1}^{m}$ be
any $m-$tupla of smooth vector fields on $N$ such that $E_{k},\mc{E}_{k}$ are $\phi-$related,
existing by the property $\star(\phi)$, thus since Lemma \ref{11031637} and the hypothesis 
$m=n$ we obtain that 
\begin{equation}
\label{11031705}
\{\mc{E}_{k}\}_{k=1}^{m}
\text{ is a frame field on $\mc{N}$.} 
\end{equation}
Next
\begin{equation*}
\begin{aligned}
Ric^{\mc{M}}(X,Y)
&=
\sum_{k=1}^{m}
\lr{E_{k}}{E_{k}}_{\mc{M}}
\lr{R^{\mc{M}}_{XE_{i}}Y}{E_{i}}_{\mc{M}}
\\
&=
\sum_{k=1}^{m}
(\lr{\mc{E}_{k}}{\mc{E}_{k}}_{\mc{N}}\circ\phi)
\lr{R^{\mc{N}}_{\mc{X}\mc{E}_{i}}\mc{Y}}{\mc{E}_{i}}_{\mc{N}}\circ\phi
\\
&=
Ric^{\mc{N}}(\mc{X},\mc{Y})\circ\phi,
\end{aligned}
\end{equation*}
where the first equality follows since \cite[Lemma 3.52]{one}, 
the second one by Lemma \ref{11031600} and the equality in Lemma \ref{11031637},
the third equality follows by \cite[Lemma 3.52]{one} and \eqref{11031705}.
\end{proof}
\begin{lemma}
\label{11031721}
Let $\upxi$ be a coordinate system on $U\subseteq N$ 
and assume that $\phi$ is a diffeomorphism, then 
$\partial_{j}^{\upxi},\partial_{j}^{\upxi\circ\phi}$
and 
$(dx_{\upxi}^{j})_{\ast},(dx_{\upxi\circ\phi}^{j})_{\ast}$
are $\phi\up\phi^{-1}(U)-$related for any $j\in\{1,\dots n\}$
respectively,
where $(\cdot)_{\ast}$ means metrically equivalent vector field.
\end{lemma}
\begin{proof}
For the coordinate vector fields follows by the definition and 
\cite[Lemma 1.21]{one}. 
$(dx_{\upxi}^{j})_{\ast},(dx_{\upxi\circ\phi}^{j})_{\ast}$
are $\phi\up\phi^{-1}(U)-$related since the formula in 
the proof of \cite[Prp. 3.10]{one} applied to $\theta=dx_{\upxi}^{j}$,
Lemma \ref{11151106} applied for $A=g_{N}$ and to the coordinate vector 
fields just now proved to be $\phi-$related.
\end{proof}
\begin{lemma}
\label{11031907}
Let $\upxi$ be a coordinate system on $U\subseteq N$ 
and assume that $\phi$ is a diffeomorphism, then 
$S^{\mc{M}}=S^{\mc{N}}\circ\phi$.
\end{lemma}
\begin{proof}
We have 
\begin{equation*}
\begin{aligned}
C(\uparrow_{1}^{1}Ric^{\mc{M}})
&=
\sum_{i=1}^{n}
\uparrow_{1}^{1}Ric^{\mc{M}}
(dx_{\upxi\circ\phi}^{i},\partial_{i}^{\upxi\circ\phi})
\\
&=
\sum_{i=1}^{n}
Ric^{\mc{M}}
((dx_{\upxi\circ\phi}^{i})_{\ast},\partial_{i}^{\upxi\circ\phi})
\\
&=
\sum_{i=1}^{n}
Ric^{\mc{N}}
((dx_{\upxi}^{i})_{\ast},\partial_{i}^{\upxi})
\circ\phi
\\
&=
C(\uparrow_{1}^{1}Ric^{\mc{N}})
\circ\phi.
\end{aligned}
\end{equation*}
where the first and the fourth equalities follow since the formula 
in the proof of \cite[Lemma 2.6]{one} applied to $\upxi\circ\phi$, $S^{\mc{M}}$ 
and $\upxi$, $S^{\mc{N}}$ respectively,
the second equality follows since the definition of the metric contraction,
while the third one follows by Lemma \ref{11031721} and Lemma \ref{11031638}. 
\end{proof}
\begin{lemma}
\label{11040652}
If $\phi$ is a diffeomorphism, then
for any $X,Y\in\mf{X}(M)$ we have  
$G^{\mc{M}}(X,Y)=G^{\mc{N}}(d\phi X,d\phi Y)\circ\phi$.
\end{lemma}
\begin{proof}
Since Lemma \ref{11031638}, Lemma \ref{11031907} and Prp. \ref{11151106}
applied to $A=g_{\mc{N}}$.
\end{proof}
Now we can state the following
\begin{theorem}
\label{11151519}
Let $\mc{M},\mc{N}$ be semi-Riemannian manifolds, 
$\phi:\mc{M}\to\mc{N}$ smooth, 
$U\in\mf{X}(\mc{M})$ and $\mc{U}\in\mf{X}(\mc{N})$
such that $g_{\mc{M}}=\phi^{\ast}g_{\mc{N}}$ 
and $U$ and $\mc{U}$ are $\phi-$related.
Then 
\begin{enumerate}
\item
If the property $\star(\phi)$ is satisfied,
then
\eqref{11151519st1a}
$\Rightarrow$
\eqref{11151519st1b}
where
\begin{enumerate}
\item
$\mc{U}$ geodesic 
\label{11151519st1a}
\item
$U$ geodesic.
\label{11151519st1b}
\end{enumerate}
\label{11151519st1}
\item
If the property $\ast(\phi)$ is satisfied
and
$\phi$ has a dense range then 
\eqref{11151519st1b}
$\Rightarrow$
\eqref{11151519st1a},
in particular if $\phi$ is a diffeomorphism, then
\eqref{11151519st1a}
$\Leftrightarrow$
\eqref{11151519st1b}.
\label{11151519st2}
\item
If $\mc{M}$ and $\mc{N}$ are spacetimes and 
$\phi$ is a diffeomorphism preserving the orientation, 
then 
\eqref{11151519st2a}
$\Leftrightarrow$
\eqref{11151519st2b}
where
\begin{enumerate}
\item
$(\mc{U},\uprho,p)$ perfect fluid on $\mc{N}$,
\label{11151519st2a}
\item
$(U,\phi_{\ast}\uprho,\phi_{\ast}p)$ perfect fluid on $\mc{M}$.
\label{11151519st2b}
\end{enumerate}
\label{11151519st3}
\end{enumerate}
\end{theorem}
\begin{proof}
St. \eqref{11151519st1} follows 
since Lemma \ref{11151232} and \cite[p. 60, Prp. 10]{one},
st. \eqref{11151519st2} follows 
since Lemma \ref{11151232}, \cite[p. 60, Prp. 10]{one},
and since $\lr{\cdot}{\cdot}_{\mc{N}}$ 
takes values in $\mf{F}(N)\subset\mc{C}(\mc{N})$. 
Assume that $(\mc{U},\uprho,p)$ perfect fluid on $\mc{N}$, then since 
Lemma \ref{11040652} we obtain 
\begin{equation}
\label{11041002}
\begin{aligned}
T^{\mc{M}}(U,U)
&=
T^{\mc{N}}(\mc{U},\mc{U})\circ\phi
\\
&=
\uprho\circ\phi=
\phi_{\ast}\uprho.
\end{aligned}
\end{equation}
Next let $\mc{X}\in\mc{U}^{\perp}$, 
thus $d\phi^{-1}\mc{X}\in U^{\perp}$
since Prp. \ref{11151106} applied to $A=g_{\mc{N}}$,
hence $d\phi^{-1}\mc{U}^{\perp}\subseteq U^{\perp}$,
similarly 
$d\phi U^{\perp}\subseteq\mc{U}^{\perp}$,
thus
\begin{equation}
\label{11041000}
d\phi^{-1}\mc{U}^{\perp}=U^{\perp}.
\end{equation}
Next let $\mc{X},\mc{Y}\in\mc{U}^{\perp}$, then
since Lemma \ref{11040652} 
and Prp. \ref{11151106} applied to $g_{\mc{N}}$.
we obtain 
\begin{equation}
\label{11041001}
\begin{aligned}
T^{\mc{M}}(d\phi^{-1}\mc{X},U)
&=
T^{\mc{N}}(\mc{X},\mc{U})\circ\phi
=\ze,
\\
T^{\mc{M}}(U,d\phi^{-1}\mc{X})
&=
T^{\mc{N}}(\mc{U},\mc{X})\circ\phi
=\ze,
\\
T^{\mc{M}}(d\phi^{-1}\mc{X},d\phi^{-1}\mc{Y})
&=
T^{\mc{N}}(\mc{X},\mc{Y})\circ\phi
\\
&=(p\lr{\mc{X}}{\mc{Y}}_{\mc{N}})\circ\phi
\\
&=(p\circ\phi)\lr{\mc{X}}{\mc{Y}}_{\mc{N}}\circ\phi
\\
&=\phi_{\ast}p\lr{d\phi^{-1}\mc{X}}{d\phi^{-1}\mc{Y}}_{\mc{M}}.
\end{aligned}
\end{equation}
Since \eqref{11041002}, \eqref{11041000} and \eqref{11041001}
to prove st. \eqref{11151519st3}
it remains to prove that 
$U$ is timelike future-pointing and unit.
$U$ is timelike and unit since Prp. \ref{11151106} applied to $g_{\mc{N}}$.
By hypothesis $\phi$ preserves the orientation, 
which implies since $\phi$ is a diffeomorphism that
if $\mf{T}^{\mc{M}}$ and $\mf{T}^{\mc{N}}$ are the timelike
smooth vector fields of $\mc{M}$ and $\mc{N}$ respectively
determining the orientations of $\mc{M}$ and $\mc{N}$ according
\cite[Lemma 5.32]{one}, then we require that 
$\mf{T}^{\mc{N}}=d\phi\mf{T}^{\mc{M}}$.
Therefore $U$ is future-pointing since Prp. \ref{11151106} applied to $g_{\mc{N}}$.
\end{proof}
\begin{proposition}
\label{09201700}
There exists a unique category $\mr{vf}$ such that 
\begin{enumerate}
\item
$\mr{Obj}(\mr{vf})=\{(M,U)\,\vert\, M\text{ is a manifold }, U\in\mf{X}(M)\text{ complete}\}$;
\item
for all $(M,U),(N,V),(L,W)\in \mr{Obj}(\mr{vf})$, we have 
\begin{enumerate}
\item
$\mr{Mor}_{\mr{vf}}((M,U),(N,V))
=
\{\phi:N\to M\text{ smooth}\,\vert\, V\text{ and }U\text{ are $\phi-$related}\}$;
\item
$\psi\circ_{\mr{vf}}\phi=\phi\circ\psi$, for any
$\phi\in \mr{Mor}_{\mr{vf}}((M,U),(N,V))$
and
$\psi\in \mr{Mor}_{\mr{vf}}((N,V),(L,W))$;
\item
$\un_{(M,U)}=\un_{M}$.
\end{enumerate}
\end{enumerate}
\end{proposition}
\begin{proof}
Since the chain rule, see for example \cite[p. 10, Lemma 15]{one},
the morphism composition is a well-set internal associative partial operation.
The remaining of the statement is easy to prove.
\end{proof}
\begin{remark}
\label{12131411}
Since Thm. \ref{11151519} we obtain that for any $(M,U),(N,V)\in\mr{vf}$ 
and $\phi\in \mr{Mor}_{\mr{vf}}((M,U),(N,V))$ 
such that $M$ and $N$ are 
supports of semi-Riemannian manifolds 
$(M,g_{M})$ and $(N,g_{N})$,
$g_{N}=\phi^{\ast}g_{M}$
and $\phi$ has dense range,
then $V$ geodesic implies $U$ geodesic, 
if in addition $\phi$ is a diffeomorphism then $V$ geodesic is equivalent to $U$ geodesic.
\end{remark}
\begin{definition}
\label{09251618}
Let $\mr{vf}_{0}$ be the subcategory of $\mr{vf}$ with the same object set and 
$\mr{Mor}_{\mr{vf}_{0}}((M,U),(N,V))$ the subset of the $\phi\in \mr{Mor}_{\mr{vf}}((M,U),(N,V))$
such that $\phi$ is an open map.
\end{definition}
Since \cite[Cor. 2.3.]{mic} any submersion is an open map, in particular any
diffeomorphism is an open map.
\begin{definition}
\label{11051643}
For any $n\in\Z_{0}^{+}$, let $\mr{St}_{n}$ be the category
whose object set is the set of the couples $(\mc{M},U)$ where
$\mc{M}$ is an $n-$dimensional connected time-oriented Lorentz manifold
and $U\in\mf{X}(\mc{M})$ complete, timelike and unit future-pointing,
while $\mr{Mor}_{\mr{St}_{n}}((\mc{M},U),(\mc{N},V))$ 
is the set of the submersions $\phi:\mc{N}\to\mc{M}$ preserving the orientation,
such that $\phi^{\ast}g_{\mc{M}}=g_{\mc{N}}$ and $V$ and $U$ are $\phi-$related; 
while the morphism composition is the opposite of map composition.
Let $\imath_{n}$ be the couple of maps defined the first on 
$\mr{Obj}(\mr{St}_{n})$ and the second on $Mor(\mr{St}_{n})$
such that the first one maps any $((M,g),U)$ into $(M,U)$ 
and the second one any morphism $\phi$ into itself.
\end{definition}
Clearly $\imath_{n}$ is a functor from $\mr{St}_{n}$ to $\mr{vf}_{0}$, 
In Prp. \ref{09201647} we will show that 
with any object of $\mr{vf}_{0}$ remains associated 
a dynamical pattern.
\begin{proposition}
\label{09201646}
Let $M$ be a manifold and $U\in\mf{X}(M)$ be complete. 
There exists a unique $\mr{top}-$quasi enriched category $[M,U]$ such that 
\begin{enumerate}
\item
$\mr{Obj}([M,U])=\{W\,\vert\,W\in Op(M)\wedge(\exists\,C\in\mr{Cl}(M))(C\subseteq W)\}$;
\item
for all $X,Y,Z\in[M,U]$
\begin{enumerate}
\item
$\mr{Mor}_{[M,U]}(X,Y)=\{(X,Y)\}\times \mr{mor}_{[M,U]}(X,Y)$,
where 
\begin{equation*}
\mr{mor}_{[M,U]}(X,Y)=\{t\in\R\,\vert\, 
\vartheta^{U}(t)Y\subseteq X\}
\end{equation*}
\item
for any 
$t\in \mr{mor}_{[M,U]}(X,Y)$
and
$s\in \mr{mor}_{[M,U]}(Y,Z)$
\begin{equation*}
((Y,Z),s)\circ((X,Y),t)=((X,Z),s+t),
\end{equation*}
and 
$((X,X),0)$ is the identity morphism of $X$;
\item
$\mr{Mor}_{[M,U]}(X,Y)$ is provided by the topology inherited by that in $\R$.
\end{enumerate}
\end{enumerate}
\end{proposition}
\begin{proof}
The composition is a well-defined associative partial operation 
since $\vartheta^{U}$ is a group morphism, and the category is 
$\mr{top}-$enriched since the sum in $\R$ is continuous.
\end{proof}
We shall prove in Prp. \ref{09201647} 
that a dynamical pattern remains associated with any object of $\mr{vf}$,
then determining the object map of a functor from $\mr{vf}_{0}$ to $\mf{dp}$
as established in Thm. \ref{09201707}.
Both these results reside in the pertinent choice of a class of  
topologies of the algebra $\mc{A}(M)$. 
This class of topologies results intrinsically related to the category $\mr{vf}$
as we can see in the following
\begin{definition}
\label{12051029}
We call $\mr{vf}-$topology any map $\upxi$ defined on $\mr{Obj}(\mr{vf})$ such that 
for all $[M,U]\in\mr{vf}$ we have 
\begin{enumerate}
\item
$\upxi_{[M,U]}$ is a Hausdorff topology on $\mc{A}(M)$
providing it with the structure of topological $\ast-$algebra;
\label{12051029a}
\item
$t\mapsto f\circ\uptheta^{U}(t)$
is a continuous map from $\R$ to $\lr{\mc{A}(M)}{\upxi_{[M,U]}}$,
for all $f\in\mc{A}(M)$; 
\label{12051029b}
\item
$f\mapsto f\circ\phi$
is a continuous map from 
$\lr{\mc{A}(M)}{\upxi_{[M,U]}}$
to
$\lr{\mc{A}(N)}{\upxi_{[N,V]}}$,
for all 
$[N,V]\in\mr{vf}$ 
and all
$\phi\in \mr{Mor}_{\mr{vf}}([M,U],[N,V])$.
\label{12051029c}
\end{enumerate}
We call $\upxi$ locally convex if all its values are locally convex topologies.
For any $[M,U]\in\mr{vf}$ and any submanifold $X$ of $M$ we let 
$\lr{\mc{A}(X)}{\upxi_{[M,U]}}$ denote the algebra $\mc{A}(X)$ endowed with the \emph{final} topology 
for the map $r_{X}:\lr{\mc{A}(M)}{\upxi_{[M,U]}}\to\mc{A}(X)$ $f\mapsto f\up X$.
\end{definition}
Note that Def. \ref{12051029}\eqref{12051029b}
is well-set since the composition of smooth maps is smooth.
\begin{proposition}
\label{12051106}
The map assigning to any object $[M,U]\in\mr{vf}$
the pointwise topology on $\mc{A}(M)$ is a locally convex $\mr{vf}-$topology,
in particular the collection of $\mr{vf}-$topologies is nonempty.
\end{proposition}
\begin{proof}
Let $f\in\mc{A}(M)$ and $p\in M$ we have that
$(f\circ\vartheta^{U}(t))(p)=(f\circ\alpha_{p}^{U})(t)$,
next $\alpha_{p}^{U}$ is continuous so
$t\mapsto(f\circ\vartheta^{U}(t))(p)$ is continuous,
hence  
$t\mapsto(f\circ\vartheta^{U}(t))$ is continuous in the pointwise topology.
The remaining of the statement is easy to prove. 
\end{proof}
\begin{remark}
\label{12051119}
The map assigning to any object $[M,U]\in\mr{vf}$ 
the Schwartz topology \cite[Ch. 4, n$^{\circ}$ 4.2]{mall}
on $\mc{A}(M)$ satisfies 
Def. \ref{12051029}(\ref{12051029a},\ref{12051029c}),
presently it is not clear if it satisfies 
also Def. \ref{12051029}\eqref{12051029b}.
\end{remark}
\begin{convention}
For the remaining of the paper we consider fixed a 
$\mr{vf}-$topology $\upxi$.
\end{convention}
\begin{proposition}
\label{09201647}
Let $M$ be a manifold and $U\in\mf{X}(M)$ be complete.
There exists a unique 
$F_{[M,U]}\in\mr{Fct}_{\mr{top}}([M,U],\mr{tsa})$
such that for all $X,Y\in[M,U]$, $t\in \mr{mor}_{[M,U]}(X,Y)$
and $f\in\mc{A}(X)$
\begin{enumerate}
\item
$F_{[M,U]}(X)=\lr{\mc{A}(X)}{\upxi_{[M,U]}}$;
\item
$F_{[M,U]}((X,Y),t)f=f\circ\vartheta^{U}(t)\up Y$\footnote{we 
convein to simplify the notation letting
$F_{[M,U]}((X,Y),t)$
denote
$F_{[M,U]}(((X,Y),t))$.}.
\end{enumerate}
\end{proposition}
\begin{proof}
Let us assume the notation of the statements.
The second statement is well-set since $\vartheta^{U}(t)Y\subseteq X$.
Let $\tilde{f}\in\mc{A}(M)$ extending $f$ whose existence is ensured for instance by \cite[Lemma 2.26]{lee}.
Next $r_{Y}:F_{[M,U]}(M)\to F_{[M,U]}(Y)$ is continuous since \cite[Prp. 6 I.14]{BourGT}, 
while $\zeta:\R\to F_{[M,U]}(M)\,t\mapsto\tilde{f}\circ\uptheta^{U}(t)$ is continuous by definition. 
Therefore $r_{Y}\circ\zeta:\R\to F_{[M,U]}(Y)$ is continuous moreover 
$r_{Y}\circ\zeta:t\mapsto f\circ\uptheta^{U}(t)\up Y$ since $\uptheta^{U}(t)(Y)\subseteq X$.
Next clearly $F_{[M,U]}$ is uniquely determined by the properties in the statement,
thus by recalling that for all $\mc{X},\mc{Y}\in\mr{tsa}$
the space $\mr{Mor}_{\mr{tsa}}(\mc{X},\mc{Y})$ is provided with the topology of simple convergence, we obtain 
by the above shown continuity of $r_{Y}\circ\zeta$ that 
\begin{equation*}
F_{[M,U]}\in\mc{C}(\mr{Mor}_{[M,U]}(X,Y),\mr{Mor}_{\mr{tsa}}(F_{[M,U]}(X),F_{[M,U]}(Y))).
\end{equation*}
Next 
\begin{equation*}
\begin{aligned}
F_{[M,U]}(((Y,Z),s)\circ((X,Y),t))
f
&=
\\
f\circ\vartheta^{U}(t+s)\up Z
&=
\\
f\circ
(\vartheta^{U}(t)\up Y)
\circ
\vartheta^{U}(s)\up Z
&=
\\
(F_{[M,U]}((Y,Z),s)\circ F_{[M,U]}((X,Y),t))
f,
\end{aligned}
\end{equation*}
and 
$(F_{[M,U]}((X,X),0)f=f$.
\end{proof}
\begin{remark}
\label{12111210}
Prp. \ref{09201647} states that 
$\lr{[M,U]}{F_{[M,U]}}\in\mr{Obj}(\mf{dp})$, see Def. \ref{09081303}.
\end{remark}
Now we are in the position of stating the following 
\begin{theorem}
\label{09201707}
There exists a unique $\mf{a}\in\mr{Fct}(\mr{vf}_{0},\mf{dp})$ 
such that for all $(M,U),(N,V)\in\mr{vf}_{0}$ and $\phi\in \mr{Mor}_{\mr{vf}_{0}}((M,U),(N,V))$
\begin{enumerate}
\item
$\mf{a}((M,U))=\lr{[M,U]}{F_{[M,U]}}$,
\item
$\mf{a}(\phi)=(f_{\phi},T_{\phi})$;
\end{enumerate}
where 
$f_{\phi}\in\mr{Fct}_{\mr{top}}([N,V],[M,U])$ 
and 
$T_{\phi}\in \mr{Mor}_{\mr{Fct}([N,V],\mr{tsa})}(F_{[M,U]}\circ f_{\phi},F_{[N,V]})$
such that for all $Y,Z\in[N,V]$ and 
$t\in \mr{mor}_{[N,V]}(Y,Z)$
\begin{enumerate}
\item
$f_{\phi}(Y)=\phi(Y)$;
\item
$f_{\phi}((Y,Z),t)=((\phi(Y),\phi(Z)),t)$;
\item
$T_{\phi}(Y):\mc{A}(\phi(Y))\ni h\mapsto h\circ\phi\up Y\in\mc{A}(Y)$.
\end{enumerate}
In particular $\ps{\Uppsi}\circ\mf{a}\in \mr{Mor}_{2-\mf{dp}}(\mr{vf}_{0},\mf{Chdv})$,
namely $\ps{\Uppsi}\circ\mf{a}\in\mf{Sp}(\mr{vf}_{0})$
where $\ps{\Uppsi}$ is the functor in Thm. \ref{12312222}.
\end{theorem}
\begin{proof}
Let us use the notation of the statement.
The last sentence of the statement follows since the first and Thm. \ref{12312222}.
The properties in the statement determine uniquely the maps $f_{\phi}$, $T_{\phi}$ and $\mf{a}$,
let us prove that these properties are well-set and that the remaining of the statement holds.
First claim: $f_{\phi}$ is well-set and $f_{\phi}\in\mr{Fct}_{\mr{top}}([N,V],[M,U])$. 
Since the naturality of flows, see for instance \cite[Prp. 9.13]{lee}, 
we deduce that $\phi\circ\vartheta^{V}(t)=\vartheta^{U}(t)\circ\phi$, hence 
\begin{equation*}
\vartheta^{U}(t)(\phi(Z))=\phi(\vartheta^{V}(t)Z)\subseteq\phi(Y);
\end{equation*}
moreover $\phi(Y)$ and $\phi(Z)$ are open since $\phi$ is open by construction, 
so $\phi(Y),\phi(Z)\in[M,U]$ and $t\in \mr{mor}_{[M,U]}(\phi(Y),\phi(Z))$, hence $f_{\phi}$ is well-set, clearly 
continuous and composition preserving, hence our first claim has been proven. 
Second claim: $T_{\phi}$ is well-set and
$T_{\phi}\in \mr{Mor}_{\mr{Fct}([N,V],\mr{tsa})}(F_{[M,U]}\circ f_{\phi},F_{[N,V]})$.
$T_{\phi}$ is well-set since the first proven claim and Prp. \ref{09201647}.
Next $T_{\phi}(Y)$ clearly is a morphism of $\ast-$algebras let us now prove that it is continuous.
Let $\chi:F_{[M,U]}(M)\to F_{[N,V]}(N)\,k\mapsto k\circ\phi$ thus by Def. \ref{12051029} $\chi$ is continuous,
next $r_{Y}:F_{[N,V]}(N)\to F_{[N,V]}(Y)$ is continuous by \cite[Prp. 6 I.14]{BourGT}, therefore 
$r_{Y}\circ\chi:F_{[M,U]}(M)\to F_{[N,V]}(Y)$ is continuous. Now easily we see that 
$T_{\phi}(Y)\circ r_{\phi(Y)}=r_{Y}\circ\chi$ which is therefore continuous, 
thus since \cite[Prp. 6 I.14]{BourGT} we deduce that $T_{\phi}(Y):F_{[M,U]}(\phi(Y))\to F_{[N,V]}(Y)$ 
is continuous. It remains to show that the following diagram in $\mr{tsa}$ is commutaive
\begin{equation*}
\xymatrix{
\mc{A}(\phi(Z))\ar[rr]^{T_{\phi}(Z)}
&&
\mc{A}(Z)
\\
&&
\\
\mc{A}(\phi(Y))
\ar[rr]_{T_{\phi}(Y)}
\ar[uu]^{F_{[M,U]}((\phi(Y),\phi(Z)),t)}
&&
\mc{A}(Y)
\ar[uu]_{F_{[N,V]}((Y,Z),t)}
}
\end{equation*}
Let $h\in\mc{A}(\phi(Y))$ then 
\begin{equation*}
\begin{aligned}
(T_{\phi}(Z)
\circ 
F_{[M,U]}((\phi(Y),\phi(Z)),t))h
&=
\\
T_{\phi}(Z)(h\circ\vartheta^{U}(t)\up\phi(Z))
&=
\\
h\circ\vartheta^{U}(t)\circ\phi\up Z
&=
\\
h\circ\phi\circ\vartheta^{V}(t)
\up Z
&=
\\
(h\circ\phi\up Y)\circ\vartheta^{V}(t)
\up Z
&=
\\
(F_{[N,V]}((Y,Z),t)\circ T_{\phi}(Y))h,
\end{aligned}
\end{equation*}
where the $3$th equality follows since the naturality of flows,
thus the previous diagram is commutative and the second claim is shown.
In order to conclude the proof we need to show the following
$3$th claim: $\mf{a}(\psi\circ_{\mr{vf}_{0}}\phi)=\mf{a}(\psi)\circ_{\mf{dp}}\mf{a}(\phi)$
for any $\phi\in \mr{Mor}_{\mr{vf}_{0}}((M,U),(N,V))$ and $\psi\in \mr{Mor}_{\mr{vf}_{0}}((N,V),(L,W))$.
To this end note that
$\mf{a}(\psi\circ_{\mr{vf}_{0}}\phi)
=
\mf{a}(\phi\circ\psi)
=
(f_{\phi\circ\psi},T_{\phi\circ\psi})$
and
\begin{equation*}
\begin{aligned}
\mf{a}(\psi)\circ_{\mf{dp}}\mf{a}(\phi)
&=
(f_{\psi},T_{\psi})\circ_{\mf{dp}}(f_{\phi},T_{\phi})
\\
&=
(f_{\phi}\circ f_{\psi}, T_{\psi}\circ(T_{\phi}\ast\un_{f_{\psi}}))
\\
&=
(f_{\phi}\circ f_{\psi}, T_{\psi}\circ(T_{\phi}\circ (f_{\psi})_{o})),
\end{aligned}
\end{equation*}
where the last equality follows since \eqref{20061403},
thus our claim is equivalent to the following one
\begin{equation}
\label{10181007}
\begin{aligned}
f_{\phi\circ\psi}&=f_{\phi}\circ f_{\psi},
\\
T_{\phi\circ\psi}&=T_{\psi}\circ(T_{\phi}\circ (f_{\psi})_{o}).
\end{aligned}
\end{equation}
The first equality is trivially true.
Next let $D\in[L,W]$ and $k\in\mc{A}((\phi\circ\psi)(D))$
then
$T_{\phi\circ\psi}(D)k=(k\circ\phi\up\psi(D))\circ\psi\up D$,
while 
\begin{equation*}
\begin{aligned}
(T_{\psi}\circ(T_{\phi}\circ (f_{\psi})_{o}))(D)k
&=
\\
(T_{\psi}(D)\circ T_{\phi}(\psi(D)))k
&=
\\
T_{\psi}(D)(k\circ\phi\up\psi(D))
&=
(k\circ\phi\up\psi(D))\circ\psi\up D,
\end{aligned}
\end{equation*}
so the second equality and then our $3$th claim has been proven. 
\end{proof}
\begin{definition}
\label{10041441}
$\mr{a}^{n}\coloneqq\ps{\Uppsi}\circ\mf{a}\circ\imath_{n}$.
\end{definition}
\begin{corollary}
[$\mr{a}^{n}$ is a species]
\label{01191431}
$\mr{St}_{n}\in\mr{Cat}$ while $\mr{a}^{n}\in \mr{Mor}_{2-\mf{dp}}(\mr{St}_{n},\mf{Chdv})$,
namely $\mr{a}^{n}\in\mf{Sp}(\mr{St_{n}})$, in particular 
$\un_{\mr{a}^{n}}\in\mf{Sct}(\mr{a}^{n})$.
\end{corollary}
\begin{proof}
Since Thm. \ref{09201707}.
\end{proof}
This result permits the following
\begin{definition}
\label{10181456}
$\mf{u}(\mr{a}^{n})\equiv$ of classical $n-$dimensional gravity.
\end{definition}
Next we shall see that the equiformity principle 
$\mc{E}(\un_{\mr{a}^{n}})(\mf{P}^{\mr{a}^{n}},\mr{G}^{\mr{a}^{n}})(\un)$
-
of the connector, actually a sector,
$\un_{\mr{a}^{n}}$ evaluated on the standard experimental setting
associated with $\mr{a}^{n}$
-
is nothing else than the generalization at $n$ dimensions 
of the equivalence principle in general relativity.
In the next section we shall consider instead equiformity principles between 
$\mr{a}^{n}$ and quantum gravity species.
\begin{corollary}
[Equiformity principle of $\un_{\mr{a}^{n}}$ 
alias equivalence principle of general relativity]
\label{01201634}
Let $n\in\Z_{0}^{+}$, $(\mc{M},U),(\mc{N},V)\in\mr{St}_{n}$, 
where $\mc{M}=(M,g)$ and $\mc{N}=(N,h)$.
Thus for any $\phi\in \mr{Mor}_{\mr{St}_{n}}((\mc{M},U),(\mc{N},V))$, 
$Y,Z\in[N,V]$ and $t\in \mr{mor}_{[N,V]}(Y,Z)$,
we have for all $\uppsi\in\mf{P}_{F_{[N,V]}(Z)}$ 
and $A\in\mc{A}(\phi(Y))_{ob}$ 
\begin{enumerate}
\item
invariant form
\begin{equation*}
\mf{f}_{(\uppsi,A\circ\phi\up Y)}^{\mr{a}^{n},(\mc{N},V),Y,Z}(((Y,Z),t))
=
\mf{f}_{(\uppsi\circ(\phi\up Z)^{\ast},A)}^{\mr{a}^{n},
(\mc{M},U),\phi(Y),\phi(Z)}(((\phi(Y),\phi(Z)),t));
\end{equation*}
\label{01201634st1}
\item
equivariant form
\begin{equation*}
(\phi\up Y)^{\ast}
\mf{t}^{\mr{a}^{n},(\mc{N},V),Y,Z,\uppsi}
=
f_{\phi}^{\intercal}
\mf{t}^{\mr{a}^{n},(\mc{M},U),\phi(Y),\phi(Z),\uppsi\circ(\phi\up Z)^{\ast}}.
\end{equation*}
\end{enumerate}
Moreover if $\phi$ is a diffeomorphism,
then $U$ is geodedic if and only if $V$ is geodedic
and if in addition $\phi$ preserves the orientation
then
$U$ is a vector field of a perfect fluid on $\mc{M}$
if and only if
$V$ is a vector field of a perfect fluid on $\mc{N}$.
\end{corollary}
\begin{proof}
Since Cor. \ref{01191431}, Thm. \ref{10081910}\eqref{10081910st5}
and Prp. \ref{01162038} applied for the positions
$\mr{a}=\mr{b}=\mr{a}^{n}$ and $\mr{T}$ the identity morphism of $\mr{Fct}(\mr{St}_{n},\mf{Chdv})$
relative to $\mr{a}^{n}$. The final sentence follows since Thm. \ref{11151519}.
\end{proof}
\subsection{Equiformity principle between classical and quantum gravity}
\label{10181450}
We start the present section by defining the concept of realization which 
represents the equiformity principle when the context morphism $\phi$ is trivial. 
Roughly speaking a species $\mr{b}$ realizes a species $\mr{a}$ when 
the dynamics $\uptau_{\mr{b}}$ mirrors the dynamics $\uptau_{\mr{a}}$.
Since any connector induces an equiformity principle according to Thm. \ref{10081910},
one shows that the target species of any connector is a realization of its source species
Cor. \ref{04081852}.
Based on the structure of $\mr{a}^{n}$ we constructed in Cor. \ref{01191431},
we define the collection of $n-$dimensional gravity species, 
the subcollection of classical ones to which $\mr{a}^{n}$ belongs
and diverse subcollections of $n-$dimensional quantum gravity species.
Then we establish that for any connector $\mr{T}$ from $\mr{a}^{n}$ to 
a $n-$dimensional strict quantum gravity species,
exist quantum realizations of the velocity 
of maximal integral curves of complete vector fields 
Cor. \ref{11091638}.
As an application to a Robertson-Walker spacetime $M(x,f)$ 
we establish that the Hubble parameter $H=f^{\prime}/f$,
the acceleration $f^{\prime\prime}$ of the scale function
and the constraints of its positivity evaluated over a subset of the range 
of the galactic time $t$ of a geodesic $\alpha$, 
are expressed in terms of a quantum realization of the velocity of $\alpha$
Thm. \ref{04151343} and Cor. \ref{04171614}.
\begin{definition}
[Realizations]
\label{04081745}
Let $(\mr{a},\mr{b})\in\mf{Sp}_{\ast}$. 
We say that 
$\mr{b}$ is a realization of $\mr{a}$
or that
$\mr{b}$ realizes the dynamics of $\mr{a}$
if for all 
$\mf{E}^{\mr{b}}=(\mf{S}^{\mr{b}},\mr{R}^{\mr{b}})\in Exp(\mr{b})$
there exists
$\mf{E}^{\mr{a}}=(\mf{S}^{\mr{a}},\mr{R}^{\mr{a}})\in Exp(\mr{a})$,
for all $M\in d(\mr{a})$ and all $y,z\in\mr{R}_{M}^{\mr{b}}$
there exist
$y^{\prime},z^{\prime}\in\mr{R}_{M}^{\mr{a}}$,
for all $g\in\mr{R}_{M}^{\mr{b}}(y,z)$ there exists
$g^{\prime}\in\mr{R}_{M}^{\mr{a}}(y^{\prime},z^{\prime})$, 
for all 
$\uppsi\in\mf{S}_{M}^{\mr{b}}(z)$ 
there exists
$\uppsi^{\prime}\in\mf{S}_{M}^{\mr{a}}(z^{\prime})$ 
and 
for all $A\in\mc{A}_{\mr{a}(M)}(y^{\prime})_{ob}$ 
there exists
$A^{\prime}\in\mc{A}_{\mr{b}(M)}(y)_{ob}$ 
satisfying
\begin{equation*}
\mf{f}_{(\uppsi,A^{\prime})}^{\mr{b},M,y,z}(g)
=
\mf{f}_{(\uppsi^{\prime},A)}^{\mr{a},M,y^{\prime},z^{\prime}}(g^{\prime}).
\end{equation*}
\end{definition}
As a result of the equiformity principle 
the target species of a connector realizes the dynamics
of the source species, namely
\begin{corollary}
[Realizations induced by connectors]
\label{04081852}
Let $(\mr{a},\mr{b})\in\mf{Sp}_{\ast}$ and $\mr{T}\in\mf{Cnt}(\mr{a},\mr{b})$,
thus $\mr{b}$ is a realization of $\mr{a}$ such that 
for all $\mf{E}\in Exp(\mr{b})$
we can select $\mf{E}^{\mr{a}}=(\mf{P}^{\mr{a}},\mr{G}^{\mr{a}})$;
while if $\Upgamma(\mf{E},\mr{T})\neq\emptyset$\footnote{for instance whenever 
$\mr{T}_{1}^{o}(M)$ is injective for all $M\in d(\mr{a})$
see Rmk. \ref{10201649}.},
then for any $s\in\Upgamma(\mf{E},\mr{T})$ we can select
$\mf{E}^{\mr{a}}=\mr{T}[\mf{E},s]$.
Moreover letting $\mf{E}=(\mf{S},\mr{R})$, 
for any 
$M\in d(\mr{a})$, $y,z\in\mr{R}_{M}$, $g\in\mr{R}_{M}(y,z)$ and $\uppsi\in\mf{S}_{M}(z)$
we can select
$y^{\prime}=\mr{T}_{1}^{o}(M)y$,
$z^{\prime}=\mr{T}_{1}^{o}(M)z$,
$g^{\prime}=\mr{T}_{1}^{m}(M)g$
and
$\uppsi^{\prime}=\mr{T}_{3}^{\dagger}(M)(z)\uppsi$,
while for any
$A\in\mc{A}_{\mr{a}(M)}(y^{\prime})_{ob}$
we can select 
$A^{\prime}=\mr{T}_{3}(M)(y)A$.
\end{corollary}
\begin{proof}
The general case follows since Prp. \ref{01162038} applied for $M=N$ and $\phi=\un_{M}$,
and by Thm. \ref{01181342}\eqref{01181342st1},
the case $\Upgamma(\mf{E},\mr{T})\neq\emptyset$
by Thm. \ref{10081910}\eqref{10081910st2}.
\end{proof}
Next we shall define the collection of $n-$dimensional gravity species 
as the set of species contextualized on $\mr{St}_{n}$ and factorizable
through $\ps{\Uppsi}$ and $\imath_{n}$. 
In this way we provide at the same time
a sort of minimal extension of the construction of $\mr{a}^{n}$, the possibility 
to select diverse subcollections of quantum gravity species, and most importantly 
a path to construct connectors between the classical and a quantum gravity.
\begin{definition}
[$n-$dimensional gravity species]
\label{10031749}
Let $n\in\Z_{0}^{+}$, define
\begin{equation*}
\begin{aligned}
G_{n}
&\coloneqq
\{\mr{b}
\in\mf{Sp}(\mr{St}_{n})
\,\vert\,
(\exists\,\mf{b}\in\mr{Fct}(\mr{vf}_{0},\mf{dp})) 
(\mr{b}=\ps{\Uppsi}\circ\mf{b}\circ\imath_{n})\},
\\
CG_{n}
&\coloneqq
\{\mr{b}\in G_{n}
\,\vert\,
(\forall M\in\mr{St}_{n})(\forall x\in\mr{G}_{M}^{\mr{b}})
(\mc{A}_{\mr{b}(M)}(x)\text{ is commutative})\};
\\
QG_{n}
&\coloneqq
\{\mr{b}\in G_{n}
\,\vert\,
(\forall M\in\mr{St}_{n})(\forall x\in\mr{G}_{M}^{\mr{b}})
(\mc{A}_{\mr{b}(M)}(x)\text{ is noncommutative})\};
\\
sQG_{n}
&\coloneqq
\{\mr{b}\in QG_{n}
\,\vert\,
(\forall M\in\mr{St}_{n})
(\mr{G}_{M}^{\mr{b}}=\mr{G}_{M}^{\mr{a}^{n}})
\}.
\end{aligned}
\end{equation*}
We call 
\begin{enumerate}
\item
$G_{n}$ the collection of $n-$dimensional gravity species;
\item
$CG_{n}$ the collection of $n-$dimensional classical gravity species;
\item
$QG_{n}$ the collection of $n-$dimensional quantum gravity species;
\item
$sQG_{n}$ the collection of $n-$dimensional strict quantum gravity species.
\end{enumerate}
\end{definition}
Clearly $\mr{a}^{n}\in CG_{n}$.
We can define the relativistic species
roughly speaking by restriction of $4-$dimensional gravity species 
$\mr{b}$ over the category of vector fields on subregions of a fixed Minkowski spacetime,
namely $\mr{b}$ composed the identity 
functor valued in $\mr{St}_{4}$ 
and defined in the full subcategory $\mr{Sr}(\M)$ 
of submanifolds of a given Minkowski spacetime $\M$
and observer fields on them.
\begin{definition}
Let $\M$ be a fixed Minkowski spacetime.
Define $\mr{Sr}(\M)$ to be the full subcategory of $\mr{St}_{4}$ whose object set
is the subset of all the $(\mc{M},U)$ such that 
$\mc{M}$ is a semi-Riemanniann submanifold of $\M$.
Moreover let $\mr{sr}(\M)$ be the subcategory of $\mr{Sr}(\M)$ whose object set 
is $\mr{Sr}(\M)$ and for any object $P,Q$
we have that $\mr{Mor}_{\mr{sr}(\M)}(P,Q)$ is the 
subset of those elements of $\mr{Mor}_{\mr{Sr}(\M)}(P,Q)$ which are 
local diffeomorphisms.
\end{definition}
Note that if for any $P\in\mr{sr}(\M)$ and $p\in P$, 
we identify the tangent space $T_{p}P$ with 
the Minkowski $4-$space $\R_{1}^{4}$,
then $\{d\phi_{p}\,\vert\,\phi\in \mr{Mor}_{\mr{sr}(\M)}(P,P)\}=O_{1}(4)$ 
the subgroup of linear isometries of the space $\R_{1}^{4}$. 
We recall that given a category $B$ and a subcategory $A$ of $B$,
thus $\mr{I}_{A\to B}$ denotes the functor from $A$ to $B$
whose object and morphism maps are the identity maps.
\begin{definition}
Define
$\mr{I}_{\M}\coloneqq\mr{I}_{\mr{Sr}(\M)\to\mr{St}_{4}}$, 
$\imath^{\M}\coloneqq\imath_{4}\circ\mr{I}_{\M}$,
and
$\mr{a}^{rt}\coloneqq\mr{a}^{4}\circ\mr{I}_{M}$
\end{definition}
\begin{definition}
Define
\begin{equation*}
\begin{aligned}
RT(\M)
&\coloneqq
\{\mr{c}\in\mf{Sp}(\mr{Sr}(\M))
\,\vert\,
(\exists\,\mr{b}\in G_{4})
(\mr{c}=\mr{b}\circ\mr{I}_{\M})\};
\\
CRT(\M)
&\coloneqq
\{\mr{c}\in\mf{Sp}(\mr{Sr}(\M))
\,\vert\,
(\exists\,\mr{b}\in CG_{4})
(\mr{c}=\mr{b}\circ\mr{I}_{\M})\};
\\
QRT(\M)
&\coloneqq
\{\mr{c}\in\mf{Sp}(\mr{Sr}(\M))
\,\vert\,
(\exists\,\mr{b}\in QG_{4})
(\mr{c}=\mr{b}\circ\mr{I}_{\M})\};
\\
sQRT(\M)
&\coloneqq
\{\mr{c}\in\mf{Sp}(\mr{Sr}(\M))
\,\vert\,
(\exists\,\mr{b}\in sQG_{4})
(\mr{c}=\mr{b}\circ\mr{I}_{\M})\}.
\end{aligned}
\end{equation*}
We call 
\begin{enumerate}
\item
$RT(\M)$ the collection of relativistic species in $\M$;
\item
$CRT(\M)$ the collection of classical relativistic species in $\M$;
\item
$QRT(\M)$ the collection of quantum relativistic species in $\M$;
\item
$sQRT(\M)$ the collection of strict quantum relativistic species in $\M$.
\end{enumerate}
\end{definition}
\begin{remark}
Clearly $\mr{a}^{rt}\in CRT(\M)$ moreover we have that 
\begin{equation*}
\begin{aligned}
RT(\M)
&\subseteq
\{\mr{c}
\in\mf{Sp}(\mr{Sr}(\M))
\,\vert\,
(\exists\,\mf{b}
\in\mr{Fct}(\mr{vf}_{0},\mf{dp})) 
(\mr{c}=\ps{\Uppsi}\circ\mf{b}\circ\imath^{\M})\};
\\
CRT(\M)
&\subseteq
\{\mr{c}\in RT(\M)
\,\vert\,
(\forall P\in\mr{Sr}(\M))(\forall x\in\mr{G}_{P}^{\mr{c}})
(\mc{A}_{\mr{c}(P)}(x)\text{ is commutative})\};
\\
QRT(\M)
&\subseteq
\{\mr{c}\in RT(\M)
\,\vert\,
(\forall P\in\mr{Sr}(\M))(\forall x\in\mr{G}_{P}^{\mr{c}})
(\mc{A}_{\mr{c}(P)}(x)\text{ is noncommutative})\};
\\
sQRT(\M)
&\subseteq
\{\mr{c}\in QRT(\M)
\,\vert\,
(\forall P\in\mr{Sr}(\M))
(\mr{G}_{P}^{\mr{c}}=\mr{G}_{P}^{\mr{a}^{rt}})
\}.
\end{aligned}
\end{equation*}
\end{remark}
Next let us come back to the general case and examine 
the outcomes of a realization of $\mr{a}^{n}$ belonging to $sQG_{n}$. 
For any context $(\mc{M},U)$ in $\mr{St}_{n}$
and any open set $Z$ of $\mc{M}$, 
there exists a variety
of observables in $\mc{A}(Z)$ obtained by using the metric tensor 
of $\mc{M}$ and the vector field $U$,
indeed $\lr{B}{U}\up Z\in\mc{A}(Z)_{ob}$
for any smooth vector field $B$ on $\mc{M}$ 
in particular any component of a frame field on $\mc{M}$.
Thus let us set
\begin{definition}
\label{04050854}
Let $n\in\Z_{0}^{+}$, 
$O=(\mc{M},U)\in\mr{St}_{n}$, where $\mc{M}=(M,g_{\mc{M}})$,
$\mc{F}=\{E_{r}\}_{r=1}^{n}$ be a frame field of $\mc{M}$, 
$Y,Z\in Op(M)$,
$t\in\R$ satisfying $\uptheta^{U}(t)Z\subseteq Y$,
$\chi\in\mf{P}_{F_{[M,U]}(Z)}$ 
thus for any $k\in\{1,\dots,n\}$ we define
\begin{equation*}
u_{k,\chi,\mc{F}}^{O,Y,Z}(t)
\coloneqq
\chi(\lr{E_{k}}{U}_{\mc{M}}\circ\uptheta^{U}(t)\up Z).
\end{equation*}
If in addition the topology $\upxi_{[M,U]}$ is stronger than 
or equal to the pointwise topology on $\mc{A}(M)$,
then for all $p\in Z$ set 
$u_{k,p,\mc{F}}^{O,Y,Z}\coloneqq u_{k,\delta_{p}^{Z},\mc{F}}^{O,Y,Z}$,
where $\updelta_{p}^{Z}$ is the Dirac distribution on $\mc{A}(Z)$ centered in $p$.
\end{definition}
\begin{remark}
\label{04050941}
If the topology $\upxi_{[M,U]}$ is stronger than the 
pointwise topology on $\mc{A}(M)$, then the topological dual 
of $\mc{A}(Z)$ w.r.t. the pointwise topology is a subset 
of the topological dual of $F_{[M,U]}(Z)$, 
hence $u_{k,p,\mc{F}}^{O,Y,Z}$ is well defined.
Moreover since for all $p\in Z$ and $t\in\R$ 
we have $\uptheta^{U}(t)p=\alpha_{p}^{U}(t)$
and $U(\alpha_{p}^{U}(t))=(\alpha_{p}^{U})^{\prime}(t)$,
we obtain if $\uptheta^{U}(t)Z\subseteq Y$
\begin{equation*}
u_{k,p,\mc{F}}^{O,Y,Z}(t)
=
g_{\mc{M}}(\alpha_{p}^{U}(t))
\left(E_{k}(\alpha_{p}^{U}(t)),(\alpha_{p}^{U})^{\prime}(t)\right),
\end{equation*}
where we recall that 
$\alpha_{p}^{U}$ is 
the maximal integral curve of $U$ 
such that $\alpha_{p}^{U}(0)=p$, 
while
$(\alpha_{p}^{U})^{\prime}(t)$
is the velocity vector of $\alpha_{p}^{U}$ at $t$.
\end{remark}
\begin{convention}
Let $n\in\Z_{0}^{+}$, $\mr{b}\in sQG_{n}$, 
$O=(\mc{M},U)\in\mr{St}_{n}$, where $\mc{M}=(M,g_{\mc{M}})$,
$Y,Z\in Op(M)$ and $t\in \mr{mor}_{[M,U]}(Y,Z)$.
Then let $\uptau_{\mr{b}(O)}(t)$ denote 
$\uptau_{\mr{b}(O)}(((Y,Z),t))$ 
whenever it will not cause confusion. 
\end{convention}
Next we will apply Cor. \ref{04081852} when $\mr{a}$ is 
the classical $n-$dimensional gravity 
and $\mr{b}$ is a $n-$dimensional strict quantum gravity,
but let us start with the following
\begin{definition}
\label{05101016}
Let $n\in\Z_{0}^{+}$, $\mr{b}\in sQG_{n}$, 
$O=(\mc{M},U)\in\mr{St}_{n}$, where $\mc{M}=(M,g_{\mc{M}})$,
$\mc{F}=\{E_{r}\}_{r=1}^{n}$ be a frame field of $\mc{M}$ 
and $Y,Z\in Op(M)$.
We define $\{U_{r}\}_{r=1}^{n}$ to be
a quantum realization from $Y$ to $Z$ through $L$
of the velocity of the maximal integral curves of $U$ on $\mc{M}$
and relative to $\mc{F}$,
if $\{U_{r}\}_{r=1}^{n}\subset\mc{A}_{\mr{b}(O)}(Y)_{ob}$
and if there exist $Y^{\prime},Z^{\prime}\in Op(M)$ 
with the following properties.
$L: \mr{mor}_{[M,U]}(Y,Z)\to \mr{mor}_{[M,U]}(Y^{\prime},Z^{\prime})$
and for any $\uppsi\in\mf{P}_{\mc{A}_{\mr{b}(O)}(Z)}$ 
there exists $\chi\in\mf{P}_{F_{[M,U]}(Z^{\prime})}$ 
such that for all $t\in \mr{mor}_{[M,U]}(Y,Z)$ and all $k\in\{1,\dots,n\}$ 
we have
\begin{equation*}
\uppsi(\uptau_{\mr{b}(O)}(t)U_{k})
=
u_{k,\chi,\mc{F}}^{O,Y^{\prime},Z^{\prime}}(L(t)).
\end{equation*}
\end{definition}
We call the map $t\mapsto\uppsi(\uptau_{\mr{b}(O)}(t)U_{k})$ 
on $\mr{mor}_{[M,U]}(Y,Z)$ the $k-$evaluation of 
$\{U_{r}\}_{r=1}^{n}$ in $\uppsi$.
Now we can state the following
\begin{corollary}
[Quantum realization of the velocity of maximal integral curves of 
complete vector fields]
\label{11091638}
Let $n\in\Z_{0}^{+}$, $\mr{b}\in sQG_{n}$, 
$\mr{T}\in\mf{Cnt}(\mr{a}^{n},\mr{b})$
and 
$O=(\mc{M},U)\in\mr{St}_{n}$, where $\mc{M}=(M,g_{\mc{M}})$.
Thus for any $Y,Z\in Op(M)$ and any frame field 
$\mc{F}=\{E_{r}\}_{r=1}^{n}$  of $\mc{M}$
there exists a quantum realization  $\{U_{r}\}_{r=1}^{n}$ 
from $Y$ to $Z$ through $\mr{T}_{1}^{m}(O)$
of the velocity of the maximal integral curves of $U$ on $\mc{M}$
and relative to $\mc{F}$.
Namely there exist $Y^{\prime},Z^{\prime}\in Op(M)$ with the following properties.
For any $\uppsi\in\mf{P}_{\mc{A}_{\mr{b}(O)}(Z)}$ 
there exists $\chi\in\mf{P}_{F_{[M,U]}(Z^{\prime})}$ 
such that for all $t\in\R$ satisfying $\uptheta^{U}(t)Z\subseteq Y$ 
and all $k\in\{1,\dots,n\}$ we obtain
\begin{equation}
\label{11091638a}
\begin{aligned}
\uppsi(\uptau_{\mr{b}(O)}(t)U_{k})
&=
u_{k,\chi,\mc{F}}^{O,Y^{\prime},Z^{\prime}}(\ov{t}),
\\
\ov{t}
&=
\mr{T}_{1}^{m}(O)(t).
\end{aligned}
\end{equation}
Moreover we can select 
$Y^{\prime}=\mr{T}_{1}^{o}(O)Y$, $Z^{\prime}=\mr{T}_{1}^{o}(O)Z$,
$\chi=\mr{T}_{3}^{\dagger}(O)(Z)\uppsi$,
and $U_{k}=\mr{T}_{3}(O)(Y)A$ with $A=\lr{E_{k}}{U}_{\mc{M}}\up Y^{\prime}$
for any $k\in\{1,\dots,n\}$.
If in addition the topology $\upxi_{[M,U]}$ is stronger than 
or equal to the pointwise topology on $\mc{A}(M)$ 
and if there exists $p\in Z^{\prime}$ such that 
$\chi=\updelta_{p}^{Z^{\prime}}$,
then 
\begin{equation}
\label{11091638b}
\uppsi(\uptau_{\mr{b}(O)}(t)U_{k})
=
g_{\mc{M}}(\alpha_{p}^{U}(\ov{t}))
\left(
E_{k}(\alpha_{p}^{U}(\ov{t})),
(\alpha_{p}^{U})^{\prime}(\ov{t})
\right).
\end{equation}
\end{corollary}
\begin{proof}
According to Cor. \ref{04081852} 
$\mr{b}$ is a realization of $\mr{a}^{n}$
such that we can select 
$\mf{E}^{\mr{b}}=(\mf{P}^{\mr{b}},\mr{G}^{\mr{b}})$,
$\mf{E}^{\mr{a}^{n}}=(\mf{P}^{\mr{a}^{n}},\mr{G}^{\mr{a}^{n}})$,
$Y^{\prime},Z^{\prime},\chi$ as in the statement
and $A^{\prime}=\mr{T}_{3}(O)(Y)A$.
Thus \eqref{11091638a} follows since Def. \ref{04081745} applied to
the object $O$ and to $A=\lr{E_{k}}{U}_{\mc{M}}\up Y^{\prime}\in F_{[M,U]}(Y^{\prime})_{ob}$.
\eqref{11091638b} follows since \eqref{11091638a} and Rmk. \ref{04050941}.
\end{proof}
Let $M(x,f)=I\times_{f}S$ be the Robertson-Walker spacetime 
with sign $x$ and scale function $f$ see \cite[Def. 12.7]{one}.
Here concerning $M(x,f)$ we follow the general notation in \cite[p. 204]{one} in particular
$\pi$ (galactic time) and $\sigma$ are the projections defined on $M(x,f)$ 
onto $I$ and $S$ respectively.
Let $H:I\to\R$ be the Hubble parameter relative to $M(x,f)$
defined $H\coloneqq f^{\prime}/f$, where
$f^{\prime}$ is the derivative of the map $f$, 
see for instance \cite[eq (1.10)]{mkh}, see also \cite[p.347]{one}.
Let $H^{\prime}$ denote the derivative of the map $H$.
\par
Now let $\mr{T}$ be a connector from $\mr{a}^{4}$ to a $4-$dimensional strict quantum gravity species
and $\alpha$ be a geodesic in $M(x,f)$ such that it is complete 
one of the existing vector fields $V$ of $M(x,f)$ for which a possibly restriction of
$\alpha$ is an integral curve. Moreover let $\{E_{r}\}_{r=1}^{3}$ be a frame field on $S$,
$t=\pi\circ\alpha$ be the galactic time of $\alpha$
and $t_{0}$ be a suitable restriction of $t\circ\mr{T}_{1}^{m}(O)$ where $O=(M(x,f),V)$.
In Thm. \ref{04151343} (see also \eqref{04221653} in a more compact form)
we establish that $H\circ t_{0}$ is the quotient of the classical factor $c_{k}$ and 
the quantum factor $q_{k}$ equal to the $k-$evaluation of 
a quantum realization of the velocity of $\alpha$ 
relative to any frame field on $M(x,f)$ extending the lift of $\{E_{r}\}_{r=1}^{3}$ to $M(x,f)$. 
\par
Since \cite[Prp. 12.22(2)]{one} 
$H\circ t$ enters in the expression 
of the velocity $\beta^{\prime}$ of the projection $\beta=\sigma\circ\alpha$. 
Thus the strategy is to extract the component 
$\beta^{\prime}$ from the velocity $\alpha^{\prime}$
by using the lift to $M(x,f)$ of $\{E_{r}\}_{r=1}^{3}$ 
and then to apply Cor. \ref{11091638} to any frame field on $M(x,f)$ extending this lift 
and to the complete vector field $V$.
\begin{theorem}
[Quantum and classical factors of the Hubble parameter]
\label{04151343}
Let $\mr{b}\in sQG_{4}$, $\mr{T}\in\mf{Cnt}(\mr{a}^{4},\mr{b})$ and 
$\alpha$ be a geodesic on $M(x,f)$.
Assume that $0\in dom(\alpha)$ eventually by a reparametrization, 
that $\alpha$ is regular in $0$, that it is complete one of the vector fields $V$ on $M(x,f)$ 
for which the restriction of $\alpha$ to an open neighbourhood $K_{0}$ of $0$
is an integral curve of $V$. 
Thus for any $Y,Z\in Op(M(x,f))$
and any frame field $\{E_{r}\}_{r=1}^{3}$ of $S$, 
said $O=(M(x,f),V)$ there exists $\{V_{r}\}_{r=1}^{3}\subset\mc{A}_{\mr{b}(O)}(Y)_{ob}$
with the following properties.
For any $\uppsi\in\mf{P}_{\mc{A}_{\mr{b}(O)}(Z)}$, 
for all $s\in\R$ satisfying $\uptheta^{V}(s)Z\subseteq Y$
and for all $k\in\{1,2,3\}$ we obtain
\begin{equation}
\label{04151128a}
\begin{aligned}
\uppsi(\uptau_{\mr{b}(O)}(s)V_{k})
&=
\chi(\lr{\ov{E}_{k}}{V}_{M(x,f)}
\circ
\uptheta^{V}(\ov{s})\up Z^{\prime}),
\\
\ov{s}
&\coloneqq
\mr{T}_{1}^{m}(O)(s),
\end{aligned}
\end{equation}
where $\chi=\mr{T}_{3}^{\dagger}(O)(Z)\uppsi$,
$\ov{E}_{k}$ is the lift of $E_{k}$ to $M(x,f)$
and $Z^{\prime}=\mr{T}_{1}^{o}(O)Z$.
Moreover said $Y^{\prime}=\mr{T}_{1}^{o}(O)Y$ we can select 
$V_{k}=\mr{T}_{3}(O)(Y)(\lr{\ov{E}_{k}}{V}_{M(x,f)}\up Y^{\prime})$.
Set
\begin{equation*}
K^{Y,Z}
\coloneqq\{\lambda\in\R\,\vert\,
\uptheta^{V}(\lambda)Z\subseteq Y,\,\mr{T}_{1}^{m}(O)(\lambda)\in K_{0}\}.
\end{equation*}
If the topology $\upxi_{[M(x,f),V]}$ is stronger than 
or equal to the pointwise topology on $\mc{A}(M(x,f))$,
$\alpha(0)\in Z^{\prime}$ and $\chi=\updelta_{\alpha(0)}^{Z^{\prime}}$
then 
for all 
$s\in K^{Y,Z}$
we obtain
\begin{equation}
\label{04151128b}
H(t(\ov{s}))
=
\frac{-f^{2}(t(\ov{s}))\,g_{S}(\beta(\ov{s}))
\left(E_{k}(\beta(\ov{s})),\beta^{\prime\prime}(\ov{s})\right)}
{2\uppsi(\uptau_{\mr{b}(O)}(s)V_{k})}
\left(\frac{dt}{ds}(\ov{s})\right)^{-1},
\end{equation}
where $t=\pi\circ\alpha$ and $\beta=\sigma\circ\alpha$.
\end{theorem}
\begin{proof}
Let $\alpha$ be a geodesic on $M(x,f)$ such that $0\in dom(\alpha)$ eventually by a reparametrization, 
and let $t=\pi\circ\alpha$ and $\beta=\sigma\circ\alpha$.
Assume $\alpha$ be regular in $0$ then by an application of \cite[Ex. 3.12]{one} 
we can find local vector fields $N$ of $I$ and $J$ of $S$, 
and an open neighbourhood $K_{0}$ of $0$ in $dom(\alpha)$ such that 
$N$ is defined on an open neighbourhood of $t(0)$, 
$J$ is defined on an open neighbourhood $\beta(0)$ 
and $t\up K_{0}$ and $\beta\up K_{0}$ are integral curves of $N$ and $J$ respectively.
Now since \cite[Ex. 1.18]{one} we deduce that $N$ and $J$ admit extensions to vector fields 
on $I$ and $S$ respectively, let us denote such extensions with the same symbols,
thus we can take the lifts $\ov{N}$ and $\tilde{J}$ on $M(x,f)$ of $N$ and $J$ respectively.
It is easy to see that $\alpha\up K_{0}$ 
is an integral curve of $V=\ov{N}+\tilde{J}$, moreover for any $r\in K_{0}$
we have that $\alpha^{\prime}(r)=\tilde{t}^{\prime}(r)+\tilde{\beta}^{\prime}(r)$,
where $\tilde{t}^{\prime}(r)$ and $\tilde{\beta}^{\prime}(r)$ 
are the lifts on $M(x,f)$ of $t^{\prime}(r)$ and $\beta^{\prime}(r)$ respectively.
We recall that $g_{\mc{N}}$ and $\lr{\cdot}{\cdot}_{\mc{N}}$ denote the metric tensor of 
any semi-Riemannian manifold $\mc{N}$,
thus by the above equality, for any frame field $\{E_{k}\}_{k=1}^{3}$ of $S$
and by letting $\ov{E}_{k}$ be the lift of $E_{k}$ on $M(x,f)$
we obtain for all $r\in K_{0}$
\begin{equation}
\label{04151143a}
g_{M(x,f)}(\alpha(r))
\left(\ov{E}_{k}(\alpha(r)),\alpha^{\prime}(r)\right)
=
f^{2}(t(r))g_{S}(\beta(r))
\left(E_{k}(\beta(r)),\beta^{\prime}(r)\right).
\end{equation}
Now if $V$ is complete, then we can apply 
Cor. \ref{11091638} to $O=(M(x,f),V)$
and to any frame field on $M(x,f)$ extending $\{\ov{E}_{k}\}_{k=1}^{3}$,
thus \eqref{04151128a} follows since \eqref{11091638a}.
Next $\alpha\up K_{0}$ is an integral curve of $V$, 
thus since the uniqueness of the integral curves see for instance \cite[Cor. 1.50]{one}, 
we obtain by \eqref{04151143a} that for all $r\in K_{0}$ and $k\in\{1,2,3\}$ 
\begin{equation}
\label{04151143}
(\lr{\ov{E}_{k}}{V}_{M(x,f)}\circ\uptheta^{V}(r))
(\alpha(0))
=
f^{2}(t(r))g_{S}(\beta(r))
\left(E_{k}(\beta(r)),\beta^{\prime}(r)\right).
\end{equation}
Now we can use \cite[Prp. 12.22(2)]{one} to express $\beta^{\prime}$ 
in terms of the Hubble parameter to obtain by \eqref{04151143} 
\begin{equation}
\label{04151111}
(\lr{\ov{E}_{k}}{V}_{M(x,f)}\circ\uptheta^{V}(r))
(\alpha(0))
=
\frac{-f^{2}(t(r))g_{S}(\beta(r))
\left(E_{k}(\beta(r)),\beta^{\prime\prime}(r)\right)}
{2H(t(r))}
\left(\frac{dt}{ds}(r)\right)^{-1}.
\end{equation}
\eqref{04151128b} follows since \eqref{04151111} and \eqref{04151128a}. 
\end{proof}
\begin{definition}
\label{04171613}
Assume the hypothesis and notation of Thm. \ref{04151343}. 
For all $k\in\{1,2,3\}$ define 
\begin{equation*}
\begin{aligned}
\mf{q}_{k}:K^{Y,Z}\to\R
\quad 
s&\mapsto\uppsi(\uptau_{\mr{b}(O)}(s)V_{k}),
\\
e_{k}:K_{0}\to\R
\quad
r&\mapsto
-\frac{1}{2}
f^{2}(t(r))\,g_{S}(\beta(r))
\left(E_{k}(\beta(r)),\beta^{\prime\prime}(r)\right)
\left(\frac{dt}{ds}(r)\right)^{-1};
\end{aligned}
\end{equation*}
\begin{equation*}
\begin{aligned}
\mf{c}_{k}&\coloneqq e_{k}\circ\mr{T}_{1}^{m}(O)\up K^{Y,Z},
\\
\mf{t}_{0}&\coloneqq t\circ\mr{T}_{1}^{m}(O)\up K^{Y,Z},
\\
q_{k}&\coloneqq\mf{q}_{k}\up\mathring{K}^{Y,Z},
\\
c_{k}&\coloneqq\mf{c}_{k}\up\mathring{K}^{Y,Z},
\\
t_{0}&\coloneqq\mf{t}_{0}\up\mathring{K}^{Y,Z};
\end{aligned}
\end{equation*}
with $\mathring{K}^{Y,Z}$ the interior of $K^{Y,Z}$. 
Call $\mf{q}_{k}$ and $\mf{c}_{k}$ as well their restrictions to $\mathring{K}^{Y,Z}$
the quantum and classical factors of the Hubble parameter relative to 
$\mr{T}$, $\{E_{r}\}_{r=1}^{3}$, $Y,Z$ and $k$ respectively.
Let $q_{k}^{\prime}$, $c_{k}^{\prime}$ and $dt_{0}/ds$
be the derivatives of the maps 
$q_{k}$, $c_{k}$ and $t_{0}$ respectively.
Define 
\begin{equation*}
c_{k}^{\pm}
\coloneqq
\frac{1}{2}
\left(\frac{dt_{0}}{ds}\right)^{-1}
\left(
q_{k}^{\prime}
\pm
\left(
(q_{k}^{\prime})^{2}-4c_{k}^{\prime}q_{k}\frac{dt_{0}}{ds}
\right)^{1/2}
\right).
\end{equation*}
and set for all $s\in\mathring{K}^{Y,Z}$ 
\begin{equation*}
\mf{I}(\mr{T},O,k,s)
\coloneqq
]-\infty,c_{k}^{-}(s)[\cup]c_{k}^{+}(s)+\infty[.
\end{equation*}
\end{definition}
The above designations emerge since 
$\mf{q}_{k}$ enroles only the quantum system $\mr{b}(O)$, 
$\mf{c}_{k}$ engages only the classical system $\mr{a}^{4}(O)$
and since according to \eqref{04151128b} we have 
for any $k\in\{1,2,3\}$ 
\begin{equation}
\label{04221653}
H\circ\mf{t}_{0}=\frac{\mf{c}_{k}}{\mf{q}_{k}}.
\end{equation}
By using the above equality, in the next result we express
$(f^{\prime\prime}/f)\circ t_{0}$ 
and the conditions for the positivity
of $f^{\prime\prime}\circ t_{0}$ 
as functions of $q_{k}$, $c_{k}$ and their derivatives.
\begin{corollary}
[Acceleration in terms of the
quantum and classical factors of the Hubble parameter]
\label{04171614}
Assume the hypothesis and notation of Thm. \ref{04151343}. 
Thus for all $k\in\{1,2,3\}$ we have 
\begin{equation}
\label{04181346}
\frac{f^{\prime\prime}}{f}\circ t_{0}
=
q_{k}^{-2}
\left(
c_{k}^{2}
-
(c_{k}q_{k}^{\prime}-c_{k}^{\prime}q_{k})
\left(\frac{dt_{0}}{ds}\right)^{-1}
\right).
\end{equation}
In particular for any $s\in\mathring{K}^{Y,Z}$ 
if $c_{k}^{\pm}(s)\in\C-\R$ then $f^{\prime\prime}(t_{0}(s))>0$,
while if $c_{k}^{\pm}(s)\in\R$, then
\begin{equation}
\label{04171746b}
f^{\prime\prime}(t_{0}(s))>0
\Leftrightarrow
c_{k}(s)\in
\mf{I}(\mr{T},O,k,s).
\end{equation}
\end{corollary}
\begin{proof}
Let $k\in\{1,2,3\}$ thus $f^{\prime\prime}/f=H^{\prime}+H^{2}$ on $I$,
clearly $(H^{\prime}\circ t_{0})=(H\circ t_{0})^{\prime}(\frac{dt_{0}}{ds})^{-1}$.
Thus \eqref{04181346} and \eqref{04171746b} follow since \eqref{04221653}. 
\end{proof}
Let us discuss the results obtained in the present section.
We define the concept of realization as the equiformity principle
applied for identity context morphisms Def. \ref{04081745}.
Then the target species of a connector realizes the dynamics
of the source species as a result of its equiformity principle 
Cor. \ref{04081852}.
We define the collection of gravity species 
and subcollections of quantum gravity species 
modeled by the structure of $\mr{a}^{n}$ Def. \ref{10031749}.
Thus we apply Cor. \ref{04081852} to establish the existence of 
quantum realizations of the velocity of the maximal integral curves of complete vector fields on spacetimes
(Def. \ref{05101016})
provided the existence of a connector from $\mr{a}^{n}$
to an $n-$dimensional strict quantum gravity Cor. \ref{11091638}.
\par
Then we employ this result in the special case of a Robertson-Walker spacetime $M(x,f)$ 
with sign $k$ and scale function $f$.
The outcomes can be so summarized.
Let 
$\mr{T}$ be a connector from $\mr{a}^{4}$ 
to a $4-$dimensional strict quantum gravity species $\mr{b}$ 
and $\alpha$ be a geodesic on $M(x,f)$, then under the hypothesis of  
Thm. \ref{04151343} we establish what follows.
\begin{enumerate}
\item
There exists a quantum realization of the velocity of $\alpha$
Thm. \ref{04151343}\eqref{04151128a}.
\item
$H\circ\mf{t}_{0}=\mf{c}_{k}/\mf{q}_{k}$ Thm. \ref{04151343}\eqref{04151128b}
(see \eqref{04221653}).
Here $H$ is the Hubble parameter, $\mf{t}_{0}$ is the restriction on $K^{Y,Z}$
of the composition of the galactic time of $\alpha$ with $\mr{T}_{1}^{m}(O)$, 
while $\mf{q}_{k}$ and $\mf{c}_{k}$ are the quantum and classical factors of the Hubble parameter
Def. \ref{04171613}.
\item
$(f^{\prime\prime}/f)\circ t_{0}$ is a function of $q_{k}$, $c_{k}$ and their derivatives
Cor. \ref{04171614}\eqref{04181346},
where $q_{k}$, $c_{k}$ and $t_{0}$ are the restrictions of 
$\mf{q}_{k}$, $\mf{c}_{k}$ and $\mf{t}_{0}$ 
to the interior $\mathring{K}^{Y,Z}$ of $K^{Y,Z}$
respectively.
\item
For any $s\in\mathring{K}^{Y,Z}$ we have that 
$f^{\prime\prime}(t_{0}(s))>0$ is equivalent to $c_{k}(s)\in \mf{I}(\mr{T},O,k,s)$
constraining the quantum and classical factors   
of the Hubble parameter
Cor. \ref{04171614}\eqref{04171746b}.
\end{enumerate}
It is well-known that for the Robertson-Walker perfect fluid $(U,\rho,p)$
with energy density $\rho$, pressure $p$ and where $U=\partial_{t}$ see \cite[Thm. 12.11]{one},
the detected positivity of the acceleration $f^{\prime\prime}$ \cite{rie2,per}
is equivalent to the negative pressure $p<-\rho/3$ see for example \cite[Cor. 12.12]{one},
occurrence commonly ascribed to the dark energy \cite[p. 65]{mkh}.
\par
However if in Thm. \ref{04151343} we take $\alpha$ an integral curve of $U=\partial_{t}$ 
so we can choose $V=U$, well-done since $U$ is geodesic being $D_{U}U=\ze$ see \cite[p.346]{one}, 
then we have that for any $s\in\mathring{K}^{Y,Z}$, 
\emph
{$f^{\prime\prime}(t(\ov{s}))>0$ 
namely the positivity of the acceleration when evaluated over 
the galactic time $t(\ov{s})=t_{0}(s)$ of $\alpha(\ov{s})$,
as well the negative pressure $p(\alpha(\ov{s}))<-\rho(\alpha(\ov{s}))/3$,
are explained via the equivalence \eqref{04171746b}
as a consequence of the existence of 
a connector $\mr{T}$,
from $\mr{a}^{4}$ to a strict quantum gravity species, 
satisfying the constraints $c_{k}(s)\in \mf{I}(\mr{T},O,k,s)$,
rather than the existence of dark energy}.
\par
Diagrammatically we can summarize what said as follows.
If we assume the hypothesis of Thm. \ref{04151343}
and let $\mf{De}$ denote the dark energy hypothesis,
then for all $s\in\mathring{K}^{Y,Z}$ and $k\in\{1,2,3\}$ we have 
\begin{equation*}
\xymatrix{
\boxed{
\begin{cases}
\mr{T}&\in\mf{Cnt}(\mr{a}^{4},\mr{b})
\\
c_{k}(s)&\in\mf{I}(\mr{T},O,k,s)
\end{cases}
}
\ar@{=>}[rr]^{\eqref{04171746b}}
&&
\boxed{f^{\prime\prime}(t(\ov{s}))>0}
\\
\boxed{\mf{De}}
\ar@{=>}[rr]
&&
\boxed{p(\alpha(\ov{s}))<-\frac{\rho(\alpha(\ov{s}))}{3}}
\ar@{<=>}[u]
}
\end{equation*}
It is aim of a future work to find $\mr{T}\in\mf{Cnt}(\mr{a}^{4},\mr{b})$
such that $c_{k}(s)\in\mf{I}(\mr{T},O,k,s)$.
We conclude with the following observation.
If any of the following occurrences
\begin{enumerate}
\item
\eqref{04221653}
by replacing $H\circ\mf{t}_{0}$ with the constant map 
equal to the Hubble parameter at the present time 
$73.02\pm 1.79\, \rm{km}\, \rm{s}^{-1}\,\rm{Mpc}^{-1}$
as determined by Riess et al. in \cite{rie},
\item
$c_{k}(s)\in\mf{I}(\mr{T},O,k,s)$,
for any $s\in\mathring{K}^{Y,Z}$;
\end{enumerate}
would be experimentally confirmed, 
then the empirical existence of $\mr{T}$ as stated in Posit \ref{10031521}\eqref{10031521rps},
in particular that the equiformity principle of $\mr{T}$ is part of its empirical representation,
will be experimentally validated.


\bibliography{references}
\bibliographystyle{amsplain}

\end{document}